\documentclass[
twocolumn,
superscriptaddress,
amsmath,amssymb,
aps,
noeprint,
prx,
]{revtex4-2}

\usepackage[utf8]{inputenc}
\usepackage[T1]{fontenc}

\usepackage{bbm,mathtools}
\usepackage{amsthm,amssymb,amsmath}

\usepackage{graphicx}
\usepackage{xcolor}
\usepackage{dcolumn}
\usepackage{booktabs}
\usepackage{physics}
\usepackage{xcolor}
\usepackage{hyperref}
\hypersetup{
  colorlinks=true,
  hypertexnames=false,
  linktocpage,
  colorlinks=true, 
  urlcolor=magenta!90!black,    
  linkcolor=blue!60!black, 
  citecolor=black!60 
}

\usepackage{enumitem}
\usepackage[capitalize,compress]{cleveref}
\usepackage{tikz}

\usepackage{complexity}

\usepackage{pgfplots}
\usepackage{tikz}
\usepackage{pgffor}
\usepackage{algorithm}
\usepackage{algorithmicx}
\usepackage{algpseudocode}
\algrenewcommand\algorithmicrequire{\textbf{Input:}}
\algrenewcommand\algorithmicensure{\textbf{Output:}}

\newclass{\stoqma}{StoqMA}
\newclass{\classP}{P}
\newclass{\bqp}{BQP}
\newclass{\qcam}{QCAM}
\newclass{\postbqp}{postBQP}
\newclass{\posta}{postA}
\newclass{\postiqp}{postIQP}
\newclass{\classa}{A}
\newclass{\bpp}{BPP}
\newclass{\fbpp}{FBPP}
\newclass{\pp}{PP}
\newclass{\cocp}{coC_=P}
\newclass{\ph}{PH}
\newclass{\np}{NP}
\newclass{\conp}{coNP}
\newclass{\gapp}{GapP}
\newclass{\approxclass}{Apx}
\newclass{\gapclass}{Gap}
\newclass{\sharpP}{\#P}
\newclass{\ma}{MA}
\newclass{\am}{AM}
\newclass{\qma}{QMA}

\newclass{\hog}{HOG}
\newclass{\quath}{QUATH}
\newclass{\bog}{BOG}
\newclass{\xeb}{XEB}
\newclass{\xhog}{XHOG}
\newclass{\xquath}{XQUATH}
\newclass{\maxcut}{MAXCUT}
\newclass{\sat}{SAT}
\newclass{\maxtwosat}{MAX2SAT}
\newclass{\twosat}{2SAT}
\newclass{\threesat}{3SAT}
\newclass{\sharpsat}{\#SAT}
\newclass{\se}{Sign Easing}
\newclass{\classx}{X}

\newtheorem{theorem}{Theorem}

\newtheorem{conjecture}[theorem]{Conjecture}
\newtheorem{definition}[theorem]{Definition}
\newtheorem{lemma}[theorem]{Lemma}

\newtheorem{corollary}[theorem]{Corollary}

\DeclareMathOperator{\supp}{supp}

\newcommand{\mc}{\mathcal}
\newcommand{\mb}{\mathbb}

\newcommand{\mf}{\mathfrak}

\newcommand{\hdim}{{\mf D}}
\newcommand{\hdepth}{{\mf L}}

\newcommand{\Xz}{\mb X_z}
\newcommand{\Xx}{\mb X_x}

\newcommand{\ix}{{\sf I}}
\newcommand{\sx}{{\sf S}}
\newcommand{\px}{{\sf P}}
\newcommand{\xx}{{\sf X}}

\newcommand{\notx}{{\neg \xx}}

\newcommand{\cnot}{\mathrm{CNOT}}
\newcommand{\notc}{\mathrm{NOTC}}
\newcommand{\tof}{\mathrm{TOF}}
\newcommand{\tc}{^{\otimes 2}}
\newcommand{\nc}{^{\otimes n}}
\newcommand{\ketx}[1]{\ket{#1}_x}
\newcommand{\ketz}[1]{\ket{#1}_z}
\newcommand{\brax}[1]{\bra{#1}_x}
\newcommand{\braz}[1]{\bra{#1}_z}
\newcommand{\projx}[1]{\proj{#1}_x}

\newcommand{\id}{\mathbbm{1}}

\newcommand{\pqw}{\mathrm{pqw}}

\newcommand{\bin}{\{0,1\}}

\newcommand{\proj}[1]{\ket{#1}\bra{#1}}

\newcommand{\code}[1]{[\![#1]\!]}

\newcommand{\nn}{\nonumber}

\newcommand{\T}{\mathrm{T}}

\newcommand{\HarvardPhysics}{Department of Physics, Harvard University, Cambridge, Massachusetts 02138, USA}
\newcommand{\QuICS}{
Joint Center for Quantum Information and Computer Science, NIST/University of Maryland, College Park, Maryland 20742, USA}
\newcommand{\simons}{
Simons Insitute for the Theory of Computing, University of California at Berkeley, Berkeley, California 94720, USA}
\newcommand{\CUBoulder}{JILA and Department of Physics, University of Colorado, Boulder, Colorado 80309, USA}
\newcommand{\caltech}{California Institute of Technology, Pasadena, California 91125, USA}
\newcommand{\aws}{AWS Center for Quantum Computing, Pasadena, California 91125, USA}

\begin{document}

\title{Fault-tolerant compiling of classically hard IQP circuits on hypercubes}

\author{Dominik Hangleiter}
\thanks{DH, MK and DB contributed equally.}
\affiliation{\QuICS}
\affiliation{\simons}

\author{Marcin Kalinowski}
\thanks{DH, MK and DB contributed equally.}
\affiliation{\HarvardPhysics}

\author{Dolev Bluvstein}
\thanks{DH, MK and DB contributed equally.}
\affiliation{\HarvardPhysics}

\author{Madelyn Cain}
\affiliation{\HarvardPhysics}
\author{Nishad Maskara}
\affiliation{\HarvardPhysics}
\author{Xun Gao}
\affiliation{\HarvardPhysics}
\affiliation{\CUBoulder}
\author{Aleksander Kubica}
\affiliation{\aws}
\affiliation{\caltech}
\author{Mikhail D. Lukin}
\affiliation{\HarvardPhysics}
\author{Michael J. Gullans}
\affiliation{\QuICS}

 \date{\today}

\begin{abstract}
Realizing computationally complex quantum circuits in the presence of noise and imperfections is a challenging task. 
While fault-tolerant quantum computing provides a route to reducing noise, it requires a large overhead for generic algorithms.
Here, we develop and analyze a hardware-efficient, fault-tolerant approach to realizing complex sampling circuits. 
We co-design the circuits with the appropriate quantum error correcting codes for efficient implementation in a reconfigurable neutral atom array architecture, constituting what we call a \emph{fault-tolerant compilation} of the sampling algorithm. 
Specifically, we consider a family of $\code{2^D,D,2}$ quantum error detecting codes whose transversal and permutation gate set can realize arbitrary degree-$D$ instantaneous quantum polynomial (IQP) circuits. 
Using native operations of the code and the atom array hardware, we compile a fault-tolerant and fast-scrambling family of such IQP 
circuits in a hypercube geometry, realized recently in the experiments by Bluvstein \textit{et al.} [Nature 626, 7997 (2024)]. 
We develop a theory of second-moment properties of degree-$D$ IQP circuits for analyzing hardness and verification of random sampling by mapping to a statistical mechanics model.  
We provide strong evidence that sampling from these hypercube IQP circuits is classically hard to simulate even at relatively low depths.
We analyze the linear cross-entropy benchmark (XEB) in comparison to the average fidelity and, depending on the local noise rate, find two different asymptotic regimes.
To realize a fully scalable approach, we first show that Bell sampling from degree-$4$ IQP circuits is classically intractable and can be efficiently validated.
We further devise new families of $\code{O(d^D),D,d}$ color codes of increasing distance $d$, permitting exponential error suppression for transversal IQP sampling. 
Our results highlight fault-tolerant compiling as a powerful tool in co-designing algorithms with specific error-correcting codes and realistic hardware.

\end{abstract}

\maketitle

\let\oldaddcontentsline\addcontentsline
\renewcommand{\addcontentsline}[3]{}

\section{Introduction}
\label{sec:intro}

Quantum computers hold a promise to significantly outperform classical computers at various tasks. 
However, for many envisioned applications, very low error rates  below ${\sim}\,10^{-10}$ are required \cite{reiher_elucidating_2017,gidney_how_2021,clinton_hamiltonian_2021,litinski_how_2023,watson_quantum_2023}, in stark contrast to the state of the art 
experimental physical error rates of ${\sim}\,10^{-3}$. 
Quantum error correction (QEC) provides a potential solution to this challenge by encoding error-corrected ``logical'' qubits across many redundant physical qubits \cite{shor_fault-tolerant_1996,steane_error_1996,aharonov_fault-tolerant_1997}.
In principle, QEC can exponentially suppress the logical error rate by increasing the code distance $d$,  thereby promising a realistic route to low error rates required for large-scale algorithms. 
However, implementing QEC in practice is challenging. 
In addition to the large physical qubit overheads, QEC codes typically realize a discrete gate set using native operations \footnote{When restricting to transversal operations, it can be proven that the implemented gate set in fact has to be discrete \cite{Eastin-Knill.2009}. }. 
Although universal computation can be realized through various techniques such as magic state distillation \cite{knill_fault-tolerant_2004,knill_fault-tolerant_2004-1,bravyi_universal_2005} and code switching \cite{paetznick_universal_2013,anderson_fault-tolerant_2014}, these are generally very resource intensive. 
Thus devising hardware-efficient and fault-tolerant realizations of quantum algorithms is a non-trivial task, requiring co-design of error correcting code and physical implementation with the algorithm. 
We call this task \emph{fault-tolerant compiling}.

Realizing computationally hard sampling algorithms  
is an interesting goal for logical qubit processors for a number of reasons.
First, using QEC the logical noise rate can be exponentially suppressed and high circuit fidelities can be maintained while system size is increased. 
In contrast, the non-corrected signal decays exponentially with increasing circuit depth and system size \cite{arute_quantum_2019,Morvan.2023}. 
Second, computationally complex sampling circuits can be implemented using significantly fewer resources compared to universal quantum computation. 
In particular, a quantum computation does not have to be universal in order to be classically hard to simulate \cite{bremner_classical_2010,aaronson_computational_2013}. 
This opens up intriguing opportunities for co-designing an algorithmic implementation with a QEC code.
Specifically, complex quantum 
sampling 
circuits can be based on a variety of non-universal gate sets~\cite{bremner_classical_2010,aaronson_computational_2013,hangleiter_computational_2023}, which allows us to restrict the circuits to native gate sets of QEC codes.
Moreover, they profit from even a limited amount of error detection and correction to improve the sample quality. 
These features dramatically reduce the overhead in fault-tolerant compiling. 
Finally, such circuits which are also fast scrambling can be used to benchmark the performance of a quantum processor \cite{boixo_characterizing_2018,arute_quantum_2019,choi_preparing_2023,Ware.2023,Morvan.2023}.

In this work, we propose a viable path to systematically improve experimental implementations of complex sampling circuits 
with encoded qubits on the near-term quantum processors.  
Our approach is based on $D$-dimensional color codes with parameters $\code{2^D,D,2}$ (distance-$2$ codes with $D$ logical qubits encoded in $2^D$ physical qubits). 
For $D\ge 3$, these codes support transversally implemented $D$-qubit non-Clifford gates, as well as CNOT and SWAP gates realized by qubit permutations.  
We show that along with transversal CNOT gates these native operations  allow us to realize arbitrary 
degree-$D$ instantaneous quantum polynomial (IQP) circuits \cite{bremner_classical_2010,shepherd_temporally_2009};
sampling from such circuits is believed to be a hard task for classical computers based on complexity-theoretic arguments \cite{Bremner.2016,Bremner.2017}. 
The distance-2 codes allow us to detect errors directly from the classical samples, yielding an improvement over bare circuits even without intermediate measurements. 
Nonetheless, an approach based on error detection is not scalable. 
To achieve scalability, we devise new families of $\code{O(d^D), D,d}$ color codes based on the $\code{2^D,D,2}$ family, allowing for repeated rounds of error correction throughout the circuit execution.
For these code families and transversal IQP sampling, there is a noise threshold~\cite{gottesman_faulttolerant_2014} below which the noisy output distribution converges exponentially fast towards the ideal output distribution as the code size is increased.

In order to maximize hardware efficiency, we focus on the capabilities 
of the recently realized logical quantum processor \cite{bluvstein_logical_2024} that is based on reconfigurable arrays of neutral atoms in optical tweezers \cite{beugnon_two-dimensional_2007,schlosser_scalable_2011,evered_high-fidelity_2023}. 
In this setting, many physical quantum operations can be naturally parallelized, including single-qubit gates on  blocks of  physical qubits and transversal entangling gates between large blocks.
We design a hardware-efficient family of degree-$D$ IQP circuits with a connectivity graph given by a $\hdim$-dimensional hypercube which we call \emph{hypercube IQP (hIQP)} circuits. We show that this family rapidly converges to uniform IQP circuits and can therefore be thought of as a fault-tolerant compilation of the uniform IQP family.
In an hIQP circuit, transversal degree-$D$ and permutation CNOT circuits are performed in each block and the blocks are coupled by transversal CNOT gates. 
We analyze the conditions under which we expect random degree-$D$ hIQP circuits to be sufficiently scrambling for quantum advantage as well as benchmarking applications. 
Finally, we address the issue of \emph{efficiently} verifying quantum advantage in this model, by showing that degree-$D$ IQP sampling can be efficiently validated by measuring two copies of a logical degree-$(D+1)$ circuit in the Bell basis.

\begin{figure*}
\centering
    \includegraphics[width=\linewidth]{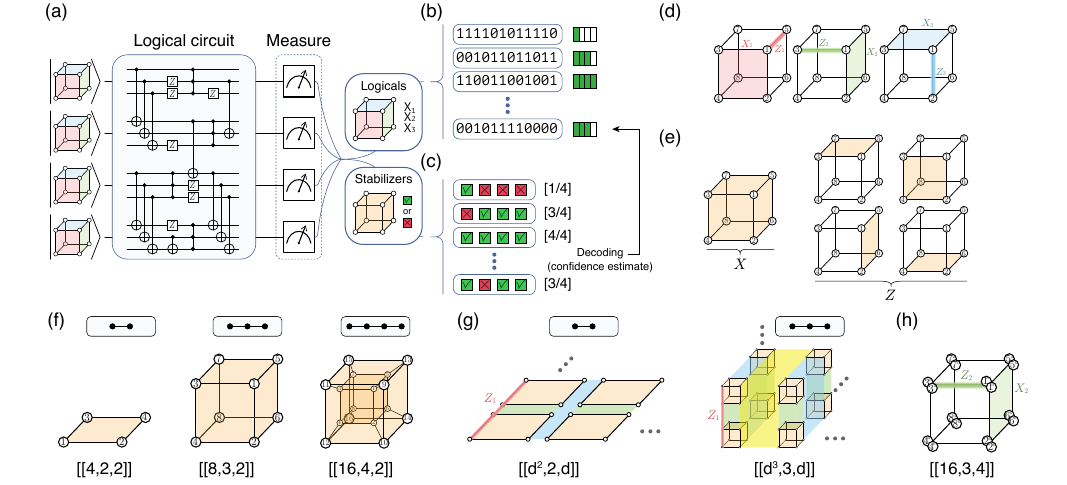}
    \caption{\textbf{Sampling from IQP circuits encoded in hypercube codes.} (a) The $\code{2^D,D,2}$ code family features a transversal gate set that enables classically hard computation. 
    The logical in-block operations are performed with high-fidelity single-qubit rotations and permutations on physical qubits, and inter-block operations are realized via transversal CNOT gates. 
    (b,c) The measured physical qubit bitstring provides a logical bitstring as well as stabilizer check outcomes. Classical post-processing (decoding) can be used to correct errors or post-select (error detection) based on the confidence for each measured logical bitstring; both procedures improve sampling performance. 
    (d) The $\code{8,3,2}$ code is an instance of the family ($D\,{=}\,3$, $d\,{=}\,2$). It can be visualized as a cube with logical $X$ and $Z$ operators corresponding to the three orientations of faces and edges, respectively. 
    (e) The single $X$-basis stabilizer is a product of Pauli-$X$ on all eight qubits, and the $Z$-basis stabilizers consists of the Pauli-$Z$ products on four independent faces of the cube. (f) A $D$-dimensional code in the $\code{O(d^D),D,d}$ family supports a transversal C$^k$Z operation within a block~\cite{Kubica.2015}, where $ k \in \{1, \ldots, D\,{-}\,1\}$. For distance $d\,{=}\,2$, a $D$-dimensional code is represented by a hypercube of dimension $D$. 
    The $X$ stabilizer is the product of all qubits, and the $Z$ stabilizers are formed from all independent faces. 
    (g) A distance-$d$ code can be constructed by placing distance-2 codes on a $D$-dimensional lattice (for even $d$).
    This construction preserves the single-qubit realization of transversal in-block gates, and we argue that only codes with even distance can have this property. (h) An alternative approach is to decorate the four vertices with small repetition codes. This results in smaller codes with a higher rate of encoding, but it comes at the expense of introducing multi-qubit gates for transversal in-block operations.}
    \label{fig:fig1}
\end{figure*}

\subsection{Summary of Results}
\label{ssec:implications}

Our approach to fault-tolerant implementations of computationally hard algorithms is based on fault-tolerant compiling of IQP circuit sampling using high-dimensional color codes. 
Our goal is to devise a scalable path towards the quantum advantage regime 
starting from experiments that are possible today, as demonstrated in Ref.~\cite{bluvstein_logical_2024}.
Thus, we address the finite-size and asymptotic regimes of three distinct properties of our proposed algorithm---fault tolerance, complexity, and verification. 
Concretely, we start from co-designing the quantum code and logical circuits with the reconfigurable atom array hardware as in Refs.~\cite{Bluvstein.2022,bluvstein_logical_2024}, 
giving rise to the hIQP circuits implemented natively on hypercube color codes, see \cref{fig:fig1}(a-e). 
In what follows we summarize our key results. 

First, in \cref{sec:architecture}, we analyze efficiently simulatable instances of hIQP circuits in which error detection is performed at the end of the circuit (\cref{fig:fig1}(b)) in terms of the achievable error reduction. 
To understand the behaviour of the logical errors, we introduce a notion of the average gate fidelity for logical quantum gates and study the performance of the transversal and permutation gates of the $\code{8,3,2}$ code using classical simulation. 
We analyze the behaviour of the average logical fidelity as a function of the amount of error detection and explain its power-law behaviour using a simple noise model. 

Second, in \cref{sec:hiqp}, we study the properties of hIQP circuits that are relevant to quantum advantage demonstrations: complexity and verification. 
We give a complexity-theoretic argument that sampling from hIQP circuits is classically intractable (\cref{ssec:complexity}). 
Our argument is based on an analytical and numerical study of the scrambling properties of the hIQP circuits. 
We find that already after two rounds of gates on all hypercube edges the output states are close to maximally scrambled. 
We also show that the runtime of existing classical simulation methods, in particular the recently developed near-Clifford simulator for degree-$D$ circuits \cite{maslov_fast_2024}, roughly scales as $\Omega(2^n)$. 
The increase in complexity compared to the $O(2^{n/2})$ scaling from Ref.~\cite{bluvstein_logical_2024} stems from additional gate layers; while the difference to the $O(2^{n/3})$ scaling from Ref.~\cite{maslov_fast_2024} is due to additional random in-block permutation CNOT gates.

We also analytically and numerically study the behaviour of the linear cross-entropy benchmark (XEB) which has been used to benchmark global circuit behaviour~\cite{arute_quantum_2019,zhu_quantum_2022,Morvan.2023,mark_benchmarking_2023} (\cref{ssec:xeb and noise}). 
We show that---similar to random quantum circuits \cite{gao_limitations_2024,Ware.2023,Morvan.2023}---the XEB in sparse degree-$D$ IQP circuits (a toy model of hIQP circuits) undergoes a transition as a function of the local noise rate between being a good proxy for the global state fidelity and being much larger than the fidelity. 
At the same time, we argue that the relation between fidelity and XEB is much tighter in the case of IQP circuits compared to Haar-random circuits, making it a good measure for quantum advantage. 
In particular, we show that the transition can be arbitrarily shifted by adding fixed gates to an otherwise random circuit. 
For small instances, we numerically confirm a tight relation between the average fidelity and average linear XEB at low noise rates 
for both physical and logical noise. 
In this context, we discuss intriguing aspects of how to think about logical fidelities in encoded sampling circuits (\cref{ssec:noisy encoded xeb}). 

To analyze the scrambling properties as well as the linear XEB of the hIQP circuits, we develop a broadly applicable theory for second-moment quantities of degree-$D$ IQP circuits (\cref{ssec:stat mech}). 
This theory is based on a mapping of second-moment quantities to a classical statistical-mechanics model, analogous to a similar model for Haar-random circuits~\cite{zhou_emergent_2019,hunter-jones_unitary_2019,barak_spoofing_2021,dalzell_random_2024,Ware.2023}. 
It can handle CNOT entangling gates as well as noise on one or two copies of the circuit. 
As a first application of the statistical-mechanics model we prove that sparse degree-$2$ IQP circuits anticoncentrate in logarithmic depth. 
The understanding of the dynamics of the model gleaned from the proof of this result then helps us analyze the asymptotic properties of the hIQP circuits as well as the dynamics of noisy IQP circuits. 

To go beyond the inefficient verification of the hIQP circuits via XEB, we propose a natural two-copy hIQP protocol which involves a measurement in the transversal Bell basis at the end of the circuit \cite{hangleiter_bell_2024} (\cref{sec:bell sampling}). 
We show that sampling from the corresponding distribution is classically intractable for degree-$4$ and higher circuits. 
Interestingly, the simulation cost of these circuits halves to roughly $\Omega(2^{n/2})$ while the number of logical qubits needs to be doubled. 
At the same time, the samples can be efficiently classically validated using properties of the two-copy measurement \cite{hangleiter_bell_2024}. 
hIQP Bell sampling thus constitutes a near-term achievable means of performing efficiently validated quantum advantage in a fault-tolerant setting. 

Finally, in \cref{sec:scaling codes}, we consider the question of scalable fault tolerance of hIQP circuits beyond implementations in the $\code{8,3,2}$ color codes. 
To this end, we devise two code families with the same transversal gate set. 
The first one is a family of 3D color codes based on the $\code{8,3,2}$ code and has parameters $\code{O(d^3),3,d}$, see \cref{fig:fig1}(g). 
The other family is a 3D toric/color code and has parameters $\code{O(d^3),3,d}$.
Both constructions result in families of topological quantum codes and have a fault-tolerance threshold under a local stochastic noise model~\cite{gottesman_faulttolerant_2014}, proving the scalability of this approach.
We also quantitatively compare the performance of error detection in the hIQP circuits using the $\code{8,3,2}$ code with the same circuits implemented in slightly larger codes that support error correction---the $\code{15,1,3}$ code and a $\code{16,3,4}$ code (depicted in \cref{fig:fig1}(h)). 
Interestingly, we find that the $\code{8,3,2}$ code with error detection outperforms the other small codes.

Overall, our work  studies computational complexity in the context of early fault-tolerant quantum computing architectures, using the $\code{8,3,2}$ code as its model system.
Our results highlight that the design of fault-tolerant quantum algorithms requires careful analysis of the algorithm as well as the fault-tolerance properties of the employed error correcting codes for the particular circuits. 
It analyzes differences to non-fault-tolerant architectures, where hardware-efficiency obeys quite different constraints \cite{arute_quantum_2019,zhu_quantum_2022}, and noise behaves in very different ways.
The study of noise becomes crucial to assessing the prospects of quantum advantage in a twofold way---in terms of success on the chosen benchmark (the linear XEB in our case) but also in terms of the error reduction capabilities of a chosen code and error correction/detection strategy. 
Most pointedly, this is highlighted by our comparisons of the performance of the 3D hypercube code with error detection with similar small codes where error correction is possible. 
Our study shows that, just like for unencoded quantum computations, random quantum circuits remain an important  setting for understanding the capabilities of a given architecture in terms of its noise resilience and potential for quantum speedups.

\subsection{Relation to prior work}
\label{ssec:prior work}

Our work builds on prior work from a range of different fields. We now briefly compare our results to the most related prior works. 

Our results on the complexity of noiseless IQP circuits build on prior work by \textcite{bremner_average-case_2016,Bremner.2017}, who showed quantum advantage for IQP circuits. 
In particular, our result on anticoncentration of sparse IQP circuits reproduces the scaling and improves its constant factors compared to the anticoncentration result of Ref.~\cite{Bremner.2017}. 
Our hardness result of hIQP sampling that follows from this follows a similar line of thought compared to the hardness of sparse IQP sampling shown by Ref.~\cite{Bremner.2017}.

Our statistical-mechanics mapping for IQP circuits follows the ideas of Refs.~
\cite{zhou_emergent_2019,hunter-jones_unitary_2019,dalzell_random_2024}, developing stat-mech mappings for circuits with Haar-random single-qubit gates. 
The IQP-circuits mapping differs from this, however. 
We use the mapping for our anticoncentration results, as well as to for the first time show the existence of a phase transition in the noisy XEB for IQP circuits, a phenomenon which has also been observed in Haar-random circuits \cite{dalzell_random_2024,deshpande_tight_2022}. 

Refs.~\cite{mezher_fault-tolerant_2020-2,paletta_robust_2023} show a constant-depth fault-tolerant quantum advantage using IQP circuits with a single round of intermediate measurements and feed-forward. 
Compared to these results, our fault-tolerant sampling architecture is codesigned for a different family of IQP circuits using on high-degree C$^k$Z gates for efficient implementation in the reconfigurable-atom architecture.
While those works focus on removing as much intermediate measurement and feed-forward as possible and give asymptotic results, our work focuses on the finite-scale performance of a scalable architecture. 
Compared to Ref.~\cite{Wang.2023} which demonstrates one-bit addition using $\code{8,3,2}$ codes on the other hand, our work gives a scalable algorithmic advantage for a simple task and analyzes the performance at scale. 

Finally, our results on noisy encoded computations are the first to define a composable measure of logical  average-gate fidelity using the notion of $r$-filters. 
Previously, logical state fidelities have been reported using logical Pauli measurements or by projecting a reconstructed physical density matrix onto the logical subspace~\cite{andersen_repeated_2020,gupta_encoding_2024}, and Ref.~\cite{postler_demonstration_2022} reported a logical process matrix. 
In contrast to these measures of logical state and process fidelity, our $r$-filtered average gate fidelity takes into account the fact that the space of potential errors decomposes into different syndrome sectores which can accumulate if uncorrected. 


\subsection{Implications}

Our analysis indicates that sampling from the output distribution of low-depth hIQP circuits 
is a viable path towards a fault-tolerant quantum advantage over classical computation. 
hIQP sampling is hardware-efficient on the recently realized logical processor using reconfigurable atom arrays \cite{bluvstein_logical_2024}. 
For intermediate-scale implementations with error detection and/or correction at the end of the computation, we find that the $\code{8,3,2}$ 3D color code has not only the highest rate but also the best fault-tolerance properties when compared to similar small codes, see \cref{fig:code_comparison}. 
In the experiment of \textcite{bluvstein_logical_2024}, hIQP sampling without in-block permutation CNOT gates was performed on 48 logical qubits encoded in $2^4$ blocks of the $\code{8,3,2}$ code. 
These specific circuits exhibit asymptotic quantum advantage but can be simulated in time $\sim 2^{n/3}$ \cite{maslov_fast_2024}.  This classical scaling is considerably  more favorable than what we expect for hIQP circuits studied in this work that include in-block CNOT gates. 
For these circuits, which can be directly implemented using the techniques demonstrated in Ref.~\cite{bluvstein_logical_2024}, we expect the best classical simulation algorithms to run in time at least roughly $\sim 2^n$.

Our results showing the existence of a transition in the relation between XEB and fidelity imply that experiments should require a careful analysis of the noise regime in order to use XEB as a reliable benchmark. 
Our analysis suggest that the relation between XEB and fidelity may be tighter for IQP circuits compared to random circuits \cite{gao_limitations_2024,Ware.2023,Morvan.2023}, but exactly how so remains an open question. 
The existence of a transition highlights the need for fault tolerance in experiments benchmarked via XEB since the local error rate needs to be suppressed to $\sim 1/n$ for the XEB to be a good estimator of fidelity. 
For the error rates and circuit parameters of Ref.~\cite{bluvstein_logical_2024}, we find that the experiments are in the ``healthy'' noise regime in which the average logical XEB score is a good measure of an appropriately defined logical fidelity.  

Our results also suggest a path to scale up hIQP sampling experiments to achieve fault-tolerant quantum advantage. 
Two natural improvements can be achieved using the $\code{2^D,D,2}$ hypercube color codes: 
First, one can perform quantum advantage experiments that can be efficiently validated from the classical samples by using transversal Bell measurements between two copies of similar hIQP circuits based on the $\code{16,4,2}$ code. Note that using the $\code{8,3,2}$ code, a classically simulatable Bell sampling experiment has already been performed for 24 logical qubits by \textcite[][Fig.~6]{bluvstein_logical_2024}.
Second, one can improve the fault-tolerance properties by performing intermediate measurements to detect more errors and improve the final state fidelity. 
Going beyond error detection will require scaling up the code sizes.  
We construct some candidate families of 3D color codes  with an error correction threshold that support the classically hard hIQP circuits.
If the experimental error rate is below the threshold, the quantum output distribution converges exponentially to the target distribution as the code size is increased.
However, the rate of these codes decreases with the code distance. 
Finding high-performing quantum codes similar to~\cite{bravyi_high-threshold_2024,xu_constant-overhead_2023}, that also support native non-Clifford gates remains an important open question to further scale up quantum advantage demonstrations of the type proposed here.  For example, one interesting possibility we leave for future work is to use a $\code{64,15,4}$  CSS code where a transversal $T$-gate implements a logical circuit with 15 CCZ gates \cite{Rengaswamy20}.

\subsection{Outlook}
\label{sec:outlook}

There are a range of interesting directions that arise from the new experimental capabilities of the reconfigurable atom platform in general, and our study of its transversal degree-$D$ IQP circuits in particular.
An immediate technical open questions raised by our study of noise in degree-$D$ IQP circuits  is what type of transition the XEB undergoes as a function of the noise strength---whether it is a sharp phase transition as for Haar random circuits, or a different type of transition.
While we have found that the same two regimes do exist for noisy IQP circuits, we have also presented evidence that the XEB behaves more benignly for random IQP circuits compared to Haar-random circuits. 
Furthermore, it is an intriguing question how the transition behaves for logical XEB and fidelity in the presence of physical noise in an encoded circuit. 

An important question raised by the advent of encoded quantum computations in experiments is how to best make use of limited degree of error detection and correction. 
In this work, we have focused on a minimal amount of error detection at the very end of the computation, allowing us to eliminate certain errors. 
 Next steps include performing multiple rounds of error detection throughout the circuit execution via stabilizer measurements without feed-forward. 
In addition, it is important to explore  the most efficient ways of using a limited number of measure-and-correct rounds. 
The availability of the additional information gleaned from mid-circuit stabilizer measurements brings with it the question of how to best make use of that data. 
While one can perform postselection to improve the overall sample quality, in addition to that, the stabilizer data can yield diagnostic data about circuit performance and even error mechanisms \cite{fujiwara_instantaneous_2014,fowler_scalable_2014,huo_learning_2017}. 
In particular, it has been shown that stabilizer data can be used to learn correctable Pauli errors \cite{wagner_pauli_2022,iyer_efficient_2022,wagner_learning_2023} since the stabilizer measurement lets one measure the Fourier transform of the Pauli noise distribution \cite{flammia_efficient_2020}.
Stabilizer data might also be used to infer the average fidelity of transversal gate layers preceding a stabilizer measurement, analogous to cycle benchmarking which includes noise tailoring \cite{wallman_noise_2016,erhard_characterizing_2019}. 
This is particularly intriguing in the context of verifying quantum advantage demonstrations, since it may allow us to efficiently estimate the average fidelity of the output state, bypassing the need to resort to the inefficient XEB, or verification schemes that come with a significant resource overhead~\cite{ringbauer_verifiable_2024,hangleiter_bell_2024}. 
Random quantum circuits constitute a well-understood model in which the relation between different benchmarks can be studied. 
As our work shows, they also  provide an ideal playground for questions regarding the nature of early fault tolerance in general.

Beyond sampling circuits, our work can be understood as a case study of circuit compilation in co-design with code and hardware constraints. 
Optimized towards fast scrambling and simulation complexity, the family of hIQP circuits is highly resource efficient in terms of the chosen code, since it only uses transversal and permutation operations, and in terms of the reconfigurable atom hardware, which allows for parallel long-range gates. 
One can ask how to compile algorithms with different use cases to the same or similar code and hardware requirements. 
At the same time, it is interesting to explore which algorithms or sub-routines are most suited towards mostly transversal implementations in high-dimensional color codes with non-Clifford transversal gate sets. 
For example, it might be possible to encode arithmetic operations using only phase states \cite{Wang.2023}. 
From a different perspective, we can think of the transversally prepared states as a resource for more complex algorithms, since these states can be prepared with comparably low errors. 
In the near term, one can also envision logical-physical quantum circuits. 
In this scenario high-fidelity transversally prepared states could supply resources that are expensive for physical circuit implementations such as non-local magic and entanglement. 
These resources could then be exploited by a low-depth universal, physical circuit.  
Although not a viable path in the long-run, such an approach might allow novel applications on currently available hardware.
We hope that the present work will inspire the community to address these and other interesting questions regarding the physics and computational properties of early fault-tolerant quantum computation.

\subsection{Overview}

\cref{sec:architecture} outlines the fault-tolerant logical sampling architecture based on $\code{2^D,D,2}$ codes and degree-$D$ hIQP circuits with error detection. 
\cref{ssec:complexity} describes the hardness properties of our circuits, \cref{ssec:xeb and noise} the behaviour of the XEB under noise, 
\cref{ssec:noisy encoded xeb} the relation between the logical XEB and fidelity in encoded circuits,
and \cref{ssec:stat mech} our statistical mechanics mapping.
\cref{sec:bell sampling} addresses efficient verification by introducing and showing hardness of degree-$D$ Bell sampling.
\cref{sec:scaling codes} compares the performance of different small codes that allow limited error correction to the $\code{8,3,2}$ code and describes the new family of $\code{O(d^D),D,d}$ color codes that allow scalable transversal IQP sampling in the presence of noise.
\cref{sec:outlook} concludes with an outlook. 
\cref{app:prelim,app:noisy encoded,app:powerlaw,app:statmech,app:second moment behaviour,app:xebproof,app:bell sampling,app:codes} contain various technical details and proofs.

\section{Logical Sampling Architecture}
\label{sec:architecture}

In this section, we detail the concrete architecture for IQP circuit sampling using logical circuits that employ only transversal and permutation gates. 
The architecture is co-designed with the capabilities of a neutral-atom logical quantum processor~\cite{bluvstein_logical_2024,beugnon_two-dimensional_2007,schlosser_scalable_2011,Bluvstein.2022} in order to generate highly scrambled quantum states that are classically intractable to simulate and at the same time  minimize errors incurred during the computation. 
To this end, we exploit parallel transversal control of logical qubits by moving atoms in arrays of tweezers, which can naturally realize variable grid-like entangling patterns~\cite{bluvstein_logical_2024}.
These patterns will be reflected in a nested-hypercube structure of the physical qubit quantum circuits we study.

As described in the introduction, our proposed architecture uses a family of distance-$2$ error detecting codes, which have a family of non-Clifford transversal gates. This circumvents the need for magic state distillation, since the transversal gates efficiently realize classically intractable operations.
The computational task we consider is sampling from the distribution of \textit{logical qubit} states. The redundancy afforded by the error-detecting code is used to improve the quality of the sampled distribution, thereby yielding an
advantage compared to unencoded computations. 
Moreover, several complex many-qubit operations can be executed much more cheaply than in a comparable bare physical circuit.  The associated cost of the relatively complex encoding step is mitigated by the fact that this state preparation process can be made fault-tolerant across the qubit array in a scalable manner through parallelized postselection and re-preparation.

Before proceeding, let us clarify  our definition of fault-tolerance that follows standard approaches \cite{gottesman_surviving_2024}. 
For a physical quantum circuit, we define locations in the circuit to be the positions of circuit elements, i.e., gates, state-preparations and measurements. 
Each of those elements affects $k$ physical qubits which in our case will be $k=1$ for state preparation, measurements and single-qubit gates, and $k=2$ for entangling gates. 
We allow each circuit location to be faulty, which means that after its application an arbitrary $k$-local quantum channel acts on the qubits or classical bits directly affected by the circuit element. 
For error detection to be fault-tolerant to distance $d$, it means that any fault path---i.e., collection of error locations---of weight $d-1$ in the circuit is detectable by the syndrome data. For error correction to be fault-tolerant to distance $d$, it means that any fault path of weight-$\lfloor(d-1)/2\rfloor$ is correctable.  
Note that if fault-path involves a $k$-body circuit element then a weight-$1$ fault at this circuit location allows for errors on all of the~$k$ qubits the circuit element acts on.
In this work, unless otherwise specified, we consider fault-tolerance with respect to $k$-qubit Pauli error channels at all error locations, and just the $X$-type syndromes extracted at the end of the circuit. 
The latter corresponds to the experimental setting of Ref.~\cite{bluvstein_logical_2024}.
Where specified, we also consider perfect syndrome extraction or Bell measurements at the end of the circuit.

\begin{figure*}
    \centering
    \includegraphics[width=\textwidth]{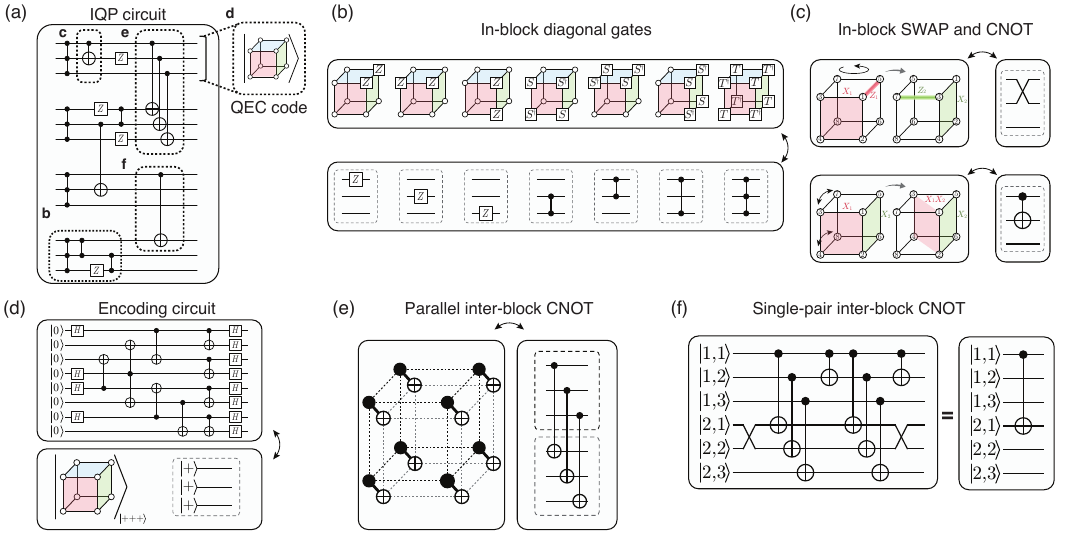}
    \caption{\textbf{IQP circuits as fault-tolerant operations on the  $\code{8,3,2}$ code.} (a) Arbitrary degree-3 IQP circuits can be decomposed into the transversal gate-set of the $\code{8,3,2}$ code and its extensions. The logical qubits are grouped into sets of three and encoded into code blocks consisting of eight physical qubits. (b) The diagonal gates between logical qubits within a single block consist of all $Z$, CZ, and CCZ gates. These operations and are implemented with parallel single-qubit rotations on the physical qubits. 
    (c) Additional transversal in-block operations include the CNOT and SWAP gates, which are realized by simple re-labeling of the qubits. In hardware, re-labeling can be realized by exchanging the qubits' positions.
    (d) State preparation circuit for preparing a symmetric logical state $\ket{{+}{+}{+}}$. In this preparation circuit and its generalizations, the method can be made fully fault-tolerant by utilizing flag protocols or Steane error detection~\cite{Wang.2023} at each stage of the protocol in a recursive manner. (e) Two encoded blocks can be entangled via a transversal CNOT, which is guaranteed by the CSS nature of the code. (f) Two layers of transversal out-block CNOTs interspersed with in-block CNOTs can realize a CNOT between individual logical qubits from different blocks, which enables transversal realization of arbitrary IQP circuits.}
    \label{fig:832code}
\end{figure*}

\subsection{A family of distance-$2$ codes and their native operations}
\label{ssec:gadgets}

Our architecture is based on the $D$-hypercubic code with code parameters $\code{2^D, D, 2}$ for $D \geq 2$ \cite[Example~3]{Vasmer2022}.
The hypercubic code is a color code defined on a $D$-dimensional hypercube, that is, its stabilizers are given by products of $X$ and $Z$ operators on the faces of the hypercube. 
The colorability of the $D$-dimensional hypercube determines the number of redundant stabilizers and hence the number of encoded logical qubits, which for the hypercube is just given by $D$. 

The $\code{2^D,D,2}$ code supports a logical gate-set comprising all Pauli gates, and logical
\begin{multline}
    \text{C}^{k}\text {Z} = \left(\id_{k} - \proj {1^{k}}\right)\otimes \id + \proj {1^{k}} \otimes Z
\end{multline}
gates for $k = 1, \ldots, D-1$ on any $k+1$ of the $D$ encoded qubits. 
Importantly, for $k \geq 2$, the C$^k$Z gate is a non-Clifford gate.
All of these operations are physically implemented with single-qubit rotations on a subset of the physical qubits. 
These single-qubit operations can be implemented in parallel with very high fidelities of ${\sim}0.9999$~\cite{sheng_high-fidelity_2018,levine_dispersive_2022}. 

The three-dimensional instance of the code with parameters $\code{8,3,2}$ is considered ``the smallest interesting color code'' \cite{campbell_blog}. 
It has a high encoding rate and allows us to perform error-detected non-Clifford computations. 
In the following, we will use the $\code{8,3,2}$ code as the main focus of our study to illustrate concepts that apply more generally to the $\code{2^D, D, 2}$ code. 
The logical operators and stabilizers of the $\code{8,3,2}$ code are illustrated in \cref{fig:fig1}(d,e), the transversal in-block operations in \cref{fig:832code}(b).

The Hadamard gate completes the universal gate set for the hypercubic codes, as for other high-dimensional codes, so it cannot have a transversal implementation~\cite{Eastin-Knill.2009}. 
Nonetheless, since the $\code{2^D,D,2}$ code is a CSS code, the logical $\ket{+^D}$-state can be prepared fault-tolerantly \cite{calderbank_good_1996,steane_error_1996}.
In \cref{fig:832code}(d), we show a circuit that prepares the $\ket{{+}{+}{+}}$ state of a single $\code{8,3,2}$ code block in a non-fault-tolerant way. 
This circuit can be understood as preparing a GHZ state of two $\code{4,2,2}$ blocks in opposite bases and then applying a transversal CNOT between these blocks, and directly generalizes to higher $D$. 
A fault-tolerant preparation can be achieved with three additional physical flag qubits~\cite{Wang.2023}. 
Moreover, measurements in both the logical $X$ and the logical $Z$-basis are transversal. 

Since the $\code{2^D,D,2}$ is a CSS code~\cite{calderbank_good_1996,steane_error_1996} the transversal CNOT gate implements a transversal CNOT gate on the~$D$ encoded qubit pairs, see \cref{fig:832code}(e). 
We entangle code blocks by applying the transversal CNOT gate in parallel on all physical qubit pairs.
We can further implement in-block CNOT and SWAP gates between arbitrary pairs of logical qubits within each block and, using this capability, also between any inter-block pair. 
Both of those gates can be realized just by permutations/relabelling of the physical qubits, see \cref{fig:832code}(c), which comes at almost no extra cost.
While permutations on their own do not grow the weight of errors, when combined with transversal entangling gates they are no longer fault-tolerant.
Generically, in order to maintain fault tolerance with both transversal and permutation gates, a round of error detection should  be applied after each permutation gate \cite{gottesman_surviving_2024}.  However, it is also important to emphasize that we can check for the fault-tolerance up to distance $d$ of specific physical circuits by enumerating the effects of all weight-$(d-1)$ Pauli errors at the possible error locations throughout the circuit on the circuit output.  
We find that for all circuits  studied in this work, fault tolerance is actually maintained even in the absence of this additional round of error detection for all $Z$ errors. 
In fact, a specific sequence of instructions needs to occur to break it for $X$ and $Y$ errors. 
In practice, we find that the specific circuits we consider are fault-tolerant even without these additional rounds of error detection, see the results in \cref{sec:code comparison} \footnote{This is implied by the quadratic decay of the XEB of the simulation of the $\code{8,3,2}$-code with error detection at the end of the circuit in \cref{fig:code_comparison}.  }.
Moreover, since one can efficiently check if fault tolerance is preserved for any random circuit instance, additional rounds of error detection/correction can be inserted into the circuit to recover fault-tolerance.
The CNOT gate between two specific qubits in two distinct blocks can be realized using two transversal CNOTs interlaced with in-block CNOT and SWAP gates, see \cref{fig:832code}(f) and the inter-block SWAP gate by compiling it from inter-block CNOT gates. 
Combined with in-block permutation gates,  arbitrary configurations of controls and targets can be realized in this way. 

Despite the blocked structure of the gates, with these gadgets we can realize any-to-any logical qubit connectivity by implementing SWAP gates between two specific logical qubits in different blocks. 
This capability enables us to extend the in-block operations to any three qubits in the system---we can simply swap targeted qubits into a block, perform the desired diagonal operation, and then swap them back to their initial locations. Together with the transversal CNOT gates, the capability to measure in the logical $X$ and $Z$ basis also makes the Bell measurement transversal.

\subsection{Random logical IQP sampling}
\label{ssec:random iqp}

Our proposal is based on sampling from encoded degree-$D$ IQP circuits. 
These circuits comprise $Z$, and C$^k$Z gates for $k =1, \ldots, D-1$ between arbitrary subsets of qubits with state preparation and measurement in the $X$ basis. 
We can implement those circuits as follows. 
Consider a system comprising $b$ blocks of the $\code{2^D, D, 2}$ code. 
We can then prepare the entire system in the $\ket{+^{bD}}$ state, and implement arbitrary in-block degree-$D$ IQP circuits using the transversal gate set.
In order to apply an IQP gate to an arbitrary subset of the qubits, we can use the SWAP gadget discussed in the previous section (\cref{ssec:gadgets}) to swap the involved logical qubits into a single block, implement the transversally realized IQP gate in that block, and then swap the logical qubits back (or to a different block). 
Finally, we can measure all blocks in the logical $X$ basis. 
Thus, we can implement sampling from arbitrary encoded degree-$D$ IQP circuits. 

We will now consider random degree-$D$ IQP circuits. 
In a \emph{uniformly} random degree-$D$ IQP circuit $C$ every gate from the gate set \{$Z$, CZ, $\ldots$,  C$^{D-1}$Z\} is applied with probability one half to every subset of qubits with the corresponding size, i.e., for every $k$-subset of qubits, a C$^{k-1}$Z gate is applied with uniform probability. 
\textcite{Bremner.2016} have shown that approximately sampling from the output distribution $p_C$
of an $n$-qubit uniformly random degree-$D$ IQP circuit $C$
with probabilities
\begin{equation}\label{eq:IQP}
    p_{C}(x) = |\bra{x}H^{\otimes n} C \ket{+^n}|^2,
\end{equation}
is classically intractable under reasonable complexity-theoretic assumptions for any $D \geq 3$.

In fact, the same is believed to hold true for a \emph{sparse} ensemble of IQP circuits comprising only $ O (n \log n) $ gates from a slightly different gate set comprising the control-phase gate CS and the $T$ gate~\cite{Bremner.2017}. 
This ensemble is significantly more hardware-efficient to implement than uniformly random IQP circuits since it can be implemented using only $O(\log n)$ parallel gate layers. 
Recently, \citet{paletta_robust_2023} have demonstrated that this ensemble of IQP circuits can be implemented fault-tolerantly in unit depth using a certain family of quantum codes they call \emph{tetrahelix codes}.
Tetrahelix codes are constructed from another three-dimensional color code, namely, a tetrahedral code whose smallest instance is the $\code{15,1,3}$ Reed-Muller code with a transversal $T$ gate.
Similar constructions make use of fault-tolerant measurement-based quantum computing~\cite{bravyi_quantum_2020,mezher_fault-tolerant_2020-2,raussendorf_fault-tolerant_2006}.

A plausible analogous and resource-efficient approach that we could take would be to define a sparse ensemble of uniform degree-$D$ IQP circuits, argue that it is hard to simulate, and then implement this ensemble using the gadgets from the previous section.
This would yield a possible path towards demonstrating quantum advantage using only native operations of the $\code{2^D, D, 2}$ code.
However, this approach has the disadvantage that the SWAP gadgets required to implement C$^k$Z gates on arbitrary subsets of qubits require an overhead compared to a single gate: 
three CNOT gates are required to  implement a single SWAP gate between an arbitrary pair of qubits, each of which requires two transversal inter-block CNOT gates, in-block qubit permutations, and additional error correction gadgets.  As detailed in the next subsection, we develop for a more hardware-efficient approach to demonstrating quantum advantage that we show converges asymptotically to the behavior of random IQP circuits.
In this sense, it constitutes a (fault-tolerant) compilation of random IQP circuits. 

We close this subsection with some general remarks on the relation between logical circuits realizable via transversal gates and IQP circuits.  In the case of qubit stabilizer codes, a rigorous proof exists that all transversal gates must be within a finite level of the Clifford hierarchy \cite{jochym-oconnor_disjointness_2018,newman_limitations_2018}. 
 The set of unitaries in the hierarchy has also been partially classified and it is conjectured  (as well as proven for the 3rd level of the hierarchy) that they all take the form of so-called generalized semi-Clifford operations~\cite{beigi_c3_2010,Zeng08}.
 Moreover, as noted above, the set of transversal gates forms a discrete group on an error detecting code~\cite{eastin_restrictions_2009}.  The finite groups formed out of generalized semi-Clifford unitaries in the $k$th level of the hierarchy have generators of the form (up to some technical caveats) \cite{anderson_groups_2022}
 \begin{equation}
     U = C^\dag P D C,
 \end{equation}
 where $C$ is a fixed Clifford unitary that is the same for all generators, $D$ is a diagonal circuit in the $k$th level (classified in \cite{Cui17}), and $P$ is a permutation circuit in the $k$th level.  Up to the basis change by the Clifford circuit $C$, these unitaries are exactly of the IQP type.  As a result, there seems to be a fundamental connection between transversal computation on stabilizer codes and IQP-like circuit dynamics.  Interestingly, if one considers non-stabilizer codes, the family of transversal operations that is achievable on qubit error detecting codes includes gates that lie completely outside the Clifford hierarchy \cite{Kubischta23,kubischta_not-so-secret_2023}.

\begin{figure*}
    \includegraphics{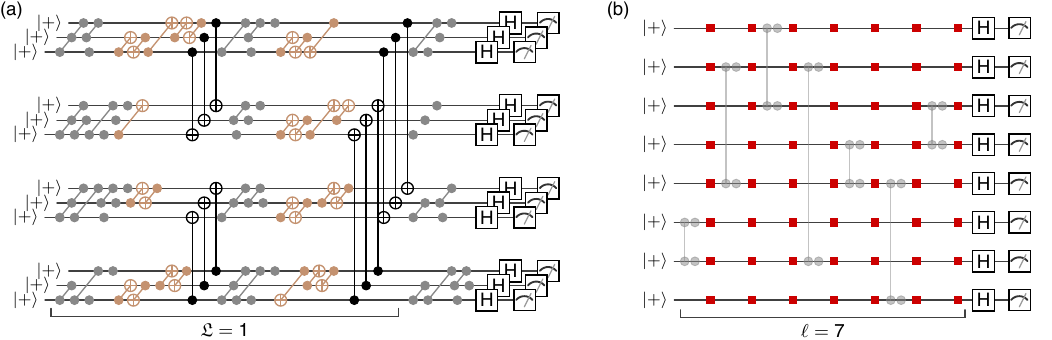}
    \caption{\label{fig:iqp circuits} \textbf{IQP circuit families.} (a) Random $\hdim=2$ hypercube IQP (hIQP) circuit on $4$ blocks of the $\code{8,3,2}$ code with hyperdepth $\hdepth=1$. 
    In-block transversal single-qubit gates realize a degree-$3$ in-block circuit (gray). 
    In-block CNOT gates are realized through permutations of the physical qubits (ochre). 
    Transversal CNOT gates couple the code blocks (black).
    The measurement is performed transversally in the $X$ basis.
    (b) In a sparse degree-$2$ IQP (sIQP) circuit as we consider it here, in every one of the $\ell$  circuit layers, a random degree-$2$ circuit consisting of a random $CZ$ and $Z$ gates (vertices) is applied to a random pair of qubits, which in the noisy setting is followed by a parallel layer of Pauli noise (red squares). 
     }
\end{figure*}

\subsection{Hypercube IQP circuits}
\label{ssec:hiqp intro}

In order to optimize the hardware efficiency, we focus on a gate set which can be fully parallelized, namely, inter-block transversal CNOT gates as well as in-block degree-$D$ IQP and CNOT circuits. 
Any in-block degree-$D$ circuit can be compiled in a single layer of physical single-qubit gates which are powers of the $T$ gate, and an in-block CNOT circuit just requires physically permuting the atoms constituting the code block. 
In order to achieve fast scrambling of the quantum circuits while exploiting the long-range parallelization possible in the reconfigurable atom arrays, we will consider an interaction graph between code blocks given by a $\hdim$-dimensional hypercube.
%
In general, quantum dynamics in hypercube geometry exhibit fast (Page) scrambling~\cite{Page.1993,Hashizume.2021} due to the expansion properties of hypercube graphs.

When compiling the hypercube IQP circuits, our first observation is that CNOT gates leave the family of degree-$D$ IQP circuits invariant under conjugation. 
In other words, a degree-$D$ IQP circuit with CNOT gates is equivalent to some other degree-$D$ IQP circuit. 
To see this, we exploit the fact that a circuit comprised of CNOT gates just acts as a linear map on bit strings,
see \cref{ssec:equivalence iqp cnot}. 
Given this observation, we define the ensemble of \emph{hypercube IQP (hIQP)} circuits (with degree $D$, dimension $\hdim$, and hyperdepth $\hdepth$) as follows, see \cref{fig:iqp circuits}(a):
\begin{enumerate}
     \item Prepare the logical $\ket{+^D}^{\otimes 2^\hdim}$ state. 
    \item Perform $\hdepth$ \emph{hypercube layers} each comprising 
    \begin{enumerate}[label=\roman*.]
        \item uniformly random logical degree-$D$ circuits and uniformly random logical CNOT gates in every block alternated with 
        \item transversal CNOT gates with random orientation between blocks 
    \end{enumerate}
        on all $\hdim$ sets of parallel edges of the hypercube.
     \item Perform a last layer of uniformly random logical degree-$D$ in-block circuits. 
     \item Measure all logical qubits transversally in the logical $X$-basis.
\end{enumerate}

After encoding of the code blocks, a hyperdepth-$\hdepth$ hIQP circuit (i.e., a circuit with $\hdepth$ hypercube layers) thus requires $\hdim \cdot \hdepth + 1 = \hdepth \cdot \log n +1$ physical gate layers each comprising a layer of in-block degree-$D$ circuits, a layer of in-block CNOT gates, and (except the last) a layer of transversal CNOT gates.
Finally, the transversal CNOT gates along parallel edges of the hypercube can be physically realized in the reconfigurable processor by interlacing a pair of two-dimensional grids of qubits, see ED Fig.~6 of Ref.~\cite{bluvstein_logical_2024}

We can also consider variants of random hyperdepth-$\hdepth$ hIQP or combinations thereof that are motivated by experimental feasibility, or fault-tolerant properties.
In one variant, random IQP gates are applied to a block if and only if it was the target of a CNOT gate in the previous layer. 
This is because diagonal gates commute through the control of a CNOT gate so that the two random in-block IQP circuits before and after the gate give rise to a new random IQP circuit.
In this variant, C$^{D-1}$Z gates are also applied deterministically, and single-qubit $Z$ gates are only applied in the last circuit layer, since they commute through the circuit. 
This variant of hIQP circuits, which furthermore did not involve the in-block CNOT gates as realized by physical atom permutations, was implemented in Ref.~\cite{bluvstein_logical_2024}.
This has the advantage that of requiring less error detection to maintain fault tolerance compared to the circuits considered here.  
As was shown recently, however, the resulting symmetry in the circuit can be exploited to reduce the exponent in the time required to simulate these circuits \cite{maslov_fast_2024}. 
This motivates including the in-block random CNOT layers.

The logical hIQP circuit has several interesting interpretations in terms of nested hypercubes. 
Consider for concreteness the $\code{8,3,2}$ code. 
First, observe that the interaction graph of the entire physical circuit is a $(\hdim\,{+}\,3)$-dimensional hypercube of individual physical qubits, since the $\code{8,3,2}$ code blocks themselves are defined on a three-dimensional cube. 
Then, consider the encoding circuit of the $\ket{{+}{+}{+}}$ state. 
This state is created by transversally entangling two logical $\ket {++}$ states of the $\code{4,2,2}$ code, see \cref{fig:fig1}(d). 
Thus, the $\hdim$-dimensional hypercube of $\code{8,3,2}$ codes can be reinterpreted as a $(\hdim\,{+}\,1)$-dimensional hypercube of $\code{4,2,2}$ codes. 
This idea generalizes to encoding the logical $\ket{+^D}$ state of the $\code{2^{D-1},D-1,2}$ code using two sets of logical $\ket{+^{D-1}}$ states of the $\code{2^{D-1},D-1,2}$ code. 
If the single-qubit rotations are altered, we can thus understand the physical hIQP circuit as state preparation of a more general $\code{2^{\hdim+D},\hdim+D,2}$ code.

Finally, we also note that the output states of degree-$D$ IQP circuits before the last layer of Hadamard gates are also known as hypergraph states with hyperedges comprising up to $D$ qubits \cite{rossi_quantum_2013}. 
Hypergraph states have been studied in the quantum information literature \cite{guhne_entanglement_2014,lyons_local_2015} and can serve as a resource for measurement-based quantum computing \cite{miller_hierarchy_2016,takeuchi_quantum_2019,gachechiladze_changing_2019}.

\subsection{Fault-tolerant sampling with error detection}
\label{ssec:error detection}

The main advantage of sampling using logical encodings is the resilience to errors. Typically, quantum error correction consists of many rounds of syndrome extraction repeated throughout the circuit, in order to prevent errors from spreading and, thus, suppress the probability of logical faults. While this approach will be eventually necessary to realize large-scale quantum algorithms, our fault-tolerant compilation of IQP circuits allows us to benefit from using encoded qubits with limited rounds of mid-circuit measurement.  In the extreme case, we can perform the syndrome measurement only at the end which significantly simplifies experimental implementation. We focus primarily on the $\code{8,3,2}$ color code, which is the smallest code with the desired properties, due to its high rate and small size.

In \cref{fig:encoded sampling}(a), we show the linear XEB score for hIQP circuits with $\hdim \,{=} \,4$ and $\hdepth\, {=}\, 2$ using $\code{8,3,2}$ codes, with various amounts of postselection based on evaluating the $X$-type stabilizer from the measurement data. 
For each measurement, we set a threshold on the number of violated stabilizers of which there are one per block
and discard the samples that fail this test. 
The fraction of the remaining samples is the \emph{acceptance ratio}. 
We observe a significant increase in the XEB score as we increase the amount of postselection in the system. The XEB curve approximately obeys a power law $\propto x^{7/8}$, which is consistent with a simple independent and identically distributed (i.i.d.) error model description; see \cref{app:powerlaw}.
Roughly speaking, the scaling arises from the fact that one out of the eight physical qubits does not participate in any logical operator. 
Using this approach, Ref.~\cite{bluvstein_logical_2024} realized IQP circuits on 48 logical qubits with logical fidelities above those achievable with physical qubit implementations.

Importantly, the in-block CZ and CCZ gates (single-qubit rotations), as well as the in-block CNOT (atom transport) many-body logical operations in the $\code{8,3,2}$ code are implemented with single-qubit operations on the physical level. Since these single-qubit operations can be realized with high fidelity, a logical implementation of degree-$D$ IQP circuits can be advantageous compared to a physical one even in the absence of error detection. 
In fact, we expect that errors in the circuits studied here are dominated by the inter-block CNOT gates, rather than the in-block CZ and CCZ gates.

\begin{figure}
    \centering
    \includegraphics{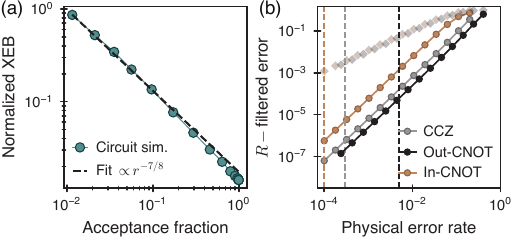}
    \caption{\textbf{Performance of encoded sampling in hIQP circuits.} (a) The logical XEB score can be improved with increased postselection, at the cost of fewer samples. The relationship between the XEB and acceptance fraction follows an approximate power law, which for the $\code{8,3,2}$ code has an exponent close to $-7/8$. The simulation was performed on a single random instance of ($\hdim\,{=}\,4$, $\hdepth\,{=}\,2$) hIQP circuit (128 physical, 48 logical qubits) with single-qubit depolarizing noise corresponding to 99.9\%  and 99.5\% fidelities for one- and two-qubit gates, respectively. The circuit includes only the Clifford subset of available operations, and the data points are averaged over ${\sim}10^6$ samples. The XEB \eqref{eq:def xeb} is evaluated relative for the single random circuit and normalized by the ideal value for this circuit $\chi_{\rm ideal}\,{=}\,3$. (b) Relative performance of logical operations within a circuit can be compared using the notion of $R$-filtered infidelities or errors with $R\,{=}\,0$ (opaque) and $R\,{=}\,2$ (semi-transparent). The three main operations in our circuits are the in-block diagonal gates, out-block CNOT gates, and in-block CNOT gates which are realized on the physical level by single-qubit gates, atom transport, and two-qubit gates, respectively. The vertical lines correspond to characteristic error rates for the relevant operations on the physical level. See \cref{app:noisy encoded} for more details.}
    \label{fig:encoded sampling}
\end{figure}

To quantify the relative performance of our circuit elements we introduce and use the concept of $R$-filtered fidelity, which measures how a given (many-body) operation affects an input state with up to weight-$R$ physical errors, see \cref{app:noisy encoded} for details.  
Refining the fidelity in this way is helpful in analyzing encoded computations. 
There, in contrast to bare computations, entanglement builds up between syndrome and logical degrees of freedom in a context-dependent way, giving rise to an intrinsically non-Markovian error model. 
For example, for a computation encoded in a distance-$2$ fault-tolerant circuit, the $0$-filtered fidelity of any single-qubit operation is always~$1$. This is because no single-qubit error can propagate to a logical one. 
In contrast, the $1$-filtered fidelity is far from unity, since a ``background'' physical error can combine with the single-qubit operation to cause a logical fault.
In \cref{fig:encoded sampling}(b), we plot the $R$-filtered infidelity and see that the inter-block CNOT indeed introduces the most errors out of our basic operations. 
The capability to realize the non-Clifford CCZ gates with high fidelity is therefore a central element of our proposal, since in most scenarios it is the most costly gate: 
physical implementations of the CCZ often have much lower fidelity than two-qubit gates, and typical QEC codes require expensive magic state distillation protocols to realize them. 
Thus, we expect that this fault-tolerant approach to classically hard circuits can result in very high XEB scores even in early logical quantum processors, as demonstrated in Ref.~\cite{bluvstein_logical_2024}.

\section{Complexity and verification of (h)IQP circuits}
\label{sec:hiqp}

Let us now turn to a detailed investigation of random hIQP circuits as defined in \cref{ssec:hiqp intro} and their properties relevant to quantum advantage demonstrations. 
We will address two questions in detail. 
The first question regards the classical complexity of sampling from the output distribution of the circuits and, second, how to verify the correctness of the samples. 
We do so both in the finite-size and the asymptotic regimes. 
As mentioned above, the complexity of various ensembles of IQP circuits has been studied by \textcite{Bremner.2016,Bremner.2017} using a complexity-theoretic argument for the hardness of sampling from random quantum circuits based on Stockmeyer's algorithm; see Ref.~\cite{hangleiter_computational_2023} for a review of that argument.
Here, we will argue that this argument applies also to random hyperdepth-$2$ hIQP circuits.
To do so, we make use of a measure of scrambling called \emph{anticoncentration} that features in this  argument and serves as an indicator of classical hardness~\cite{aaronson_computational_2013,Bremner.2016,hangleiter_computational_2023}.
We study anticoncentration of random hIQP circuits in \cref{ssec:complexity}.

The second question regards verification. 
Verifying classical samples from random quantum circuits is a challenging problem, one that in fact requires exponentially many samples if the goal is unconditional verification \cite{hangleiter_sample_2019}. 
What has become the standard approach to verification in this scenario is to make use of the linear \emph{cross-entropy benchmark (XEB)} \cite{Boixo.2018,arute_quantum_2019}. 
This benchmark can be evaluated using only a small (polynomial) number of samples but involves computing ideal output probabilities, which can require exponential time. 
It is appealing, however, because it has been argued that it can be used as a proxy for the fidelity of the quantum pre-measurement states averaged over instances of the random circuit \cite{arute_quantum_2019,dalzell_random_2024}.
It may also witness the achievement of a computationally complex task \cite{aaronson_complexity-theoretic_2017,aaronson_classical_2020,gao_limitations_2024}.
The XEB has been studied in detail for random universal circuits, revealing the value of the XEB that we would obtain for ideal circuits, as well as the behaviour of the XEB and fidelity under local circuit noise~\cite{barak_spoofing_2021,gao_limitations_2024,Ware.2023,Morvan.2023}.
We study the linear XEB for hIQP circuits in \cref{ssec:xeb and noise}.

It turns out that both the anticoncentration property and the linear XEB can be written as a second moment of potentially noisy hIQP circuits. 
The tool of our analytical study of both complexity and verification of hIQP circuits will be a mapping of the behvaiour of such second-moment quantities for general degree-$D$ IQP circuits of varying circuit depth to the dynamics of four classical states. 
We introduce this mapping and discuss its most important properties in \cref{ssec:stat mech}, deferring details to \cref{app:statmech,app:xebproof}.

\subsection{Complexity of hIQP circuits}
\label{ssec:complexity}

To set the stage for our results on the scrambling and complexity properties of random hIQP circuits, let us recap the results on uniform IQP circuits by 
\textcite{Bremner.2016}. 
They show that output distributions of circuits comprising \{$Z$,CZ,CCZ\} in uniformly random locations, i.e. uniform degree-$3$ circuits, are classically intractable to simulate under certain complexity-theoretic assumptions. 
Specifically, there is no efficient classical sampling algorithm for uniformly random IQP circuits, assuming that the output probabilities \eqref{eq:IQP} are \#P-hard to approximate on average over random IQP circuits, unless the widely believed conjecture that the polynomial hierarchy does not collapse is false.
The key assumption on which the result hinges is therefore the approximate average-case hardness of computing the outcome probabilities. 
To give evidence for this conjecture, \textcite{Bremner.2016} show approximation hardness for degree-$3$ IQP circuits \emph{in the worst case} and a strong version of the so-called \emph{anticoncentration} property. 

At a high level, the anticoncentration property ensures that the output distribution of random IQP circuits has barely any structure that might be exploited by an approximate classical sampling algorithm to speed up the computation compared to the worst case. 
Technically, the (strong) anticoncentration property is the statement that the average second moment of the output distribution of a random circuit family $\mc C$ as measured by the ideal linear XEB score is constant as 
\begin{align}
\label{eq:anticoncentration}
    \overline \chi_{\mc C} \coloneqq 2^{2n} \cdot \mb E_{C \leftarrow \mc C }\left[p_C(x)^2\right] -1 \in O(1). 
\end{align}
The anticoncentration property ensures that in order to sample from the correct output distribution it is not merely sufficient to identify the dominant outcomes in the distribution.
Rather, almost all probabilities need to be computed to exponential precision in order to perform the sampling correctly.
In the following, we will use anticoncentration as a key---but not the sole indicator---that sampling from non-uniform IQP circuits remains intractable. 
For uniform degree-$D$ (uIQP) circuits with any $D\geq 2$, \textcite{Bremner.2016} show that 
\begin{align}
    \overline \chi_{\mathrm{uIQP}} = 2 -2^{-n+1}. 
\end{align}

\begin{figure}
    \centering
    \includegraphics[width=\linewidth]{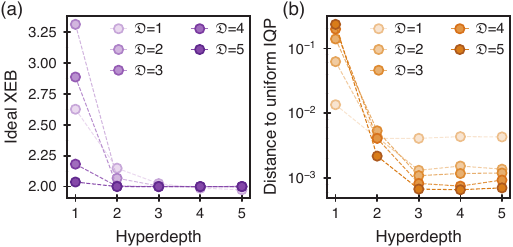}
    \caption{\textbf{Convergence of hypercube IQP to uniform IQP.} 
    The effective IQP circuits implemented by random hIQP circuits converge to uniform IQP circuits as a function of hyperdepth. For each data point, we average ${\sim}\,10^6$ random degree-$2$ hIQP circuits.
    (a) Convergence of the ideal XEB score $\overline \chi_{\mathrm{hIQP}}(\hdepth)$ \eqref{eq:anticoncentration} of hIQP circuits to the uniform value $\overline \chi_{\mathrm{uIQP}} = 2$ as a function of hyperdepth $\hdepth$.
    (b) Total-variation distance between the distribution over IQP circuits defined by hIQP circuits  and uniform IQP. 
    }
    \label{fig:random hiqp}
\end{figure}

How do hIQP circuits fare in terms of their complexity? 
We expect degree-$3$ and higher hIQP circuits to be classically hard to simulate already at constant hyperdepth because of their large number of non-Clifford gates and the expander properties of the hypercube. 
Therefore, we expect that the quantum dynamics scramble very quickly in this geometry. Fast scrambling on hypercubes has in fact been observed in Ref.~\cite{Hashizume.2021}. 

We provide analytical evidence for this in two steps. 
We first show that degree-$D$ circuits anticoncentrate if they have $\Omega(n \log n)$ gates, but not for any constant number of gates. 
To this end, we consider a model of sparse random degree-$2$ IQP (sIQP) circuits in which a CZ gate and $Z$ gates are applied with probability $1/2$ to $\ell$ random qubit pairs, see \cref{fig:iqp circuits}(b).
For us, sparse IQP serves as a toy model of low-depth hIQP circuits, which have significantly more structure and are therefore more difficult to analyze. 

\begin{theorem}[Anticoncentration of sparse IQP]
\label{thm:sparse iqp anticoncentration main}
The ideal linear XEB score of sparse degree-$2$ IQP circuits acting on $n$ qubits with uniformly random $Z$ gates and random CZ gates acting on $\ell$ random pairs is given by 
\begin{align}
    \overline \chi_{\mathrm{sIQP}}(\ell) &= \overline \chi_{\mathrm{uIQP}}+ 2^{-\Omega(\log n)} & \text{ if } \ell \geq \frac {11} 4 n \log n \\
    \overline \chi_{\mathrm{sIQP}}(\ell) & \in \Omega (2^n) & \text{ if } \ell \in O(n). 
\end{align}
\end{theorem} 

This result is analogue to the sparse anticoncentration result of \textcite[Lemma 6]{bremner_achieving_2017} for degree-$2$ circuits. 
Next, we show that the ideal XEB score of hIQP circuits approaches the uniform score as $\hdepth \rightarrow \infty$. 
This complements the sparse IQP result by showing that the structure of hIQP circuits is asymptotically irrelevant to the XEB.

\begin{theorem}[Anticoncentration of deep hIQP]
\label{thm:anticoncentration deep hiqp}
    The ideal linear XEB score of degree-$2$, hyperdepth-$\hdepth$ hIQP circuits acting on $2 ^\hdim$ blocks of $D$ qubits satisfies
   \begin{align}
    \label{eq:xeb score hiqp main}
        \overline \chi_{\mathrm{hIQP}}(\hdepth) \xrightarrow{\hdepth \rightarrow \infty} \overline \chi_{\mathrm{uIQP}}
    \end{align}
\end{theorem}
We defer the proofs of \cref{thm:sparse iqp anticoncentration main,thm:anticoncentration deep hiqp} to  \cref{app:second moment behaviour}. 
We note that \cref{thm:sparse iqp anticoncentration main,thm:anticoncentration deep hiqp}, while formulated for degree-$2$ circuits, provide upper bounds for the XEB value of IQP circuits with additional higher-degree gates which may be random or fixed.  
This is because all gates commute and therefore additional gates cannot increase the XEB score. 
In particular, \cref{thm:sparse iqp anticoncentration main} implies anticoncentration for $\log(n)$-sparse random degree-$D$ circuits, in which uniformly random degree-$D$ circuits are applied to $\ell$ random subsets of $D$ qubits.

The two anticoncentration results complement each other. 
While \cref{thm:sparse iqp anticoncentration main} implies that uniformly sparse degree-$D$ circuits can anticoncentrate in logarithmic circuit depth, 
\cref{thm:anticoncentration deep hiqp} implies anticoncentration of hIQP circuits as $\hdepth \rightarrow \infty$. 
Since IQP and hIQP share the same gate set, this suggests that the ensemble of degree-$D$ hIQP circuits approaches degree-$D$ uniform IQP in this limit and, thus, that very deep hIQP circuits are also classically hard to simulate.
Together, they lend credence to our claim that hIQP circuits anticoncentrate in constant hyperdepth,  meaning that the total circuit comprises $\Theta(n \log n )$ (nonuniform) IQP gates.

To better understand the behaviour of the hypercube circuits at low depths, we numerically compute the ideal XEB value for random hIQP circuits as a function of hyperdepth $\hdepth$ and hypercube dimension $\hdim$. 
We do so by classically sampling random degree-$2$ circuits, which are efficiently simulatable. 
The results, shown in \cref{fig:random hiqp}(a), demonstrate that the XEB of hIQP circuits quickly converges to the uniform value as a function of hyperdepth, with hyperdepth-$2$ circuits being basically converged. 
We also compare the full ensemble of degree-$2$ hIQP circuits to uniform IQP and find a very similar behaviour. 
Indeed, for any $\hdepth \geq 2$ the hIQP distribution quickly converges to uniform IQP, see \cref{fig:random hiqp}(b).

Together, these observations provide evidence that after hyperdepth $2$ any structure which might be exploited by a classical simulation algorithm has essentially vanished, rendering classical simulation inefficient, and we have the following conjecture. 

\begin{conjecture}[hIQP approximate average-case hardness]
    \label{conj:average case hiqp}
    Approximating the output probabilities of hyperdepth-$2$, degree-$3$ hIQP circuits up to error $O(2^{-n})$ is \#P hard for a constant fraction of the instances. 
\end{conjecture}
Note that the output probabilities of degree-$D$ circuits are related to the \emph{normalized gap} of degree-$D$ Boolean polynomials \cite{Bremner.2016} and hence we can equivalently phrase \cref{conj:average case hiqp} in terms of approximating the gap of those polynomials, see Ref.~\cite{dalzell_how_2020} for an in-depth discussion.
\cref{conj:average case hiqp} implies hardness of sampling under a widely believed complexity-theoretic conjecture by a standard argument, see Ref.~\cite{hangleiter_computational_2023} for details. 
\begin{theorem}[Hardness of hIQP sampling]
  There is no efficient classical algorithm that samples from the output distribution of hyperdepth-$2$, degree-$3$ hIQP circuits up to constant total-variation distance error, unless \cref{conj:average case hiqp} is false or the polynomial hierarchy collapses to its third level. 
\end{theorem}

\begin{figure}

    \centering
    \includegraphics{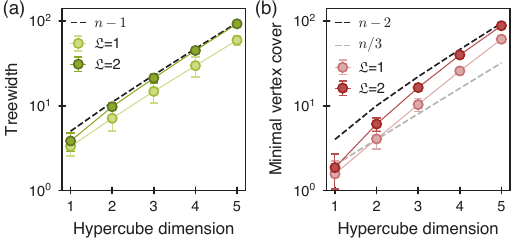}
    \caption{\textbf{Evidence for classical hardness.} To give evidence for the hardness of simulating random hIQP circuits, we analyze the properties of the effective hypergraphs defined by degree-$3$ hIQP circuits for 100 random instances at hyperdepth $\hdepth =1$ and $\hdepth = 2$. 
    Error bars are one standard deviation.
    (a) The treewidth determines the complexity of tensor-network contractions \cite{liu_tropical_2021}. (b) The minimum vertex cover of the degree-3 hyperedges determines the run-time of the near-Clifford simulator of Ref.~\cite{maslov_fast_2024}.
    \label{fig:classical}
    }
\end{figure}

How are these complexity-theoretic results reflected in the concrete runtime of classical simulation algorithms?
To answer this question, we now discuss the runtime scaling of the available classical algorithms. 

The most important classes of general-purpose simulation algorithms are tensor-network algorithms. 
The runtime of these algorithms is governed by the treewidth of the circuits' interaction graph~\cite{Markov.2008}. 
While the hyperdepth-$1$ circuits can be simulated in time $O(2^{n/2})$~\cite{bluvstein_logical_2024}, as discussed in more detail in Ref.~\cite{bluvstein_logical_2024}, we expect the treewidth of hyperdepth-$2$ circuits to (nearly) saturate $n$. 
Tensor-network techniques can also be applied to the degree-$D$ IQP circuit equivalent to the logically implemented hIQP circuit which comprises degree-$D$ and CNOT gates \cite{liu_tropical_2021,liu_computing_2023,generictensornetworks_website}.
We show the treewidth of the hypergraphs defined by the effective IQP circuits in \cref{fig:classical}(a), and find that it is significantly smaller than $n$ by a constant factor for hyperdepth $\hdepth = 1$ but nearly saturating the maximal value of $n-1$ at $\hdepth \geq 2$. 

Another family of simulation algorithms are near-Clifford simulators \cite{bravyi_improved_2016,bravyi_simulation_2019}. 
Since the number of non-Clifford gates of a random degree-$3$ hIQP circuit scales as $n \cdot \hdim \cdot \hdepth /3 $ we would expect those algorithms to have a large runtime. 
Nonetheless, IQP circuits appear to be prone to relatively low-rank stabilizer decompositions.
Notably, very recently Ref.~\cite{maslov_fast_2024} has demonstrated that in the absence of the in-block permutation gates, there is a decomposition of the output states of hIQP circuits of any depth into $2^{n/3}$ stabilizer states. 
To this end they exploited the fact that the effective IQP hypergraph has a vertex cover of size $n/3$. 
For the hIQP circuits considered here, the minimal vertex cover approaches its maximal size given by $n-D+1$, however, see \cref{fig:classical}(b).
We note that one could attempt to classically optimize the minimum vertex cover over CNOT circuits just before the $X$-basis measurement, since the addition of such a circuit can be simulated in classical post-processing, see \cref{ssec:equivalence iqp cnot}. 
A closely related approach to that of Ref.~\cite{maslov_fast_2024} has also been explored in Ref.~\cite{codsi_classically_2023} where the authors find low-rank stabilizer decompositions of 2-local classically hard IQP circuits using large independent sets of the IQP graph.

Further specific simulation algorithms for IQP circuits can exploit the fact that one can additively approximate their outcome probabilities in conjunction with sparsity of the output distribution \cite{shepherd_binary_2010,schwarz_simulating_2013,pashayan_estimation_2020}. 
But such algorithms fail due to anticoncentration. 
This is because anticoncentration implies that the output distribution has exponentially large support and hence the individual probabilities need to be estimated with exponential precision.  
The best known algorithm for exactly computing output probabilities of IQP circuits scales as $\tilde \Omega(2^{0.9965 n})$ \cite{lokshtanov_beating_2017} and exploits a clever way to find roots of degree-$3$ polynomials.
As discussed in Ref.~\cite{dalzell_how_2020} it seems unlikely that this scaling can be improved beyond $2^{n/2}$ for uniform degree-$3$ circuits.
Importantly, the best known algorithm for general IQP circuits thus has a runtime of $\Omega(2^{0.9965 n})$ per amplitude. 

The results above all apply to the computation of outcome probabilities of IQP circuits.  
It is not immediately clear that they also imply simulability in terms of sampling. 
The best available approaches to sampling are, first, the gate-by-gate simulation algorithm of \textcite{bravyi_how_2022}. 
This algorithm efficiently turns any method to compute amplitudes of degree-$D$ phase states where some qubits are measured in the $X$ basis and others in the $Z$ basis into a sampling algorithm. 
For instance, this has been done in Ref.~\cite{maslov_fast_2024} to simulate the experiments in Ref.~\cite{bluvstein_logical_2024}. 
Second, one can use rejection sampling if the output distribution is known to be very flat \cite{markov_quantum_2018}. 
It is not clear, however, that the output distributions of IQP circuits are sufficiently flat for this approach to work. 

Finally, let us note that when comparing to a noisy experiment, the question of classical simulability becomes more subtle. 
Now, we can also allow the classical algorithm to exploit ``the noise'', but what exactly this means is up for debate. 
One could either require a classical algorithm to output samples close to the physically noisy samples. 
In this scenario, under a simple noise model, any constant amount of local noise makes IQP circuits classically simulatable after constant depth \cite{rajakumar_polynomial-time_2024,bremner_achieving_2017}.
However, these algorithms are not competitive in practice because, while they are technically speaking ``polynomial-runtime'' algorithms, the exponent depends very unfavorably on the noise strength as $1/\epsilon^{O(1/\gamma)}$, where $\gamma$ is the local noise rate, and $\epsilon$ is the targeted total-variation distance error~\cite{aharonov_polynomial-time_2023}. 
Alternatively, one could merely require the samples to pass a test such as the cross-entropy benchmark with a similar score, or ``spoof the benchmark'' \cite{barak_spoofing_2021,gao_limitations_2024}. 
This approach relates closely to the question of verifying the classical samples, a topic we discuss in detail in the next section. 
For hIQP circuits Ref.~\cite[][ED Fig.~8]{bluvstein_logical_2024} finds that an approach similar to that of \textcite{gao_limitations_2024} is not successful.
Ultimately, the question is whether quantum error correction is needed for scalable quantum advantage---progress towards this goal has been made in restricted IQP settings~\cite{bremner_achieving_2017,bravyi_quantum_2020}.
While we do not study noisy simulability of IQP circuits in this work, the question is therefore highly interesting and deserves further study, in particular in the context of encoded circuits. 

\subsection{Verification of noisy IQP circuits}
\label{ssec:xeb and noise}

We now turn to the question how to verify samples from potentially noisy implementations of random  hIQP circuits. 
We will show that IQP circuits have appealing properties in this respect. 
Generally speaking, unconditionally verying samples drawn from random quantum circuits which anticoncentrate cannot be done from less than exponentially many samples \cite{hangleiter_sample_2019}. 
This is why sample-efficient benchmarks are typically used to validate the correctness of an experiment \cite{Boixo.2018,arute_quantum_2019,bouland_complexity_2019}. 
To the extent that a benchmark does not reflect all intricacies of the targeted task, `spoofing' the experiment---that is, achieving a similar score on the benchmark---can become easier compared to simulating the full distribution. Nevertheless, these benchmarks remain an important tool in the verification of sampling experiments. 
Most sampling experiments have used the average linear XEB  \cite{arute_quantum_2019,zhu_quantum_2022,Morvan.2023}
\begin{align}
\label{eq:def xeb}
    \overline \chi \coloneqq \mb E_{C}\, 2^n  \sum_x  q_C(x) p_C(x) -1 , 
\end{align}
between the ideal and noisy output distributions of a random circuit $C$ denoted by $p_C$ and $q_C$, respectively, as a benchmark of quantum advantage. 
The linear XEB has the favourable property that it can be estimated efficiently from few samples by averging the ideal probabilities $p_C(x_i)$ corresponding to samples $x_i$ from a randomly chosen quantum circuits. 
It has the unfavourable property that computing those ideal probabilities can require exponential time.

The second feature which makes the XEB appealing is that it can be used to measure the many-body average fidelity 
\begin{align}
    \label{eq:average fidelity}
    \overline F \coloneqq \mb E_C \bra C \rho_C \ket C 
\end{align}
of the pre-measurement state $\rho_C$ compared to the target state $\ket C$ for sufficiently scrambling circuits or dynamics~\cite{Boixo.2018,arute_quantum_2019,mark_benchmarking_2023,shaw_benchmarking_2024}.  
For circuits with Haar-random two-qubit gates and local noise models, this relation holds so long as the single-qubit error rate per gate layer $\epsilon < \gamma/n$ for some constant~$\gamma$~\cite{dalzell_random_2024,gao_limitations_2024}. 
For error rates $\epsilon > \gamma/n$ on the other hand, the XEB decays much slower than the fidelity, namely with the circuit depth~\cite{gao_limitations_2024,deshpande_tight_2022,Ware.2023, Morvan.2023}. 
The slower decay for large noise rates constitutes a vulnerability of XEB that can be exploited when designing classical algorithms using bespokely planted noise at specific locations of the targeted ideal circuit. 
These algorithms achieve similar XEB scores~\cite{gao_limitations_2024,pan_solving_2022} as the first experimental demonstration of quantum advantage \cite{arute_quantum_2019}. 
In fact, it has been argued by \textcite{gao_limitations_2024} that polynomial-time simulations should be able to achieve a comparably high XEB score asymptotically that decays only with the circuit depth.

Intuitively, the deviation between XEB and fidelity for locally Haar-random circuits can be understood in terms of the comparison of error events close to the beginning and the end of the circuit, respectively \cite{gao_limitations_2024}. 
Consider a noise model in which single-qubit Pauli errors occur uniformly across the entire circuit volume. 
Errors at the beginning of the circuit tend to affect XEB and fidelity in the same way. 
To see this, consider their effect on the initial state which can be gleaned from propagating them backwards into a string with weight determined by the backwards lightcone. 
Any Pauli error which propagates into a string with a single Pauli-$X$ term will flip a bit of the all-zero input state giving rise to an orthogonal state.
This will cause the fidelity as well as the linear XEB to be zero \footnote{To see the latter observe that $2^{2n} \mb E_C p_C(x) p_C(y) - 1 \,{=}\, 0 $ for $x \neq y$.}. 
Pauli errors which propagate into a string with only $Z$ errors do not affect XEB and fidelity.
In contrast, an error close to the end of the circuit will affect XEB and fidelity differently. 
For the XEB, a comparably high fraction of such errors propagate into a low-weight $Z$-type string at the end of the circuit. 
Such errors do not affect the outcome probabilities. 
But they do affect the fidelity: to compute it, we compute the overlap with an ideal state, and therefore an error towards the end of the circuit has a large lightcone. 
Thus, errors which propagate into a low-weight $Z$-type string at the end of the circuit will not affect XEB but do affect fidelity, giving rise to the deviation. 

Consider in contrast degree-$D$ IQP circuits. 
Here, all $Z$-type errors commute with the circuit and therefore yield an orthogonal state. 
When commuting a single-qubit $X$ error through a degree-$k$ gate, it will turn into an $ X \otimes Z^{\otimes (k-1)}$ error.
Unless there are cancellations we can therefore commute an $X$ error to the beginning or the end of the circuit, where it will have picked up a $Z$-string, which yields an orthogonal state. 
The only deviation between XEB and fidelity should therefore stem from $X$ errors in the very last circuit layer. 

In the following, we will rigorously assess this intuition by analyzing the relationship between XEB and fidelity for different circuit ensembles and noise regimes. 
We find that, indeed, at low noise rates $\lesssim 1/n$ there is a correspondence between XEB and fidelity for general local noise.
However, that correspondence breaks down for high noise rates and while the fidelity decays exponentially in the circuit volume, the XEB only decays exponentially in the circuit depth, as is the case for Haar-random circuits \cite{Boixo.2018,gao_limitations_2024,Ware.2023,Morvan.2023}.
We then show the transition between XEB and fidelity can be shifted to higher noise rates by adding a fixed layer of gates, which does not have an analogue for Haar-random circuits. 
Altogether, these results give evidence that the intuition discussed above gives the correct mechanism, while highlighting that error cancellations do in fact play a significant role at high noise rates.

\begin{figure}
    \centering
        \includegraphics[width=\linewidth]{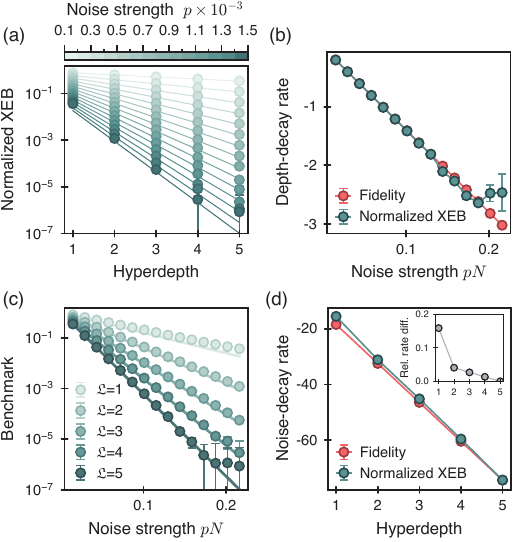}
    \caption{\textbf{XEB versus fidelity in hIQP circuits at low noise rates.} 
    Measures of quality for noisy random degree-2 hIQP circuit ensemble at $\hdim\,{=}\,4$ with circuit-level single-qubit depolarizing noise following a gate application with total error probability $p$, i.e., each $X/Y/Z$ error occurs with probability~$p/3$.
    Single-qubit and two-qubit fidelities are therefore $1-p$ and $1-2p$, respectively. 
    To obtain the data, we average~${\sim}10^4$ circuit realizations with~$10^7$ noise realizations each and measure the average fidelity $\overline F$ \eqref{eq:average fidelity} and normalized average  XEB $\overline \chi/\overline \chi_{\mathrm{hIQP}}$ with $\overline \chi,\overline \chi_{\mathrm{hIQP}} $ defined in \cref{eq:def xeb,eq:anticoncentration}, respectively. 
    (a) Decay of the normalized XEB (circles) and fidelity (lines) with hyperdepth for various noise rate $p$. 
    (b) Decay rate with hyperdepth $\alpha$ of the normalized XEB (teal) and fidelity (orange) fitted as $\propto\exp(\alpha \hdepth)$, respectively, shown as a function of noise strength $pN$. 
    (c) Decay of the normalized XEB (circles) and fidelity (lines) with noise strength $pN$ (as in (b)) for various hyperdepths. 
    (d) Decay rate with noise strength~$\beta$ of the normalized XEB (teal) and fidelity (orange) fitted as $\propto\exp(\beta \times pN)$, respectively, shown as a function of the hyperdepth $\hdepth$. 
    The inset shows the relative difference~$\beta_{\overline \chi}/\beta_{\overline F}$ between the decay rates of normalized XEB~$\beta_{\overline \chi}$ and fidelity~$\beta_{\overline F}$.
     The last two values of~$pN$ from are ommited from fitting due to their large uncertainty.
     }
    \label{fig:noisy xeb}
\end{figure}

First, we show that in a general setting, at low noise rates, the XEB of IQP circuits can be used as a proxy for the average fidelity. 
Specifically, we consider general noise occurring in encoded circuits as modeled by quantum instruments and any quantum circuit that is equivalent to an IQP circuit. 
In this setting, we show that if the noise rate $\epsilon < c/n \log n$ and the circuit depth $\ell \in \Omega(\log n)$, then 
\begin{align}
    \overline \chi = 2 \overline F - 2^{-n+1} + O(c).
\end{align}
We elaborate the setting and concrete noise model in \cref{app:low error rates}. 

We quantitatively observe this behaviour in numerical simulations of degree-$2$ IQP circuits with circuit-level noise, see \cref{fig:noisy xeb}. 
There, we compare the average fidelity with the average noisy XEB normalized by the ideal XEB score $\overline \chi/\overline \chi_{\mathrm{hIQP}}$. 
This \emph{normalized XEB} removes some finite-size differences between the average XEB and fidelity at very low noise rates and is therefore a better estimator for the average fidelity (see Refs.~\cite{choi_preparing_2023,ringbauer_verifiable_2024} for more details). 
At the same time, the normalization does not change the noisy behaviour qualitatively, i.e., in terms of the decay in terms of the noise and depth shown in \cref{fig:noisy xeb}(b,d).  
We find that for sufficiently low noise rates, the (normalized) average XEB and fidelity match extremely well. 
Together, the analytical and numerical results justify the use of the XEB as a proxy for the average fidelity at sufficiently low (Pauli) noise rates and sufficiently high depth $\Omega(\log n)$.

But how far does the correspondence between XEB and fidelity carry over to the high-noise regime? 
We find that the relationship between XEB and fidelity breaks down for high local Pauli noise rates and we qualitatively recover its behaviour in Haar-random circuits. 
Concretely, we consider the sparse degree-$D$ model with a layer of single-qubit noise following every gate, see \cref{fig:iqp circuits}(b). We analyze the most general noise model amenable to our proof techniques, which is a local $X$-$Y$ symmetric Pauli noise with $X$ and $Y$ noise rates $p_x = p_y \eqqcolon q_\perp$ and $Z$ noise rate $p_z$. 
Defining $q \coloneqq q_\perp + 2 p_z$ we have the following result. 
\begin{theorem}[XEB versus fidelity transition]
\label{thm:xeb bounds main}
    The average XEB for a local random degree-$D$ IQP circuits with $\ell $ uniformly random degree-$D$ gates alternated with $\ell$ layers of $X$-$Y$ symmetric Pauli noise is bounded as 
     \begin{align}
        \overline{\chi} \ge \,\,&  2^{1-n} [ (1-q_\perp)^\ell+(1-q)^\ell]^n \nonumber\\
        & \hspace{2cm}+ 2^{-D! \cdot \ell/n}[(1-q)(1-q_\perp)]^\ell,\\
        \overline{\chi}  \le \,\,&  2^{1-n}[1+(1-q)^\ell]^n + p(n) 2^{-2\ell/n} ,
    \end{align}
    where $p(n)= n^{11/2}/(2\pi)^{3/2}$.
\end{theorem}

We defer the proof of \cref{thm:xeb bounds main} to \cref{app:xeb bounds,ssec:extension_to_degree_3_circuits}.

The lower bound on the XEB shows directly that it has two regimes depending on the strength of the noise.  
When $(p_x+p_y+p_z) < 2  \log 2/n^2$, then the XEB is dominated by the exponential decay that is linear in the circuit volume $n \ell$.  Whereas for larger noise rates, the XEB decay is dominated by the decay that is linear in the effective depth $\ell/n$ (that one layer of gates requires $\ell = n/2$ in this model). This sharp change in behavior of the XEB is exactly analogous to what has been found in noisy Haar random circuits \cite{gao_limitations_2024,Morvan.2023,Ware.2023}.  
In both cases, it indicates that the XEB ceases to be good proxy for the fidelity, which always decays at a rate linear in the circuit volume until it reaches a value near $1/2^n$.  
The upper bound of the XEB shows that the two scaling regimes observed in the lower bound are actually tight at low and large enough noise rates.

One way to understand the deviation between the intuition we discuss above and the lower bound for high noise rates is to consider the random gate choices in the sparse IQP ensemble. 
In the model considered in \cref{thm:xeb bounds main}, for every one of the $\ell$ randomly chosen qubit pairs a CZ gate is applied with probability $1/2$. 
Think of a circuit comprising $\ell/n$ parallel gate layers, and consider an $X$-error occurring in the $(\ell/n -k)$-th layer of the circuit. 
Then the probability that no gate is applied in the last~$k$ layers decays as $2^{-k}$, leading to a deviation between XEB and fidelity on that order of magnitude. 

This reasoning suggests that the deviation should vanish if we ensure that every $X$-noise event in the circuit will meet a degree-$2$ or higher gate when we commute it to the end of the circuit, and that its resulting attached $Z$-type string is not cancelled by any other gates. 
As a first step towards this, we can add a fixed parallel layer of CZ gates between all qubit pairs at the very end of an otherwise random degree-$2$ circuit. 
This will only leave deviations due to cancellations of $Z$ strings.
We find that a fixed gate layer indeed has the effect of lowering the bounds on the average XEB from \cref{thm:xeb bounds main}, an effect that does not have an analogue for Haar-random circuits.

\begin{lemma}[Shifting the transition]
\label{lem:shift}
     The average XEB for the local random degree-2 IQP model on an even number $n$ of qubits with $\ell$ gates and $\ell$ layers of $X$-$Y$-symmetric Pauli noise with a last layer of parallel CZ gates between all neighbouring qubits is bounded as 
    \begin{align}
        \overline{\chi} &\ge 2^{1-n} [ (1-q_\perp)^\ell+(1-q)^\ell]^n \\
        &+ 2^{-4\ell/n}[(1-q)(1-q_\perp)]^{2\ell}. \\
        \overline \chi & \le 2^{1-n}[1+(1-q)^\ell]^n + p(n) 2^{-3\ell/n} . 
    \end{align}
\end{lemma}
We show \cref{lem:shift} in \cref{ssec:fixed gates}.
The new upper bound implies that the transition in the XEB has been shifted by a factor of $3/2$ to larger circuit depths. 

The bounds can be exponentially lowered further by adding additional, non-overlapping fixed gate layers to the circuit. 
More formally, this expectation can be understood via the proofs of the analytical results in this section. 
In these proofs, we make use of a mapping of second-moment quantities of degree-$2$ IQP circuits to a classical statistical model whose states are evolved deterministically by the circuit. 
In the following section, we will detail this statistical model, and sketch the proofs of our results regarding the ideal and noisy average XEB values.

Before we do so, let us briefly summarize the findings of this section. 
We have found that---as for Haar-random circuits---the linear XEB is a good estimator of the global state fidelity in the regime of relatively low noise. 
This is witnessed by our noisy simulations of hIQP circuits shown in \cref{fig:noisy xeb} as well as our tight bounds in \cref{thm:xeb bounds main}. 
That correspondence extends to very general types of noise. 
As for Haar-random circuits---the correspondence breaks down for large noise rates at which the XEB only decays with the circuit depth while the fidelity decays with the circuit volume. 
We have argued, however, that XEB should be a better estimator of fidelity in IQP circuits compared to Haar-random circuits because of the way the noise affects the circuit. 
This intuition is made rigorous in \cref{lem:shift}, which shows that fixed gates in the circuit shift the transition further towards higher noise rates until it eventually vanishes. 
Altogether, our results shine light on the differences and similarities between random IQP and Haar-random circuits, which we hope will further illuminate the underlying mechanisms.

An open question relevant to our findings is to what extent the noisy XEB score respects typicality or, in other words, how it behaves across different instances of the random circuits. 
For Haar-random circuits, it is well known that the XEB indeed concentrates rapidly around the mean and hence that typical instances are close to the mean value, since these circuits are higher designs.
Second, while we have shown that a transition exists, the precise nature of that transition remains an intriguing field for future studies. 
This is true in particular when considering the logical XEB and fidelity of encoded circuits subject to physical noise, as we will discuss it in the next section. 

Another question left open by our study is to understand the reasons for the transition between hyperdepth $\hdepth=1$ and $\hdepth\geq2$ circuits. 
We find this transition in the behaviour of the noisy XEB, as well as in measures of complexity of the underlying IQP graph. 
This matches the results of \textcite{Hashizume.2021} who find the same scrambling transition in a setting of only CZ entangling gates on the hypercube. 
A more detailed study of a similar scrambling transition in a sparse random model analogous to the random sparse IQP circuits we studied here is given in Ref.~\cite{kuriyattil_onset_2023}. 
But the exact mechanism and type of transition in the hypercube circuits remains unclear. 
From a naive interpretation of the statistical-mechanics mapping---elaborated in \cref{ssec:stat mech}---we expect an exponentially decaying deviation with every circuit layer from fully scrambled circuits rather than a sharp phase transition at a given point.

\subsection{Verification of noisy encoded IQP circuits}
\label{ssec:noisy encoded xeb}

\begin{figure}
    \centering
    \includegraphics{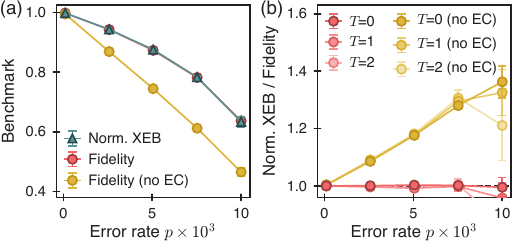}
    \caption{\textbf{Logical XEB and fidelity for an encoded circuit.}
    We simulate 30 instances of random noisy encoded degree-$2$ hIQP circuits ($\hdim\,{=}\,4,\,\hdepth\,{=}\,2$)  with the property that the ideal XEB value $\overline \chi_{\rm ideal}\,{=}\,3$, taking $10^5$ shots per circuit.  
    The circuits are subject to i.i.d.\ \emph{physical} single-qubit Pauli errors such that the single- and two-qubit gate fidelities are $1\,{-}\,p{/}10$ and $1\,{-}\,p$, respectively. Error bars represent one standard deviation. (a) The average logical fidelity \eqref{eq:average fidelity} and  logical XEB \eqref{eq:def xeb} normalized by $\overline \chi_{\rm ideal}$ postselected on perfect $X$-basis stabilizers with and without error correction (EC) performed in the internal state representation. (b) Relative value of the average normalized XEB and fidelity  for various $X$-axis postselection thresholds $T$.
    }
    \label{fig:encoded_xeb_fidelity}
\end{figure}

In the previous section, we showed that the XEB and fidelity in IQP circuits match for local logical noise with low rates. 
We now assess this correspondence for encoded IQP circuits subject to physical noise. 
For encoded qubits, there is an ambiguity in defining the logical fidelity in relation to the XEB: 
The XEB is evaluated directly from measurement outcomes in the $X$-basis and is not affected by (physical and logical) $X$ errors just before the measurement, but the fidelity can be drastically reduced. 
This reduction is determined by the physical error density and happens because we are not applying any correction in the $Z$ basis and therefore the state after $X$ measurements and postselection remains outside of the code space. 
Therefore, a genuine comparison, akin to that from the previous section, necessitates incorporation of an error-correcting procedure in the fidelity definition that leaves a valid code state.

In this procedure, we would typically assume a perfect round of stabilizer measurements and decoding. 
However, in the sampling protocol we consider, the $Z$-stabilizers are not measured and a fault-tolerant protocol using the $Z$-stabilizer information might change the sampling output. 
Therefore, a meaningful notion of logical fidelity to compare with the XEB should only consider corrections that do not change the logical $X$ basis samples.
This is why we define the reference logical fidelity as a state overlap on the level of logical operators after returning the state to the $Z$-basis code space by performing a virtual (in the simulation code) error correction step that does not modify the XEB score. 
To achieve this, we use a lookup-table decoder constructed by enumerating low-weight $X$ errors in an $\code{8,3,2}$ code.

We simulate random encoded degree-$2$ hIQP circuits (with the property that the ideal XEB value $\chi_{\rm ideal}\,{=}\,3$) at experimentally realistic two-qubit gate noise rates around $99.5\%$~\cite{evered_high-fidelity_2023}.
In \cref{fig:encoded_xeb_fidelity} we then compare the logical XEB with the logical fidelity with and without correction of the $X$ errors prior to the $X$ measurement. 
We observe that for the corrected logical fidelity the correspondence with XEB is tight and consistent with \cref{fig:noisy xeb}. 
In contrast, the deviation between XEB and fidelity increases when comparing to the state fidelity without error correction.

These results show that intermediate-size encoded circuits with experimentally relevant noise rates exhibit a close correspondence between logical XEB and an appropriate notion of logical fidelity, i.e., they are in the ``healthy regime'' of the correspondence between XEB and fidelity. They also point to some of the inherent subtleties in assessing the performance of encoded circuits where the syndrome qubits necessitate context-dependent definitions of logical process, state preparation, and measurement fidelities.

\subsection{A statistical model for second moments of random IQP circuits}
\label{ssec:stat mech}

In order to derive the results on the noisy and ideal linear XEB score, we derive in this section a statistical-mechanics mapping for second-moments of degree-$2$ IQP circuits analogous to those used in the analysis of locally Haar-random circuits \cite{zhou_emergent_2019,hunter-jones_unitary_2019,dalzell_random_2022,Ware.2023}. We provide a broad overview of the properties and behavior of our mapping in \cref{fig:statmech}.

We begin by observing that the ideal and noisy linear XEB score is a second-moment quantity of the IQP circuit. 
To be precise, we define the \emph{second-moment operator} of a family $\mc I$ of IQP circuits as the two-copy projector 
\begin{align}
    M_2 \coloneqq \mb E_{C \sim \mc I} C^{\otimes 2} \proj{+^n}^{\otimes 2} C^{\otimes 2}. 
\end{align}
We can also allow for noise in one or both copies of the circuit; 
denote by $\mc N$ a noise model specified as a sequence of quantum channels after each gate in $C$ and by $\rho_C(\mc N)$ the noisy state induced by the application of $C \in \mc I$ and the noise model. 
Then the noisy second-moment operator is just given by $M_2(\mc N, \mc N') =\mb E_{C \sim \mc I} \rho_C(\mc N) \otimes \rho_C(\mc N')$. 
With this definition at hand, we can immediately see that the (noisy) XEB is given by 
\begin{align}
    \overline \chi = 2^n \sum_x \bra x^{\otimes 2} M_2(\{\id\}, \mc N) \ket x ^{\otimes 2} -1, 
\end{align}
and hence we can derive all its properties from the (noisy) second moment operator. 
The statistical-mechanics model will describe the evolution of the second-moment operator as a function of the depth of the IQP circuit. 

The ideal XEB score of IQP circuits has been computed for uniform degree-$2$ circuits by \textcite{Bremner.2016}, while the general moment operator for degree-$n$ circuits has been found by \textcite{nechita_graphical_2021}.
The statistical-mechanics model will allow us to go beyond both settings and analyze arbitrary circuits composed of random degree-$2$ gates, as well as arbitrary IQP and CNOT gates and (potentially distinct) Pauli noise on each copy. 
To motivate our approach, we build on the derivation of the ideal XEB score of uniform degree-$2$ IQP circuits by \textcite{Bremner.2016}. 
Their argument makes use of the fact that the output states of degree-$2$ IQP circuits can be written as a \emph{polynomial phase state}
\begin{align}
    C \ket {+^n} = \frac{1}{\sqrt{2^n}} \sum_{x \in \bin^n} (-1)^{f(x)} \ket x, \label{eq:poly_phase_state}
\end{align}
where $f(x) = \sum_{ij} b_{ij} x_i x_j + \sum_i a_i x_i$ is a degree-$2$ Boolean polynomial whose coefficients $a_i$ ($b_{ij}$) are 1 if a degree-$1$ (degree-$2$) gate is applied to the qubits labeled by the indices in the circuit $C$, and 0 otherwise. 
We can then write the two-copy average over such phase states as 
\begin{align}
    M_2 & =  \frac 1 {2^{2n}}\sum_{z_1, \ldots, z_4} \mb E_{f} \left[(-1)^{\sum_i f(z_{i})}\right] \ket{z_1 z_3} \bra{z_2z_4}\\
    & = \frac 1 {2^{2n}} (\id^{\otimes n} + \mb S^{\otimes n} + \mb P^{\otimes n} - 2 \mb X^{\otimes n}), \label{eq:bremner iqp formula}
\end{align}
where $\mb S = \sum_{a,b\in \bin} \ket {ab}\bra{ba}$, $\mb P = \sum_{a,b\in \bin} \ket {aa}\bra{bb}$, and $\mb X = \sum_{a\in \bin} \ket {aa}\bra{aa}$. 
This follows from evaluating conditions under which the averages over the phase $(-1)^{\sum_i f(z_i)}$ does not vanish for any fixed configuration $z_1, z_2, z_3, z_4$, through the randomness of the coefficients~$a,b$. 

The formula \eqref{eq:bremner iqp formula} is reminiscent of the one for global Haar-random circuits, where we find $M_2 \propto (\id^{\otimes n} + \mb S^{\otimes n})$, but additional invariances appear in the case of IQP averaging. 
In fact, one can show that the action of a single-qubit Haar-random gate in the two-copy average is to project onto the space spanned by $\id$ and $\mb S$. 
If a quantum circuit contains layers of single-qubit Haar-random gates one can therefore describe the average two-copy evolution as an evolution of these two `states', which represent the local invariances of the circuit.

\begin{figure*}
  \centering
  \includegraphics{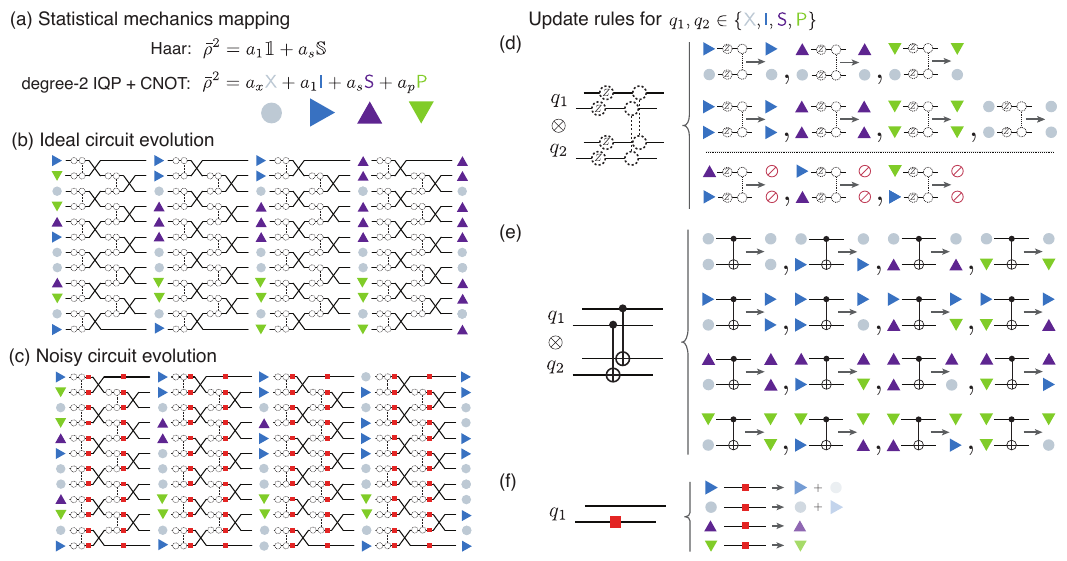}
  \caption{\label{fig:statmech} \textbf{Statistical-mechanics mapping of degree-$2$ second moments.} (a) To analyze the ideal and noisy XEB of degree-$2$ IQP circuits we map the second moment $\overline \rho^2  = \mb E_C \rho_C\otimes \rho_C$ of IQP + CNOT circuits with random degree-$2$ circuits on subsets of qubits to the evolution of four classical states $\xx, \ix, \sx, \px$ at each physical site.
  This mapping is analogous to the mapping used for Haar-random circuits, where the invariant states correspond to the identity and swap operator between the two copies, $\id$ and $\mb S $, respectively
  (b) Under a local random circuit in which random degree-$2$ circuits are applied to random qubits or alternated with layers of SWAP gates the system (depicted here), the system will percolate to states with only one $Q \in \{\ix, \sx, \px\}$ and $\xx$ states otherwise, yielding a globally invariant state in $\{Q,X\}^{n}$ for $Q \in \{ \ix, \sx, \px\} $.  The intermediate states shown here represent typical states at the respective circuit layer since under the ideal circuit evolution states only survive or die.
  (c) Under the noisy circuit evolution, all but the $\xx$ and $ \ix$ state eventually die out. 
  (d-f) Update rules under random and fixed gates. (d) Subjected to a random IQP circuit a state on two sites either survives (if it contains $\xx$ or is of the form $Q Q $ for $Q \in \{ \ix, \sx, \px\}$) or `dies' otherwise (crossed out circles). 
  (e) Under a deterministic CNOT gate a state $PQ $ is permuted to a state $P \varepsilon_P(Q)$.
  (f) To analyze the noisy value of the linear XEB, we need to apply noise to one copy of the circuit. Under $X$-$Y$ symmetric Pauli noise the amplitude of all states is reduced, while the total amplitude of $\xx + \ix $ is conserved. 
  } 
\end{figure*}

The same effect occurs for degree-$2$ IQP circuits. 
We find that a good basis for the local state space of degree-$2$ circuits is given by $\mc S = \{\xx, \ix, \sx, \px \}$ with 
\begin{align}
    \xx & \coloneqq \proj {00} + \proj {11},\\
    \ix &\coloneqq \proj {01} + \proj {10},\\
    \sx &\coloneqq \ket {01}\bra {10} + \ket {10} \bra{01},\\
    \px &\coloneqq \ket {00}\bra {11} + \ket {11} \bra{00}. 
\end{align}
Intuitively, this is because these states are invariant under swaps between the two copies and contain two $\ket 1$ states each so that a $Z$ gate on two copies leaves them invariant. 
We illustrate the statistical model in \cref{fig:statmech} and describe it in detail in \cref{app:statmech}. 

The second-moment operator expressed in the $\mc S$-basis undergoes simple evolution under random degree-$2$ gates. 
Without loss of generality, we can move a layer of random single-qubit $Z$ gates to the beginning of the circuit in all of the models we study.
Writing $Z^z = \prod_i Z_i^{z_i}$ we can then write the second moment of the input state
\begin{align}
   \mb E_{z\in \bin^n} (Z^z \proj{+^n} Z^z)^{\otimes 2} = \frac{1}{4^{n}} (\ix + \sx + \px + \xx)^{\otimes n},
\end{align}
which can be expressed as a uniform mixture of all states in~$\mc S^{\otimes n}$. 
If we apply a random degree-$2$ circuit $C$ to a pair of qubits, we find that a state $P \otimes Q \in \mc S^{\otimes 2}$ evolves as
\begin{multline}
\label{eq:iqp update main}
    \mb E_C C^{\otimes 2} (P \otimes Q) C^{\otimes 2} \\= \begin{cases}
        P\otimes Q & \text{ if }   P = \xx \text{ or } Q = \xx \text{ or } Q = P \\
        0 & \text{ else,} 
    \end{cases} 
\end{multline} 
see also \cref{fig:statmech}(d).
In other words, the states of our model can only survive or die under random IQP circuit evolution.

Let us pause here for a moment and consider the evolution of the states. 
In particular, we would like to understand under which condition the second moments of a sparse random degree-$2$ circuit converges to its uniform value \eqref{eq:bremner iqp formula} as a function of depth. 
It is apparent from \cref{eq:iqp update main} that the $\xx$-state plays a special role in the model, and that the model tends towards completely `polarized states' in which the only appearing states are $\xx$ and one of $\ix, \sx, \px$.  
These states are invariant under the circuit evolution---they are immortal.
We can understand the behaviour of the circuit evolution to the completely polarized states $\mc S_{\text{imm}} = \{\xx,\ix\}^{\otimes n}\cup \{\xx,\sx\}^{\otimes n}\cup\{\xx,\px\}^{\otimes n}$ in terms of percolation, see \cref{fig:statmech}(b). 
As soon as the degree-$2$ gates form a large connected component of the graph with high probability most states outside of $\mc S_{\text{imm}}$ will have died so that the second moment approximately recovers the form \eqref{eq:bremner iqp formula}.
It is at this point that we expect anticoncentration, since every state contributes exactly $2^{-n}$ to the overall XEB. 
Indeed, in \cref{thm:sparse iqp anticoncentration main} we find a sharp transition as the number of random degree-$2$ gates exceeds $2.75\, n \log n$ which is precisely the scaling we expect from the percolation argument \cite{broadbent_percolation_1957,bollobas_percolation_2006}.

Next, we notice that $C^kZ$ gates leave the immortal states in $\mc S_{\text{imm}}$ invariant. 
However, they can have a non-trivial effect on states outside of that set that leaves the space spanned by those states.
For instance, a CCZ gate maps $\ix \sx \xx$ state to $\ix \otimes \sx \otimes \proj {00} - \ix\otimes \sx\otimes \proj {11}$.
In our circuit ensembles, we always consider random higher-degree gates in conjunction with uniformly lower-degree gates on the same set of qubits, however.
They therefore cannot alter the invariant state. 
Indeed, even degree-$n$ circuits have the same second moment as degree-$2$ circuits~\cite{nechita_graphical_2021}. 
Generally, adding additional higher-degree random gates can only speed up anticoncentration. 

Furthermore, we find that the action of a CNOT gate on a string $P \otimes Q$
\begin{align}
    \mb E_C \cnot^{\otimes 2} (P \otimes Q)\cnot^{\otimes 2} = 
        P\otimes \varepsilon_P(Q), 
\end{align}
is just a permutation $\varepsilon_P(Q) \in \mc S$ that depends on the control of the CNOT gate, illustrated in \cref{fig:statmech}(e); see \cref{lem:cnot update} in \cref{app:stat mech mapping} for details. 
In combination with local random degree-$D$ circuits on fixed qubits, CNOT gates thus have a similar effect compared to random degree-$D$ circuits on random subsets of qubits, since they transport `excitations', i.e., non-$\xx$ states through the system.
We exploit this property when we compute the asymptotic XEB score of block-random IQP and hIQP circuits in \cref{thm:anticoncentration deep hiqp}. 

Finally, consider the effect of noise. 
Noise is reflected in an asymmetry between the two circuit copies, wherein---in the case of Pauli noise---a Pauli operator is applied somewhere in the circuit on one copy but not the other. 
We find that Pauli $X$ and $Z$ noise leave the states $\xx, \ix$ invariant, and have a nontrivial effect on $\sx, \px$: 
$X$ noise permutes $\sx \leftrightarrow \px$, while $Z$ noise adds a $-1$ phase to $\sx$ and $\px$. 
If $X$ noise occurs during the circuit evolution, it therefore maps some immortal states to states which will eventually die.
$Z$ noise can be moved to the end of the circuit and adds a phase to $\sx$ and $\px$ states. 
In the XEB, $Z$ noise thus leads to a cancellation of the contributions of states with nontrivial $\sx$ and $\px$ contribution at the end of the circuit. 
$X$ noise on the other hand leads to a decay of all states outside of the set $\{\xx, \ix\}^{\otimes n}$. 
If we therefore consider Pauli noise which is $X$-$Y$ symmetric, we find a decay of the $\sx,\px$ states and a stochastic permutation of the $\xx,\ix$ states, see \cref{fig:statmech}(f). 
Under noisy degree-$2$ circuit evolution, the system thus undergoes a damped percolation to $(\xx + \ix)^n$, while all other states die out, eventually leading to an XEB score of $0$, see \cref{fig:statmech}(c). 

In \cref{thm:xeb bounds main} we compute upper and lower bounds to the noisy XEB of sparse IQP circuits with Pauli noise during the circuit evolution as a function of the circuit depth. 
To show the theorem, we find a set of long-lived states of the form 
\begin{align}
\label{eq:long lived states}
    \{\ix, \xx\}^{\otimes(n-2)} \otimes \{\ix\} \otimes \{\sx\},
\end{align}
i.e., states with a single `excitation' which slowly decays. 
It is these states which lead to a deviation between XEB and fidelity for circuits with depth $\ell \in O(n \log n)$. 

At the same time, the contribution of those states to the XEB vanishes if we consider a last, fixed layer of CZ gates. 
To see this, observe that $\ix \sx \xmapsto{\rm CZ} -\ix \sx$ while $\xx \sx \xmapsto{\rm CZ} \xx \sx$. 
Exactly half the states of the form \eqref{eq:long lived states} will incur a negative sign at the end of the circuit since a CZ gate is guaranteed to act on the excitation $\sx$.
In order to lower-bound XEB, we now need to consider states with two excitations, e.g., $\px \px\ix^{n-2} $. This reduces the upper and lower bounds by a factor of $3/2$ and $2$ in the exponent, respectively.

To summarize, our statistical model can be used to qualitatively understand and quantitatively analyze the behaviour of second-moment properties of IQP circuits. 
We have done so here for the noisy and ideal values of the XEB score in the hIQP and sparse IQP model, but other architectures can also be considered. 
An interesting case to study in future work is, for instance, that of geometrically local IQP circuits. A convenient model to study in this case is an IQP + SWAP model wherein degree-2 IQP circuits act on fixed pairs of qubits while the qubits are coupled by SWAP gates acting in a geometrically local fashion.
We expect the behaviour in this setting to be quite similar to that of Haar-random circuits, in particular, we expect anticoncentration in logarithmic depth~\cite{barak_spoofing_2021,dalzell_random_2022}.

\section{Validated advantage via transversal Bell sampling}
\label{sec:bell sampling}

In the previous section, we have argued that sampling from random degree-$D$ hIQP and sIQP circuits for $D \geq 3$ is classically intractable, and that the linear XEB is a good measure of the output state fidelity in the regime of low noise. 
Together, these results allow us to perform verified quantum advantage experiments in realistic settings. 
The drawback of this approach, however, is the fact that estimating the XEB requires evaluating ideal outcome probabilities, which by definition are classically hard to compute in the quantum advantage regime. 

This raises the question whether there are more efficient ways to validate the quantum advantage of IQP sampling.  
In contrast to universal circuits, it is known that the output states of IQP circuits can be verified using the ability to perform (near-)perfect single-qubit measurements for certain families. 
These include IQP circuits with two-local rotations around arbitrary angles~\cite{bermejo-vega_architectures_2018,hayashi_verifying_2019} and also general degree-$3$ circuits which can be thought of as preparing hypergraph states~\cite{morimae_verification_2017,zhu_efficient_2019}. 
Such tests require very high-fidelity state preparations, however. 
Since the stabilizers of IQP circuits can be computed efficiently, it is also possible to perform direct fidelity estimation~\cite{flammia_direct_2011}, which is highly sample-efficient, but beyond geometrically local architectures~\cite{ringbauer_verifiable_2024} requires potentially long-range entangled measurements whose correctness must be guaranteed. 
All of these approaches furthermore require different measurements for sampling and verification.

An alternative, unified approach to sampling and verification is to make use of a Bell measurement on two copies of an IQP circuit \cite{hangleiter_bell_2024}. 
For universal random circuits, this gives rise to a sampling task which is classically intractable to simulate while at the same time the output distribution has certain structure that can be exploited for efficient validation of the samples as well as learning properties of the underlying quantum circuit. 
While this test can be efficiently spoofed by an adversary, it can be used to efficiently validate samples from a noisy quantum device in a `trusted' laboratory environment. 

The output probabilities of a Bell measurement on two copies of a quantum state $\ket C$ are given by 
\begin{align}
    P_C(x, z) = \frac 1 {2^n} | \bra C X_{x} Z_{z} \ket C|^2, 
\end{align}
where by $X_{x} = \prod_{i=1}^n X_i^{x_i}$ and likewise for $Z$ we denote the Pauli $X$ and $Z$ operator on the sites labeled by $x,z \in \bin^n$, respectively. 
Since a degree-$D$ IQP circuit $C$ is in the $D$-th level of the Clifford hierarchy we thus find that the Bell sampling probabilities of degree-$D$ IQP sampling are given by the output probabilities of certain degree-$(D-1)$ IQP circuit $C_{x}$ that depends on $X^x$. 
Concretely, we find 
\begin{align}
\label{eq:bell expression circuit}
    P_C(x,z) = \frac 1 {2^n} | \bra {+^n} C_x Z_z \ket{+^n} |^2.
\end{align}
The gates of $C_x$ derive from the gates in $C$, observing that for $x \in \bin^{k+1}$, $\text C^k \text Z\cdot X_1 \cdot \text C ^k \text Z = X \otimes \text C^{k-1}\text Z$. 
In other words, a degree-$k$ gate which meets one or several $X$ operators will reduce to a lower degree gate on the complement sites. 
All other gates cancel. 
Thus, the highest-degree gates in $C_x$ are those gates in $C$ which meet a single $X$ operator and are therefore reduced by one degree. 

It follows that computing the output probabilities of Bell sampling with degree-$3$ circuits, which is hard in the standard ($X$) basis, becomes easy since all amplitudes are just amplitudes of various degree-$2$ Clifford circuits. 
However, if we consider Bell sampling from degree-$4$ circuits, we recover hardness of sampling. 
\begin{theorem}[Hardness of degree-4 Bell sampling (informal)]
\label{thm:hardness bell sampling main}
    Approximately sampling from the output distribution of Bell measurements on uniformly random degree-$D$ IQP circuits for $D \geq 4$ is classically intractable under certain complexity-theoretic assumptions. 
\end{theorem}
We elaborate the theorem more formally as \cref{thm:bell sampling hardness} in \cref{app:bell sampling}. 

Let us mention here that in contrast to Bell sampling from universal random circuits \cite{hangleiter_bell_2024}, the fact that degree-$D$ circuits lie in the $D$-th level of the Clifford hierarchy introduces significant structure in the output distribution of degree-$D$ Bell sampling. 
In particular, the output distribution `shatters' into $2^n$ sectors determined by the $x$ outcome, and for each such sector, the probabilities are the output probabilities of the smaller IQP circuit $C_x$. 
In terms of the statistics of uniform degree-$D$ Bell sampling, we thus find $n+1$ sectors depending on the number $|x|$ of nonzero $x$ outcomes which determine the support of $C_x$.
While within each sector the statistics are the same, they differ between the sectors.
This is reminiscent of the situation in Gaussian boson sampling, where the distribution falls into distinct sectors determined by the total photon number~\cite{kruse_detailed_2019,deshpande_quantum_2022,ehrenberg_transition_2024}. 
Similarly, in boson sampling with a low number of modes the statistics of the outcomes differ depending on the `collision pattern', that is, the number and size of the collision outcomes~\cite{bouland_complexity-theoretic_2023}.  

In~\cref{app:bell sampling}, we detail an adapted hardness proof based on Stockmeyer's counting algorithm~\cite{stockmeyer_complexity_1983}for this scenario. 
A key observation in the proof is that the probability weight is concentrated on outcomes $(x,z)$ for which $|x| = n/2 \pm O(\sqrt n)$, and therefore hardness of computing outcome probabilities still holds for most outcomes. 

For verification of the Bell samples we observe that the degreee-$D$ Bell distribution is supported only on outcomes $(x,z) \in \bin^{2n}$ for which there is no index $i$ such that $1 = x_i = z_i$. 
In the terminology of Ref.~\cite{hangleiter_bell_2024}, such indices correspond to $Y$-outcomes or singlets.  
Let $\supp_Y(x,z)\coloneqq \{i \in [n]: x_i = z_i = 1\}$ be the $Y$-support of $(x,z)$ and $\pi_Y(x,z) = |\supp_Y(x,z)| \mod 2$ be its parity.
For universal Bell sampling, the output distribution is supported on outcomes with even $Y$-support $\pi_Y(x,z) = 0$ \cite{hangleiter_bell_2024}. 
Here, we find $|\supp_Y(x,z) = 0|$ for all outcome strings. 

In case the state preparation is noisy, it is discussed in Ref.~\cite{hangleiter_bell_2024} how and in which scenarios the Bell samples may be used to estimate the average fidelity of random two-copy state preparations. 
The key element of this verification protocol is to estimate the purity 
\begin{multline}
    \tr[\rho^2] \approx \frac 1 M \big(|\{(x^i,z^i): \pi_Y(x^i,z^i)=0\}| \\- |\{(x^i,z^i): \pi_Y(x^i,z^i)=1\}| \big),
\end{multline}
from a set of $M$ Bell samples $\{(x^i,z^i)\}_i$. 
Importantly, this estimate can be computed fully efficiently from the Bell samples. 
The outcomes with odd $Y$-parity thus indicate errors in the circuit. 
In degree-$D$ IQP sampling we have an additional mechanism to detect errors which is given by the even $Y$-parity outcomes. 
We leave a detailed study of the possibility to exploit this additional structure in the output distribution in order to diagnose errors in degree-$D$ Bell sampling to future work.  

Suffice it to stress here that transversal degree-$4$ Bell sampling using $\code{2^4,4,2}$ color codes provides a viable route towards demonstrating efficiently validated quantum advantage. 
Efficient validation is possible using the additional structure afforded by the Bell samples.
The price at which the efficient validation comes is twofold. 
First, as in universal Bell sampling, we need to double the number of qubits and perform a Bell measurement across the two circuit copies. 
This is feasible in the setting considered here since the Bell measurement is transversal in stabilizer codes. 
Second, there is an additional price to be paid in the specific setting of degree-$D$ Bell sampling:
As mentiond above, most outcomes $(x,z)$ have support $|x| = n/2 \pm O(\sqrt n )$. 
By \cref{eq:bell expression circuit} this implies that the naive simulation cost scales exponentially in the worst case as $n/2$ compared  to $n$. 
On the flipside, that additional price is counterbalanced by doubling the amount of error detection that is possible, since the Bell distribution is guaranteed to have zero weight on $Y$-outcomes.

To summarize, degree-$(D+1)$ Bell sampling using transversal logical circuits as discussed here is a feasible route towards efficiently validated quantum advantage that requires a linear overhead  in the number of qubits compared to standard degree-$D$ sampling. 
This overhead is given by  a factor of four due to the structure of the Bell distribution times a factor of 2 due to the larger codes. 
This approach can be compared to other means of efficiently verified quantum advantage using structured states such as cluster states \cite{bermejo-vega_architectures_2018,ringbauer_verifiable_2024}, where the quadratic overhead can be thought of as a space-time tradeoff. 
The space-time-overhead might be traded for intermediate \emph{logical} stabilizer measurements in the context of encoded computations. 
An alternative route towards error-detected or -corrected quantum advantage could therefore be to use three-dimensional cluster states~\cite{raussendorf_fault-tolerant_2006}. 

\section{Generalized codes for scalable computation}
\label{sec:scaling codes}

The $\code{8,3,2}$ code is the smallest realization of the 3D color code on a cube-like region with a logical $\overline{\rm CCZ}$ gate implemented via a transversal $T$ gate~\cite{Kubica.2015}.\footnote{For the discussion in this section, we will emphasize the fact that a gate $G$ is a logical gate using the standard notation $\overline G$ in contrast to physically implemented gates. }
Alternatively, the $\code{8,3,2}$ code can be viewed as the $D=3$ instance of the $D$-hyperoctahedron code~\cite{Vasmer2022} with parameters $\code{2^D,D,2}$ for any $D\geq 2$.
In what follows, we present four constructions that lead to codes with similar transversal implementations of a logical $\overline{\rm CCZ}$ but an increased code distance $d\geq 4$.  Before describing them in more detail, we first provide an overview of each code and their advantages and disadvantages.  

In the first code construction, we concatenate a two-qubit $Z$-repetition code into the $\code{2^D,D,2}$ code to form a $\code{2^{D+1},D,4}$ code.  The simplicity of this approach makes it suitable for near-term devices; however, this simplicity comes at the cost of performance because the transversal IQP gates now have to be implemented via two-qubit gates.  Furthermore, fault tolerance in this code requires some amount of error detection, so it is not a fully fault-tolerant error correction scheme. We compare the performance of this code to the $\code{8,3,2}$ code, as well as IQP circuits implemented via transversal gates on the $\code{15,1,3}$ code.  For a physically reasonable noise model, our conclusion is that the $\code{8,3,2}$ code has the best error detection performance, while the $\code{15,1,3}$ code can be used to reduce the postselection overhead of error detection through error correction at low noise rates.

The next family of codes we consider is a direct generalization of the $\code{2^D,D,2}$ code to a higher distance color code with parameters $\code{2^{2D},D,4}$.  The qubit overhead increases substantially compared to the $\code{2^{D+1},D,4}$ concatenated codes, but the IQP gates can now all be implemented via transversal single-qubit gates.  These are the minimal size color codes with distance $\ge 3$ we have found with transversally implemented non-Clifford IQP gates, but the downside of the construction is that scaling to higher distance leads to macroscopic stabilizers that compromise fault-tolerance performance. 

These considerations lead us to the third family of color codes we consider, which maintain the LDPC property in their check operators as the code distance is scaled.
To achieve this favorable property the qubit overhead (asymptotically) increases by a factor of 5 compared to the previous color code construction with macroscopic stabilizers.  

The final family of codes we consider are variants of the three-dimensional toric codes. In the case of the standard 3D toric code with certain boundary conditions, one can implement a logical $\overline{\text{CCZ}}$ gate via transversal physical CCZ gates on each code block.  Unfortunately, this is a lower-fidelity physical operation than single-qubit rotations. 
To overcome this constraint, we propose to concatenate $\code{8,3,2}$ codes with three copies of the 3D toric code so that the local $\overline{\text{CCZ}}$ gates can be implemented via single-qubit gates on the $\code{8,3,2}$ code blocks.
This method requires 72 qubits to realize a distance-4 code.

Implementing degree-$D$ IQP circuits in these codes can be achieved with an overhead in physical circuit depth of at most $d$. 
If these codes permit single-shot decoding along the lines of~\cite{Bombin15}, or algorithmic fault-tolerance with transversal non-Clifford gates~\cite{zhou_algorithmic_2024} it can even be reduced to constant depth overhead, since all required gates and measurements are transversal.

\subsection{The $\code{16,3,4}$ and $\code{15,1,3}$ codes}
\label{sec:code comparison}

As a CSS code the $\code{8,3,2}$ code has an unbalanced $X/Z$-distance of $4/2$.  Thus, one way to improve the code distance is to increase the $Z$ distance by replacing each qubit by a two-qubit repetition code with the stabilizer group generated by the two-site $XX$ operator.  This $\code{16,3,4}$ code remains a CSS code with transversal CNOT gates and $X/Z$ measurements; however, the in-block $\overline{\text{CZ}}$ and $\overline{\text{CCZ}}$ gates must now be implemented via transversally applied two-qubit $e^{i \theta ZZ}$ gates.  
The distance of a full transversal circuit thus remains 2, but, depending on the noise model, one may find improved performance of this code.

To test whether such an improvement is possible in our transversal IQP sampling circuits, we compare the performance of the $\code{16,3,4}$ code to the $\code{8,3,2}$ code in Fig.~\ref{fig:code_comparison}.
We use a simplified noise model with perfect state preparation of the logical code blocks, perfect single-qubit gates, and noisy physical two-qubit gates, i.e., after each two-qubit gate we apply i.i.d. depolarizing noise with rate $p$ to each qubit involved in the gate.   While the results of these simulations might vary with the noise model, we expect the qualitative conclusions of our overall analysis to remain similar.
In the last round of the $\code{16,3,4}$-code protocol, we perform error correction.
We also perform postselection depending on the measurement outcomes.  
To implement error correction in the $\code{16,3,4}$ code we use a decoder that corrects all syndromes generated by weight-1 errors with the associated weight-1 error and, using postselection, discards samples with syndromes outside this set. 
Figure~\ref{fig:code_comparison}(a) shows a comparison between the error corrected $\code{16,3,4}$ code and error detection in the deviation of the logical XEB from 1.  
Both codes have an error rate that scales at low-noise rates as $p^2$, consistent with expectations; however, we notice that the error detected $\code{8,3,2}$ code outperforms the corrected $\code{16,3,4}$ code and matches the performance of the $\code{16,3,4}$ code with full error detection.  Figure \ref{fig:code_comparison}(b) then shows the performance of the codes in terms of the postselection overhead.  In this noise model, the $\code{16,3,4}$ code with error correction still has a larger postselection overhead in our decoder than the $\code{8,3,2}$ code.  As a result, when the noise is dominated by two-qubit gate errors, we do not find evidence that the $\code{16,3,4}$ code offers improved performance  over the $\code{8,3,2}$ code, despite its improved distance.  

\begin{figure}
    \centering
    \includegraphics{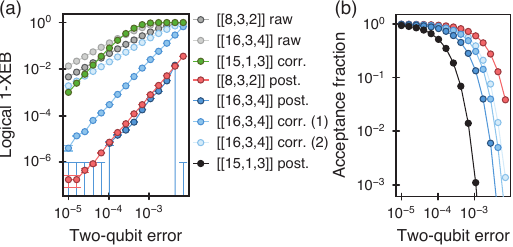}
    \caption{\textbf{Comparison with the $\code{16,3,4}$ and $\code{15,1,3}$  codes.} The same  hIQP logical circuit instance ($\hdim\,{=}\,4$, $\hdepth\,{=}\,2$) is executed using the $\code{8,3,2}$, $\code{16,3,4}$, and $\code{15,1,3}$ encodings. 
    The state preparation (encoding circuit) is noiseless, and the transversal part of the circuit is realized with perfect single-qubit gates and two-qubit gates with fidelity $1-p$, implemented as independent single-qubit depolarizing channels on the affected sites. The error bars represent standard deviation of the mean from $10^7$ samples and are within markers for most data points. (a) The deviation from 1 of the normalized XEB $\overline \chi/\overline \chi_{\text{ideal}}$ for various operation modes of the studied encodings. The \textit{raw}, \textit{corr.}, \textit{post.} denote the raw XEB (for a fixed logical operator set), XEB after a correction step, and XEB after postselection, respectively. In the case of the $\code{16,3,4}$, \textit{corr.\,(1)} corresponds to a decoder that corrects all distance-1 errors and discards all samples with syndromes that it cannot correct while \textit{corr.\,(2)} implements a sliding-scale postselection protocol that discards a sample if two or more uncorrectable syndromes appear in the entire system.  The \textit{corr.\,(1)} decoder has fidelity below one since logical errors occur that give the all-zero syndrome but reduce the fidelity.  Once the log-log-decay for the $\code{15,1,3}$ code becomes linear below the pseudo-threshold, it is parallel to the dark blue and red lines, which decay quadratically.  (b) Acceptance fraction for protocols utilizing some degree of postselection.}
    \label{fig:code_comparison}
\end{figure}

To demonstrate an improvement over the $\code{8,3,2}$ code in this noise model, we instead turn to transversal IQP circuits implemented with the punctured Reed-Mueller $\code{15,1,3}$ CSS code that has a logical $\overline T$-gate implemented by transversal $T$-gates.  
In addition, to transversal $T$-gates, this code also has transversal logical $\overline{\text{CS}}$ gates implemented via transversal CS gates between code blocks~\footnote{This follows from the fact that CNOT is transversal for CSS codes since $\text{CS} = \text{CNOT} (\id \otimes T^\dagger) \text{CNOT}(\id \otimes T)$ \cite{koutsioumpas_smallest_2022}.}.  Thus, we can realize a similar transversal IQP sampling experiment with the $\code{15,1,3}$ code using physical one and two-qubit gates by doing the same CNOT circuit on a hypercube interspersed with block-random $\overline{\text{CS}}$ gates on subsets of $\code{15,1,3}$ code blocks.   
In this way, we expect to achieve very similar hardness results and circuit families as the degree-$3$ circuits we studied using the $\code{8,3,2}$ code, see Ref.~\cite{Bremner.2017}.  
In \cref{fig:code_comparison}(a), we compare the performance of the error corrected $\code{15,1,3}$ circuit implemented in the same architecture as the $\code{8,3,2}$ code circuit, both with 48 logical qubits.  In the comparison both circuits are implemented with only Clifford gates ($\overline Z$, $\overline{\text{CZ}}$, $\overline{\text{SWAP}}$ and $\overline{\text{CNOT}}$ gates) to make classical simulations feasible.  
Although the $\code{15,1,3}$ code removes the postselection overhead, we see that its pseudothreshold behavior for this circuit is below error rates of $10^{-3}$ and fails to go beyond break-even for the uncorrected $\code{8,3,2}$ circuits until error rates below $10^{-4}$.  Applying postselection to the $\code{15,1,3}$ circuit data boosts the estimated logical fidelity to nearly~$1$ for these noise rates and samples\footnote{In the $10^7$ samples we generated to obtain \cref{fig:code_comparison}, we did not observe a single weight-$3$ error after postselection.}, but the postselection overhead is much more severe than the other codes.  
It is likely that these results could change significantly with repeated rounds of mid-circuit error detection or correction, but that comes with additional qubit overhead.  For the present model with error detection/correction implemented only in the final layer of measurements, we find the $\code{8,3,2}$ code to be the highest-performing code.

Finally we remark that there is a closely related code to the $\code{15,1,3}$ code, but with a higher code rate and distance.  
It is a $\code{64,15,4}$ code where transversal $T$ gates implement a 15-qubit logical $\overline{\text{CCZ}}$ circuit where each logical qubit is involved in 3 distinct $\overline{\text{CCZ}}$ gates \cite[][Example 6]{Rengaswamy20}.  Such a code may have promising applications for transversal IQP sampling, but its fault-tolerant properties are not well understood; 
therefore, we leave the investigation of this and related codes for future work.

\subsection{The 64-qubit code}

For any $D\geq 2$, we define the distance-$4$ $D$-hyperoctahedron code by taking $2^D$ $D$-hypercubes and arranging them to form a $D$-hypercube; see \cref{fig:fig1}(g) for an illustration of this construction for $D=2$ and $D=3$.
Qubits are placed at the vertices of the resulting lattice, with $Z$ and $X$ stabilizer generators associated with the faces and $D$-dimensional cells, respectively. 
The code has parameters $\code{2^{2D},D,4}$ and can be interpreted as a $D$-dimensional color code on a $D$-hypercube-like region, with $D$-dimensional cells colored in $D+1$ colors.
It has a logical $\overline{{\rm C}^{D-1}{\rm Z}}$ gate implemented via a transversal $R_D={\rm{diag}}\left(1,i\pi/2^{D-1}\right)$ gate.
In particular, for $D=3$ we obtain the 64-qubit code with a logical $\overline{\rm CCZ}$ gate.

\begin{figure}
\centering
\includegraphics[width=\columnwidth]{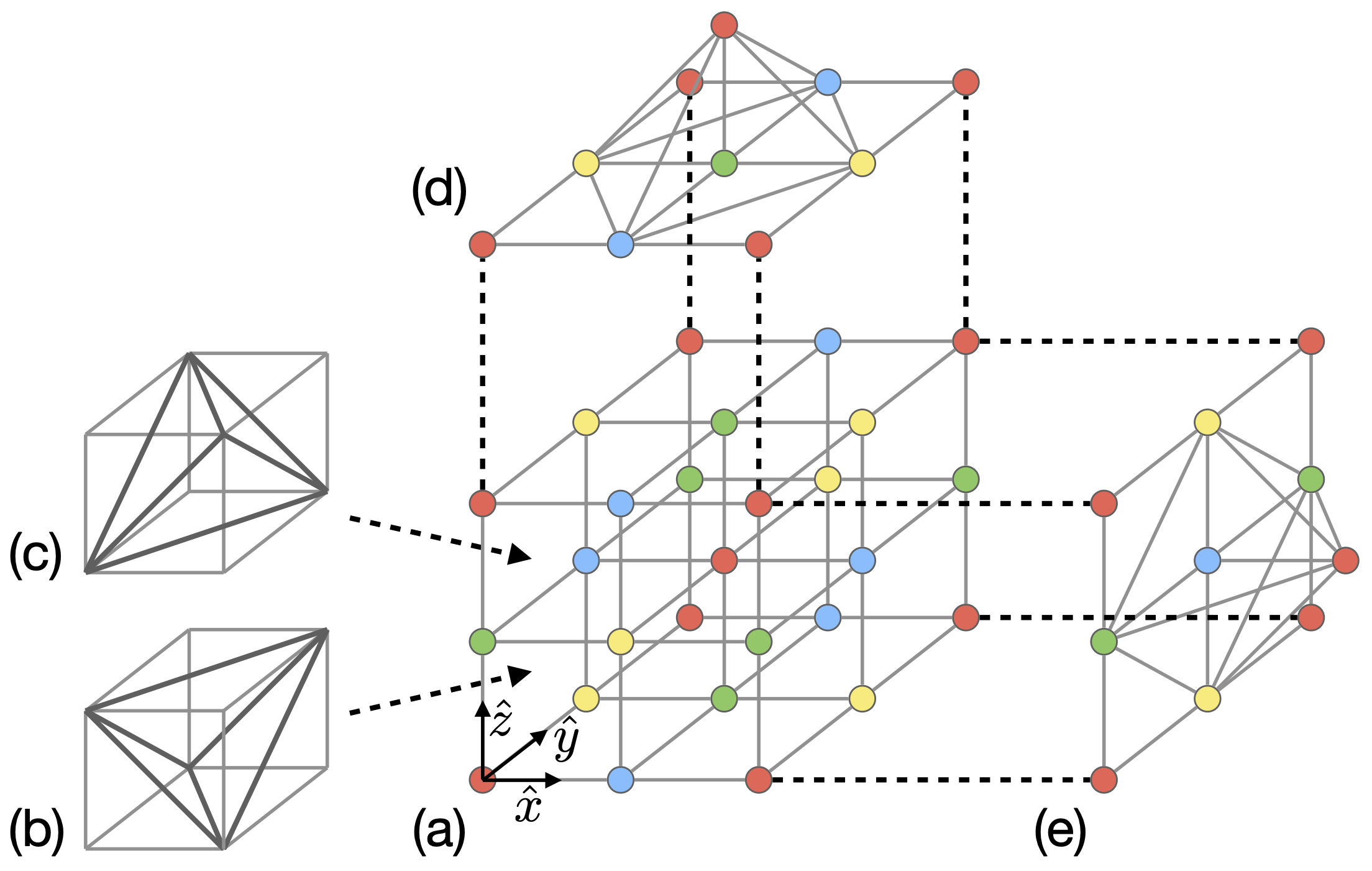}
\caption{\textbf{Large-distance codes with logical CCZ via transversal single-qubit rotations.}
Construction of a lattice $\mathcal L$ that facilitates a 3D color code with distance $d = L+2$ and a logical $\overline{\rm CCZ}$ gate for $L=2$.
Note that we use the dual-lattice picture, where qubits are identified with tetrahedra.
(a) A cubic lattice $\mathcal L'$ of size $L\times L\times L$, with its vertices colored in four colors, $R$, $G$, $B$ and $Y$.
(b)(c) Two possible ways of splitting each cube of $\mathcal L'$ into five tetrahedra.
We obtain $\mathcal L$ from $\mathcal L'$ by adding additional tetrahedra along the boundaries of $\mathcal L'$.
To avoi d clutter, we only illustrate (d) the top and (e) right boundaries.}
\label{fig_lattices_all}
\end{figure}

In order to reduce the weight of stabilizer operators, we can introduce some gauge operators.
For the 64-qubit code in \cref{fig:fig1}(g), whose volumes are colored in four colors, $R$, $G$, $B$ and $Y$, we choose the following gauge group
\begin{equation}
\mathcal G = \langle Z_f, X_c, X_{f'} | f\in F, c \in C, f'\in F_{BY}\rangle,
\end{equation}
where $C$, $F$ and $F_{BY}$ denote the set of all volumes, faces, and $BY$ faces (shared between $B$ and $Y$ volumes), respectively.
The parameters of the resulting code are the same, however we now measure operators of weight at most 8, compared to the original version with one stabilizer of weight 32 and two of weight 16 that were associated with $Y$ and $B$ volumes, respectively. 

We remark that, in an analogous way, for any even $d$ we can obtain the distance-$d$ $D$-hyperoctahedron code by taking an arrangement of $D$-hypercubes forming a $(d/2)\times\ldots\times (d/2)$ region of the $D$-hypercubic lattice.
The parameters of the code are $\code{d^{D},D,d}$.
Unfortunately, some of its stabilizers are macroscopic.
In particular, for $D=3$ there are $d/2-1$ stabilizers of weight $2d^2$, $d(d-2)/4$ stabilizers of weight $4d$, $(d+4)(d-2)/2$ stabilizers of weight $2d$; weight of other stabilizers is at most 8.
From the perspective of fault tolerance we may require that all stabilizers are of constant weight.
This motivates a definition of a family of 3D color codes, which we present in the following subsection. 

\subsection{A family of 3D color codes}

In order to define a 3D color code with even distance $d\geq 4$ and a logical $\overline{\rm CCZ}$ gate, we need to be able to construct an appropriate lattice supporting this code.
Since describing our construction seems easier in the the dual-lattice picture,
 we will use it in this subsection.

We start with a cubic lattice $\mathcal L'$ of size $L\times L\times L$, where $L=d-2$, which will constitute the bulk of the color code lattice; see \cref{fig_lattices_all}(a).
Each cube of $\mathcal L'$ is then split into five tetrahedra in one of two ways~\cite{Brown2016}; see \cref{fig_lattices_all}(b)(c).
We require that any two neighboring cubes that share a face are split differently. 
The vertices of $\mathcal L'$ are colored in one of four colors, $R$, $G$, $B$ or $Y$.
Then, we obtain $\mathcal L$ by adding additional tetrahedra to $\mathcal L'$, such that
every vertex of color $G$ (respectively, $B$ and $Y$) on the top or bottom (respectively, left or right, and front of rear) boundary of $\mathcal L'$ is covered; see \cref{fig_lattices_all}(d)(e).
Let $\partial\mathcal L_2$ denote the set of triangular faces on the boundary of $\mathcal L$.
We define $\partial\mathcal L_1$ and $\partial\mathcal L_0$ to be the sets of outermost edges and vertices bounding the initial lattice $\mathcal L'$, respectively.
By construction, $|\mathcal L| = 5L^3+6L^2$, $|\partial\mathcal L_2|=12L^2$, $|\partial\mathcal L_2|=12L^2$, $|\partial\mathcal L_1|=12L$ and $|\partial\mathcal L_0|=8$.

We define a 3D color code associated with the lattice~$\mathcal L$ as follows.
We place one qubit on every tetrahedron of~$\mathcal L$, triangular face in~$\partial\mathcal L_2$, edge in~$\partial\mathcal L_1$ and vertex in~$\partial\mathcal L_0$, and associate~$X$ and~$Z$ stabilizers with vertices and edges of~$\mathcal L$, respectively.
This construction is reminiscent of how a 3D tetrahedral color code is defined~\cite{Bombin2015,Kubica2015universal}.
The resulting code has parameters $\code{5d^3-12d^2+16,3,d}$.
Note that for $d=4$ we require $144$ qubits, which is $2.25\times$ more than for the $64$-qubit code.
While this construction is less qubit efficient (roughly by $5\times$ for large $d$) than the one from the previous subsection, asymptotically it only has stabilizers of constant weight.
By introducing gauge qubits we can further reduce the weight of the largest stabilizer, which is 32, requiring to measure operators of weight at most 6.
For more details on the lattice, see \cref{app:codes}.
Decoding these codes is possible using established decoding approaches for 3D color codes \cite{Kubica23}, including the potential for single-shot decoding \cite{Bombin15}.

\subsection{A family of 3D toric/color codes}

To construct a family of codes with growing code distance and constant-weight stabilizers we can take three copies of the 3D toric code with open boundary conditions, each encoding one logical qubit and rotated with respect to each other (via a cyclic permutation of the $x$, $y$ and $z$ axes).
The three logical qubits admit a logical $\overline{\rm CCZ}$ implemented via a transversal ${\rm CCZ}$ gate~\cite{Kubica.2015,Vasmer2019}.
By encoding each triple of qubits supporting a physical ${\rm CCZ}$ gate into the $\code{8,3,2}$ code, we obtain a 3D color code with a logical $\overline{\rm CCZ}$ implemented via a transversal $T$ gate.
The resulting code has parameters $\code{8n,3,2d}$, where $n$ and $d$ are the number of qubits and code distance of each copy of the toric code, and typically $n = O(d^3)$.

\begin{figure}
\centering
\includegraphics[width=.95\columnwidth]{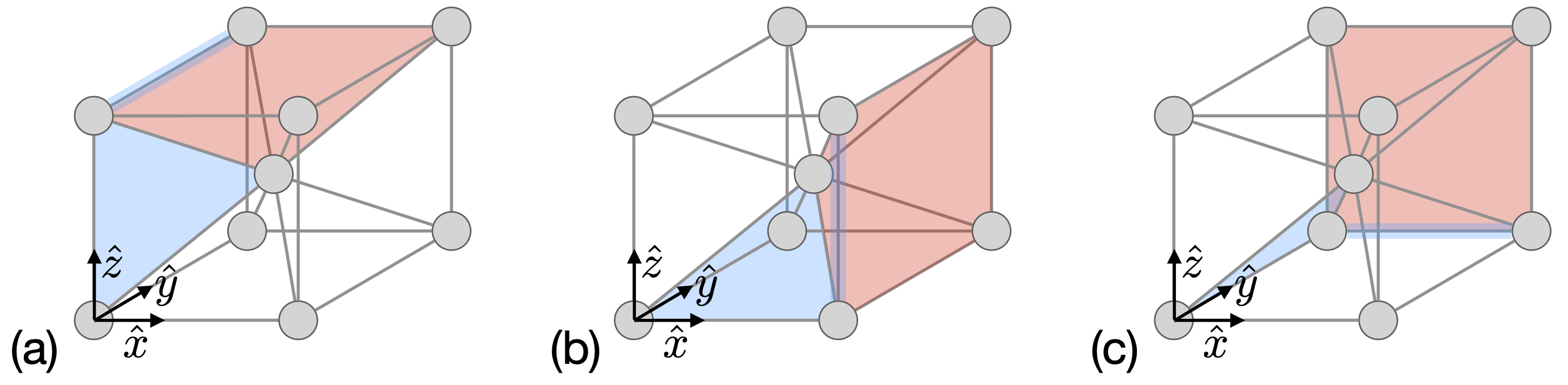}
\caption{\textbf{3D toric/color codes with transversal CCZ}.
Three copies of the 3D toric code with distance two, which admit a transversal logical $\overline{{\rm CCZ}}$ gate. 
To avoid clutter, we only depict some $X$ (red) and $Z$ (blue) stabilizers.
The copy in (a) has two $X$ stabilizers (pyramids spanned by the top/bottom faces and the cube center) and eight $Z$ stabilizers (four edges parallel to $\hat y$ and four triangular faces spanned by edges parallel to $z$ and the cube center).
The copies in (b) and (c) are obtained by rotating the copy in (a).
}
\label{fig_lattices_toric}
\end{figure}

The smallest realization of this construction uses 72 qubits and has distance four.
It is based on the 3D toric code with distance two that we depict in \cref{fig_lattices_toric}.
Analogously, by considering $D$ copies of the $D$-dimensional toric code with distance two concatenated with the $D$-hyperocatedron code we obtain the code with parameters $\code{2^D(2^D+1),D,4}$ and a logical $\overline{{\rm C}^{D-1}{\rm Z}}$ gate implemented via a transversal $R_D={\rm{diag}}\left(1,i\pi/2^{D-1}\right)$ gate.

\subsection{Threshold Theorem for Transversal IQP Sampling}

An important property of the previous two code constructions is that they are asymptotically low density parity check (LDPC) codes.  As shown by \textcite{gottesman_faulttolerant_2014}, any quantum LDPC code whose distance diverges at least logarithmically with increasing system size will have a fault-tolerance threshold under a local stochastic noise model.  
The proof includes the case of transversal gates applied to the code.  
More precisely under a minimum weight decoder, the theorem implies that for an $\code{n,k,d}$ LDPC code with $p$ below the threshold value $p_{\rm th}$ the logical error rate scales as
\begin{equation}
    O[ n (p/p_{\rm th})^{d/2}]
\end{equation}
as $n \to \infty$.  In particular, for the two families of LDPC codes introduced in the previous section, $d = O(n^{1/3})$, which leads to an exponential suppression of the error rate with increasing system size.  As a result, the output distribution of random transversal IQP circuits implemented with these codes will converge exponentially close in total variation distance to the ideal noiseless distribution.  

A similar threshold behavior in IQP circuit sampling was also studied recently in Ref.~\cite{paletta_robust_2023}.  There, the authors consider a different circuit architecture that uses ancilla qubits to compress the IQP circuit to a constant depth.  
They also consider a different family of 3D color codes with logical $\overline T$-gates instead of logical $\overline{\text{CCZ}}$ gates implemented via transversal $T$-gates.

 \begin{acknowledgements}
    We thank Abhinav Deshpande, Bill Fefferman, Daniel Grier, Jinguo Liu, Brayden Ware, Sepehr Ebadi, Simon Evered, Alexandra Geim, Sophie Li, Tom Manovitz, and Hengyun Zhou for helpful discussions.
    DH gratefully acknowledges the hospitality of the Simons Institute for the Theory of Computing during Summer 2023 and Spring 2024 supported by DOE QSA grant \#FP00010905, where part of this work was conducted.
    DH acknowledges funding from the US DoD through a QuICS Hartree fellowship and from the Simons Institute for the Theory of Computing, supported by DOE QSA.
    This research was supported in part by NSF QLCI grant OMA-2120757. 
    We acknowledge financial support from IARPA and the Army Research Office, under the Entangled Logical Qubits program (Cooperative Agreement Number W911NF-23-2-0219), the DARPA ONISQ program (grant number W911NF2010021), the DARPA IMPAQT program (grant number HR0011-23-3-0012), the Center for Ultracold Atoms (a NSF Physics Frontiers Center, PHY-1734011), the National Science Foundation (grant number PHY-2012023 and grant number CCF-2313084), the Army Research Office MURI (grant number W911NF-20-1-0082), and QuEra Computing. 
    XG acknowledges support from U.S. Department of Energy, Office of Science, National Quantum Information Science Research Centers, Quantum Systems Accelerator, NSF PFC grant No. PHYS 2317149 and start-up grants from CU Boulder.
\end{acknowledgements}

\bibliography{bibliography,doms_refs}

\begin{thebibliography}{150}%
\makeatletter
\providecommand \@ifxundefined [1]{%
 \@ifx{#1\undefined}
}%
\providecommand \@ifnum [1]{%
 \ifnum #1\expandafter \@firstoftwo
 \else \expandafter \@secondoftwo
 \fi
}%
\providecommand \@ifx [1]{%
 \ifx #1\expandafter \@firstoftwo
 \else \expandafter \@secondoftwo
 \fi
}%
\providecommand \natexlab [1]{#1}%
\providecommand \enquote  [1]{``#1''}%
\providecommand \bibnamefont  [1]{#1}%
\providecommand \bibfnamefont [1]{#1}%
\providecommand \citenamefont [1]{#1}%
\providecommand \href@noop [0]{\@secondoftwo}%
\providecommand \href [0]{\begingroup \@sanitize@url \@href}%
\providecommand \@href[1]{\@@startlink{#1}\@@href}%
\providecommand \@@href[1]{\endgroup#1\@@endlink}%
\providecommand \@sanitize@url [0]{\catcode `\\12\catcode `\$12\catcode
  `\&12\catcode `\#12\catcode `\^12\catcode `\_12\catcode `\%12\relax}%
\providecommand \@@startlink[1]{}%
\providecommand \@@endlink[0]{}%
\providecommand \url  [0]{\begingroup\@sanitize@url \@url }%
\providecommand \@url [1]{\endgroup\@href {#1}{\urlprefix }}%
\providecommand \urlprefix  [0]{URL }%
\providecommand \Eprint [0]{\href }%
\providecommand \doibase [0]{https://doi.org/}%
\providecommand \selectlanguage [0]{\@gobble}%
\providecommand \bibinfo  [0]{\@secondoftwo}%
\providecommand \bibfield  [0]{\@secondoftwo}%
\providecommand \translation [1]{[#1]}%
\providecommand \BibitemOpen [0]{}%
\providecommand \bibitemStop [0]{}%
\providecommand \bibitemNoStop [0]{.\EOS\space}%
\providecommand \EOS [0]{\spacefactor3000\relax}%
\providecommand \BibitemShut  [1]{\csname bibitem#1\endcsname}%
\let\auto@bib@innerbib\@empty
\bibitem [{\citenamefont {Reiher}\ \emph {et~al.}(2017)\citenamefont {Reiher},
  \citenamefont {Wiebe}, \citenamefont {Svore}, \citenamefont {Wecker},\ and\
  \citenamefont {Troyer}}]{reiher_elucidating_2017}%
  \BibitemOpen
  \bibfield  {author} {\bibinfo {author} {\bibfnamefont {M.}~\bibnamefont
  {Reiher}}, \bibinfo {author} {\bibfnamefont {N.}~\bibnamefont {Wiebe}},
  \bibinfo {author} {\bibfnamefont {K.~M.}\ \bibnamefont {Svore}}, \bibinfo
  {author} {\bibfnamefont {D.}~\bibnamefont {Wecker}},\ and\ \bibinfo {author}
  {\bibfnamefont {M.}~\bibnamefont {Troyer}},\ }\href
  {https://doi.org/10.1073/pnas.1619152114} {\bibfield  {journal} {\bibinfo
  {journal} {PNAS}\ }\textbf {\bibinfo {volume} {114}},\ \bibinfo {pages}
  {7555} (\bibinfo {year} {2017})}\BibitemShut {NoStop}%
\bibitem [{\citenamefont {Gidney}\ and\ \citenamefont
  {Eker{\aa}}(2021)}]{gidney_how_2021}%
  \BibitemOpen
  \bibfield  {author} {\bibinfo {author} {\bibfnamefont {C.}~\bibnamefont
  {Gidney}}\ and\ \bibinfo {author} {\bibfnamefont {M.}~\bibnamefont
  {Eker{\aa}}},\ }\href {https://doi.org/10.22331/q-2021-04-15-433} {\bibfield
  {journal} {\bibinfo  {journal} {Quantum}\ }\textbf {\bibinfo {volume} {5}},\
  \bibinfo {pages} {433} (\bibinfo {year} {2021})}\BibitemShut {NoStop}%
\bibitem [{\citenamefont {Clinton}\ \emph {et~al.}(2021)\citenamefont
  {Clinton}, \citenamefont {Bausch},\ and\ \citenamefont
  {Cubitt}}]{clinton_hamiltonian_2021}%
  \BibitemOpen
  \bibfield  {author} {\bibinfo {author} {\bibfnamefont {L.}~\bibnamefont
  {Clinton}}, \bibinfo {author} {\bibfnamefont {J.}~\bibnamefont {Bausch}},\
  and\ \bibinfo {author} {\bibfnamefont {T.}~\bibnamefont {Cubitt}},\ }\href
  {https://doi.org/10.1038/s41467-021-25196-0} {\bibfield  {journal} {\bibinfo
  {journal} {Nat Commun}\ }\textbf {\bibinfo {volume} {12}},\ \bibinfo {pages}
  {4989} (\bibinfo {year} {2021})},\ \Eprint {https://arxiv.org/abs/2003.06886}
  {arXiv:2003.06886} \BibitemShut {NoStop}%
\bibitem [{\citenamefont {Litinski}(2023)}]{litinski_how_2023}%
  \BibitemOpen
  \bibfield  {author} {\bibinfo {author} {\bibfnamefont {D.}~\bibnamefont
  {Litinski}},\ }\href@noop {} {\  (\bibinfo {year} {2023})},\ \Eprint
  {https://arxiv.org/abs/2306.08585} {arXiv:2306.08585} \BibitemShut {NoStop}%
\bibitem [{\citenamefont {Watson}\ \emph {et~al.}(2023)\citenamefont {Watson},
  \citenamefont {Bringewatt}, \citenamefont {Shaw}, \citenamefont {Childs},
  \citenamefont {Gorshkov},\ and\ \citenamefont
  {Davoudi}}]{watson_quantum_2023}%
  \BibitemOpen
  \bibfield  {author} {\bibinfo {author} {\bibfnamefont {J.~D.}\ \bibnamefont
  {Watson}}, \bibinfo {author} {\bibfnamefont {J.}~\bibnamefont {Bringewatt}},
  \bibinfo {author} {\bibfnamefont {A.~F.}\ \bibnamefont {Shaw}}, \bibinfo
  {author} {\bibfnamefont {A.~M.}\ \bibnamefont {Childs}}, \bibinfo {author}
  {\bibfnamefont {A.~V.}\ \bibnamefont {Gorshkov}},\ and\ \bibinfo {author}
  {\bibfnamefont {Z.}~\bibnamefont {Davoudi}},\ }\href@noop {} {\  (\bibinfo
  {year} {2023})},\ \Eprint {https://arxiv.org/abs/2312.05344}
  {arXiv:2312.05344} \BibitemShut {NoStop}%
\bibitem [{\citenamefont {Shor}(1996)}]{shor_fault-tolerant_1996}%
  \BibitemOpen
  \bibfield  {author} {\bibinfo {author} {\bibfnamefont {P.}~\bibnamefont
  {Shor}},\ }in\ \href {https://doi.org/10.1109/SFCS.1996.548464} {\emph
  {\bibinfo {booktitle} {Proceedings of 37th {{Conference}} on {{Foundations}}
  of {{Computer Science}}}}}\ (\bibinfo {year} {1996})\ pp.\ \bibinfo {pages}
  {56--65}\BibitemShut {NoStop}%
\bibitem [{\citenamefont {Steane}(1996)}]{steane_error_1996}%
  \BibitemOpen
  \bibfield  {author} {\bibinfo {author} {\bibfnamefont {A.~M.}\ \bibnamefont
  {Steane}},\ }\href {https://doi.org/10.1103/PhysRevLett.77.793} {\bibfield
  {journal} {\bibinfo  {journal} {Phys. Rev. Lett.}\ }\textbf {\bibinfo
  {volume} {77}},\ \bibinfo {pages} {793} (\bibinfo {year} {1996})}\BibitemShut
  {NoStop}%
\bibitem [{\citenamefont {Aharonov}\ and\ \citenamefont
  {{Ben-Or}}(1997)}]{aharonov_fault-tolerant_1997}%
  \BibitemOpen
  \bibfield  {author} {\bibinfo {author} {\bibfnamefont {D.}~\bibnamefont
  {Aharonov}}\ and\ \bibinfo {author} {\bibfnamefont {M.}~\bibnamefont
  {{Ben-Or}}},\ }in\ \href {https://doi.org/10.1145/258533.258579} {\emph
  {\bibinfo {booktitle} {Proceedings of the Twenty-Ninth Annual {{ACM}}
  Symposium on {{Theory}} of Computing}}},\ \bibinfo {series and number}
  {{{STOC}} '97}\ (\bibinfo  {publisher} {Association for Computing
  Machinery},\ \bibinfo {address} {New York, NY, USA},\ \bibinfo {year}
  {1997})\ pp.\ \bibinfo {pages} {176--188}\BibitemShut {NoStop}%
\bibitem [{Note1()}]{Note1}%
  \BibitemOpen
  \bibinfo {note} {When restricting to transversal operations, it can be proven
  that the implemented gate set in fact has to be discrete \cite
  {Eastin-Knill.2009}.}\BibitemShut {Stop}%
\bibitem [{\citenamefont
  {Knill}(2004{\natexlab{a}})}]{knill_fault-tolerant_2004}%
  \BibitemOpen
  \bibfield  {author} {\bibinfo {author} {\bibfnamefont {E.}~\bibnamefont
  {Knill}},\ }\href@noop {} {\  (\bibinfo {year} {2004}{\natexlab{a}})},\
  \Eprint {https://arxiv.org/abs/quant-ph/0402171} {arXiv:quant-ph/0402171}
  \BibitemShut {NoStop}%
\bibitem [{\citenamefont
  {Knill}(2004{\natexlab{b}})}]{knill_fault-tolerant_2004-1}%
  \BibitemOpen
  \bibfield  {author} {\bibinfo {author} {\bibfnamefont {E.}~\bibnamefont
  {Knill}},\ }\href@noop {} {\  (\bibinfo {year} {2004}{\natexlab{b}})},\
  \Eprint {https://arxiv.org/abs/quant-ph/0404104} {arXiv:quant-ph/0404104}
  \BibitemShut {NoStop}%
\bibitem [{\citenamefont {Bravyi}\ and\ \citenamefont
  {Kitaev}(2005)}]{bravyi_universal_2005}%
  \BibitemOpen
  \bibfield  {author} {\bibinfo {author} {\bibfnamefont {S.}~\bibnamefont
  {Bravyi}}\ and\ \bibinfo {author} {\bibfnamefont {A.}~\bibnamefont
  {Kitaev}},\ }\href {https://doi.org/10.1103/PhysRevA.71.022316} {\bibfield
  {journal} {\bibinfo  {journal} {Phys. Rev. A}\ }\textbf {\bibinfo {volume}
  {71}},\ \bibinfo {pages} {022316} (\bibinfo {year} {2005})}\BibitemShut
  {NoStop}%
\bibitem [{\citenamefont {Paetznick}\ and\ \citenamefont
  {Reichardt}(2013)}]{paetznick_universal_2013}%
  \BibitemOpen
  \bibfield  {author} {\bibinfo {author} {\bibfnamefont {A.}~\bibnamefont
  {Paetznick}}\ and\ \bibinfo {author} {\bibfnamefont {B.~W.}\ \bibnamefont
  {Reichardt}},\ }\href {https://doi.org/10.1103/PhysRevLett.111.090505}
  {\bibfield  {journal} {\bibinfo  {journal} {Phys. Rev. Lett.}\ }\textbf
  {\bibinfo {volume} {111}},\ \bibinfo {pages} {090505} (\bibinfo {year}
  {2013})}\BibitemShut {NoStop}%
\bibitem [{\citenamefont {Anderson}\ \emph {et~al.}(2014)\citenamefont
  {Anderson}, \citenamefont {{Duclos-Cianci}},\ and\ \citenamefont
  {Poulin}}]{anderson_fault-tolerant_2014}%
  \BibitemOpen
  \bibfield  {author} {\bibinfo {author} {\bibfnamefont {J.~T.}\ \bibnamefont
  {Anderson}}, \bibinfo {author} {\bibfnamefont {G.}~\bibnamefont
  {{Duclos-Cianci}}},\ and\ \bibinfo {author} {\bibfnamefont {D.}~\bibnamefont
  {Poulin}},\ }\href {https://doi.org/10.1103/PhysRevLett.113.080501}
  {\bibfield  {journal} {\bibinfo  {journal} {Phys. Rev. Lett.}\ }\textbf
  {\bibinfo {volume} {113}},\ \bibinfo {pages} {080501} (\bibinfo {year}
  {2014})}\BibitemShut {NoStop}%
\bibitem [{\citenamefont {Arute}\ \emph {et~al.}(2019)\citenamefont {Arute},
  \citenamefont {Arya}, \citenamefont {Babbush}, \citenamefont {Bacon},
  \citenamefont {Bardin}, \citenamefont {Barends}, \citenamefont {Biswas},
  \citenamefont {Boixo}, \citenamefont {Brandao}, \citenamefont {Buell},
  \citenamefont {Burkett},\ and\ \citenamefont
  {{others}}}]{arute_quantum_2019}%
  \BibitemOpen
  \bibfield  {author} {\bibinfo {author} {\bibfnamefont {F.}~\bibnamefont
  {Arute}}, \bibinfo {author} {\bibfnamefont {K.}~\bibnamefont {Arya}},
  \bibinfo {author} {\bibfnamefont {R.}~\bibnamefont {Babbush}}, \bibinfo
  {author} {\bibfnamefont {D.}~\bibnamefont {Bacon}}, \bibinfo {author}
  {\bibfnamefont {J.~C.}\ \bibnamefont {Bardin}}, \bibinfo {author}
  {\bibfnamefont {R.}~\bibnamefont {Barends}}, \bibinfo {author} {\bibfnamefont
  {R.}~\bibnamefont {Biswas}}, \bibinfo {author} {\bibfnamefont
  {S.}~\bibnamefont {Boixo}}, \bibinfo {author} {\bibfnamefont {F.~G. S.~L.}\
  \bibnamefont {Brandao}}, \bibinfo {author} {\bibfnamefont {D.~A.}\
  \bibnamefont {Buell}}, \bibinfo {author} {\bibfnamefont {B.}~\bibnamefont
  {Burkett}},\ and\ \bibinfo {author} {\bibnamefont {{others}}},\ }\href
  {https://doi.org/10.1038/s41586-019-1666-5} {\bibfield  {journal} {\bibinfo
  {journal} {Nature}\ }\textbf {\bibinfo {volume} {574}},\ \bibinfo {pages}
  {505} (\bibinfo {year} {2019})}\BibitemShut {NoStop}%
\bibitem [{\citenamefont {Morvan}\ \emph {et~al.}(2023)\citenamefont {Morvan},
  \citenamefont {Villalonga}, \citenamefont {Mi}, \citenamefont {Mandrà},
  \citenamefont {Bengtsson}, \citenamefont {Klimov}, \citenamefont {Chen},
  \citenamefont {Hong}, \citenamefont {Erickson}, \citenamefont {Drozdov},\
  and\ \citenamefont {{others}}}]{Morvan.2023}%
  \BibitemOpen
  \bibfield  {author} {\bibinfo {author} {\bibfnamefont {A.}~\bibnamefont
  {Morvan}}, \bibinfo {author} {\bibfnamefont {B.}~\bibnamefont {Villalonga}},
  \bibinfo {author} {\bibfnamefont {X.}~\bibnamefont {Mi}}, \bibinfo {author}
  {\bibfnamefont {S.}~\bibnamefont {Mandrà}}, \bibinfo {author} {\bibfnamefont
  {A.}~\bibnamefont {Bengtsson}}, \bibinfo {author} {\bibfnamefont {P.~V.}\
  \bibnamefont {Klimov}}, \bibinfo {author} {\bibfnamefont {Z.}~\bibnamefont
  {Chen}}, \bibinfo {author} {\bibfnamefont {S.}~\bibnamefont {Hong}}, \bibinfo
  {author} {\bibfnamefont {C.}~\bibnamefont {Erickson}}, \bibinfo {author}
  {\bibfnamefont {I.~K.}\ \bibnamefont {Drozdov}},\ and\ \bibinfo {author}
  {\bibnamefont {{others}}},\ }\href@noop {} {\  (\bibinfo {year} {2023})},\
  \Eprint {https://arxiv.org/abs/2304.11119} {arXiv:2304.11119} \BibitemShut
  {NoStop}%
\bibitem [{\citenamefont {Bremner}\ \emph {et~al.}(2010)\citenamefont
  {Bremner}, \citenamefont {Jozsa},\ and\ \citenamefont
  {Shepherd}}]{bremner_classical_2010}%
  \BibitemOpen
  \bibfield  {author} {\bibinfo {author} {\bibfnamefont {M.~J.}\ \bibnamefont
  {Bremner}}, \bibinfo {author} {\bibfnamefont {R.}~\bibnamefont {Jozsa}},\
  and\ \bibinfo {author} {\bibfnamefont {D.~J.}\ \bibnamefont {Shepherd}},\
  }\href {https://doi.org/10.1098/rspa.2010.0301} {\bibfield  {journal}
  {\bibinfo  {journal} {Proceedings of the Royal Society A: Mathematical,
  Physical and Engineering Sciences}\ }\textbf {\bibinfo {volume} {467}},\
  \bibinfo {pages} {459} (\bibinfo {year} {2010})}\BibitemShut {NoStop}%
\bibitem [{\citenamefont {Aaronson}\ and\ \citenamefont
  {Arkhipov}(2013)}]{aaronson_computational_2013}%
  \BibitemOpen
  \bibfield  {author} {\bibinfo {author} {\bibfnamefont {S.}~\bibnamefont
  {Aaronson}}\ and\ \bibinfo {author} {\bibfnamefont {A.}~\bibnamefont
  {Arkhipov}},\ }\href {https://doi.org/10.4086/toc.2013.v009a004} {\bibfield
  {journal} {\bibinfo  {journal} {Th. Comp.}\ }\textbf {\bibinfo {volume}
  {9}},\ \bibinfo {pages} {143} (\bibinfo {year} {2013})}\BibitemShut {NoStop}%
\bibitem [{\citenamefont {Hangleiter}\ and\ \citenamefont
  {Eisert}(2023)}]{hangleiter_computational_2023}%
  \BibitemOpen
  \bibfield  {author} {\bibinfo {author} {\bibfnamefont {D.}~\bibnamefont
  {Hangleiter}}\ and\ \bibinfo {author} {\bibfnamefont {J.}~\bibnamefont
  {Eisert}},\ }\href {https://doi.org/10.1103/RevModPhys.95.035001} {\bibfield
  {journal} {\bibinfo  {journal} {Rev. Mod. Phys.}\ }\textbf {\bibinfo {volume}
  {95}},\ \bibinfo {pages} {035001} (\bibinfo {year} {2023})},\ \Eprint
  {https://arxiv.org/abs/2206.04079} {arXiv:2206.04079} \BibitemShut {NoStop}%
\bibitem [{\citenamefont {Boixo}\ \emph
  {et~al.}(2018{\natexlab{a}})\citenamefont {Boixo}, \citenamefont {Isakov},
  \citenamefont {Smelyanskiy}, \citenamefont {Babbush}, \citenamefont {Ding},
  \citenamefont {Jiang}, \citenamefont {Bremner}, \citenamefont {Martinis},\
  and\ \citenamefont {Neven}}]{boixo_characterizing_2018}%
  \BibitemOpen
  \bibfield  {author} {\bibinfo {author} {\bibfnamefont {S.}~\bibnamefont
  {Boixo}}, \bibinfo {author} {\bibfnamefont {S.~V.}\ \bibnamefont {Isakov}},
  \bibinfo {author} {\bibfnamefont {V.~N.}\ \bibnamefont {Smelyanskiy}},
  \bibinfo {author} {\bibfnamefont {R.}~\bibnamefont {Babbush}}, \bibinfo
  {author} {\bibfnamefont {N.}~\bibnamefont {Ding}}, \bibinfo {author}
  {\bibfnamefont {Z.}~\bibnamefont {Jiang}}, \bibinfo {author} {\bibfnamefont
  {M.~J.}\ \bibnamefont {Bremner}}, \bibinfo {author} {\bibfnamefont {J.~M.}\
  \bibnamefont {Martinis}},\ and\ \bibinfo {author} {\bibfnamefont
  {H.}~\bibnamefont {Neven}},\ }\href
  {https://doi.org/10.1038/s41567-018-0124-x} {\bibfield  {journal} {\bibinfo
  {journal} {Nature Phys}\ }\textbf {\bibinfo {volume} {14}},\ \bibinfo {pages}
  {595} (\bibinfo {year} {2018}{\natexlab{a}})},\ \Eprint
  {https://arxiv.org/abs/1608.00263} {arXiv:1608.00263} \BibitemShut {NoStop}%
\bibitem [{\citenamefont {Choi}\ \emph {et~al.}(2023)\citenamefont {Choi},
  \citenamefont {Shaw}, \citenamefont {Madjarov}, \citenamefont {Xie},
  \citenamefont {Finkelstein}, \citenamefont {Covey}, \citenamefont {Cotler},
  \citenamefont {Mark}, \citenamefont {Huang}, \citenamefont {Kale},
  \citenamefont {Pichler}, \citenamefont {Brand{\~a}o}, \citenamefont {Choi},\
  and\ \citenamefont {Endres}}]{choi_preparing_2023}%
  \BibitemOpen
  \bibfield  {author} {\bibinfo {author} {\bibfnamefont {J.}~\bibnamefont
  {Choi}}, \bibinfo {author} {\bibfnamefont {A.~L.}\ \bibnamefont {Shaw}},
  \bibinfo {author} {\bibfnamefont {I.~S.}\ \bibnamefont {Madjarov}}, \bibinfo
  {author} {\bibfnamefont {X.}~\bibnamefont {Xie}}, \bibinfo {author}
  {\bibfnamefont {R.}~\bibnamefont {Finkelstein}}, \bibinfo {author}
  {\bibfnamefont {J.~P.}\ \bibnamefont {Covey}}, \bibinfo {author}
  {\bibfnamefont {J.~S.}\ \bibnamefont {Cotler}}, \bibinfo {author}
  {\bibfnamefont {D.~K.}\ \bibnamefont {Mark}}, \bibinfo {author}
  {\bibfnamefont {H.-Y.}\ \bibnamefont {Huang}}, \bibinfo {author}
  {\bibfnamefont {A.}~\bibnamefont {Kale}}, \bibinfo {author} {\bibfnamefont
  {H.}~\bibnamefont {Pichler}}, \bibinfo {author} {\bibfnamefont {F.~G. S.~L.}\
  \bibnamefont {Brand{\~a}o}}, \bibinfo {author} {\bibfnamefont
  {S.}~\bibnamefont {Choi}},\ and\ \bibinfo {author} {\bibfnamefont
  {M.}~\bibnamefont {Endres}},\ }\href
  {https://doi.org/10.1038/s41586-022-05442-1} {\bibfield  {journal} {\bibinfo
  {journal} {Nature}\ }\textbf {\bibinfo {volume} {613}},\ \bibinfo {pages}
  {468} (\bibinfo {year} {2023})}\BibitemShut {NoStop}%
\bibitem [{\citenamefont {Ware}\ \emph {et~al.}(2023)\citenamefont {Ware},
  \citenamefont {Deshpande}, \citenamefont {Hangleiter}, \citenamefont
  {Niroula}, \citenamefont {Fefferman}, \citenamefont {Gorshkov},\ and\
  \citenamefont {Gullans}}]{Ware.2023}%
  \BibitemOpen
  \bibfield  {author} {\bibinfo {author} {\bibfnamefont {B.}~\bibnamefont
  {Ware}}, \bibinfo {author} {\bibfnamefont {A.}~\bibnamefont {Deshpande}},
  \bibinfo {author} {\bibfnamefont {D.}~\bibnamefont {Hangleiter}}, \bibinfo
  {author} {\bibfnamefont {P.}~\bibnamefont {Niroula}}, \bibinfo {author}
  {\bibfnamefont {B.}~\bibnamefont {Fefferman}}, \bibinfo {author}
  {\bibfnamefont {A.~V.}\ \bibnamefont {Gorshkov}},\ and\ \bibinfo {author}
  {\bibfnamefont {M.~J.}\ \bibnamefont {Gullans}},\ }\Eprint
  {https://arxiv.org/abs/2305.04954} {arXiv:2305.04954}  (\bibinfo {year}
  {2023})\BibitemShut {NoStop}%
\bibitem [{\citenamefont {Shepherd}\ and\ \citenamefont
  {Bremner}(2009)}]{shepherd_temporally_2009}%
  \BibitemOpen
  \bibfield  {author} {\bibinfo {author} {\bibfnamefont {D.}~\bibnamefont
  {Shepherd}}\ and\ \bibinfo {author} {\bibfnamefont {M.~J.}\ \bibnamefont
  {Bremner}},\ }\href {https://doi.org/10.1098/rspa.2008.0443} {\bibfield
  {journal} {\bibinfo  {journal} {Proceedings of the Royal Society of London A:
  Mathematical, Physical and Engineering Sciences}\ }\textbf {\bibinfo {volume}
  {465}},\ \bibinfo {pages} {1413} (\bibinfo {year} {2009})}\BibitemShut
  {NoStop}%
\bibitem [{\citenamefont {Bremner}\ \emph
  {et~al.}(2016{\natexlab{a}})\citenamefont {Bremner}, \citenamefont
  {Montanaro},\ and\ \citenamefont {Shepherd}}]{Bremner.2016}%
  \BibitemOpen
  \bibfield  {author} {\bibinfo {author} {\bibfnamefont {M.~J.}\ \bibnamefont
  {Bremner}}, \bibinfo {author} {\bibfnamefont {A.}~\bibnamefont {Montanaro}},\
  and\ \bibinfo {author} {\bibfnamefont {D.~J.}\ \bibnamefont {Shepherd}},\
  }\href {https://doi.org/10.1103/PhysRevLett.117.080501} {\bibfield  {journal}
  {\bibinfo  {journal} {Phys. Rev. Lett.}\ }\textbf {\bibinfo {volume} {117}},\
  \bibinfo {pages} {080501} (\bibinfo {year} {2016}{\natexlab{a}})}\BibitemShut
  {NoStop}%
\bibitem [{\citenamefont {Bremner}\ \emph
  {et~al.}(2017{\natexlab{a}})\citenamefont {Bremner}, \citenamefont
  {Montanaro},\ and\ \citenamefont {Shepherd}}]{Bremner.2017}%
  \BibitemOpen
  \bibfield  {author} {\bibinfo {author} {\bibfnamefont {M.~J.}\ \bibnamefont
  {Bremner}}, \bibinfo {author} {\bibfnamefont {A.}~\bibnamefont {Montanaro}},\
  and\ \bibinfo {author} {\bibfnamefont {D.~J.}\ \bibnamefont {Shepherd}},\
  }\href {https://doi.org/10.22331/q-2017-04-25-8} {\bibfield  {journal}
  {\bibinfo  {journal} {{Quantum}}\ }\textbf {\bibinfo {volume} {1}},\ \bibinfo
  {pages} {8} (\bibinfo {year} {2017}{\natexlab{a}})}\BibitemShut {NoStop}%
\bibitem [{\citenamefont {Gottesman}()}]{gottesman_faulttolerant_2014}%
  \BibitemOpen
  \bibfield  {author} {\bibinfo {author} {\bibfnamefont {D.}~\bibnamefont
  {Gottesman}},\ }\href@noop {} {\ }\Eprint {https://arxiv.org/abs/1310.2984}
  {arXiv:1310.2984} \BibitemShut {NoStop}%
\bibitem [{\citenamefont {Bluvstein}\ \emph {et~al.}(2024)\citenamefont
  {Bluvstein}, \citenamefont {Evered}, \citenamefont {Geim}, \citenamefont
  {Li}, \citenamefont {Zhou}, \citenamefont {Manovitz}, \citenamefont {Ebadi},
  \citenamefont {Cain}, \citenamefont {Kalinowski}, \citenamefont {Hangleiter},
  \citenamefont {Bonilla~Ataides}, \citenamefont {Maskara}, \citenamefont
  {Cong}, \citenamefont {Gao}, \citenamefont {Sales~Rodriguez}, \citenamefont
  {Karolyshyn}, \citenamefont {Semeghini}, \citenamefont {Gullans},
  \citenamefont {Greiner}, \citenamefont {Vuleti{\'c}},\ and\ \citenamefont
  {Lukin}}]{bluvstein_logical_2024}%
  \BibitemOpen
  \bibfield  {author} {\bibinfo {author} {\bibfnamefont {D.}~\bibnamefont
  {Bluvstein}}, \bibinfo {author} {\bibfnamefont {S.~J.}\ \bibnamefont
  {Evered}}, \bibinfo {author} {\bibfnamefont {A.~A.}\ \bibnamefont {Geim}},
  \bibinfo {author} {\bibfnamefont {S.~H.}\ \bibnamefont {Li}}, \bibinfo
  {author} {\bibfnamefont {H.}~\bibnamefont {Zhou}}, \bibinfo {author}
  {\bibfnamefont {T.}~\bibnamefont {Manovitz}}, \bibinfo {author}
  {\bibfnamefont {S.}~\bibnamefont {Ebadi}}, \bibinfo {author} {\bibfnamefont
  {M.}~\bibnamefont {Cain}}, \bibinfo {author} {\bibfnamefont {M.}~\bibnamefont
  {Kalinowski}}, \bibinfo {author} {\bibfnamefont {D.}~\bibnamefont
  {Hangleiter}}, \bibinfo {author} {\bibfnamefont {J.~P.}\ \bibnamefont
  {Bonilla~Ataides}}, \bibinfo {author} {\bibfnamefont {N.}~\bibnamefont
  {Maskara}}, \bibinfo {author} {\bibfnamefont {I.}~\bibnamefont {Cong}},
  \bibinfo {author} {\bibfnamefont {X.}~\bibnamefont {Gao}}, \bibinfo {author}
  {\bibfnamefont {P.}~\bibnamefont {Sales~Rodriguez}}, \bibinfo {author}
  {\bibfnamefont {T.}~\bibnamefont {Karolyshyn}}, \bibinfo {author}
  {\bibfnamefont {G.}~\bibnamefont {Semeghini}}, \bibinfo {author}
  {\bibfnamefont {M.~J.}\ \bibnamefont {Gullans}}, \bibinfo {author}
  {\bibfnamefont {M.}~\bibnamefont {Greiner}}, \bibinfo {author} {\bibfnamefont
  {V.}~\bibnamefont {Vuleti{\'c}}},\ and\ \bibinfo {author} {\bibfnamefont
  {M.~D.}\ \bibnamefont {Lukin}},\ }\href
  {https://doi.org/10.1038/s41586-023-06927-3} {\bibfield  {journal} {\bibinfo
  {journal} {Nature}\ }\textbf {\bibinfo {volume} {626}},\ \bibinfo {pages}
  {58} (\bibinfo {year} {2024})}\BibitemShut {NoStop}%
\bibitem [{\citenamefont {Beugnon}\ \emph {et~al.}(2007)\citenamefont
  {Beugnon}, \citenamefont {Tuchendler}, \citenamefont {Marion}, \citenamefont
  {Ga{\"e}tan}, \citenamefont {Miroshnychenko}, \citenamefont {Sortais},
  \citenamefont {Lance}, \citenamefont {Jones}, \citenamefont {Messin},
  \citenamefont {Browaeys},\ and\ \citenamefont
  {Grangier}}]{beugnon_two-dimensional_2007}%
  \BibitemOpen
  \bibfield  {author} {\bibinfo {author} {\bibfnamefont {J.}~\bibnamefont
  {Beugnon}}, \bibinfo {author} {\bibfnamefont {C.}~\bibnamefont {Tuchendler}},
  \bibinfo {author} {\bibfnamefont {H.}~\bibnamefont {Marion}}, \bibinfo
  {author} {\bibfnamefont {A.}~\bibnamefont {Ga{\"e}tan}}, \bibinfo {author}
  {\bibfnamefont {Y.}~\bibnamefont {Miroshnychenko}}, \bibinfo {author}
  {\bibfnamefont {Y.~R.~P.}\ \bibnamefont {Sortais}}, \bibinfo {author}
  {\bibfnamefont {A.~M.}\ \bibnamefont {Lance}}, \bibinfo {author}
  {\bibfnamefont {M.~P.~A.}\ \bibnamefont {Jones}}, \bibinfo {author}
  {\bibfnamefont {G.}~\bibnamefont {Messin}}, \bibinfo {author} {\bibfnamefont
  {A.}~\bibnamefont {Browaeys}},\ and\ \bibinfo {author} {\bibfnamefont
  {P.}~\bibnamefont {Grangier}},\ }\href {https://doi.org/10.1038/nphys698}
  {\bibfield  {journal} {\bibinfo  {journal} {Nature Phys}\ }\textbf {\bibinfo
  {volume} {3}},\ \bibinfo {pages} {696} (\bibinfo {year} {2007})}\BibitemShut
  {NoStop}%
\bibitem [{\citenamefont {Schlosser}\ \emph {et~al.}(2011)\citenamefont
  {Schlosser}, \citenamefont {Tichelmann}, \citenamefont {Kruse},\ and\
  \citenamefont {Birkl}}]{schlosser_scalable_2011}%
  \BibitemOpen
  \bibfield  {author} {\bibinfo {author} {\bibfnamefont {M.}~\bibnamefont
  {Schlosser}}, \bibinfo {author} {\bibfnamefont {S.}~\bibnamefont
  {Tichelmann}}, \bibinfo {author} {\bibfnamefont {J.}~\bibnamefont {Kruse}},\
  and\ \bibinfo {author} {\bibfnamefont {G.}~\bibnamefont {Birkl}},\ }\href
  {https://doi.org/10.1007/s11128-011-0297-z} {\bibfield  {journal} {\bibinfo
  {journal} {Quantum Inf Process}\ }\textbf {\bibinfo {volume} {10}},\ \bibinfo
  {pages} {907} (\bibinfo {year} {2011})}\BibitemShut {NoStop}%
\bibitem [{\citenamefont {Evered}\ \emph {et~al.}(2023)\citenamefont {Evered},
  \citenamefont {Bluvstein}, \citenamefont {Kalinowski}, \citenamefont {Ebadi},
  \citenamefont {Manovitz}, \citenamefont {Zhou}, \citenamefont {Li},
  \citenamefont {Geim}, \citenamefont {Wang}, \citenamefont {Maskara},
  \citenamefont {Levine}, \citenamefont {Semeghini}, \citenamefont {Greiner},
  \citenamefont {Vuleti{\'c}},\ and\ \citenamefont
  {Lukin}}]{evered_high-fidelity_2023}%
  \BibitemOpen
  \bibfield  {author} {\bibinfo {author} {\bibfnamefont {S.~J.}\ \bibnamefont
  {Evered}}, \bibinfo {author} {\bibfnamefont {D.}~\bibnamefont {Bluvstein}},
  \bibinfo {author} {\bibfnamefont {M.}~\bibnamefont {Kalinowski}}, \bibinfo
  {author} {\bibfnamefont {S.}~\bibnamefont {Ebadi}}, \bibinfo {author}
  {\bibfnamefont {T.}~\bibnamefont {Manovitz}}, \bibinfo {author}
  {\bibfnamefont {H.}~\bibnamefont {Zhou}}, \bibinfo {author} {\bibfnamefont
  {S.~H.}\ \bibnamefont {Li}}, \bibinfo {author} {\bibfnamefont {A.~A.}\
  \bibnamefont {Geim}}, \bibinfo {author} {\bibfnamefont {T.~T.}\ \bibnamefont
  {Wang}}, \bibinfo {author} {\bibfnamefont {N.}~\bibnamefont {Maskara}},
  \bibinfo {author} {\bibfnamefont {H.}~\bibnamefont {Levine}}, \bibinfo
  {author} {\bibfnamefont {G.}~\bibnamefont {Semeghini}}, \bibinfo {author}
  {\bibfnamefont {M.}~\bibnamefont {Greiner}}, \bibinfo {author} {\bibfnamefont
  {V.}~\bibnamefont {Vuleti{\'c}}},\ and\ \bibinfo {author} {\bibfnamefont
  {M.~D.}\ \bibnamefont {Lukin}},\ }\href
  {https://doi.org/10.1038/s41586-023-06481-y} {\bibfield  {journal} {\bibinfo
  {journal} {Nature}\ }\textbf {\bibinfo {volume} {622}},\ \bibinfo {pages}
  {268} (\bibinfo {year} {2023})},\ \Eprint {https://arxiv.org/abs/2304.05420}
  {arXiv:2304.05420} \BibitemShut {NoStop}%
\bibitem [{\citenamefont {Kubica}\ \emph {et~al.}(2015)\citenamefont {Kubica},
  \citenamefont {Yoshida},\ and\ \citenamefont {Pastawski}}]{Kubica.2015}%
  \BibitemOpen
  \bibfield  {author} {\bibinfo {author} {\bibfnamefont {A.}~\bibnamefont
  {Kubica}}, \bibinfo {author} {\bibfnamefont {B.}~\bibnamefont {Yoshida}},\
  and\ \bibinfo {author} {\bibfnamefont {F.}~\bibnamefont {Pastawski}},\ }\href
  {https://doi.org/10.1088/1367-2630/17/8/083026} {\bibfield  {journal}
  {\bibinfo  {journal} {New J. Phys.}\ }\textbf {\bibinfo {volume} {17}},\
  \bibinfo {pages} {083026} (\bibinfo {year} {2015})}\BibitemShut {NoStop}%
\bibitem [{\citenamefont {Bluvstein}\ \emph {et~al.}(2022)\citenamefont
  {Bluvstein}, \citenamefont {Levine}, \citenamefont {Semeghini}, \citenamefont
  {Wang}, \citenamefont {Ebadi}, \citenamefont {Kalinowski}, \citenamefont
  {Keesling}, \citenamefont {Maskara}, \citenamefont {Pichler}, \citenamefont
  {Greiner}, \citenamefont {Vuletić},\ and\ \citenamefont
  {Lukin}}]{Bluvstein.2022}%
  \BibitemOpen
  \bibfield  {author} {\bibinfo {author} {\bibfnamefont {D.}~\bibnamefont
  {Bluvstein}}, \bibinfo {author} {\bibfnamefont {H.}~\bibnamefont {Levine}},
  \bibinfo {author} {\bibfnamefont {G.}~\bibnamefont {Semeghini}}, \bibinfo
  {author} {\bibfnamefont {T.~T.}\ \bibnamefont {Wang}}, \bibinfo {author}
  {\bibfnamefont {S.}~\bibnamefont {Ebadi}}, \bibinfo {author} {\bibfnamefont
  {M.}~\bibnamefont {Kalinowski}}, \bibinfo {author} {\bibfnamefont
  {A.}~\bibnamefont {Keesling}}, \bibinfo {author} {\bibfnamefont
  {N.}~\bibnamefont {Maskara}}, \bibinfo {author} {\bibfnamefont
  {H.}~\bibnamefont {Pichler}}, \bibinfo {author} {\bibfnamefont
  {M.}~\bibnamefont {Greiner}}, \bibinfo {author} {\bibfnamefont
  {V.}~\bibnamefont {Vuletić}},\ and\ \bibinfo {author} {\bibfnamefont
  {M.~D.}\ \bibnamefont {Lukin}},\ }\href
  {https://doi.org/10.1038/s41586-022-04592-6} {\bibfield  {journal} {\bibinfo
  {journal} {Nature}\ }\textbf {\bibinfo {volume} {604}},\ \bibinfo {pages}
  {451} (\bibinfo {year} {2022})}\BibitemShut {NoStop}%
\bibitem [{\citenamefont {Maslov}\ \emph {et~al.}(2024)\citenamefont {Maslov},
  \citenamefont {Bravyi}, \citenamefont {Tripier}, \citenamefont {Maksymov},\
  and\ \citenamefont {Latone}}]{maslov_fast_2024}%
  \BibitemOpen
  \bibfield  {author} {\bibinfo {author} {\bibfnamefont {D.}~\bibnamefont
  {Maslov}}, \bibinfo {author} {\bibfnamefont {S.}~\bibnamefont {Bravyi}},
  \bibinfo {author} {\bibfnamefont {F.}~\bibnamefont {Tripier}}, \bibinfo
  {author} {\bibfnamefont {A.}~\bibnamefont {Maksymov}},\ and\ \bibinfo
  {author} {\bibfnamefont {J.}~\bibnamefont {Latone}},\ }\href@noop {} {\
  (\bibinfo {year} {2024})},\ \Eprint {https://arxiv.org/abs/2402.03211}
  {arXiv:2402.03211} \BibitemShut {NoStop}%
\bibitem [{\citenamefont {Zhu}\ \emph {et~al.}(2022)\citenamefont {Zhu},
  \citenamefont {Cao}, \citenamefont {Chen}, \citenamefont {Chen},
  \citenamefont {Chen}, \citenamefont {Chung}, \citenamefont {Deng},
  \citenamefont {Du}, \citenamefont {Fan}, \citenamefont {Gong}, \citenamefont
  {Guo}, \citenamefont {Guo}, \citenamefont {Guo}, \citenamefont {Han},
  \citenamefont {Hong}, \citenamefont {Huang}, \citenamefont {Huo},
  \citenamefont {Li},\ and\ \citenamefont {{others}}}]{zhu_quantum_2022}%
  \BibitemOpen
  \bibfield  {author} {\bibinfo {author} {\bibfnamefont {Q.}~\bibnamefont
  {Zhu}}, \bibinfo {author} {\bibfnamefont {S.}~\bibnamefont {Cao}}, \bibinfo
  {author} {\bibfnamefont {F.}~\bibnamefont {Chen}}, \bibinfo {author}
  {\bibfnamefont {M.-C.}\ \bibnamefont {Chen}}, \bibinfo {author}
  {\bibfnamefont {X.}~\bibnamefont {Chen}}, \bibinfo {author} {\bibfnamefont
  {T.-H.}\ \bibnamefont {Chung}}, \bibinfo {author} {\bibfnamefont
  {H.}~\bibnamefont {Deng}}, \bibinfo {author} {\bibfnamefont {Y.}~\bibnamefont
  {Du}}, \bibinfo {author} {\bibfnamefont {D.}~\bibnamefont {Fan}}, \bibinfo
  {author} {\bibfnamefont {M.}~\bibnamefont {Gong}}, \bibinfo {author}
  {\bibfnamefont {C.}~\bibnamefont {Guo}}, \bibinfo {author} {\bibfnamefont
  {C.}~\bibnamefont {Guo}}, \bibinfo {author} {\bibfnamefont {S.}~\bibnamefont
  {Guo}}, \bibinfo {author} {\bibfnamefont {L.}~\bibnamefont {Han}}, \bibinfo
  {author} {\bibfnamefont {L.}~\bibnamefont {Hong}}, \bibinfo {author}
  {\bibfnamefont {H.-L.}\ \bibnamefont {Huang}}, \bibinfo {author}
  {\bibfnamefont {Y.-H.}\ \bibnamefont {Huo}}, \bibinfo {author} {\bibfnamefont
  {L.}~\bibnamefont {Li}},\ and\ \bibinfo {author} {\bibnamefont {{others}}},\
  }\href {https://doi.org/10.1016/j.scib.2021.10.017} {\bibfield  {journal}
  {\bibinfo  {journal} {Science Bulletin}\ }\textbf {\bibinfo {volume} {67}},\
  \bibinfo {pages} {240} (\bibinfo {year} {2022})}\BibitemShut {NoStop}%
\bibitem [{\citenamefont {Mark}\ \emph {et~al.}(2023)\citenamefont {Mark},
  \citenamefont {Choi}, \citenamefont {Shaw}, \citenamefont {Endres},\ and\
  \citenamefont {Choi}}]{mark_benchmarking_2023}%
  \BibitemOpen
  \bibfield  {author} {\bibinfo {author} {\bibfnamefont {D.~K.}\ \bibnamefont
  {Mark}}, \bibinfo {author} {\bibfnamefont {J.}~\bibnamefont {Choi}}, \bibinfo
  {author} {\bibfnamefont {A.~L.}\ \bibnamefont {Shaw}}, \bibinfo {author}
  {\bibfnamefont {M.}~\bibnamefont {Endres}},\ and\ \bibinfo {author}
  {\bibfnamefont {S.}~\bibnamefont {Choi}},\ }\href
  {https://doi.org/10.1103/PhysRevLett.131.110601} {\bibfield  {journal}
  {\bibinfo  {journal} {Phys. Rev. Lett.}\ }\textbf {\bibinfo {volume} {131}},\
  \bibinfo {pages} {110601} (\bibinfo {year} {2023})}\BibitemShut {NoStop}%
\bibitem [{\citenamefont {Gao}\ \emph {et~al.}(2024)\citenamefont {Gao},
  \citenamefont {Kalinowski}, \citenamefont {Chou}, \citenamefont {Lukin},
  \citenamefont {Barak},\ and\ \citenamefont {Choi}}]{gao_limitations_2024}%
  \BibitemOpen
  \bibfield  {author} {\bibinfo {author} {\bibfnamefont {X.}~\bibnamefont
  {Gao}}, \bibinfo {author} {\bibfnamefont {M.}~\bibnamefont {Kalinowski}},
  \bibinfo {author} {\bibfnamefont {C.-N.}\ \bibnamefont {Chou}}, \bibinfo
  {author} {\bibfnamefont {M.~D.}\ \bibnamefont {Lukin}}, \bibinfo {author}
  {\bibfnamefont {B.}~\bibnamefont {Barak}},\ and\ \bibinfo {author}
  {\bibfnamefont {S.}~\bibnamefont {Choi}},\ }\href
  {https://doi.org/10.1103/PRXQuantum.5.010334} {\bibfield  {journal} {\bibinfo
   {journal} {PRX Quantum}\ }\textbf {\bibinfo {volume} {5}},\ \bibinfo {pages}
  {010334} (\bibinfo {year} {2024})}\BibitemShut {NoStop}%
\bibitem [{\citenamefont {Zhou}\ and\ \citenamefont
  {Nahum}(2019)}]{zhou_emergent_2019}%
  \BibitemOpen
  \bibfield  {author} {\bibinfo {author} {\bibfnamefont {T.}~\bibnamefont
  {Zhou}}\ and\ \bibinfo {author} {\bibfnamefont {A.}~\bibnamefont {Nahum}},\
  }\href {https://doi.org/10.1103/PhysRevB.99.174205} {\bibfield  {journal}
  {\bibinfo  {journal} {Phys. Rev. B}\ }\textbf {\bibinfo {volume} {99}},\
  \bibinfo {pages} {174205} (\bibinfo {year} {2019})}\BibitemShut {NoStop}%
\bibitem [{\citenamefont {{Hunter-Jones}}(2019)}]{hunter-jones_unitary_2019}%
  \BibitemOpen
  \bibfield  {author} {\bibinfo {author} {\bibfnamefont {N.}~\bibnamefont
  {{Hunter-Jones}}},\ }\href@noop {} {\  (\bibinfo {year} {2019})},\ \Eprint
  {https://arxiv.org/abs/1905.12053} {arXiv:1905.12053} \BibitemShut {NoStop}%
\bibitem [{\citenamefont {Barak}\ \emph {et~al.}(2021)\citenamefont {Barak},
  \citenamefont {Chou},\ and\ \citenamefont {Gao}}]{barak_spoofing_2021}%
  \BibitemOpen
  \bibfield  {author} {\bibinfo {author} {\bibfnamefont {B.}~\bibnamefont
  {Barak}}, \bibinfo {author} {\bibfnamefont {C.-N.}\ \bibnamefont {Chou}},\
  and\ \bibinfo {author} {\bibfnamefont {X.}~\bibnamefont {Gao}},\ }in\ \href
  {https://doi.org/10.4230/LIPIcs.ITCS.2021.30} {\emph {\bibinfo {booktitle}
  {12th {{Innovations}} in {{Theoretical Computer Science Conference}}
  ({{ITCS}} 2021)}}},\ \bibinfo {series} {Leibniz {{International Proceedings}}
  in {{Informatics}} ({{LIPIcs}})}, Vol.\ \bibinfo {volume} {185},\ \bibinfo
  {editor} {edited by\ \bibinfo {editor} {\bibfnamefont {J.~R.}\ \bibnamefont
  {Lee}}}\ (\bibinfo  {publisher} {{Schloss Dagstuhl\textendash Leibniz-Zentrum
  f\"ur Informatik}},\ \bibinfo {address} {{Dagstuhl, Germany}},\ \bibinfo
  {year} {2021})\ pp.\ \bibinfo {pages} {30:1--30:20},\ \Eprint
  {https://arxiv.org/abs/2005.02421} {arXiv:2005.02421} \BibitemShut {NoStop}%
\bibitem [{\citenamefont {Dalzell}\ \emph {et~al.}(2024)\citenamefont
  {Dalzell}, \citenamefont {{Hunter-Jones}},\ and\ \citenamefont
  {Brand{\~a}o}}]{dalzell_random_2024}%
  \BibitemOpen
  \bibfield  {author} {\bibinfo {author} {\bibfnamefont {A.~M.}\ \bibnamefont
  {Dalzell}}, \bibinfo {author} {\bibfnamefont {N.}~\bibnamefont
  {{Hunter-Jones}}},\ and\ \bibinfo {author} {\bibfnamefont {F.~G. S.~L.}\
  \bibnamefont {Brand{\~a}o}},\ }\href
  {https://doi.org/10.1007/s00220-024-04958-z} {\bibfield  {journal} {\bibinfo
  {journal} {Commun. Math. Phys.}\ }\textbf {\bibinfo {volume} {405}},\
  \bibinfo {pages} {78} (\bibinfo {year} {2024})}\BibitemShut {NoStop}%
\bibitem [{\citenamefont {Hangleiter}\ and\ \citenamefont
  {Gullans}(2024)}]{hangleiter_bell_2024}%
  \BibitemOpen
  \bibfield  {author} {\bibinfo {author} {\bibfnamefont {D.}~\bibnamefont
  {Hangleiter}}\ and\ \bibinfo {author} {\bibfnamefont {M.~J.}\ \bibnamefont
  {Gullans}},\ }\href {https://doi.org/10.1103/PhysRevLett.133.020601}
  {\bibfield  {journal} {\bibinfo  {journal} {Phys. Rev. Lett.}\ }\textbf
  {\bibinfo {volume} {133}},\ \bibinfo {pages} {020601} (\bibinfo {year}
  {2024})},\ \Eprint {https://arxiv.org/abs/2306.00083} {arXiv:2306.00083}
  \BibitemShut {NoStop}%
\bibitem [{\citenamefont {Bremner}\ \emph
  {et~al.}(2016{\natexlab{b}})\citenamefont {Bremner}, \citenamefont
  {Montanaro},\ and\ \citenamefont {Shepherd}}]{bremner_average-case_2016}%
  \BibitemOpen
  \bibfield  {author} {\bibinfo {author} {\bibfnamefont {M.~J.}\ \bibnamefont
  {Bremner}}, \bibinfo {author} {\bibfnamefont {A.}~\bibnamefont {Montanaro}},\
  and\ \bibinfo {author} {\bibfnamefont {D.~J.}\ \bibnamefont {Shepherd}},\
  }\href {https://doi.org/10.1103/PhysRevLett.117.080501} {\bibfield  {journal}
  {\bibinfo  {journal} {Physical Review Letters}\ }\textbf {\bibinfo {volume}
  {117}},\ \bibinfo {pages} {080501} (\bibinfo {year}
  {2016}{\natexlab{b}})}\BibitemShut {NoStop}%
\bibitem [{\citenamefont {Deshpande}\ \emph
  {et~al.}(2022{\natexlab{a}})\citenamefont {Deshpande}, \citenamefont
  {Niroula}, \citenamefont {Shtanko}, \citenamefont {Gorshkov}, \citenamefont
  {Fefferman},\ and\ \citenamefont {Gullans}}]{deshpande_tight_2022}%
  \BibitemOpen
  \bibfield  {author} {\bibinfo {author} {\bibfnamefont {A.}~\bibnamefont
  {Deshpande}}, \bibinfo {author} {\bibfnamefont {P.}~\bibnamefont {Niroula}},
  \bibinfo {author} {\bibfnamefont {O.}~\bibnamefont {Shtanko}}, \bibinfo
  {author} {\bibfnamefont {A.~V.}\ \bibnamefont {Gorshkov}}, \bibinfo {author}
  {\bibfnamefont {B.}~\bibnamefont {Fefferman}},\ and\ \bibinfo {author}
  {\bibfnamefont {M.~J.}\ \bibnamefont {Gullans}},\ }\href
  {https://doi.org/10.1103/PRXQuantum.3.040329} {\bibfield  {journal} {\bibinfo
   {journal} {PRX Quantum}\ }\textbf {\bibinfo {volume} {3}},\ \bibinfo {pages}
  {040329} (\bibinfo {year} {2022}{\natexlab{a}})}\BibitemShut {NoStop}%
\bibitem [{\citenamefont {Mezher}\ \emph {et~al.}(2020)\citenamefont {Mezher},
  \citenamefont {Ghalbouni}, \citenamefont {Dgheim},\ and\ \citenamefont
  {Markham}}]{mezher_fault-tolerant_2020-2}%
  \BibitemOpen
  \bibfield  {author} {\bibinfo {author} {\bibfnamefont {R.}~\bibnamefont
  {Mezher}}, \bibinfo {author} {\bibfnamefont {J.}~\bibnamefont {Ghalbouni}},
  \bibinfo {author} {\bibfnamefont {J.}~\bibnamefont {Dgheim}},\ and\ \bibinfo
  {author} {\bibfnamefont {D.}~\bibnamefont {Markham}},\ }\href
  {https://doi.org/10.1103/PhysRevResearch.2.033444} {\bibfield  {journal}
  {\bibinfo  {journal} {Phys. Rev. Research}\ }\textbf {\bibinfo {volume}
  {2}},\ \bibinfo {pages} {033444} (\bibinfo {year} {2020})}\BibitemShut
  {NoStop}%
\bibitem [{\citenamefont {Paletta}\ \emph {et~al.}(2023)\citenamefont
  {Paletta}, \citenamefont {Leverrier}, \citenamefont {Sarlette}, \citenamefont
  {Mirrahimi},\ and\ \citenamefont {Vuillot}}]{paletta_robust_2023}%
  \BibitemOpen
  \bibfield  {author} {\bibinfo {author} {\bibfnamefont {L.}~\bibnamefont
  {Paletta}}, \bibinfo {author} {\bibfnamefont {A.}~\bibnamefont {Leverrier}},
  \bibinfo {author} {\bibfnamefont {A.}~\bibnamefont {Sarlette}}, \bibinfo
  {author} {\bibfnamefont {M.}~\bibnamefont {Mirrahimi}},\ and\ \bibinfo
  {author} {\bibfnamefont {C.}~\bibnamefont {Vuillot}},\ }\href@noop {} {\
  (\bibinfo {year} {2023})},\ \Eprint {https://arxiv.org/abs/2307.10729}
  {arXiv:2307.10729} \BibitemShut {NoStop}%
\bibitem [{\citenamefont {Wang}\ \emph {et~al.}(2023)\citenamefont {Wang},
  \citenamefont {Simsek}, \citenamefont {Gatterman}, \citenamefont {Gerber},
  \citenamefont {Gilmore}, \citenamefont {Gresh}, \citenamefont {Hewitt},
  \citenamefont {Horst}, \citenamefont {Matheny}, \citenamefont {Mengle},
  \citenamefont {Neyenhuis},\ and\ \citenamefont {Criger}}]{Wang.2023}%
  \BibitemOpen
  \bibfield  {author} {\bibinfo {author} {\bibfnamefont {Y.}~\bibnamefont
  {Wang}}, \bibinfo {author} {\bibfnamefont {S.}~\bibnamefont {Simsek}},
  \bibinfo {author} {\bibfnamefont {T.~M.}\ \bibnamefont {Gatterman}}, \bibinfo
  {author} {\bibfnamefont {J.~A.}\ \bibnamefont {Gerber}}, \bibinfo {author}
  {\bibfnamefont {K.}~\bibnamefont {Gilmore}}, \bibinfo {author} {\bibfnamefont
  {D.}~\bibnamefont {Gresh}}, \bibinfo {author} {\bibfnamefont
  {N.}~\bibnamefont {Hewitt}}, \bibinfo {author} {\bibfnamefont {C.~V.}\
  \bibnamefont {Horst}}, \bibinfo {author} {\bibfnamefont {M.}~\bibnamefont
  {Matheny}}, \bibinfo {author} {\bibfnamefont {T.}~\bibnamefont {Mengle}},
  \bibinfo {author} {\bibfnamefont {B.}~\bibnamefont {Neyenhuis}},\ and\
  \bibinfo {author} {\bibfnamefont {B.}~\bibnamefont {Criger}},\ }\href@noop {}
  {\  (\bibinfo {year} {2023})},\ \Eprint {https://arxiv.org/abs/2309.09893}
  {arXiv:2309.09893} \BibitemShut {NoStop}%
\bibitem [{\citenamefont {Andersen}\ \emph {et~al.}(2020)\citenamefont
  {Andersen}, \citenamefont {Remm}, \citenamefont {Lazar}, \citenamefont
  {Krinner}, \citenamefont {Lacroix}, \citenamefont {Norris}, \citenamefont
  {Gabureac}, \citenamefont {Eichler},\ and\ \citenamefont
  {Wallraff}}]{andersen_repeated_2020}%
  \BibitemOpen
  \bibfield  {author} {\bibinfo {author} {\bibfnamefont {C.~K.}\ \bibnamefont
  {Andersen}}, \bibinfo {author} {\bibfnamefont {A.}~\bibnamefont {Remm}},
  \bibinfo {author} {\bibfnamefont {S.}~\bibnamefont {Lazar}}, \bibinfo
  {author} {\bibfnamefont {S.}~\bibnamefont {Krinner}}, \bibinfo {author}
  {\bibfnamefont {N.}~\bibnamefont {Lacroix}}, \bibinfo {author} {\bibfnamefont
  {G.~J.}\ \bibnamefont {Norris}}, \bibinfo {author} {\bibfnamefont
  {M.}~\bibnamefont {Gabureac}}, \bibinfo {author} {\bibfnamefont
  {C.}~\bibnamefont {Eichler}},\ and\ \bibinfo {author} {\bibfnamefont
  {A.}~\bibnamefont {Wallraff}},\ }\href
  {https://doi.org/10.1038/s41567-020-0920-y} {\bibfield  {journal} {\bibinfo
  {journal} {Nat. Phys.}\ }\textbf {\bibinfo {volume} {16}},\ \bibinfo {pages}
  {875} (\bibinfo {year} {2020})}\BibitemShut {NoStop}%
\bibitem [{\citenamefont {Gupta}\ \emph {et~al.}(2024)\citenamefont {Gupta},
  \citenamefont {Sundaresan}, \citenamefont {Alexander}, \citenamefont {Wood},
  \citenamefont {Merkel}, \citenamefont {Healy}, \citenamefont {Hillenbrand},
  \citenamefont {{Jochym-O'Connor}}, \citenamefont {Wootton}, \citenamefont
  {Yoder}, \citenamefont {Cross}, \citenamefont {Takita},\ and\ \citenamefont
  {Brown}}]{gupta_encoding_2024}%
  \BibitemOpen
  \bibfield  {author} {\bibinfo {author} {\bibfnamefont {R.~S.}\ \bibnamefont
  {Gupta}}, \bibinfo {author} {\bibfnamefont {N.}~\bibnamefont {Sundaresan}},
  \bibinfo {author} {\bibfnamefont {T.}~\bibnamefont {Alexander}}, \bibinfo
  {author} {\bibfnamefont {C.~J.}\ \bibnamefont {Wood}}, \bibinfo {author}
  {\bibfnamefont {S.~T.}\ \bibnamefont {Merkel}}, \bibinfo {author}
  {\bibfnamefont {M.~B.}\ \bibnamefont {Healy}}, \bibinfo {author}
  {\bibfnamefont {M.}~\bibnamefont {Hillenbrand}}, \bibinfo {author}
  {\bibfnamefont {T.}~\bibnamefont {{Jochym-O'Connor}}}, \bibinfo {author}
  {\bibfnamefont {J.~R.}\ \bibnamefont {Wootton}}, \bibinfo {author}
  {\bibfnamefont {T.~J.}\ \bibnamefont {Yoder}}, \bibinfo {author}
  {\bibfnamefont {A.~W.}\ \bibnamefont {Cross}}, \bibinfo {author}
  {\bibfnamefont {M.}~\bibnamefont {Takita}},\ and\ \bibinfo {author}
  {\bibfnamefont {B.~J.}\ \bibnamefont {Brown}},\ }\href
  {https://doi.org/10.1038/s41586-023-06846-3} {\bibfield  {journal} {\bibinfo
  {journal} {Nature}\ }\textbf {\bibinfo {volume} {625}},\ \bibinfo {pages}
  {259} (\bibinfo {year} {2024})}\BibitemShut {NoStop}%
\bibitem [{\citenamefont {Postler}\ \emph {et~al.}(2022)\citenamefont
  {Postler}, \citenamefont {Heu{$\beta$}en}, \citenamefont {Pogorelov},
  \citenamefont {Rispler}, \citenamefont {Feldker}, \citenamefont {Meth},
  \citenamefont {Marciniak}, \citenamefont {Stricker}, \citenamefont
  {Ringbauer}, \citenamefont {Blatt}, \citenamefont {Schindler}, \citenamefont
  {M{\"u}ller},\ and\ \citenamefont {Monz}}]{postler_demonstration_2022}%
  \BibitemOpen
  \bibfield  {author} {\bibinfo {author} {\bibfnamefont {L.}~\bibnamefont
  {Postler}}, \bibinfo {author} {\bibfnamefont {S.}~\bibnamefont
  {Heu{$\beta$}en}}, \bibinfo {author} {\bibfnamefont {I.}~\bibnamefont
  {Pogorelov}}, \bibinfo {author} {\bibfnamefont {M.}~\bibnamefont {Rispler}},
  \bibinfo {author} {\bibfnamefont {T.}~\bibnamefont {Feldker}}, \bibinfo
  {author} {\bibfnamefont {M.}~\bibnamefont {Meth}}, \bibinfo {author}
  {\bibfnamefont {C.~D.}\ \bibnamefont {Marciniak}}, \bibinfo {author}
  {\bibfnamefont {R.}~\bibnamefont {Stricker}}, \bibinfo {author}
  {\bibfnamefont {M.}~\bibnamefont {Ringbauer}}, \bibinfo {author}
  {\bibfnamefont {R.}~\bibnamefont {Blatt}}, \bibinfo {author} {\bibfnamefont
  {P.}~\bibnamefont {Schindler}}, \bibinfo {author} {\bibfnamefont
  {M.}~\bibnamefont {M{\"u}ller}},\ and\ \bibinfo {author} {\bibfnamefont
  {T.}~\bibnamefont {Monz}},\ }\href
  {https://doi.org/10.1038/s41586-022-04721-1} {\bibfield  {journal} {\bibinfo
  {journal} {Nature}\ }\textbf {\bibinfo {volume} {605}},\ \bibinfo {pages}
  {675} (\bibinfo {year} {2022})}\BibitemShut {NoStop}%
\bibitem [{\citenamefont {Bravyi}\ \emph {et~al.}(2024)\citenamefont {Bravyi},
  \citenamefont {Cross}, \citenamefont {Gambetta}, \citenamefont {Maslov},
  \citenamefont {Rall},\ and\ \citenamefont
  {Yoder}}]{bravyi_high-threshold_2024}%
  \BibitemOpen
  \bibfield  {author} {\bibinfo {author} {\bibfnamefont {S.}~\bibnamefont
  {Bravyi}}, \bibinfo {author} {\bibfnamefont {A.~W.}\ \bibnamefont {Cross}},
  \bibinfo {author} {\bibfnamefont {J.~M.}\ \bibnamefont {Gambetta}}, \bibinfo
  {author} {\bibfnamefont {D.}~\bibnamefont {Maslov}}, \bibinfo {author}
  {\bibfnamefont {P.}~\bibnamefont {Rall}},\ and\ \bibinfo {author}
  {\bibfnamefont {T.~J.}\ \bibnamefont {Yoder}},\ }\href
  {https://doi.org/10.1038/s41586-024-07107-7} {\bibfield  {journal} {\bibinfo
  {journal} {Nature}\ }\textbf {\bibinfo {volume} {627}},\ \bibinfo {pages}
  {778} (\bibinfo {year} {2024})}\BibitemShut {NoStop}%
\bibitem [{\citenamefont {Xu}\ \emph {et~al.}(2023)\citenamefont {Xu},
  \citenamefont {Ataides}, \citenamefont {Pattison}, \citenamefont
  {Raveendran}, \citenamefont {Bluvstein}, \citenamefont {Wurtz}, \citenamefont
  {Vasic}, \citenamefont {Lukin}, \citenamefont {Jiang},\ and\ \citenamefont
  {Zhou}}]{xu_constant-overhead_2023}%
  \BibitemOpen
  \bibfield  {author} {\bibinfo {author} {\bibfnamefont {Q.}~\bibnamefont
  {Xu}}, \bibinfo {author} {\bibfnamefont {J.~P.~B.}\ \bibnamefont {Ataides}},
  \bibinfo {author} {\bibfnamefont {C.~A.}\ \bibnamefont {Pattison}}, \bibinfo
  {author} {\bibfnamefont {N.}~\bibnamefont {Raveendran}}, \bibinfo {author}
  {\bibfnamefont {D.}~\bibnamefont {Bluvstein}}, \bibinfo {author}
  {\bibfnamefont {J.}~\bibnamefont {Wurtz}}, \bibinfo {author} {\bibfnamefont
  {B.}~\bibnamefont {Vasic}}, \bibinfo {author} {\bibfnamefont {M.~D.}\
  \bibnamefont {Lukin}}, \bibinfo {author} {\bibfnamefont {L.}~\bibnamefont
  {Jiang}},\ and\ \bibinfo {author} {\bibfnamefont {H.}~\bibnamefont {Zhou}},\
  }\href@noop {} {\  (\bibinfo {year} {2023})},\ \Eprint
  {https://arxiv.org/abs/2308.08648} {arXiv:2308.08648} \BibitemShut {NoStop}%
\bibitem [{\citenamefont {Rengaswamy}\ \emph {et~al.}(2020)\citenamefont
  {Rengaswamy}, \citenamefont {Calderbank}, \citenamefont {Newman},\ and\
  \citenamefont {Pfister}}]{Rengaswamy20}%
  \BibitemOpen
  \bibfield  {author} {\bibinfo {author} {\bibfnamefont {N.}~\bibnamefont
  {Rengaswamy}}, \bibinfo {author} {\bibfnamefont {R.}~\bibnamefont
  {Calderbank}}, \bibinfo {author} {\bibfnamefont {M.}~\bibnamefont {Newman}},\
  and\ \bibinfo {author} {\bibfnamefont {H.~D.}\ \bibnamefont {Pfister}},\
  }\href {https://doi.org/10.1109/jsait.2020.3012914} {\bibfield  {journal}
  {\bibinfo  {journal} {IEEE JSAIT}\ }\textbf {\bibinfo {volume} {1}},\
  \bibinfo {pages} {499} (\bibinfo {year} {2020})},\ \Eprint
  {https://arxiv.org/abs/1910.09333} {arXiv:1910.09333} \BibitemShut {NoStop}%
\bibitem [{\citenamefont {Fujiwara}(2014)}]{fujiwara_instantaneous_2014}%
  \BibitemOpen
  \bibfield  {author} {\bibinfo {author} {\bibfnamefont {Y.}~\bibnamefont
  {Fujiwara}},\ }\href@noop {} {\  (\bibinfo {year} {2014})},\ \Eprint
  {https://arxiv.org/abs/1405.6267} {arXiv:1405.6267} \BibitemShut {NoStop}%
\bibitem [{\citenamefont {Fowler}\ \emph {et~al.}(2014)\citenamefont {Fowler},
  \citenamefont {Sank}, \citenamefont {Kelly}, \citenamefont {Barends},\ and\
  \citenamefont {Martinis}}]{fowler_scalable_2014}%
  \BibitemOpen
  \bibfield  {author} {\bibinfo {author} {\bibfnamefont {A.~G.}\ \bibnamefont
  {Fowler}}, \bibinfo {author} {\bibfnamefont {D.}~\bibnamefont {Sank}},
  \bibinfo {author} {\bibfnamefont {J.}~\bibnamefont {Kelly}}, \bibinfo
  {author} {\bibfnamefont {R.}~\bibnamefont {Barends}},\ and\ \bibinfo {author}
  {\bibfnamefont {J.~M.}\ \bibnamefont {Martinis}},\ }\href@noop {} {\
  (\bibinfo {year} {2014})},\ \Eprint {https://arxiv.org/abs/1405.1454}
  {arXiv:1405.1454} \BibitemShut {NoStop}%
\bibitem [{\citenamefont {Huo}\ and\ \citenamefont
  {Li}(2017)}]{huo_learning_2017}%
  \BibitemOpen
  \bibfield  {author} {\bibinfo {author} {\bibfnamefont {M.-X.}\ \bibnamefont
  {Huo}}\ and\ \bibinfo {author} {\bibfnamefont {Y.}~\bibnamefont {Li}},\
  }\href {https://doi.org/10.1088/1367-2630/aa916e} {\bibfield  {journal}
  {\bibinfo  {journal} {New J. Phys.}\ }\textbf {\bibinfo {volume} {19}},\
  \bibinfo {pages} {123032} (\bibinfo {year} {2017})}\BibitemShut {NoStop}%
\bibitem [{\citenamefont {Wagner}\ \emph {et~al.}(2022)\citenamefont {Wagner},
  \citenamefont {Kampermann}, \citenamefont {Bru{\ss}},\ and\ \citenamefont
  {Kliesch}}]{wagner_pauli_2022}%
  \BibitemOpen
  \bibfield  {author} {\bibinfo {author} {\bibfnamefont {T.}~\bibnamefont
  {Wagner}}, \bibinfo {author} {\bibfnamefont {H.}~\bibnamefont {Kampermann}},
  \bibinfo {author} {\bibfnamefont {D.}~\bibnamefont {Bru{\ss}}},\ and\
  \bibinfo {author} {\bibfnamefont {M.}~\bibnamefont {Kliesch}},\ }\href
  {https://doi.org/10.22331/q-2022-09-19-809} {\bibfield  {journal} {\bibinfo
  {journal} {Quantum}\ }\textbf {\bibinfo {volume} {6}},\ \bibinfo {pages}
  {809} (\bibinfo {year} {2022})}\BibitemShut {NoStop}%
\bibitem [{\citenamefont {Iyer}\ \emph {et~al.}(2022)\citenamefont {Iyer},
  \citenamefont {Jain}, \citenamefont {Bartlett},\ and\ \citenamefont
  {Emerson}}]{iyer_efficient_2022}%
  \BibitemOpen
  \bibfield  {author} {\bibinfo {author} {\bibfnamefont {P.}~\bibnamefont
  {Iyer}}, \bibinfo {author} {\bibfnamefont {A.}~\bibnamefont {Jain}}, \bibinfo
  {author} {\bibfnamefont {S.~D.}\ \bibnamefont {Bartlett}},\ and\ \bibinfo
  {author} {\bibfnamefont {J.}~\bibnamefont {Emerson}},\ }\href
  {https://doi.org/10.1103/PhysRevResearch.4.043218} {\bibfield  {journal}
  {\bibinfo  {journal} {Phys. Rev. Res.}\ }\textbf {\bibinfo {volume} {4}},\
  \bibinfo {pages} {043218} (\bibinfo {year} {2022})}\BibitemShut {NoStop}%
\bibitem [{\citenamefont {Wagner}\ \emph {et~al.}(2023)\citenamefont {Wagner},
  \citenamefont {Kampermann}, \citenamefont {Bru{\ss}},\ and\ \citenamefont
  {Kliesch}}]{wagner_learning_2023}%
  \BibitemOpen
  \bibfield  {author} {\bibinfo {author} {\bibfnamefont {T.}~\bibnamefont
  {Wagner}}, \bibinfo {author} {\bibfnamefont {H.}~\bibnamefont {Kampermann}},
  \bibinfo {author} {\bibfnamefont {D.}~\bibnamefont {Bru{\ss}}},\ and\
  \bibinfo {author} {\bibfnamefont {M.}~\bibnamefont {Kliesch}},\ }\href
  {https://doi.org/10.1103/PhysRevLett.130.200601} {\bibfield  {journal}
  {\bibinfo  {journal} {Phys. Rev. Lett.}\ }\textbf {\bibinfo {volume} {130}},\
  \bibinfo {pages} {200601} (\bibinfo {year} {2023})}\BibitemShut {NoStop}%
\bibitem [{\citenamefont {Flammia}\ and\ \citenamefont
  {Wallman}(2020)}]{flammia_efficient_2020}%
  \BibitemOpen
  \bibfield  {author} {\bibinfo {author} {\bibfnamefont {S.~T.}\ \bibnamefont
  {Flammia}}\ and\ \bibinfo {author} {\bibfnamefont {J.~J.}\ \bibnamefont
  {Wallman}},\ }\href {https://doi.org/10.1145/3408039} {\bibfield  {journal}
  {\bibinfo  {journal} {ACM Transactions on Quantum Computing}\ }\textbf
  {\bibinfo {volume} {1}},\ \bibinfo {pages} {3:1} (\bibinfo {year}
  {2020})}\BibitemShut {NoStop}%
\bibitem [{\citenamefont {Wallman}\ and\ \citenamefont
  {Emerson}(2016)}]{wallman_noise_2016}%
  \BibitemOpen
  \bibfield  {author} {\bibinfo {author} {\bibfnamefont {J.~J.}\ \bibnamefont
  {Wallman}}\ and\ \bibinfo {author} {\bibfnamefont {J.}~\bibnamefont
  {Emerson}},\ }\href {https://doi.org/10.1103/PhysRevA.94.052325} {\bibfield
  {journal} {\bibinfo  {journal} {Phys. Rev. A}\ }\textbf {\bibinfo {volume}
  {94}},\ \bibinfo {pages} {052325} (\bibinfo {year} {2016})}\BibitemShut
  {NoStop}%
\bibitem [{\citenamefont {Erhard}\ \emph {et~al.}(2019)\citenamefont {Erhard},
  \citenamefont {Wallman}, \citenamefont {Postler}, \citenamefont {Meth},
  \citenamefont {Stricker}, \citenamefont {Martinez}, \citenamefont
  {Schindler}, \citenamefont {Monz}, \citenamefont {Emerson},\ and\
  \citenamefont {Blatt}}]{erhard_characterizing_2019}%
  \BibitemOpen
  \bibfield  {author} {\bibinfo {author} {\bibfnamefont {A.}~\bibnamefont
  {Erhard}}, \bibinfo {author} {\bibfnamefont {J.~J.}\ \bibnamefont {Wallman}},
  \bibinfo {author} {\bibfnamefont {L.}~\bibnamefont {Postler}}, \bibinfo
  {author} {\bibfnamefont {M.}~\bibnamefont {Meth}}, \bibinfo {author}
  {\bibfnamefont {R.}~\bibnamefont {Stricker}}, \bibinfo {author}
  {\bibfnamefont {E.~A.}\ \bibnamefont {Martinez}}, \bibinfo {author}
  {\bibfnamefont {P.}~\bibnamefont {Schindler}}, \bibinfo {author}
  {\bibfnamefont {T.}~\bibnamefont {Monz}}, \bibinfo {author} {\bibfnamefont
  {J.}~\bibnamefont {Emerson}},\ and\ \bibinfo {author} {\bibfnamefont
  {R.}~\bibnamefont {Blatt}},\ }\href
  {https://doi.org/10.1038/s41467-019-13068-7} {\bibfield  {journal} {\bibinfo
  {journal} {Nat Commun}\ }\textbf {\bibinfo {volume} {10}},\ \bibinfo {pages}
  {5347} (\bibinfo {year} {2019})},\ \Eprint {https://arxiv.org/abs/1902.08543}
  {arXiv:1902.08543} \BibitemShut {NoStop}%
\bibitem [{\citenamefont {Ringbauer}\ \emph {et~al.}(2024)\citenamefont
  {Ringbauer}, \citenamefont {Hinsche}, \citenamefont {Feldker}, \citenamefont
  {Faehrmann}, \citenamefont {{Bermejo-Vega}}, \citenamefont {Edmunds},
  \citenamefont {Postler}, \citenamefont {Stricker}, \citenamefont {Marciniak},
  \citenamefont {Meth}, \citenamefont {Pogorelov}, \citenamefont {Blatt},
  \citenamefont {Schindler}, \citenamefont {Eisert}, \citenamefont {Monz},\
  and\ \citenamefont {Hangleiter}}]{ringbauer_verifiable_2024}%
  \BibitemOpen
  \bibfield  {author} {\bibinfo {author} {\bibfnamefont {M.}~\bibnamefont
  {Ringbauer}}, \bibinfo {author} {\bibfnamefont {M.}~\bibnamefont {Hinsche}},
  \bibinfo {author} {\bibfnamefont {T.}~\bibnamefont {Feldker}}, \bibinfo
  {author} {\bibfnamefont {P.~K.}\ \bibnamefont {Faehrmann}}, \bibinfo {author}
  {\bibfnamefont {J.}~\bibnamefont {{Bermejo-Vega}}}, \bibinfo {author}
  {\bibfnamefont {C.}~\bibnamefont {Edmunds}}, \bibinfo {author} {\bibfnamefont
  {L.}~\bibnamefont {Postler}}, \bibinfo {author} {\bibfnamefont
  {R.}~\bibnamefont {Stricker}}, \bibinfo {author} {\bibfnamefont {C.~D.}\
  \bibnamefont {Marciniak}}, \bibinfo {author} {\bibfnamefont {M.}~\bibnamefont
  {Meth}}, \bibinfo {author} {\bibfnamefont {I.}~\bibnamefont {Pogorelov}},
  \bibinfo {author} {\bibfnamefont {R.}~\bibnamefont {Blatt}}, \bibinfo
  {author} {\bibfnamefont {P.}~\bibnamefont {Schindler}}, \bibinfo {author}
  {\bibfnamefont {J.}~\bibnamefont {Eisert}}, \bibinfo {author} {\bibfnamefont
  {T.}~\bibnamefont {Monz}},\ and\ \bibinfo {author} {\bibfnamefont
  {D.}~\bibnamefont {Hangleiter}},\ }\href@noop {} {\  (\bibinfo {year}
  {2024})},\ \Eprint {https://arxiv.org/abs/2307.14424} {arXiv:2307.14424}
  \BibitemShut {NoStop}%
\bibitem [{\citenamefont {Gottesman}(2024)}]{gottesman_surviving_2024}%
  \BibitemOpen
  \bibfield  {author} {\bibinfo {author} {\bibfnamefont {D.}~\bibnamefont
  {Gottesman}},\ }\href@noop {} {\emph {\bibinfo {title} {Surviving as a
  Quantum Computer in a Classical World}}}\ (\bibinfo {year}
  {2024})\BibitemShut {NoStop}%
\bibitem [{\citenamefont {Vasmer}\ and\ \citenamefont
  {Kubica}(2022)}]{Vasmer2022}%
  \BibitemOpen
  \bibfield  {author} {\bibinfo {author} {\bibfnamefont {M.}~\bibnamefont
  {Vasmer}}\ and\ \bibinfo {author} {\bibfnamefont {A.}~\bibnamefont
  {Kubica}},\ }\href {https://doi.org/10.1103/PRXQuantum.3.030319} {\bibfield
  {journal} {\bibinfo  {journal} {PRX Quantum}\ }\textbf {\bibinfo {volume}
  {3}},\ \bibinfo {pages} {030319} (\bibinfo {year} {2022})}\BibitemShut
  {NoStop}%
\bibitem [{\citenamefont {Sheng}\ \emph {et~al.}(2018)\citenamefont {Sheng},
  \citenamefont {He}, \citenamefont {Xu}, \citenamefont {Guo}, \citenamefont
  {Wang}, \citenamefont {Xiong}, \citenamefont {Liu}, \citenamefont {Wang},\
  and\ \citenamefont {Zhan}}]{sheng_high-fidelity_2018}%
  \BibitemOpen
  \bibfield  {author} {\bibinfo {author} {\bibfnamefont {C.}~\bibnamefont
  {Sheng}}, \bibinfo {author} {\bibfnamefont {X.}~\bibnamefont {He}}, \bibinfo
  {author} {\bibfnamefont {P.}~\bibnamefont {Xu}}, \bibinfo {author}
  {\bibfnamefont {R.}~\bibnamefont {Guo}}, \bibinfo {author} {\bibfnamefont
  {K.}~\bibnamefont {Wang}}, \bibinfo {author} {\bibfnamefont {Z.}~\bibnamefont
  {Xiong}}, \bibinfo {author} {\bibfnamefont {M.}~\bibnamefont {Liu}}, \bibinfo
  {author} {\bibfnamefont {J.}~\bibnamefont {Wang}},\ and\ \bibinfo {author}
  {\bibfnamefont {M.}~\bibnamefont {Zhan}},\ }\href
  {https://doi.org/10.1103/physrevlett.121.240501} {\bibfield  {journal}
  {\bibinfo  {journal} {Phys. Rev. Lett.}\ }\textbf {\bibinfo {volume} {121}},\
  \bibinfo {pages} {240501} (\bibinfo {year} {2018})}\BibitemShut {NoStop}%
\bibitem [{\citenamefont {Levine}\ \emph {et~al.}(2022)\citenamefont {Levine},
  \citenamefont {Bluvstein}, \citenamefont {Keesling}, \citenamefont {Wang},
  \citenamefont {Ebadi}, \citenamefont {Semeghini}, \citenamefont {Omran},
  \citenamefont {Greiner}, \citenamefont {Vuletić},\ and\ \citenamefont
  {Lukin}}]{levine_dispersive_2022}%
  \BibitemOpen
  \bibfield  {author} {\bibinfo {author} {\bibfnamefont {H.}~\bibnamefont
  {Levine}}, \bibinfo {author} {\bibfnamefont {D.}~\bibnamefont {Bluvstein}},
  \bibinfo {author} {\bibfnamefont {A.}~\bibnamefont {Keesling}}, \bibinfo
  {author} {\bibfnamefont {T.~T.}\ \bibnamefont {Wang}}, \bibinfo {author}
  {\bibfnamefont {S.}~\bibnamefont {Ebadi}}, \bibinfo {author} {\bibfnamefont
  {G.}~\bibnamefont {Semeghini}}, \bibinfo {author} {\bibfnamefont
  {A.}~\bibnamefont {Omran}}, \bibinfo {author} {\bibfnamefont
  {M.}~\bibnamefont {Greiner}}, \bibinfo {author} {\bibfnamefont
  {V.}~\bibnamefont {Vuletić}},\ and\ \bibinfo {author} {\bibfnamefont
  {M.~D.}\ \bibnamefont {Lukin}},\ }\href
  {https://doi.org/10.1103/physreva.105.032618} {\bibfield  {journal} {\bibinfo
   {journal} {Phys. Rev. A}\ }\textbf {\bibinfo {volume} {105}},\ \bibinfo
  {pages} {032618} (\bibinfo {year} {2022})}\BibitemShut {NoStop}%
\bibitem [{\citenamefont {Campbell}(2016)}]{campbell_blog}%
  \BibitemOpen
  \bibfield  {author} {\bibinfo {author} {\bibfnamefont {E.}~\bibnamefont
  {Campbell}},\ }\href
  {https://earltcampbell.com/2016/09/26/the-smallest-interest ing-colour-code/}
  {\bibinfo {title} {The smallest interesting colour code}} (\bibinfo {year}
  {2016}),\ \bibinfo {note} {accessed: 2023-01-29}\BibitemShut {NoStop}%
\bibitem [{\citenamefont {Eastin}\ and\ \citenamefont
  {Knill}(2009{\natexlab{a}})}]{Eastin-Knill.2009}%
  \BibitemOpen
  \bibfield  {author} {\bibinfo {author} {\bibfnamefont {B.}~\bibnamefont
  {Eastin}}\ and\ \bibinfo {author} {\bibfnamefont {E.}~\bibnamefont {Knill}},\
  }\href {https://doi.org/10.1103/PhysRevLett.102.110502} {\bibfield  {journal}
  {\bibinfo  {journal} {Phys. Rev. Lett.}\ }\textbf {\bibinfo {volume} {102}},\
  \bibinfo {pages} {110502} (\bibinfo {year} {2009}{\natexlab{a}})}\BibitemShut
  {NoStop}%
\bibitem [{\citenamefont {Calderbank}\ and\ \citenamefont
  {Shor}(1996)}]{calderbank_good_1996}%
  \BibitemOpen
  \bibfield  {author} {\bibinfo {author} {\bibfnamefont {A.~R.}\ \bibnamefont
  {Calderbank}}\ and\ \bibinfo {author} {\bibfnamefont {P.~W.}\ \bibnamefont
  {Shor}},\ }\href {https://doi.org/10.1103/PhysRevA.54.1098} {\bibfield
  {journal} {\bibinfo  {journal} {Phys. Rev. A}\ }\textbf {\bibinfo {volume}
  {54}},\ \bibinfo {pages} {1098} (\bibinfo {year} {1996})}\BibitemShut
  {NoStop}%
\bibitem [{Note2()}]{Note2}%
  \BibitemOpen
  \bibinfo {note} {This is implied by the quadratic decay of the XEB of the
  simulation of the $[\protect \![8,3,2]\protect \!]$-code with error detection
  at the end of the circuit in \protect \cref
  {fig:code_comparison}.}\BibitemShut {Stop}%
\bibitem [{\citenamefont {Bravyi}\ \emph {et~al.}(2020)\citenamefont {Bravyi},
  \citenamefont {Gosset}, \citenamefont {K{\"o}nig},\ and\ \citenamefont
  {Tomamichel}}]{bravyi_quantum_2020}%
  \BibitemOpen
  \bibfield  {author} {\bibinfo {author} {\bibfnamefont {S.}~\bibnamefont
  {Bravyi}}, \bibinfo {author} {\bibfnamefont {D.}~\bibnamefont {Gosset}},
  \bibinfo {author} {\bibfnamefont {R.}~\bibnamefont {K{\"o}nig}},\ and\
  \bibinfo {author} {\bibfnamefont {M.}~\bibnamefont {Tomamichel}},\ }\href
  {https://doi.org/10.1038/s41567-020-0948-z} {\bibfield  {journal} {\bibinfo
  {journal} {Nat. Phys.}\ }\textbf {\bibinfo {volume} {16}},\ \bibinfo {pages}
  {1040} (\bibinfo {year} {2020})}\BibitemShut {NoStop}%
\bibitem [{\citenamefont {Raussendorf}\ \emph {et~al.}(2006)\citenamefont
  {Raussendorf}, \citenamefont {Harrington},\ and\ \citenamefont
  {Goyal}}]{raussendorf_fault-tolerant_2006}%
  \BibitemOpen
  \bibfield  {author} {\bibinfo {author} {\bibfnamefont {R.}~\bibnamefont
  {Raussendorf}}, \bibinfo {author} {\bibfnamefont {J.}~\bibnamefont
  {Harrington}},\ and\ \bibinfo {author} {\bibfnamefont {K.}~\bibnamefont
  {Goyal}},\ }\href {https://doi.org/10.1016/j.aop.2006.01.012} {\bibfield
  {journal} {\bibinfo  {journal} {Annals of Physics}\ }\textbf {\bibinfo
  {volume} {321}},\ \bibinfo {pages} {2242} (\bibinfo {year}
  {2006})}\BibitemShut {NoStop}%
\bibitem [{\citenamefont {{Jochym-O'Connor}}\ \emph {et~al.}(2018)\citenamefont
  {{Jochym-O'Connor}}, \citenamefont {Kubica},\ and\ \citenamefont
  {Yoder}}]{jochym-oconnor_disjointness_2018}%
  \BibitemOpen
  \bibfield  {author} {\bibinfo {author} {\bibfnamefont {T.}~\bibnamefont
  {{Jochym-O'Connor}}}, \bibinfo {author} {\bibfnamefont {A.}~\bibnamefont
  {Kubica}},\ and\ \bibinfo {author} {\bibfnamefont {T.~J.}\ \bibnamefont
  {Yoder}},\ }\href {https://doi.org/10.1103/PhysRevX.8.021047} {\bibfield
  {journal} {\bibinfo  {journal} {Phys. Rev. X}\ }\textbf {\bibinfo {volume}
  {8}},\ \bibinfo {pages} {021047} (\bibinfo {year} {2018})}\BibitemShut
  {NoStop}%
\bibitem [{\citenamefont {Newman}\ and\ \citenamefont
  {Shi}(2018)}]{newman_limitations_2018}%
  \BibitemOpen
  \bibfield  {author} {\bibinfo {author} {\bibfnamefont {M.}~\bibnamefont
  {Newman}}\ and\ \bibinfo {author} {\bibfnamefont {Y.}~\bibnamefont {Shi}},\
  }\href {https://doi.org/10.48550/arXiv.1704.07798} {\bibfield  {journal}
  {\bibinfo  {journal} {Quant. Inf. Comp.}\ }\textbf {\bibinfo {volume} {18}},\
  \bibinfo {pages} {0927} (\bibinfo {year} {2018})},\ \Eprint
  {https://arxiv.org/abs/1704.07798} {arXiv:1704.07798} \BibitemShut {NoStop}%
\bibitem [{\citenamefont {Beigi}\ and\ \citenamefont
  {Shor}(2010)}]{beigi_c3_2010}%
  \BibitemOpen
  \bibfield  {author} {\bibinfo {author} {\bibfnamefont {S.}~\bibnamefont
  {Beigi}}\ and\ \bibinfo {author} {\bibfnamefont {P.~W.}\ \bibnamefont
  {Shor}},\ }\href {https://doi.org/10.48550/arXiv.0810.5108} {\bibfield
  {journal} {\bibinfo  {journal} {Quant. Inf. Comp.}\ }\textbf {\bibinfo
  {volume} {10}},\ \bibinfo {pages} {0041} (\bibinfo {year}
  {2010})}\BibitemShut {NoStop}%
\bibitem [{\citenamefont {Zeng}\ \emph {et~al.}(2008)\citenamefont {Zeng},
  \citenamefont {Chen},\ and\ \citenamefont {Chuang}}]{Zeng08}%
  \BibitemOpen
  \bibfield  {author} {\bibinfo {author} {\bibfnamefont {B.}~\bibnamefont
  {Zeng}}, \bibinfo {author} {\bibfnamefont {X.}~\bibnamefont {Chen}},\ and\
  \bibinfo {author} {\bibfnamefont {I.~L.}\ \bibnamefont {Chuang}},\ }\href
  {https://doi.org/10.1103/PhysRevA.77.042313} {\bibfield  {journal} {\bibinfo
  {journal} {Phys. Rev. A}\ }\textbf {\bibinfo {volume} {77}},\ \bibinfo
  {pages} {042313} (\bibinfo {year} {2008})}\BibitemShut {NoStop}%
\bibitem [{\citenamefont {Eastin}\ and\ \citenamefont
  {Knill}(2009{\natexlab{b}})}]{eastin_restrictions_2009}%
  \BibitemOpen
  \bibfield  {author} {\bibinfo {author} {\bibfnamefont {B.}~\bibnamefont
  {Eastin}}\ and\ \bibinfo {author} {\bibfnamefont {E.}~\bibnamefont {Knill}},\
  }\href {https://doi.org/10.1103/PhysRevLett.102.110502} {\bibfield  {journal}
  {\bibinfo  {journal} {Phys. Rev. Lett.}\ }\textbf {\bibinfo {volume} {102}},\
  \bibinfo {pages} {110502} (\bibinfo {year} {2009}{\natexlab{b}})}\BibitemShut
  {NoStop}%
\bibitem [{\citenamefont {Anderson}(2022)}]{anderson_groups_2022}%
  \BibitemOpen
  \bibfield  {author} {\bibinfo {author} {\bibfnamefont {J.~T.}\ \bibnamefont
  {Anderson}},\ }\href@noop {} {\  (\bibinfo {year} {2022})},\ \Eprint
  {https://arxiv.org/abs/2212.05398} {arXiv:2212.05398} \BibitemShut {NoStop}%
\bibitem [{\citenamefont {Cui}\ \emph {et~al.}(2017)\citenamefont {Cui},
  \citenamefont {Gottesman},\ and\ \citenamefont {Krishna}}]{Cui17}%
  \BibitemOpen
  \bibfield  {author} {\bibinfo {author} {\bibfnamefont {S.~X.}\ \bibnamefont
  {Cui}}, \bibinfo {author} {\bibfnamefont {D.}~\bibnamefont {Gottesman}},\
  and\ \bibinfo {author} {\bibfnamefont {A.}~\bibnamefont {Krishna}},\ }\href
  {https://doi.org/10.1103/PhysRevA.95.012329} {\bibfield  {journal} {\bibinfo
  {journal} {Phys. Rev. A}\ }\textbf {\bibinfo {volume} {95}},\ \bibinfo
  {pages} {012329} (\bibinfo {year} {2017})}\BibitemShut {NoStop}%
\bibitem [{\citenamefont {Kubischta}\ and\ \citenamefont
  {Teixeira}(2023{\natexlab{a}})}]{Kubischta23}%
  \BibitemOpen
  \bibfield  {author} {\bibinfo {author} {\bibfnamefont {E.}~\bibnamefont
  {Kubischta}}\ and\ \bibinfo {author} {\bibfnamefont {I.}~\bibnamefont
  {Teixeira}},\ }\href {https://doi.org/10.1103/PhysRevLett.131.240601}
  {\bibfield  {journal} {\bibinfo  {journal} {Phys. Rev. Lett.}\ }\textbf
  {\bibinfo {volume} {131}},\ \bibinfo {pages} {240601} (\bibinfo {year}
  {2023}{\natexlab{a}})}\BibitemShut {NoStop}%
\bibitem [{\citenamefont {Kubischta}\ and\ \citenamefont
  {Teixeira}(2023{\natexlab{b}})}]{kubischta_not-so-secret_2023}%
  \BibitemOpen
  \bibfield  {author} {\bibinfo {author} {\bibfnamefont {E.}~\bibnamefont
  {Kubischta}}\ and\ \bibinfo {author} {\bibfnamefont {I.}~\bibnamefont
  {Teixeira}},\ }\href@noop {} {\  (\bibinfo {year} {2023}{\natexlab{b}})},\
  \Eprint {https://arxiv.org/abs/2310.17652} {arXiv:2310.17652} \BibitemShut
  {NoStop}%
\bibitem [{\citenamefont {Page}(1993)}]{Page.1993}%
  \BibitemOpen
  \bibfield  {author} {\bibinfo {author} {\bibfnamefont {D.~N.}\ \bibnamefont
  {Page}},\ }\href {https://doi.org/10.1103/PhysRevLett.71.1291} {\bibfield
  {journal} {\bibinfo  {journal} {Phys. Rev. Lett.}\ }\textbf {\bibinfo
  {volume} {71}},\ \bibinfo {pages} {1291} (\bibinfo {year}
  {1993})}\BibitemShut {NoStop}%
\bibitem [{\citenamefont {Hashizume}\ \emph {et~al.}(2021)\citenamefont
  {Hashizume}, \citenamefont {Bentsen}, \citenamefont {Weber},\ and\
  \citenamefont {Daley}}]{Hashizume.2021}%
  \BibitemOpen
  \bibfield  {author} {\bibinfo {author} {\bibfnamefont {T.}~\bibnamefont
  {Hashizume}}, \bibinfo {author} {\bibfnamefont {G.~S.}\ \bibnamefont
  {Bentsen}}, \bibinfo {author} {\bibfnamefont {S.}~\bibnamefont {Weber}},\
  and\ \bibinfo {author} {\bibfnamefont {A.~J.}\ \bibnamefont {Daley}},\ }\href
  {https://doi.org/10.1103/PhysRevLett.126.200603} {\bibfield  {journal}
  {\bibinfo  {journal} {Phys. Rev. Lett.}\ }\textbf {\bibinfo {volume} {126}},\
  \bibinfo {pages} {200603} (\bibinfo {year} {2021})}\BibitemShut {NoStop}%
\bibitem [{\citenamefont {Rossi}\ \emph {et~al.}(2013)\citenamefont {Rossi},
  \citenamefont {Huber}, \citenamefont {Bru{\ss}},\ and\ \citenamefont
  {Macchiavello}}]{rossi_quantum_2013}%
  \BibitemOpen
  \bibfield  {author} {\bibinfo {author} {\bibfnamefont {M.}~\bibnamefont
  {Rossi}}, \bibinfo {author} {\bibfnamefont {M.}~\bibnamefont {Huber}},
  \bibinfo {author} {\bibfnamefont {D.}~\bibnamefont {Bru{\ss}}},\ and\
  \bibinfo {author} {\bibfnamefont {C.}~\bibnamefont {Macchiavello}},\ }\href
  {https://doi.org/10.1088/1367-2630/15/11/113022} {\bibfield  {journal}
  {\bibinfo  {journal} {New J. Phys.}\ }\textbf {\bibinfo {volume} {15}},\
  \bibinfo {pages} {113022} (\bibinfo {year} {2013})}\BibitemShut {NoStop}%
\bibitem [{\citenamefont {G{\"u}hne}\ \emph {et~al.}(2014)\citenamefont
  {G{\"u}hne}, \citenamefont {Cuquet}, \citenamefont {Steinhoff}, \citenamefont
  {Moroder}, \citenamefont {Rossi}, \citenamefont {Bru{\ss}}, \citenamefont
  {Kraus},\ and\ \citenamefont {Macchiavello}}]{guhne_entanglement_2014}%
  \BibitemOpen
  \bibfield  {author} {\bibinfo {author} {\bibfnamefont {O.}~\bibnamefont
  {G{\"u}hne}}, \bibinfo {author} {\bibfnamefont {M.}~\bibnamefont {Cuquet}},
  \bibinfo {author} {\bibfnamefont {F.~E.~S.}\ \bibnamefont {Steinhoff}},
  \bibinfo {author} {\bibfnamefont {T.}~\bibnamefont {Moroder}}, \bibinfo
  {author} {\bibfnamefont {M.}~\bibnamefont {Rossi}}, \bibinfo {author}
  {\bibfnamefont {D.}~\bibnamefont {Bru{\ss}}}, \bibinfo {author}
  {\bibfnamefont {B.}~\bibnamefont {Kraus}},\ and\ \bibinfo {author}
  {\bibfnamefont {C.}~\bibnamefont {Macchiavello}},\ }\href
  {https://doi.org/10.1088/1751-8113/47/33/335303} {\bibfield  {journal}
  {\bibinfo  {journal} {J. Phys. A: Math. Theor.}\ }\textbf {\bibinfo {volume}
  {47}},\ \bibinfo {pages} {335303} (\bibinfo {year} {2014})}\BibitemShut
  {NoStop}%
\bibitem [{\citenamefont {Lyons}\ \emph {et~al.}(2015)\citenamefont {Lyons},
  \citenamefont {Upchurch}, \citenamefont {Walck},\ and\ \citenamefont
  {Yetter}}]{lyons_local_2015}%
  \BibitemOpen
  \bibfield  {author} {\bibinfo {author} {\bibfnamefont {D.~W.}\ \bibnamefont
  {Lyons}}, \bibinfo {author} {\bibfnamefont {D.~J.}\ \bibnamefont {Upchurch}},
  \bibinfo {author} {\bibfnamefont {S.~N.}\ \bibnamefont {Walck}},\ and\
  \bibinfo {author} {\bibfnamefont {C.~D.}\ \bibnamefont {Yetter}},\ }\href
  {https://doi.org/10.1088/1751-8113/48/9/095301} {\bibfield  {journal}
  {\bibinfo  {journal} {J. Phys. A: Math. Theor.}\ }\textbf {\bibinfo {volume}
  {48}},\ \bibinfo {pages} {095301} (\bibinfo {year} {2015})}\BibitemShut
  {NoStop}%
\bibitem [{\citenamefont {Miller}\ and\ \citenamefont
  {Miyake}(2016)}]{miller_hierarchy_2016}%
  \BibitemOpen
  \bibfield  {author} {\bibinfo {author} {\bibfnamefont {J.}~\bibnamefont
  {Miller}}\ and\ \bibinfo {author} {\bibfnamefont {A.}~\bibnamefont
  {Miyake}},\ }\href {https://doi.org/10.1038/npjqi.2016.36} {\bibfield
  {journal} {\bibinfo  {journal} {npj Quantum Inf}\ }\textbf {\bibinfo {volume}
  {2}},\ \bibinfo {pages} {1} (\bibinfo {year} {2016})}\BibitemShut {NoStop}%
\bibitem [{\citenamefont {Takeuchi}\ \emph {et~al.}(2019)\citenamefont
  {Takeuchi}, \citenamefont {Morimae},\ and\ \citenamefont
  {Hayashi}}]{takeuchi_quantum_2019}%
  \BibitemOpen
  \bibfield  {author} {\bibinfo {author} {\bibfnamefont {Y.}~\bibnamefont
  {Takeuchi}}, \bibinfo {author} {\bibfnamefont {T.}~\bibnamefont {Morimae}},\
  and\ \bibinfo {author} {\bibfnamefont {M.}~\bibnamefont {Hayashi}},\ }\href
  {https://doi.org/10.1038/s41598-019-49968-3} {\bibfield  {journal} {\bibinfo
  {journal} {Sci Rep}\ }\textbf {\bibinfo {volume} {9}},\ \bibinfo {pages}
  {13585} (\bibinfo {year} {2019})}\BibitemShut {NoStop}%
\bibitem [{\citenamefont {Gachechiladze}\ \emph {et~al.}(2019)\citenamefont
  {Gachechiladze}, \citenamefont {G{\"u}hne},\ and\ \citenamefont
  {Miyake}}]{gachechiladze_changing_2019}%
  \BibitemOpen
  \bibfield  {author} {\bibinfo {author} {\bibfnamefont {M.}~\bibnamefont
  {Gachechiladze}}, \bibinfo {author} {\bibfnamefont {O.}~\bibnamefont
  {G{\"u}hne}},\ and\ \bibinfo {author} {\bibfnamefont {A.}~\bibnamefont
  {Miyake}},\ }\href {https://doi.org/10.1103/PhysRevA.99.052304} {\bibfield
  {journal} {\bibinfo  {journal} {Phys. Rev. A}\ }\textbf {\bibinfo {volume}
  {99}},\ \bibinfo {pages} {052304} (\bibinfo {year} {2019})}\BibitemShut
  {NoStop}%
\bibitem [{\citenamefont {Hangleiter}\ \emph {et~al.}(2019)\citenamefont
  {Hangleiter}, \citenamefont {Kliesch}, \citenamefont {Eisert},\ and\
  \citenamefont {Gogolin}}]{hangleiter_sample_2019}%
  \BibitemOpen
  \bibfield  {author} {\bibinfo {author} {\bibfnamefont {D.}~\bibnamefont
  {Hangleiter}}, \bibinfo {author} {\bibfnamefont {M.}~\bibnamefont {Kliesch}},
  \bibinfo {author} {\bibfnamefont {J.}~\bibnamefont {Eisert}},\ and\ \bibinfo
  {author} {\bibfnamefont {C.}~\bibnamefont {Gogolin}},\ }\href
  {https://doi.org/10.1103/PhysRevLett.122.210502} {\bibfield  {journal}
  {\bibinfo  {journal} {Phys. Rev. Lett.}\ }\textbf {\bibinfo {volume} {122}},\
  \bibinfo {pages} {210502} (\bibinfo {year} {2019})}\BibitemShut {NoStop}%
\bibitem [{\citenamefont {Boixo}\ \emph
  {et~al.}(2018{\natexlab{b}})\citenamefont {Boixo}, \citenamefont {Isakov},
  \citenamefont {Smelyanskiy}, \citenamefont {Babbush}, \citenamefont {Ding},
  \citenamefont {Jiang}, \citenamefont {Bremner}, \citenamefont {Martinis},\
  and\ \citenamefont {Neven}}]{Boixo.2018}%
  \BibitemOpen
  \bibfield  {author} {\bibinfo {author} {\bibfnamefont {S.}~\bibnamefont
  {Boixo}}, \bibinfo {author} {\bibfnamefont {S.~V.}\ \bibnamefont {Isakov}},
  \bibinfo {author} {\bibfnamefont {V.~N.}\ \bibnamefont {Smelyanskiy}},
  \bibinfo {author} {\bibfnamefont {R.}~\bibnamefont {Babbush}}, \bibinfo
  {author} {\bibfnamefont {N.}~\bibnamefont {Ding}}, \bibinfo {author}
  {\bibfnamefont {Z.}~\bibnamefont {Jiang}}, \bibinfo {author} {\bibfnamefont
  {M.~J.}\ \bibnamefont {Bremner}}, \bibinfo {author} {\bibfnamefont {J.~M.}\
  \bibnamefont {Martinis}},\ and\ \bibinfo {author} {\bibfnamefont
  {H.}~\bibnamefont {Neven}},\ }\href
  {https://doi.org/10.1038/s41567-018-0124-x} {\bibfield  {journal} {\bibinfo
  {journal} {Nat. Phys.}\ }\textbf {\bibinfo {volume} {14}},\ \bibinfo {pages}
  {595} (\bibinfo {year} {2018}{\natexlab{b}})}\BibitemShut {NoStop}%
\bibitem [{\citenamefont {Aaronson}\ and\ \citenamefont
  {Chen}(2017)}]{aaronson_complexity-theoretic_2017}%
  \BibitemOpen
  \bibfield  {author} {\bibinfo {author} {\bibfnamefont {S.}~\bibnamefont
  {Aaronson}}\ and\ \bibinfo {author} {\bibfnamefont {L.}~\bibnamefont
  {Chen}},\ }in\ \href {https://doi.org/10.4230/LIPIcs.CCC.2017.22} {\emph
  {\bibinfo {booktitle} {32nd {{Computational Complexity Conference}} ({{CCC}}
  2017)}}},\ \bibinfo {series} {Leibniz {{International Proceedings}} in
  {{Informatics}} ({{LIPIcs}})}, Vol.~\bibinfo {volume} {79},\ \bibinfo
  {editor} {edited by\ \bibinfo {editor} {\bibfnamefont {R.}~\bibnamefont
  {O'Donnell}}}\ (\bibinfo  {publisher} {{Schloss Dagstuhl\textendash
  Leibniz-Zentrum fuer Informatik}},\ \bibinfo {address} {{Dagstuhl,
  Germany}},\ \bibinfo {year} {2017})\ pp.\ \bibinfo {pages}
  {22:1--22:67}\BibitemShut {NoStop}%
\bibitem [{\citenamefont {Aaronson}\ and\ \citenamefont
  {Gunn}(2020)}]{aaronson_classical_2020}%
  \BibitemOpen
  \bibfield  {author} {\bibinfo {author} {\bibfnamefont {S.}~\bibnamefont
  {Aaronson}}\ and\ \bibinfo {author} {\bibfnamefont {S.}~\bibnamefont
  {Gunn}},\ }\href@noop {} {\  (\bibinfo {year} {2020})},\ \Eprint
  {https://arxiv.org/abs/1910.12085} {arXiv:1910.12085} \BibitemShut {NoStop}%
\bibitem [{\citenamefont {Bremner}\ \emph
  {et~al.}(2017{\natexlab{b}})\citenamefont {Bremner}, \citenamefont
  {Montanaro},\ and\ \citenamefont {Shepherd}}]{bremner_achieving_2017}%
  \BibitemOpen
  \bibfield  {author} {\bibinfo {author} {\bibfnamefont {M.~J.}\ \bibnamefont
  {Bremner}}, \bibinfo {author} {\bibfnamefont {A.}~\bibnamefont {Montanaro}},\
  and\ \bibinfo {author} {\bibfnamefont {D.~J.}\ \bibnamefont {Shepherd}},\
  }\href {https://doi.org/10.22331/q-2017-04-25-8} {\bibfield  {journal}
  {\bibinfo  {journal} {Quantum}\ }\textbf {\bibinfo {volume} {1}},\ \bibinfo
  {pages} {8} (\bibinfo {year} {2017}{\natexlab{b}})}\BibitemShut {NoStop}%
\bibitem [{\citenamefont {Dalzell}\ \emph {et~al.}(2020)\citenamefont
  {Dalzell}, \citenamefont {Harrow}, \citenamefont {Koh},\ and\ \citenamefont
  {Placa}}]{dalzell_how_2020}%
  \BibitemOpen
  \bibfield  {author} {\bibinfo {author} {\bibfnamefont {A.~M.}\ \bibnamefont
  {Dalzell}}, \bibinfo {author} {\bibfnamefont {A.~W.}\ \bibnamefont {Harrow}},
  \bibinfo {author} {\bibfnamefont {D.~E.}\ \bibnamefont {Koh}},\ and\ \bibinfo
  {author} {\bibfnamefont {R.~L.~L.}\ \bibnamefont {Placa}},\ }\href
  {https://doi.org/10.22331/q-2020-05-11-264} {\bibfield  {journal} {\bibinfo
  {journal} {Quantum}\ }\textbf {\bibinfo {volume} {4}},\ \bibinfo {pages}
  {264} (\bibinfo {year} {2020})}\BibitemShut {NoStop}%
\bibitem [{\citenamefont {Liu}\ \emph {et~al.}(2021)\citenamefont {Liu},
  \citenamefont {Wang},\ and\ \citenamefont {Zhang}}]{liu_tropical_2021}%
  \BibitemOpen
  \bibfield  {author} {\bibinfo {author} {\bibfnamefont {J.-G.}\ \bibnamefont
  {Liu}}, \bibinfo {author} {\bibfnamefont {L.}~\bibnamefont {Wang}},\ and\
  \bibinfo {author} {\bibfnamefont {P.}~\bibnamefont {Zhang}},\ }\href
  {https://doi.org/10.1103/PhysRevLett.126.090506} {\bibfield  {journal}
  {\bibinfo  {journal} {Phys. Rev. Lett.}\ }\textbf {\bibinfo {volume} {126}},\
  \bibinfo {pages} {090506} (\bibinfo {year} {2021})}\BibitemShut {NoStop}%
\bibitem [{\citenamefont {Markov}\ and\ \citenamefont
  {Shi}(2008)}]{Markov.2008}%
  \BibitemOpen
  \bibfield  {author} {\bibinfo {author} {\bibfnamefont {I.~L.}\ \bibnamefont
  {Markov}}\ and\ \bibinfo {author} {\bibfnamefont {Y.}~\bibnamefont {Shi}},\
  }\href {https://doi.org/10.1137/050644756} {\bibfield  {journal} {\bibinfo
  {journal} {SIAM Journal on Computing}\ }\textbf {\bibinfo {volume} {38}},\
  \bibinfo {pages} {963} (\bibinfo {year} {2008})}\BibitemShut {NoStop}%
\bibitem [{\citenamefont {Liu}\ \emph {et~al.}(2023)\citenamefont {Liu},
  \citenamefont {Gao}, \citenamefont {Cain}, \citenamefont {Lukin},\ and\
  \citenamefont {Wang}}]{liu_computing_2023}%
  \BibitemOpen
  \bibfield  {author} {\bibinfo {author} {\bibfnamefont {J.-G.}\ \bibnamefont
  {Liu}}, \bibinfo {author} {\bibfnamefont {X.}~\bibnamefont {Gao}}, \bibinfo
  {author} {\bibfnamefont {M.}~\bibnamefont {Cain}}, \bibinfo {author}
  {\bibfnamefont {M.~D.}\ \bibnamefont {Lukin}},\ and\ \bibinfo {author}
  {\bibfnamefont {S.-T.}\ \bibnamefont {Wang}},\ }\href
  {https://doi.org/10.1137/22M1501787} {\bibfield  {journal} {\bibinfo
  {journal} {SIAM J. Sci. Comput.}\ }\textbf {\bibinfo {volume} {45}},\
  \bibinfo {pages} {A1239} (\bibinfo {year} {2023})}\BibitemShut {NoStop}%
\bibitem [{gen()}]{generictensornetworks_website}%
  \BibitemOpen
  \href {https://QuEraComputing.github.io/GenericTensorNetworks.jl/} {\bibinfo
  {title} {{GenericTensorNetworks}.jl documentation}},\ \bibinfo {note}
  {accessed 2023-09-27}\BibitemShut {NoStop}%
\bibitem [{\citenamefont {Bravyi}\ and\ \citenamefont
  {Gosset}(2016)}]{bravyi_improved_2016}%
  \BibitemOpen
  \bibfield  {author} {\bibinfo {author} {\bibfnamefont {S.}~\bibnamefont
  {Bravyi}}\ and\ \bibinfo {author} {\bibfnamefont {D.}~\bibnamefont
  {Gosset}},\ }\href {https://doi.org/10.1103/PhysRevLett.116.250501}
  {\bibfield  {journal} {\bibinfo  {journal} {Phys. Rev. Lett.}\ }\textbf
  {\bibinfo {volume} {116}},\ \bibinfo {pages} {250501} (\bibinfo {year}
  {2016})}\BibitemShut {NoStop}%
\bibitem [{\citenamefont {Bravyi}\ \emph {et~al.}(2019)\citenamefont {Bravyi},
  \citenamefont {Browne}, \citenamefont {Calpin}, \citenamefont {Campbell},
  \citenamefont {Gosset},\ and\ \citenamefont
  {Howard}}]{bravyi_simulation_2019}%
  \BibitemOpen
  \bibfield  {author} {\bibinfo {author} {\bibfnamefont {S.}~\bibnamefont
  {Bravyi}}, \bibinfo {author} {\bibfnamefont {D.}~\bibnamefont {Browne}},
  \bibinfo {author} {\bibfnamefont {P.}~\bibnamefont {Calpin}}, \bibinfo
  {author} {\bibfnamefont {E.}~\bibnamefont {Campbell}}, \bibinfo {author}
  {\bibfnamefont {D.}~\bibnamefont {Gosset}},\ and\ \bibinfo {author}
  {\bibfnamefont {M.}~\bibnamefont {Howard}},\ }\href
  {https://doi.org/10.22331/q-2019-09-02-181} {\bibfield  {journal} {\bibinfo
  {journal} {Quantum}\ }\textbf {\bibinfo {volume} {3}},\ \bibinfo {pages}
  {181} (\bibinfo {year} {2019})}\BibitemShut {NoStop}%
\bibitem [{\citenamefont {Codsi}\ and\ \citenamefont {{van de
  Wetering}}(2023)}]{codsi_classically_2023}%
  \BibitemOpen
  \bibfield  {author} {\bibinfo {author} {\bibfnamefont {J.}~\bibnamefont
  {Codsi}}\ and\ \bibinfo {author} {\bibfnamefont {J.}~\bibnamefont {{van de
  Wetering}}},\ }\href@noop {} {\  (\bibinfo {year} {2023})},\ \Eprint
  {https://arxiv.org/abs/2212.08609} {arXiv:2212.08609} \BibitemShut {NoStop}%
\bibitem [{\citenamefont {Shepherd}(2010)}]{shepherd_binary_2010}%
  \BibitemOpen
  \bibfield  {author} {\bibinfo {author} {\bibfnamefont {D.}~\bibnamefont
  {Shepherd}},\ }\href@noop {} {\  (\bibinfo {year} {2010})},\ \Eprint
  {https://arxiv.org/abs/1005.1744} {arXiv:1005.1744} \BibitemShut {NoStop}%
\bibitem [{\citenamefont {Schwarz}\ and\ \citenamefont {den
  Nest}(2013)}]{schwarz_simulating_2013}%
  \BibitemOpen
  \bibfield  {author} {\bibinfo {author} {\bibfnamefont {M.}~\bibnamefont
  {Schwarz}}\ and\ \bibinfo {author} {\bibfnamefont {M.~V.}\ \bibnamefont {den
  Nest}},\ }\Eprint {https://arxiv.org/abs/1310.6749} {arXiv:1310.6749}
  (\bibinfo {year} {2013})\BibitemShut {NoStop}%
\bibitem [{\citenamefont {Pashayan}\ \emph {et~al.}(2020)\citenamefont
  {Pashayan}, \citenamefont {Bartlett},\ and\ \citenamefont
  {Gross}}]{pashayan_estimation_2020}%
  \BibitemOpen
  \bibfield  {author} {\bibinfo {author} {\bibfnamefont {H.}~\bibnamefont
  {Pashayan}}, \bibinfo {author} {\bibfnamefont {S.~D.}\ \bibnamefont
  {Bartlett}},\ and\ \bibinfo {author} {\bibfnamefont {D.}~\bibnamefont
  {Gross}},\ }\href {https://doi.org/10.22331/q-2020-01-13-223} {\bibfield
  {journal} {\bibinfo  {journal} {Quantum}\ }\textbf {\bibinfo {volume} {4}},\
  \bibinfo {pages} {223} (\bibinfo {year} {2020})}\BibitemShut {NoStop}%
\bibitem [{\citenamefont {Lokshtanov}\ \emph {et~al.}(2017)\citenamefont
  {Lokshtanov}, \citenamefont {Paturi}, \citenamefont {Tamaki}, \citenamefont
  {Williams},\ and\ \citenamefont {Yu}}]{lokshtanov_beating_2017}%
  \BibitemOpen
  \bibfield  {author} {\bibinfo {author} {\bibfnamefont {D.}~\bibnamefont
  {Lokshtanov}}, \bibinfo {author} {\bibfnamefont {R.}~\bibnamefont {Paturi}},
  \bibinfo {author} {\bibfnamefont {S.}~\bibnamefont {Tamaki}}, \bibinfo
  {author} {\bibfnamefont {R.}~\bibnamefont {Williams}},\ and\ \bibinfo
  {author} {\bibfnamefont {H.}~\bibnamefont {Yu}},\ }in\ \href
  {https://doi.org/10.1137/1.9781611974782.143} {\emph {\bibinfo {booktitle}
  {Proceedings of the 2017 {{Annual ACM-SIAM Symposium}} on {{Discrete
  Algorithms}} ({{SODA}})}}},\ \bibinfo {series and number} {Proceedings}\
  (\bibinfo  {publisher} {{Society for Industrial and Applied Mathematics}},\
  \bibinfo {year} {2017})\ pp.\ \bibinfo {pages} {2190--2202}\BibitemShut
  {NoStop}%
\bibitem [{\citenamefont {Bravyi}\ \emph {et~al.}(2022)\citenamefont {Bravyi},
  \citenamefont {Gosset},\ and\ \citenamefont {Liu}}]{bravyi_how_2022}%
  \BibitemOpen
  \bibfield  {author} {\bibinfo {author} {\bibfnamefont {S.}~\bibnamefont
  {Bravyi}}, \bibinfo {author} {\bibfnamefont {D.}~\bibnamefont {Gosset}},\
  and\ \bibinfo {author} {\bibfnamefont {Y.}~\bibnamefont {Liu}},\ }\href
  {https://doi.org/10.1103/PhysRevLett.128.220503} {\bibfield  {journal}
  {\bibinfo  {journal} {Phys. Rev. Lett.}\ }\textbf {\bibinfo {volume} {128}},\
  \bibinfo {pages} {220503} (\bibinfo {year} {2022})}\BibitemShut {NoStop}%
\bibitem [{\citenamefont {Markov}\ \emph {et~al.}(2018)\citenamefont {Markov},
  \citenamefont {Fatima}, \citenamefont {Isakov},\ and\ \citenamefont
  {Boixo}}]{markov_quantum_2018}%
  \BibitemOpen
  \bibfield  {author} {\bibinfo {author} {\bibfnamefont {I.~L.}\ \bibnamefont
  {Markov}}, \bibinfo {author} {\bibfnamefont {A.}~\bibnamefont {Fatima}},
  \bibinfo {author} {\bibfnamefont {S.~V.}\ \bibnamefont {Isakov}},\ and\
  \bibinfo {author} {\bibfnamefont {S.}~\bibnamefont {Boixo}},\ }\href@noop {}
  {\  (\bibinfo {year} {2018})},\ \Eprint {https://arxiv.org/abs/1807.10749}
  {arXiv:1807.10749} \BibitemShut {NoStop}%
\bibitem [{\citenamefont {Rajakumar}\ \emph {et~al.}(2024)\citenamefont
  {Rajakumar}, \citenamefont {Watson},\ and\ \citenamefont
  {Liu}}]{rajakumar_polynomial-time_2024}%
  \BibitemOpen
  \bibfield  {author} {\bibinfo {author} {\bibfnamefont {J.}~\bibnamefont
  {Rajakumar}}, \bibinfo {author} {\bibfnamefont {J.~D.}\ \bibnamefont
  {Watson}},\ and\ \bibinfo {author} {\bibfnamefont {Y.-K.}\ \bibnamefont
  {Liu}},\ }\href@noop {} {\  (\bibinfo {year} {2024})},\ \Eprint
  {https://arxiv.org/abs/2403.14607} {arXiv:2403.14607} \BibitemShut {NoStop}%
\bibitem [{\citenamefont {Aharonov}\ \emph {et~al.}(2023)\citenamefont
  {Aharonov}, \citenamefont {Gao}, \citenamefont {Landau}, \citenamefont
  {Liu},\ and\ \citenamefont {Vazirani}}]{aharonov_polynomial-time_2023}%
  \BibitemOpen
  \bibfield  {author} {\bibinfo {author} {\bibfnamefont {D.}~\bibnamefont
  {Aharonov}}, \bibinfo {author} {\bibfnamefont {X.}~\bibnamefont {Gao}},
  \bibinfo {author} {\bibfnamefont {Z.}~\bibnamefont {Landau}}, \bibinfo
  {author} {\bibfnamefont {Y.}~\bibnamefont {Liu}},\ and\ \bibinfo {author}
  {\bibfnamefont {U.}~\bibnamefont {Vazirani}},\ }in\ \href
  {https://doi.org/10.1145/3564246.3585234} {\emph {\bibinfo {booktitle}
  {Proceedings of the 55th {{Annual ACM Symposium}} on {{Theory}} of
  {{Computing}}}}}\ (\bibinfo {year} {2023})\ pp.\ \bibinfo {pages}
  {945--957},\ \Eprint {https://arxiv.org/abs/2211.03999} {arXiv:2211.03999
  [quant-ph]} \BibitemShut {NoStop}%
\bibitem [{\citenamefont {Bouland}\ \emph {et~al.}(2019)\citenamefont
  {Bouland}, \citenamefont {Fefferman}, \citenamefont {Nirkhe},\ and\
  \citenamefont {Vazirani}}]{bouland_complexity_2019}%
  \BibitemOpen
  \bibfield  {author} {\bibinfo {author} {\bibfnamefont {A.}~\bibnamefont
  {Bouland}}, \bibinfo {author} {\bibfnamefont {B.}~\bibnamefont {Fefferman}},
  \bibinfo {author} {\bibfnamefont {C.}~\bibnamefont {Nirkhe}},\ and\ \bibinfo
  {author} {\bibfnamefont {U.}~\bibnamefont {Vazirani}},\ }\href
  {https://doi.org/10.1038/s41567-018-0318-2} {\bibfield  {journal} {\bibinfo
  {journal} {Nature Phys}\ }\textbf {\bibinfo {volume} {15}},\ \bibinfo {pages}
  {159} (\bibinfo {year} {2019})}\BibitemShut {NoStop}%
\bibitem [{\citenamefont {Shaw}\ \emph {et~al.}(2024)\citenamefont {Shaw},
  \citenamefont {Chen}, \citenamefont {Choi}, \citenamefont {Mark},
  \citenamefont {Scholl}, \citenamefont {Finkelstein}, \citenamefont {Elben},
  \citenamefont {Choi},\ and\ \citenamefont {Endres}}]{shaw_benchmarking_2024}%
  \BibitemOpen
  \bibfield  {author} {\bibinfo {author} {\bibfnamefont {A.~L.}\ \bibnamefont
  {Shaw}}, \bibinfo {author} {\bibfnamefont {Z.}~\bibnamefont {Chen}}, \bibinfo
  {author} {\bibfnamefont {J.}~\bibnamefont {Choi}}, \bibinfo {author}
  {\bibfnamefont {D.~K.}\ \bibnamefont {Mark}}, \bibinfo {author}
  {\bibfnamefont {P.}~\bibnamefont {Scholl}}, \bibinfo {author} {\bibfnamefont
  {R.}~\bibnamefont {Finkelstein}}, \bibinfo {author} {\bibfnamefont
  {A.}~\bibnamefont {Elben}}, \bibinfo {author} {\bibfnamefont
  {S.}~\bibnamefont {Choi}},\ and\ \bibinfo {author} {\bibfnamefont
  {M.}~\bibnamefont {Endres}},\ }\href
  {https://doi.org/10.1038/s41586-024-07173-x} {\bibfield  {journal} {\bibinfo
  {journal} {Nature}\ }\textbf {\bibinfo {volume} {628}},\ \bibinfo {pages}
  {71} (\bibinfo {year} {2024})}\BibitemShut {NoStop}%
\bibitem [{\citenamefont {Pan}\ \emph {et~al.}(2022)\citenamefont {Pan},
  \citenamefont {Chen},\ and\ \citenamefont {Zhang}}]{pan_solving_2022}%
  \BibitemOpen
  \bibfield  {author} {\bibinfo {author} {\bibfnamefont {F.}~\bibnamefont
  {Pan}}, \bibinfo {author} {\bibfnamefont {K.}~\bibnamefont {Chen}},\ and\
  \bibinfo {author} {\bibfnamefont {P.}~\bibnamefont {Zhang}},\ }\href
  {https://doi.org/10.1103/PhysRevLett.129.090502} {\bibfield  {journal}
  {\bibinfo  {journal} {Phys. Rev. Lett.}\ }\textbf {\bibinfo {volume} {129}},\
  \bibinfo {pages} {090502} (\bibinfo {year} {2022})}\BibitemShut {NoStop}%
\bibitem [{Note3()}]{Note3}%
  \BibitemOpen
  \bibinfo {note} {To see the latter observe that $2^{2n} \protect \mathbb E_C
  p_C(x) p_C(y) - 1 \protect \,{=}\protect \, 0 $ for $x \protect \neq
  y$.}\BibitemShut {Stop}%
\bibitem [{\citenamefont {Kuriyattil}\ \emph {et~al.}(2023)\citenamefont
  {Kuriyattil}, \citenamefont {Hashizume}, \citenamefont {Bentsen},\ and\
  \citenamefont {Daley}}]{kuriyattil_onset_2023}%
  \BibitemOpen
  \bibfield  {author} {\bibinfo {author} {\bibfnamefont {S.}~\bibnamefont
  {Kuriyattil}}, \bibinfo {author} {\bibfnamefont {T.}~\bibnamefont
  {Hashizume}}, \bibinfo {author} {\bibfnamefont {G.}~\bibnamefont {Bentsen}},\
  and\ \bibinfo {author} {\bibfnamefont {A.~J.}\ \bibnamefont {Daley}},\ }\href
  {https://doi.org/10.1103/PRXQuantum.4.030325} {\bibfield  {journal} {\bibinfo
   {journal} {PRX Quantum}\ }\textbf {\bibinfo {volume} {4}},\ \bibinfo {pages}
  {030325} (\bibinfo {year} {2023})}\BibitemShut {NoStop}%
\bibitem [{\citenamefont {Dalzell}\ \emph {et~al.}(2022)\citenamefont
  {Dalzell}, \citenamefont {{Hunter-Jones}},\ and\ \citenamefont
  {Brand{\~a}o}}]{dalzell_random_2022}%
  \BibitemOpen
  \bibfield  {author} {\bibinfo {author} {\bibfnamefont {A.~M.}\ \bibnamefont
  {Dalzell}}, \bibinfo {author} {\bibfnamefont {N.}~\bibnamefont
  {{Hunter-Jones}}},\ and\ \bibinfo {author} {\bibfnamefont {F.~G. S.~L.}\
  \bibnamefont {Brand{\~a}o}},\ }\href
  {https://doi.org/10.1103/PRXQuantum.3.010333} {\bibfield  {journal} {\bibinfo
   {journal} {PRX Quantum}\ }\textbf {\bibinfo {volume} {3}},\ \bibinfo {pages}
  {010333} (\bibinfo {year} {2022})}\BibitemShut {NoStop}%
\bibitem [{\citenamefont {Nechita}\ and\ \citenamefont
  {Singh}(2021)}]{nechita_graphical_2021}%
  \BibitemOpen
  \bibfield  {author} {\bibinfo {author} {\bibfnamefont {I.}~\bibnamefont
  {Nechita}}\ and\ \bibinfo {author} {\bibfnamefont {S.}~\bibnamefont
  {Singh}},\ }\href {https://doi.org/10.1016/j.laa.2020.12.014} {\bibfield
  {journal} {\bibinfo  {journal} {Linear Algebra and its Applications}\
  }\textbf {\bibinfo {volume} {613}},\ \bibinfo {pages} {46} (\bibinfo {year}
  {2021})}\BibitemShut {NoStop}%
\bibitem [{\citenamefont {Broadbent}\ and\ \citenamefont
  {Hammersley}(1957)}]{broadbent_percolation_1957}%
  \BibitemOpen
  \bibfield  {author} {\bibinfo {author} {\bibfnamefont {S.~R.}\ \bibnamefont
  {Broadbent}}\ and\ \bibinfo {author} {\bibfnamefont {J.~M.}\ \bibnamefont
  {Hammersley}},\ }\href {https://doi.org/10.1017/S0305004100032680} {\bibfield
   {journal} {\bibinfo  {journal} {Mathematical Proceedings of the Cambridge
  Philosophical Society}\ }\textbf {\bibinfo {volume} {53}},\ \bibinfo {pages}
  {629} (\bibinfo {year} {1957})}\BibitemShut {NoStop}%
\bibitem [{\citenamefont {Bollob{\'a}s}\ and\ \citenamefont
  {Riordan}(2006)}]{bollobas_percolation_2006}%
  \BibitemOpen
  \bibfield  {author} {\bibinfo {author} {\bibfnamefont {B.}~\bibnamefont
  {Bollob{\'a}s}}\ and\ \bibinfo {author} {\bibfnamefont {O.}~\bibnamefont
  {Riordan}},\ }\href {https://doi.org/10.1017/CBO9781139167383} {\emph
  {\bibinfo {title} {Percolation}}}\ (\bibinfo  {publisher} {Cambridge
  University Press},\ \bibinfo {address} {Cambridge},\ \bibinfo {year}
  {2006})\BibitemShut {NoStop}%
\bibitem [{\citenamefont {{Bermejo-Vega}}\ \emph {et~al.}(2018)\citenamefont
  {{Bermejo-Vega}}, \citenamefont {Hangleiter}, \citenamefont {Schwarz},
  \citenamefont {Raussendorf},\ and\ \citenamefont
  {Eisert}}]{bermejo-vega_architectures_2018}%
  \BibitemOpen
  \bibfield  {author} {\bibinfo {author} {\bibfnamefont {J.}~\bibnamefont
  {{Bermejo-Vega}}}, \bibinfo {author} {\bibfnamefont {D.}~\bibnamefont
  {Hangleiter}}, \bibinfo {author} {\bibfnamefont {M.}~\bibnamefont {Schwarz}},
  \bibinfo {author} {\bibfnamefont {R.}~\bibnamefont {Raussendorf}},\ and\
  \bibinfo {author} {\bibfnamefont {J.}~\bibnamefont {Eisert}},\ }\href
  {https://doi.org/10.1103/PhysRevX.8.021010} {\bibfield  {journal} {\bibinfo
  {journal} {Phys. Rev. X}\ }\textbf {\bibinfo {volume} {8}},\ \bibinfo {pages}
  {021010} (\bibinfo {year} {2018})}\BibitemShut {NoStop}%
\bibitem [{\citenamefont {Hayashi}\ and\ \citenamefont
  {Takeuchi}(2019)}]{hayashi_verifying_2019}%
  \BibitemOpen
  \bibfield  {author} {\bibinfo {author} {\bibfnamefont {M.}~\bibnamefont
  {Hayashi}}\ and\ \bibinfo {author} {\bibfnamefont {Y.}~\bibnamefont
  {Takeuchi}},\ }\href {https://doi.org/10.1088/1367-2630/ab3d88} {\bibfield
  {journal} {\bibinfo  {journal} {New J. Phys.}\ }\textbf {\bibinfo {volume}
  {21}},\ \bibinfo {pages} {093060} (\bibinfo {year} {2019})}\BibitemShut
  {NoStop}%
\bibitem [{\citenamefont {Morimae}\ \emph {et~al.}(2017)\citenamefont
  {Morimae}, \citenamefont {Takeuchi},\ and\ \citenamefont
  {Hayashi}}]{morimae_verification_2017}%
  \BibitemOpen
  \bibfield  {author} {\bibinfo {author} {\bibfnamefont {T.}~\bibnamefont
  {Morimae}}, \bibinfo {author} {\bibfnamefont {Y.}~\bibnamefont {Takeuchi}},\
  and\ \bibinfo {author} {\bibfnamefont {M.}~\bibnamefont {Hayashi}},\ }\href
  {https://doi.org/10.1103/PhysRevA.96.062321} {\bibfield  {journal} {\bibinfo
  {journal} {Phys. Rev. A}\ }\textbf {\bibinfo {volume} {96}},\ \bibinfo
  {pages} {062321} (\bibinfo {year} {2017})}\BibitemShut {NoStop}%
\bibitem [{\citenamefont {Zhu}\ and\ \citenamefont
  {Hayashi}(2019)}]{zhu_efficient_2019}%
  \BibitemOpen
  \bibfield  {author} {\bibinfo {author} {\bibfnamefont {H.}~\bibnamefont
  {Zhu}}\ and\ \bibinfo {author} {\bibfnamefont {M.}~\bibnamefont {Hayashi}},\
  }\href {https://doi.org/10.1103/PhysRevApplied.12.054047} {\bibfield
  {journal} {\bibinfo  {journal} {Phys. Rev. Appl.}\ }\textbf {\bibinfo
  {volume} {12}},\ \bibinfo {pages} {054047} (\bibinfo {year}
  {2019})}\BibitemShut {NoStop}%
\bibitem [{\citenamefont {Flammia}\ and\ \citenamefont
  {Liu}(2011)}]{flammia_direct_2011}%
  \BibitemOpen
  \bibfield  {author} {\bibinfo {author} {\bibfnamefont {S.~T.}\ \bibnamefont
  {Flammia}}\ and\ \bibinfo {author} {\bibfnamefont {Y.-K.}\ \bibnamefont
  {Liu}},\ }\href {https://doi.org/10.1103/PhysRevLett.106.230501} {\bibfield
  {journal} {\bibinfo  {journal} {Physical Review Letters}\ }\textbf {\bibinfo
  {volume} {106}},\ \bibinfo {pages} {230501} (\bibinfo {year}
  {2011})}\BibitemShut {NoStop}%
\bibitem [{\citenamefont {Kruse}\ \emph {et~al.}(2019)\citenamefont {Kruse},
  \citenamefont {Hamilton}, \citenamefont {Sansoni}, \citenamefont {Barkhofen},
  \citenamefont {Silberhorn},\ and\ \citenamefont {Jex}}]{kruse_detailed_2019}%
  \BibitemOpen
  \bibfield  {author} {\bibinfo {author} {\bibfnamefont {R.}~\bibnamefont
  {Kruse}}, \bibinfo {author} {\bibfnamefont {C.~S.}\ \bibnamefont {Hamilton}},
  \bibinfo {author} {\bibfnamefont {L.}~\bibnamefont {Sansoni}}, \bibinfo
  {author} {\bibfnamefont {S.}~\bibnamefont {Barkhofen}}, \bibinfo {author}
  {\bibfnamefont {C.}~\bibnamefont {Silberhorn}},\ and\ \bibinfo {author}
  {\bibfnamefont {I.}~\bibnamefont {Jex}},\ }\href
  {https://doi.org/10.1103/PhysRevA.100.032326} {\bibfield  {journal} {\bibinfo
   {journal} {Phys. Rev. A}\ }\textbf {\bibinfo {volume} {100}},\ \bibinfo
  {pages} {032326} (\bibinfo {year} {2019})}\BibitemShut {NoStop}%
\bibitem [{\citenamefont {Deshpande}\ \emph
  {et~al.}(2022{\natexlab{b}})\citenamefont {Deshpande}, \citenamefont {Mehta},
  \citenamefont {Vincent}, \citenamefont {Quesada}, \citenamefont {Hinsche},
  \citenamefont {Ioannou}, \citenamefont {Madsen}, \citenamefont {Lavoie},
  \citenamefont {Qi}, \citenamefont {Eisert}, \citenamefont {Hangleiter},
  \citenamefont {Fefferman},\ and\ \citenamefont
  {Dhand}}]{deshpande_quantum_2022}%
  \BibitemOpen
  \bibfield  {author} {\bibinfo {author} {\bibfnamefont {A.}~\bibnamefont
  {Deshpande}}, \bibinfo {author} {\bibfnamefont {A.}~\bibnamefont {Mehta}},
  \bibinfo {author} {\bibfnamefont {T.}~\bibnamefont {Vincent}}, \bibinfo
  {author} {\bibfnamefont {N.}~\bibnamefont {Quesada}}, \bibinfo {author}
  {\bibfnamefont {M.}~\bibnamefont {Hinsche}}, \bibinfo {author} {\bibfnamefont
  {M.}~\bibnamefont {Ioannou}}, \bibinfo {author} {\bibfnamefont
  {L.}~\bibnamefont {Madsen}}, \bibinfo {author} {\bibfnamefont
  {J.}~\bibnamefont {Lavoie}}, \bibinfo {author} {\bibfnamefont
  {H.}~\bibnamefont {Qi}}, \bibinfo {author} {\bibfnamefont {J.}~\bibnamefont
  {Eisert}}, \bibinfo {author} {\bibfnamefont {D.}~\bibnamefont {Hangleiter}},
  \bibinfo {author} {\bibfnamefont {B.}~\bibnamefont {Fefferman}},\ and\
  \bibinfo {author} {\bibfnamefont {I.}~\bibnamefont {Dhand}},\ }\href
  {https://doi.org/10.1126/sciadv.abi7894} {\bibfield  {journal} {\bibinfo
  {journal} {Science Advances}\ }\textbf {\bibinfo {volume} {8}},\ \bibinfo
  {pages} {eabi7894} (\bibinfo {year} {2022}{\natexlab{b}})}\BibitemShut
  {NoStop}%
\bibitem [{\citenamefont {Ehrenberg}\ \emph {et~al.}(2024)\citenamefont
  {Ehrenberg}, \citenamefont {Iosue}, \citenamefont {Deshpande}, \citenamefont
  {Hangleiter},\ and\ \citenamefont {Gorshkov}}]{ehrenberg_transition_2024}%
  \BibitemOpen
  \bibfield  {author} {\bibinfo {author} {\bibfnamefont {A.}~\bibnamefont
  {Ehrenberg}}, \bibinfo {author} {\bibfnamefont {J.~T.}\ \bibnamefont
  {Iosue}}, \bibinfo {author} {\bibfnamefont {A.}~\bibnamefont {Deshpande}},
  \bibinfo {author} {\bibfnamefont {D.}~\bibnamefont {Hangleiter}},\ and\
  \bibinfo {author} {\bibfnamefont {A.~V.}\ \bibnamefont {Gorshkov}},\
  }\href@noop {} {\  (\bibinfo {year} {2024})},\ \Eprint
  {https://arxiv.org/abs/2312.08433} {arXiv:2312.08433} \BibitemShut {NoStop}%
\bibitem [{\citenamefont {Bouland}\ \emph {et~al.}(2023)\citenamefont
  {Bouland}, \citenamefont {Brod}, \citenamefont {Datta}, \citenamefont
  {Fefferman}, \citenamefont {Grier}, \citenamefont {Hernandez},\ and\
  \citenamefont {Oszmaniec}}]{bouland_complexity-theoretic_2023}%
  \BibitemOpen
  \bibfield  {author} {\bibinfo {author} {\bibfnamefont {A.}~\bibnamefont
  {Bouland}}, \bibinfo {author} {\bibfnamefont {D.}~\bibnamefont {Brod}},
  \bibinfo {author} {\bibfnamefont {I.}~\bibnamefont {Datta}}, \bibinfo
  {author} {\bibfnamefont {B.}~\bibnamefont {Fefferman}}, \bibinfo {author}
  {\bibfnamefont {D.}~\bibnamefont {Grier}}, \bibinfo {author} {\bibfnamefont
  {F.}~\bibnamefont {Hernandez}},\ and\ \bibinfo {author} {\bibfnamefont
  {M.}~\bibnamefont {Oszmaniec}},\ }\href@noop {} {\  (\bibinfo {year}
  {2023})},\ \Eprint {https://arxiv.org/abs/2312.00286} {arXiv:2312.00286}
  \BibitemShut {NoStop}%
\bibitem [{\citenamefont {Stockmeyer}(1983)}]{stockmeyer_complexity_1983}%
  \BibitemOpen
  \bibfield  {author} {\bibinfo {author} {\bibfnamefont {L.}~\bibnamefont
  {Stockmeyer}},\ }in\ \href {https://doi.org/10.1145/800061.808740} {\emph
  {\bibinfo {booktitle} {Proceedings of the Fifteenth Annual {{ACM}} Symposium
  on {{Theory}} of Computing}}},\ \bibinfo {series and number} {{{STOC}} '83}\
  (\bibinfo  {publisher} {Association for Computing Machinery},\ \bibinfo
  {address} {New York, NY, USA},\ \bibinfo {year} {1983})\ pp.\ \bibinfo
  {pages} {118--126}\BibitemShut {NoStop}%
\bibitem [{Note4()}]{Note4}%
  \BibitemOpen
  \bibinfo {note} {For the discussion in this section, we will emphasize the
  fact that a gate $G$ is a logical gate using the standard notation $\protect
  \overline G$ in contrast to physically implemented gates.}\BibitemShut
  {Stop}%
\bibitem [{\citenamefont {Bomb\'{\i}n}(2015)}]{Bombin15}%
  \BibitemOpen
  \bibfield  {author} {\bibinfo {author} {\bibfnamefont {H.}~\bibnamefont
  {Bomb\'{\i}n}},\ }\href {https://doi.org/10.1103/PhysRevX.5.031043}
  {\bibfield  {journal} {\bibinfo  {journal} {Phys. Rev. X}\ }\textbf {\bibinfo
  {volume} {5}},\ \bibinfo {pages} {031043} (\bibinfo {year}
  {2015})}\BibitemShut {NoStop}%
\bibitem [{\citenamefont {Zhou}\ \emph {et~al.}(2024)\citenamefont {Zhou},
  \citenamefont {Zhao}, \citenamefont {Cain}, \citenamefont {Bluvstein},
  \citenamefont {Duckering}, \citenamefont {Hu}, \citenamefont {Wang},
  \citenamefont {Kubica},\ and\ \citenamefont {Lukin}}]{zhou_algorithmic_2024}%
  \BibitemOpen
  \bibfield  {author} {\bibinfo {author} {\bibfnamefont {H.}~\bibnamefont
  {Zhou}}, \bibinfo {author} {\bibfnamefont {C.}~\bibnamefont {Zhao}}, \bibinfo
  {author} {\bibfnamefont {M.}~\bibnamefont {Cain}}, \bibinfo {author}
  {\bibfnamefont {D.}~\bibnamefont {Bluvstein}}, \bibinfo {author}
  {\bibfnamefont {C.}~\bibnamefont {Duckering}}, \bibinfo {author}
  {\bibfnamefont {H.-Y.}\ \bibnamefont {Hu}}, \bibinfo {author} {\bibfnamefont
  {S.-T.}\ \bibnamefont {Wang}}, \bibinfo {author} {\bibfnamefont
  {A.}~\bibnamefont {Kubica}},\ and\ \bibinfo {author} {\bibfnamefont {M.~D.}\
  \bibnamefont {Lukin}},\ }\href@noop {} {\  (\bibinfo {year} {2024})},\
  \Eprint {https://arxiv.org/abs/2406.17653} {arXiv:2406.17653} \BibitemShut
  {NoStop}%
\bibitem [{Note5()}]{Note5}%
  \BibitemOpen
  \bibinfo {note} {This follows from the fact that CNOT is transversal for CSS
  codes since $\protect \text {CS} = \protect \text {CNOT} (\protect \mathbbm
  {1}\otimes T^\dagger ) \protect \text {CNOT}(\protect \mathbbm {1}\otimes T)$
  \cite {koutsioumpas_smallest_2022}.}\BibitemShut {Stop}%
\bibitem [{Note6()}]{Note6}%
  \BibitemOpen
  \bibinfo {note} {In the $10^7$ samples we generated to obtain \protect \cref
  {fig:code_comparison}, we did not observe a single weight-$3$ error after
  postselection.}\BibitemShut {Stop}%
\bibitem [{\citenamefont {Brown}\ \emph {et~al.}(2016)\citenamefont {Brown},
  \citenamefont {Nickerson},\ and\ \citenamefont {Browne}}]{Brown2016}%
  \BibitemOpen
  \bibfield  {author} {\bibinfo {author} {\bibfnamefont {B.~J.}\ \bibnamefont
  {Brown}}, \bibinfo {author} {\bibfnamefont {N.~H.}\ \bibnamefont
  {Nickerson}},\ and\ \bibinfo {author} {\bibfnamefont {D.~E.}\ \bibnamefont
  {Browne}},\ }\href {https://doi.org/10.1038/ncomms12302} {\bibfield
  {journal} {\bibinfo  {journal} {Nature Communications}\ }\textbf {\bibinfo
  {volume} {7}},\ \bibinfo {pages} {12302} (\bibinfo {year}
  {2016})}\BibitemShut {NoStop}%
\bibitem [{\citenamefont {Bombín}(2015)}]{Bombin2015}%
  \BibitemOpen
  \bibfield  {author} {\bibinfo {author} {\bibfnamefont {H.}~\bibnamefont
  {Bombín}},\ }\href {https://doi.org/10.1088/1367-2630/17/8/083002}
  {\bibfield  {journal} {\bibinfo  {journal} {New Journal of Physics}\ }\textbf
  {\bibinfo {volume} {17}},\ \bibinfo {pages} {083002} (\bibinfo {year}
  {2015})}\BibitemShut {NoStop}%
\bibitem [{\citenamefont {Kubica}\ and\ \citenamefont
  {Beverland}(2015)}]{Kubica2015universal}%
  \BibitemOpen
  \bibfield  {author} {\bibinfo {author} {\bibfnamefont {A.}~\bibnamefont
  {Kubica}}\ and\ \bibinfo {author} {\bibfnamefont {M.~E.}\ \bibnamefont
  {Beverland}},\ }\href {https://doi.org/10.1103/PhysRevA.91.032330} {\bibfield
   {journal} {\bibinfo  {journal} {Phys. Rev. A}\ }\textbf {\bibinfo {volume}
  {91}},\ \bibinfo {pages} {032330} (\bibinfo {year} {2015})}\BibitemShut
  {NoStop}%
\bibitem [{\citenamefont {Kubica}\ and\ \citenamefont
  {Delfosse}(2023)}]{Kubica23}%
  \BibitemOpen
  \bibfield  {author} {\bibinfo {author} {\bibfnamefont {A.}~\bibnamefont
  {Kubica}}\ and\ \bibinfo {author} {\bibfnamefont {N.}~\bibnamefont
  {Delfosse}},\ }\href {https://doi.org/10.22331/q-2023-02-21-929} {\bibfield
  {journal} {\bibinfo  {journal} {Quantum}\ }\textbf {\bibinfo {volume} {7}},\
  \bibinfo {pages} {929} (\bibinfo {year} {2023})}\BibitemShut {NoStop}%
\bibitem [{\citenamefont {Vasmer}\ and\ \citenamefont
  {Browne}(2019)}]{Vasmer2019}%
  \BibitemOpen
  \bibfield  {author} {\bibinfo {author} {\bibfnamefont {M.}~\bibnamefont
  {Vasmer}}\ and\ \bibinfo {author} {\bibfnamefont {D.~E.}\ \bibnamefont
  {Browne}},\ }\href {https://doi.org/10.1103/physreva.100.012312} {\bibfield
  {journal} {\bibinfo  {journal} {Physical Review A}\ }\textbf {\bibinfo
  {volume} {100}},\ \bibinfo {pages} {012312} (\bibinfo {year}
  {2019})}\BibitemShut {NoStop}%
\bibitem [{\citenamefont {Koutsioumpas}\ \emph {et~al.}(2022)\citenamefont
  {Koutsioumpas}, \citenamefont {Banfield},\ and\ \citenamefont
  {Kay}}]{koutsioumpas_smallest_2022}%
  \BibitemOpen
  \bibfield  {author} {\bibinfo {author} {\bibfnamefont {S.}~\bibnamefont
  {Koutsioumpas}}, \bibinfo {author} {\bibfnamefont {D.}~\bibnamefont
  {Banfield}},\ and\ \bibinfo {author} {\bibfnamefont {A.}~\bibnamefont
  {Kay}},\ }\href@noop {} {\  (\bibinfo {year} {2022})},\ \Eprint
  {https://arxiv.org/abs/2210.14066} {arXiv:2210.14066} \BibitemShut {NoStop}%
\bibitem [{\citenamefont {Bullock}\ and\ \citenamefont
  {Markov}(2004)}]{Bullock.2004}%
  \BibitemOpen
  \bibfield  {author} {\bibinfo {author} {\bibfnamefont {S.~S.}\ \bibnamefont
  {Bullock}}\ and\ \bibinfo {author} {\bibfnamefont {I.~L.}\ \bibnamefont
  {Markov}},\ }\href@noop {} {\bibfield  {journal} {\bibinfo  {journal}
  {Quantum Info. Comput.}\ }\textbf {\bibinfo {volume} {4}},\ \bibinfo {pages}
  {27–47} (\bibinfo {year} {2004})},\ \Eprint
  {https://arxiv.org/abs/quant-ph/0303039} {arXiv:quant-ph/0303039}
  \BibitemShut {NoStop}%
\bibitem [{\citenamefont {Aaronson}\ and\ \citenamefont
  {Gottesman}(2004)}]{aaronson_improved_2004}%
  \BibitemOpen
  \bibfield  {author} {\bibinfo {author} {\bibfnamefont {S.}~\bibnamefont
  {Aaronson}}\ and\ \bibinfo {author} {\bibfnamefont {D.}~\bibnamefont
  {Gottesman}},\ }\href {https://doi.org/10.1103/PhysRevA.70.052328} {\bibfield
   {journal} {\bibinfo  {journal} {Phys. Rev. A}\ }\textbf {\bibinfo {volume}
  {70}},\ \bibinfo {pages} {052328} (\bibinfo {year} {2004})}\BibitemShut
  {NoStop}%
\bibitem [{\citenamefont {Schumacher}(1996)}]{schumacher_sending_1996}%
  \BibitemOpen
  \bibfield  {author} {\bibinfo {author} {\bibfnamefont {B.}~\bibnamefont
  {Schumacher}},\ }\href@noop {} {\  (\bibinfo {year} {1996})},\ \Eprint
  {https://arxiv.org/abs/quant-ph/9604023} {arXiv:quant-ph/9604023}
  \BibitemShut {NoStop}%
\bibitem [{\citenamefont {Nielsen}(1996)}]{nielsen_entanglement_1996}%
  \BibitemOpen
  \bibfield  {author} {\bibinfo {author} {\bibfnamefont {M.~A.}\ \bibnamefont
  {Nielsen}},\ }\href@noop {} {\  (\bibinfo {year} {1996})},\ \Eprint
  {https://arxiv.org/abs/quant-ph/9606012} {arXiv:quant-ph/9606012}
  \BibitemShut {NoStop}%
\bibitem [{\citenamefont {Horodecki}\ \emph {et~al.}(1999)\citenamefont
  {Horodecki}, \citenamefont {Horodecki},\ and\ \citenamefont
  {Horodecki}}]{horodecki_general_1999}%
  \BibitemOpen
  \bibfield  {author} {\bibinfo {author} {\bibfnamefont {M.}~\bibnamefont
  {Horodecki}}, \bibinfo {author} {\bibfnamefont {P.}~\bibnamefont
  {Horodecki}},\ and\ \bibinfo {author} {\bibfnamefont {R.}~\bibnamefont
  {Horodecki}},\ }\href {https://doi.org/10.1103/PhysRevA.60.1888} {\bibfield
  {journal} {\bibinfo  {journal} {Phys. Rev. A}\ }\textbf {\bibinfo {volume}
  {60}},\ \bibinfo {pages} {1888} (\bibinfo {year} {1999})}\BibitemShut
  {NoStop}%
\bibitem [{\citenamefont {Nielsen}(2002)}]{nielsen_simple_2002}%
  \BibitemOpen
  \bibfield  {author} {\bibinfo {author} {\bibfnamefont {M.~A.}\ \bibnamefont
  {Nielsen}},\ }\href {https://doi.org/10.1016/S0375-9601(02)01272-0}
  {\bibfield  {journal} {\bibinfo  {journal} {Physics Letters A}\ }\textbf
  {\bibinfo {volume} {303}},\ \bibinfo {pages} {249} (\bibinfo {year}
  {2002})},\ \Eprint {https://arxiv.org/abs/quant-ph/0205035}
  {arXiv:quant-ph/0205035} \BibitemShut {NoStop}%
\bibitem [{\citenamefont {Wallman}\ and\ \citenamefont
  {Flammia}(2014)}]{wallman_randomized_2014}%
  \BibitemOpen
  \bibfield  {author} {\bibinfo {author} {\bibfnamefont {J.~J.}\ \bibnamefont
  {Wallman}}\ and\ \bibinfo {author} {\bibfnamefont {S.~T.}\ \bibnamefont
  {Flammia}},\ }\href {https://doi.org/10.1088/1367-2630/16/10/103032}
  {\bibfield  {journal} {\bibinfo  {journal} {New J. Phys.}\ }\textbf {\bibinfo
  {volume} {16}},\ \bibinfo {pages} {103032} (\bibinfo {year}
  {2014})}\BibitemShut {NoStop}%
\bibitem [{\citenamefont {Gidney}(2021)}]{gidney_stim_2021}%
  \BibitemOpen
  \bibfield  {author} {\bibinfo {author} {\bibfnamefont {C.}~\bibnamefont
  {Gidney}},\ }\href {https://doi.org/10.22331/q-2021-07-06-497} {\bibfield
  {journal} {\bibinfo  {journal} {Quantum}\ }\textbf {\bibinfo {volume} {5}},\
  \bibinfo {pages} {497} (\bibinfo {year} {2021})}\BibitemShut {NoStop}%
\bibitem [{\citenamefont {{NetworkX
  developers}}(2024)}]{networkx_developers_networkx_2024}%
  \BibitemOpen
  \bibfield  {author} {\bibinfo {author} {\bibnamefont {{NetworkX
  developers}}},\ }\href {https://networkx.org/} {\bibinfo {title}
  {{{NetworkX}} version 3.2.1}} (\bibinfo {year} {2024})\BibitemShut {NoStop}%
\bibitem [{\citenamefont {Gurobi~Optimization}(2024)}]{gurobi}%
  \BibitemOpen
  \bibfield  {author} {\bibinfo {author} {\bibfnamefont {L.}~\bibnamefont
  {Gurobi~Optimization}},\ }\href
  {https://www.gurobi.com/wp-content/plugins/hd_documentations/documentation/9.0/refman.pdf}
  {\bibinfo {title} {Gurobi optimizer reference manual}} (\bibinfo {year}
  {2024})\BibitemShut {NoStop}%
\end{thebibliography}%
\let\addcontentsline\oldaddcontentsline%


\appendix
 \tableofcontents
\section{Quantum channels on phase states}
\label{app:prelim}

We begin with a brief recap of the properties of IQP computations that will be helpful in the following. 
In particular, we analyze  quantum channels which preserve phase states. 

\subsection{Some background on phase states}
\label{ssec:phase states}

In its most general form, IQP computation is about the preparation and manipulation of a class of quantum states called \emph{phase states}.  

\begin{definition}
    Given a function $\phi: \mathbb{Z}_2^n \to \mathbb{R}$, a phase state $\ket{\phi}$ on $n$ qubits is
    \begin{equation}
        \ket{\phi} = \frac{1}{\sqrt{2^n}} \sum_{x \in \{0,1\}^n} e^{ i \pi \phi(x)} \ket{x}.
    \end{equation}
\end{definition}
We denote the single-qubit $x/y/z$ Pauli operators by $X/Y/Z$ and use the notation $A^a \equiv A^{a_1} \otimes \cdots \otimes A^{a_n}$ for a bit string $a \in \{0,1\}^n$ and a single qubit operator $A$.  
A few convenient facts about phase states are
\begin{itemize}
    \item A pure quantum state $\ket{\psi}$ is a phase state \emph{iff} for every nonzero~$a \in \{0,1\}^n $, $\bra{\psi}Z^a \ket{\psi} =0 $. 
    The forward direction is clear, and the reverse holds because states satisfying $\bra{\psi}Z^a \ket{\psi} =0 $ for all nonzero $a$ are precisely those states with equal magnitude coefficients on all computational basis states, i.e., phase states.
    \item Any function $\phi:\mathbb{Z}_2^n \to \mathbb{R}$ can be expanded in terms of polynomial coefficients $\alpha_a \in \mathbb R$ as 
    \begin{equation}
    \phi(x) = \sum_{a \in \{0,1\}^n} \alpha_{a} x^ a,
    \end{equation}
    where $x^a = x_1^{a_1} \cdots x_n^{a_n}.$

    \item A phase state $\ket{\phi}$ with polynomial coefficients $\alpha_a$ can be prepared with a unitary operator composed of diagonal $Z$ interactions
    \begin{align}
    \label{eq:uphi}
        U_{\phi} = \prod_{a \in \{0,1\}^n} \exp\left[i \pi \alpha_{a} \prod_i\left( \frac{\mathbbm{1}-Z_i}{2} \right)^{a_i}\right],\hspace{-1.5em}
    \end{align}
    such that $\ket{\phi} = U_{\phi} \ket{+^n}$ where $\ket{+^n}$ is the $n$-fold tensor power of the single-qubit $+1$-eigenstate $\ket + $ of~$X$.
    \item If $\phi(x)$ is a degree-$D$ polynomial, i.e., a polynomial with maximum degree $D$, then $U_\phi$ can be implemented by a quantum circuit consisting of at most $O(2^Dn^D)$ gates: There are $O(n^D)$ terms in the polynomial and the $D$-local gates $\exp(i\pi \alpha_{a} \prod_i\big( (\mathbbm{1}-Z_i)/2)^{a_i})$ can be decomposed into at most $O(2^D)$ two-qubit gates~\cite{Bullock.2004}.
\end{itemize}

\subsection{Equivalence of degree-$D$ IQP and degree-$D$ IQP + CNOT}
\label{ssec:equivalence iqp cnot}

Let us show the equivalence between degree-$D$ IQP circuits and circuits composed of degree-$D$ IQP gates and CNOT gates. 
Let us define a binary phase state as a state $\ket \psi$ which can be written as  
\begin{align}
    \ket {\psi} = \sum_x (-1)^{f(x)} \ket x,
\end{align}
where $f: \bin^n \rightarrow \bin$ is an arbitrary Boolean function.
We now observe that a degree-$D$ IQP circuit $C_f$ acting on $n$ qubits prepares a polynomial binary phase state
\begin{align}
    C_f \ket +^{\otimes n} = \frac{1}{\sqrt{2^{n}}}\sum_x (-1)^{f(x)} \ket x , 
\end{align}
where $f(x) = \sum_{k=1}^D \sum_{j^k_1, \ldots j^k_k} c^k_{j^k}x_{j^k_1} \cdots  x_{j^k_k}$ is a Boolean degree-$D$ polynomial whose coefficients $c^k_{j^k} \in \{0,1\}$ are nonzero if a C$^{k-1}$Z gate is applied to the $k$ qubits labeled by $j^k = (j_1^k, \ldots, j_k^k) $.
The action of a CNOT gate on a two-qubit basis state $\ket {xy}$ is given by $ \text{CNOT} \ket{xy} = \ket{x, x \oplus y } $, where $\oplus$ denotes addition modulo $2$. 
This implies that the action of a circuit $\mc C $ composed of CNOT gates on a polynomial binary phase state is given by 
\begin{align}
 \mc C C_f\ket {+^n} = \frac{1}{\sqrt{2^n}}\sum_x (-1)^{f(A^{-1} x)} \ket {x}  = D_{f_A} \ket{+^n}, 
 \label{eq:iqp+cnot = iqp}
\end{align}
where $A$ is a Boolean invertible matrix defined by the CNOT circuit. 
$A^{-1}$ preserves the degree of $f$ and hence $f(A^{-1}x)$ is just a different degree-$D$ Boolean polynomial $f_A(x)$ that depends on $f$ and $A$. 
Any hIQP circuit which includes IQP and CNOT gates is therefore  equivalent to some purely IQP circuit.

\subsection{Review of quantum operations}
\label{ssec:quantum operations}

Let us now briefly review the basic formalism of quantum channel theory.  

A \emph{completely positive} (CP) map acts on density matrices $\rho$ as $\mathcal{E}(\rho) = \sum_\mu K_\mu \rho K_\mu^\dag$, where the linear operators $K_\mu$ are called \emph{Kraus operators}.
A \emph{quantum channel} (or CPT map) $\mathcal{E}$ is a completely positive map that also preserves the trace, i.e., its Kraus operators satisfy $\sum_\mu K_\mu^\dag K_\mu = \mathbbm{1}$. 
\begin{definition}
    Let $P_n$ be the Pauli group on $n$ qubits, then a Pauli channel is defined by a probability distribution $q(P)$ over $P \in P_n$ such that
\begin{equation}
    \mathcal{E}(\rho) = \sum_{P \in P_n} q(P) P \rho P^\dagger, 
\end{equation}
i.e., we can take the Kraus operators to be Pauli operators.
\end{definition}

To analyze the effects of quantum channels on phase states, we first introduce two (inequivalent) levels of conservation of the phase state property. 
First, we introduce \emph{phase-state-preserving} maps (\cref{sec:psp}), then \emph{diagonal CP maps} (\cref{sec:diagonal cp}).

\subsection{Phase-state-preserving quantum operations}
\label{sec:psp}

\begin{definition}
A CP map $\mathcal{E}$  is a \emph{phase-state-preserving (PSP)} map if, for any mixture of phase states $\rho$, $\mathcal{E}(\rho)$ can be written as a mixture $\sum_a p_a \ket{\phi_a}\bra{\phi_a}$ of phase states $\ket {\phi_a}$ with $p_a \geq 0 $.
\end{definition}
\begin{definition}
A CP map $\mathcal{E}$ is a \emph{completely-phase-state preserving (CPSP)} map if $\mathcal{E}\otimes \mathbbm{1}$ is a PSP map for any trivial extension of $\mathcal{E}$ to a larger Hilbert space.
\end{definition}
We give some examples of channels that fall into different categories according to these definitions:
\begin{itemize}
    \item The partial trace operation is a CPSP channel. Take a phase state $\ket{\phi}$ on the joint system $\mathcal{H}_A \otimes \mathcal{H}_B$ with $n_A$ and $n_B$ qubits, respectively. 
    We can expand this state in the computational basis as 
    $\ket{\phi} = \sum_{x,y}c_{x,y} \ket{x}\otimes \ket{y}. $
    Taking the partial trace over $B$, we arrive at the state
    \begin{equation}
    \begin{split}
    \mathrm{Tr}_{B}(\ket{\phi}\bra{\phi}) &= \sum_{x,x',y}c_{x,y}c^*_{x',y} \ket{x}\bra{x'} \\
    &= \sum_y \frac{1}{2^{n_B}}\ket{\phi_y}\bra{\phi_y}.
    \end{split}
    \end{equation}
    Each state $\ket {\phi_y} \coloneqq \sum_x c_{x,y} \ket x$ in this mixture is a phase state because $|c_{x,y}|^2=1/2^{n_A+n_B}$.  This argument shows that the partial trace results in a mixture of phase states, proving the PSP condition.  
    The trivial extension of the partial trace operation is also a partial trace, which shows the CPSP property.

    \item 
     Let $\mc U = U \cdot U^\dagger$ with some unitary $U$ be a unitary PSP channel on $n$ qubits, then $\mc U$ is CPSP. 
    To prove the CPSP condition we need to prove $\mc U \otimes \id$ is PSP acting on the doubled Hilbert space $\mathcal{H}\otimes \mathcal{H}$.  Taking a joint state $\ket{\phi}$ and expanding it as we did in the partial trace discussion above, we can represent the state as the superposition 
    \begin{equation}
            \ket{\phi} = \sum_z c_z \ket{\phi_z} \otimes \ket{z},
    \end{equation}
    where $\ket{z}$ is a computational basis state, $|c_z|^2 = 1/2^n$ and $\ket{\phi_z}$ is a phase state.  Applying $\mc U\otimes \mathbb{I}$ to this state will result in a similar superposition of phase states each with equal magnitude coefficients; thus, the resulting state is also a phase state, proving the CPSP condition for $\mc U$.

        \item Let $U$ be a circuit composed of CNOT gates and $X$ gates, then $U \ket{x} = \ket{A x+b}$ for some invertible $n \times n$ matrix $A$ over $\mathbb{F}_2$ and a Boolean vector $b \in \mathbb{F}_2^n$. As a result, $U \ket{\phi} = \ket{\phi'}$, where $ \phi'(x)=\phi[A^{-1} (x+b)]$, so all such $U$ give rise to CPSP channels.
        \item More generally, let $U$ be a circuit composed of Toffoli gates, CNOT and $X$ gates. 
        Then, $U$ is a CPSP map with the property that 
        $ U \ket{\phi} = \ket{\phi \circ f^{-1}}$, 
        where $f(x)$ is the invertible Boolean function implemented by the circuit on computational basis states. 
        \item Any mixture of CPSP maps is a CPSP map.
        \item Any Pauli channel is a CPSP channel.
        \item A projective measurement in the $x$-basis is a 1-qubit PSP map, but is not a CPSP map.  
        As a counterexample to the CPSP condition, consider the (unnormalized) two qubit phase state $\ket{00} + \ket{01} + \ket{10} -\ket{11}$ and measure the second qubit in the $x$-basis. 
        The state after the measurement is $\ket{0+}$ or $\ket{1-}$, neither of which are phase states.
    \item The amplitude damping channel with Kraus operators $K_0= (1+\sqrt{1-\gamma})\mathbbm{1}/2 + (1-\sqrt{1-\gamma}) Z/2$ and $K_1 = (X+i Y) \sqrt{\gamma}/2$ is not PSP for any $\gamma>0$.        
    \end{itemize}
From these examples, we can see that the set of operations that preserve the phase state property has a relatively rich structure.  

\subsection{Diagonal completely positive maps}
\label{sec:diagonal cp}

In addition to phase-state preservation, another natural class of maps that act in a simple way on phase states are \emph{diagonal} CP maps. 
\begin{definition}
    A CP map $\mathcal{E}$ is diagonal if it has a Kraus decomposition $\{K_\mu \}$ such that every Kraus operator $K_\mu = \sum_a c_a Z^a$ is diagonal in the computational basis.
\end{definition}
We denote a map as a \emph{diagonal-unitary CPSP map} if it has a representation as a mixture of diagonal maps composed with unitary CPSP maps.

\subsection{The $z$-twirl}

In the context of noisy random IQP computations, it is convenient to introduce $z$-twirling of a CP map.
\begin{definition}
    The $z$-twirl of a CP map $\mathcal{E}$  with Kraus operators $\{ K_\mu \}$ is the CP map  
    \begin{equation}
       \mathcal{E}_z(\rho)=  \frac{1}{2^n} \sum_{a \mu} Z^a K_\mu Z^a \rho Z^a K_\mu^\dag Z^a
    \end{equation}
\end{definition}

This twirling operation has the effect of converting any CP map into a diagonal-unitary CPSP map as we now show.
\begin{lemma}
\label{lemma:ztwirl}
    Let $\mathcal{E}$ be a CP map; its $z$-twirl $\mathcal{E}_z$ is a diagonal-unitary CPSP map.
\end{lemma}
\begin{proof}
    Let $K_\mu$ be the Kraus operators for $\mathcal{E}$, we expand them in the Pauli basis
    \begin{equation}
        K_\mu = \frac{1 }{2^n}\sum_{ab}c_{ab}^\mu  X^b Z^a,~c_{ab}^\mu = \mathrm{Tr}[K_{\mu} Z^a X^b].
    \end{equation}
    With this representation, we can then express the $z$-twirled map as
    \begin{align}
        \mathcal{E}_z(\rho) &= \sum_{ \mu b} p_{\mu b} X^b K_{\mu b} \rho K_{\mu b}^\dagger X^b,       
    \end{align}
    where $p_{\mu b}=\sum_{a} |c_{ab}^\mu |^2$ and the Kraus operators $K_{\mu b} =  \frac{1 }{\sqrt{p_{\mu b}}}\sum_a c_{ab}^\mu Z^a$.  The effect of the twirling is to remove the terms on which the $X$ operators act with a different power on each side of the density matrix. 
    This representation of the Kraus maps proves the twirled map is a diagonal-unitary CPSP map because $K_{\mu b}$ is a diagonal map and $X^b$ is a unitary CPSP map.
 
\end{proof}

\section{Measuring the quality of noisy encoded computations}
\label{app:noisy encoded}

Measuring the quality of encoded computations in the absence of full error correction is nontrivial, because physical errors affect encoded qubits quite differently compared to physical qubits.  
In this section, we take some initial steps towards a theory of \emph{logical gate fidelities} (\cref{ssec:logical fidelities}).
Before we do so, we recap the basics of stabilizer codes and analyze the effective noise channels arising on the logical qubits (\cref{app:logical noise}).

\subsection{Stabilizer codes and logical noise channels}
\label{app:logical noise}

Let $S$ be an abelian subgroup of the Pauli group such that $-\id \notin S$.  The size of $S$ is $2^{n-k}$ for an integer~$k$.  The code space projector
\begin{equation}
    \Pi_0 = \frac{1}{2^{n-k}}\sum_{g \in S} g,
\end{equation}
defines an $\code{n,k,d}$ quantum error correcting code where~$d$ is the code distance.  Taking a minimal size generating set for $S =\langle g_1,\ldots,g_{n-k} \rangle$, we can re-express $\Pi = \prod_{i=1}^{n-k} (\id +g_i)/2$.  

For two Pauli operators $P$ and $Q$, define the scalar commutator as $[[P,Q]] = 0 \in \mathbb{F}_2$ if $[P,Q]=0$ and $[[P,Q]] = 1 \in \mathbb{F}_2$ if $PQ = - QP$.  We can define a canonical basis for the Pauli group with reference to our generating set for $S$ as follows: 
find a set of operators $\{h_1,\ldots,h_{n-k}\}$ called destabilizers \cite{aaronson_improved_2004} that satisfy $[[h_i,h_j]]=0$ and $[[g_i,h_j]]=\delta_{ij}$, as well as a set of logical operators $\{\bar{X}_1,\bar{Z}_1,\ldots,\bar{X}_k,\bar{Z}_k\}$ that satisfy $[[\bar{X}_i,h_j]]=[[\bar{X}_i,g_j]]=[[\bar{X}_i,\bar{X}_j]] = 0$, $[[\bar{Z}_i,h_j]]=[[\bar{Z}_i,g_j]]=[[\bar{Z}_i,\bar{Z}_j]] = 0$, and $[[\bar{X}_i,\bar{Z}_j]] = \delta_{ij}$.  This canonical basis set can be found efficiently by solving a set of linear equations over $\mathbb{F}_2$.

A logical state $\ket{\bar{\psi}}$ is any state that satisfies $\Pi_0 \ket{\bar{\psi}} = g \ket{\bar{\psi}} = \ket{\bar{\psi}}$ for all $g\in S$.  A \emph{syndrome} is a set of measurement outcomes $s \in \{0,1\}^{n-k}$ for the stabilizer generators~$g_i$.  
The projector onto a given syndrome outcome is
\begin{equation}
    \Pi_{s}= \prod_{i=1}^{n-k} \frac{\mathbbm{1}+(-1)^{s_i} g_i}{2}.
\end{equation}
  We also define  $\Pi_{0r}$ as the projector onto the \emph{$r$-error subspace}
\begin{equation*}
    \mathrm{span}\Big( \{ P \ket{\bar{\psi}} : \Pi_0 \ket{\bar{\psi}}=\ket{\bar{\psi}}, \, P \in P_n, \, \mathrm{wt}(P)\le r \}\Big),
\end{equation*}
where $P = \bigotimes_{i=1}^n P_i$ is a Pauli operator and $\mathrm{wt}(P) = |\{i: P_i \ne \id\}$ is the \emph{weight} of $P$, i.e., the number of non-identity Paulis in $P$.  The projector $\Pi_{0r}$ can be interpreted as a projector onto the set of states that have at most $r$ errors. Note, $\Pi_{0n} = \id$, while the smallest $m$ such that $\Pi_{0m}= \id$ is greater than or equal to the code distance, but will generally depend on the code.

In order to formalize the effect of physical noise on the logical subspace it is convenient to introduce a generalization of a quantum channel, called \emph{quantum instrument}.
\begin{definition}
    A \emph{quantum instrument} is a collection of CP maps $\{\mathcal{E}_a \}_{a \in A}$ such that $\sum_{a \in A} \mathcal{E}_a$ is a quantum channel. 
\end{definition} 
A projective measurement $\{ \Pi_a \}_{a \in M}$ is an example of a quantum instrument, where $a$ indexes the possible measurement outcomes $M$. 

For a given physical noise channel $\mathcal{E}$, the corresponding \emph{logical noise channel} is the quantum instrument $\{ \mathcal{E}_{s} \}_{s \in \{0,1\}^{n-k}}$ associated with the syndrome outcomes 
\begin{align}
    \mathcal{E}_{s}(\rho ) &= \Pi_{s} \mathcal{E}( \rho ) \Pi_{s} = \sum_\mu \Pi_{s} K_\mu  \rho  K_\mu^\dag \Pi_{s}.
    \end{align}
    For an initial state in the code space, these are effectively logical channels with the Kraus operators
    \begin{align}
     \Pi_{s} K_\mu \Pi_0 &= h^s \Pi_0 h^s K_\mu \Pi_0 =\sum_{abc}c_{abcs}^\mu h^s g^c \bar{X}^b  \bar{Z}^a  \Pi_0,
\end{align}
where $h^s \equiv \prod_{i=1}^{n-k} h_i^{s_i}$ is a product of destabilizers that corrects the syndrome and $c_{abcs}^\mu = \mathrm{Tr}[K_\mu  \bar{Z}^a \bar{X}^b g^c h^s]$.

\subsection{Logical gate fidelities}
\label{ssec:logical fidelities}

We are now ready to define notions of logical transversal gate fidelities.  
Due to entanglement that builds up between syndrome degrees of freedom and logical degrees of freedom during the computation, the  error model for logical qubits is intrinsically non-Markovian, increasing the relative importance of circuit context dependence in defining gate fidelities.  
We provide a partial resolution to this challenge through the use of $r$-error filters that restrict the input states of the logical channel to have at most weight-$r$ errors.

Given a syndrome outcome $s$, a \emph{decoder} is a map from $s$ to a logical Pauli operator $\bar{P}_s$ chosen to correct the error and a destabilizer $\bar{h}^s$ that fixes the syndrome without affecting the logical degrees of freedom.
We define the \emph{$(r,s)$-filtered average gate fidelity} of a decoded logical operation $\bar {\mc U} = \bar{U} \cdot \bar U^\dagger$ conditioned on a particular syndrome outcome $s$ acting on an $r$-error state as
\begin{align}
    F_{sr}(\mc E, \bar{\mc U}) &= \frac{1}{p_{sr}} \int d \psi_r  \bra{\psi_r} \bar{U}^\dag \bar{P}_s h^s \nn \\
    &\times\mathcal{E}_s(\bar{U} \ket{\psi_r}\bra{\psi_r}  \bar{U}^\dag) h^s \bar{P}_s U \ket{\psi_r}, \\
    p_{sr} & = \mathrm{Tr}[\mathcal{E}_s(\Pi_{0r}/d_r)],\\
    d_r &= \mathrm{Tr}[\Pi_{0r}],
\end{align}
where $\mathcal{E} \circ \bar{\mc U}(\rho)$ is the noisy implementation of the logical gate $\bar{\mc U}$, $h^s$ maps the states with syndrome $s$ back to the logical states, $d\psi_r$ is the uniform measure on the set of $r$-error states induced by the Haar measure over logical operators and $p_{sr}$ is the probability of observing syndrome $s$ starting from a maximally mixed state in the $r$-error space.
In the case of error detection, where the optimal decoding strategy is to postselect on the all-$0$ syndrome outcome and set $\bar{P}_0 = \mathbbm{1}$, the average gate fidelity is
\begin{equation}
    F_{0r}(\mc E,\bar{\mc U}) = \frac{1}{p_{0r}} \int d\psi_r \bra{\psi_r} \bar{U}^\dag  \mathcal{E}_0 (\bar{U}\ket{\psi_r}\bra{\psi_r}\bar{U}^\dag)  U \ket{\psi_r}.
\end{equation}

In addition to the average gate fidelity, we also introduce the \emph{$(r,s)$-filtered entanglement fidelity} \cite{schumacher_sending_1996,nielsen_entanglement_1996}
\begin{equation}
    F_{esr}(\mc E, \bar{\mc U}) = \frac 1{p_{sr}} \bra{\Phi_r}  (\mathcal{U}^{\dag} \circ \mathcal{R}_s \circ \mathcal{E}_s \circ \mc U) \otimes \id  (\ket{\Phi_r}\bra{\Phi_r}) \ket{\Phi_r},
\end{equation}
where $\mathcal{R}_s(\rho) = \bar{P}_s h^s \rho h^s \bar{P}_s$ and $\ket{\Phi_r} = \operatorname{vec}(\Pi_{0r}/d_r)$ is a purification of the density matrix $\Pi_{0r}/d_r$.  
As with their standard versions, we find that the $(r,s)$-filtered entanglement fidelity is related to the $(r,s)$-filtered average gate fidelity through the formula \cite{horodecki_general_1999,nielsen_simple_2002}
\begin{equation}
\label{eq:average vs entanglement fidelity}
    F_{sr} = \frac{d_r F_{esr} +1}{d_r+1}.
\end{equation}
To see this equality, one follows the simple proof of \textcite{nielsen_simple_2002}.
Similarly, we can use the $r$-filtered average gate fidelity to upper- and lower-bound an $(r,s)$-filtered version of the diamond distance (defined analogously with respect to $r$-error input states and decoding conditioned on syndrome outcome $s$) of the logical channel from the identity \cite[e.g.][Prop. 9]{wallman_randomized_2014}.

We can analytically compute the $(r,s)$-filtered entanglement fidelity as
\begin{align}
    F_{esr} & = \frac{1}{ p_{sr} d_r^2}\sum_\mu \left\lvert \mathrm{Tr}\big[\Pi_0 \bar{P}_s h^s  K_\mu \Pi_{0r}\big]\right\lvert^2,
\end{align}
where $\{K_\mu\}$ are a set of Kraus operators for $\mathcal{E}$.  
For two-qubit logical gates between code blocks (e.g., the transversal CNOT)  and two-qubit gate error channels that factorize  $\mathcal{E} = \mathcal{E}_1 \otimes\mathcal{E}_2$, we can directly extend this formula to two-qubit logical gates
\begin{equation}
    F_{e ss' rr'} = F_{esr}^{(1)} F_{es'r'}^{(2)}.
\end{equation}

These definitions of logical gate fidelities implicitly assume either that $r$ is less than the code distance or that most states in the $r$-error space are close to a logical state.  Otherwise, the $r$-filtered gate fidelity is dominated by the projection back to the code space.  In the case of near-term experiments where accessible code distances are small, which is the focus of this work, it is convenient to define an intermediate notion of logical fidelity that looks at the fidelity of the logical operation on the whole $r$-error space without error correction and decoding following the gate.  The error rate of this fidelity does not decrease with increasing code distance, but  provides a useful bound on the fidelity of logical gates in a near-term logical circuit.  
We thus define the \emph{$r$-filtered average gate fidelity} and \emph{$r$-filtered entanglement fidelity} as
\begin{align}
    F_r(\mc E, \bar{\mc U}) &= \frac{d_r F_{er}(\mc E, \bar{\mc U})+1}{d_r+1},\\
    F_{er}(\mc E, \bar{\mc U}) &= \frac{1}{p_r}\bra{\Phi_r}  ( \bar{\mc U}^\dag \circ \mc E \circ \bar{\mc U}) \otimes \id ( \ket{\Phi_r}\bra{\Phi_r}) \ket{\Phi_r},\\
    p_r & = \frac{1}{d_r} \mathrm{Tr}[\Pi_{0r} \mathcal{E}(\Pi_{0r})],
\end{align}
respectively. 
In the case of $\code{8,3,2}$ codes, for $r=1$, $\Pi_{01}$ has support on 144 out of all possible 256 states. For $r \ge 2$, $\Pi_{0r}$ has support on the whole 8-qubit Hilbert space. Consequently, in this case $F_{sr}$ is negligible for any $r>0$; however, $F_{r}$ can still be large because of the small size of the code.

In \cref{fig:encoded sampling}(b) in the main text, we present the $r$-filtered gate fidelities  $F_r$ for transversal CCZ, in-block CNOT, and transversal CNOT gates of the $\code{8,3,2}$ code.  In the calculations for the CCZ and transversal CNOT gates we took iid depolarizing noise on each qubit following the application of a physical gate.  For the in-block CNOT, which we assume to be implemented via qubit reconfiguration, we modeled the noise as pure $Z$-dephasing noise on each qubit following the gate.  To numerically compute the $r$-filtered fidelities we directly evaluated the formulas above  by expressing them as polynomials in the local noise rate parameter $p(1-p)$.  We then used direct calculations on the $2^8$ dimensional Hilbert space to evaluate the coefficients of each polynomial.  

For a native two-qubit gate error rate of $0.5\%$, the fidelity of a compiled CCZ gate would be around $3 \%$.  We see in \cref{fig:encoded sampling}(b) that the $r$-filtered fidelity can readily exceed that for $r\,{=}\,0$, while at high enough single-qubit gate error rates $10^{-3}-10^{-4}$ the transversal CCZ fidelity exceeds the compiled version for $r\,{>}\,0$.  In the case of transversal CNOTs shown in \cref{fig:encoded sampling}(b), the $r$-filtered fidelity only exceeds the native fidelity for $r\,{=}\,0$.  The pseudothreshold for two-qubit CNOT gates is observed to be at gate error rates of $4\%$.  When comparing the fidelity of physical CNOTs implemented on 3 independent pairs, which is most comparable to the logical transversal CNOT gate for the $\code{8,3,2}$ code, the pseudothreshold increases to $14\%$.  

\section{Model for fidelity versus acceptance fraction}
\label{app:powerlaw}

As we increase the amount of postselection in the encoded circuit, the XEB improves but the number of discarded samples grows as well. In \cref{fig:encoded sampling}(a) of the main text and in the experimental data (Fig. 5(e)) of Ref.~\cite{bluvstein_logical_2024}, the relationship between the XEB and the fraction of accepted samples reduces to a power law in a certain regime. 

Here, we provide a simple model based on the $\code{8,3,2}$ code that reproduces this behavior.
While the exponent will differ for a different code, we expect the power-law behaviour to generalize. 
We consider $n$ logical qubits encoded into $n/3$ blocks of the $\code{8,3,2}$ code, which experience an i.i.d.\ dephasing channel with rate $p$ on all physical qubits, right before measurement in the $X$ basis. For each block, we can evaluate the probability of a  logical $Z$ errors conditioned on whether the single $X$ stabilizer is violated or not; this is enough to recover the power law behavior mentioned above. The probability of having $k\,{\in}\,\{0,\ldots ,8\}$ errors in a block is
\begin{equation}
    P_k = \binom{8}{k} p^k (1-p)^{8-k},
\end{equation}
and the corresponding probability of violating the stabilizer is
\begin{equation}
    \epsilon = \sum_{k\in \{1,3,5,7\}} P_k,
\end{equation}
which corresponds to having an odd number of errors. The logical operators within a block are defined up to multiplication with stabilizers, but in practice we choose a particular operator realization, e.g., the three sides of a cube highlighted in \cref{fig:fig1}(d). 
Then, for every odd number of errors, we have a $7/8$ probability of having at least one logical fault, since for any choice of logical operators, a single physical qubit does not participate in any of them. 
 Similarly, the probability of having a logical fault given an even number of physical errors is $\{0,1,4/5,1,0\}$ for $\{0,2,4,6,8\}$ errors, respectively. Together, this results in probabilities  $p_l^{s}$ of logical faults conditioned on whether a stabilizer is violated ($s = -$) or not ($s = + $), which are given by
\begin{align}
    p_l^{-} &= 7/8, \\
    p_l^{+} &= \frac{1}{1-\epsilon}(P_2+\frac{4}{5}P_4 + P_6).
\end{align}
Given $n_b=n/3$ blocks, the probability of having exactly $j$ stabilizers violated is
\begin{equation}
    S_j = \binom{n_b}{j}\epsilon^j (1-\epsilon)^{n_b-j},
\end{equation}
and the probability of having no logical error in that situation is
\begin{equation}
    F_j = (1-p_l^{-})^j (1-p_l^{+})^{n_b-j}.
\end{equation}
The acceptance fraction for up to $m$ violated stabilizers is then
\begin{align}
    \bar{S}_m = \sum_{j \leq m} S_j,
    \label{eq:accept fraction}
\end{align}
and the fidelity (i.e., the probability of no logical errors) is
\begin{align}
    \label{eq:fidelity fraction}
    \bar{F}_m = \frac{1}{\bar{S}_m} \sum_{j \leq m }S_j F_j,
\end{align}
which results in an approximate power law (when evaluated numerically) for the fidelity $\bar{F}_m$ as a function of acceptance fraction $\bar{S}_m$. This is in good qualitative agreement with circuit-level simulations such as those in \cref{fig:encoded sampling}(a). 

Finally, we provide an analytical formula for the power-law exponent in the low-noise, highly post-selected regime. Assuming a power-law relationship $\bar{F}_m = a\, \bar{S}_m^\alpha$, we can estimate $\alpha$ in the post-selected regime from $\bar{F}_1/\bar{F}_0$ and $\bar{S}_1/\bar{S}_0$ using \cref{eq:accept fraction,eq:fidelity fraction}. 
To first order in $p$, these are given by
\begin{align}
    \bar{F}_1/\bar{F}_0 &= 1-7n_b\, p + O(p^2),\\
    \bar{S}_1/\bar{S}_0 &= 1+8n_b\,p + O(p^2),
\end{align}
which gives $\alpha=-7/8$ and is, to first order, independent of both $p$ and $n_b$. In Fig.~\ref{fig:encoded sampling}(a), we show a power law fit with this value of $\alpha$, and observe good agreement with simulated data.

\section{Statistical model for second moment quantities of IQP circuits}
\label{app:statmech}

To characterize random circuits in terms of their output distributions in the noiseless case and their quality of implementation in the presence of noise, it is useful to analyze their \emph{second-moment quantities}. 
Second-moment quantities include all quantites which can be expressed as average quantities that are linear in two tensor copies of the random circuits applied to a fixed initial state.
Most importantly for us, the average collision probability,  the cross-entropy benchmark, the average fidelity and the average purity are all second-moment quantities. 
In this work, we are interested in those quantities for random degree-3 IQP circuits. 

In this section, we derive a framework for analysing the second-moment properties of ideal and noisy circuits with some random IQP degree-$2$ gates whose output is a polynomial binary phase state
\begin{align}
    \ket {\psi} = \sum_x (-1)^{f(x)} \ket x,
\end{align}
for some degree-$2$ Boolean polynomial $f$.
This is simpler than analyzing degree-$3$ circuits directly, and will allow us to make statements about random degree-$3$ circuits as well, since we can fix the degree-$3$ part of a circuit and then average over its degree-$2$ part.

Importantly, our framework will \emph{not} require the IQP circuits to be \emph{globally random} but only \emph{locally random}. 
Specifically, we consider circuit architectures, in which the first layer is a layer of Hadamard gates, on certain pairs $(i,j)$ of qubits, a random degree-$2$ IQP circuit $Z_i^{u}Z_j^v CZ_{ij}^w$ for $u,v,w \in \bin$ is applied. All other components of the (ideal) circuit are arbitrary gates which map binary phase states to binary phase states. 
We also allow for noise which preserves binary phase states.

Choose an IQP circuit $D \sim \mf I$ from some IQP ensemble $\mf I$ satisfying the properties above.
Let us then write $\rho_D(\mc N)$ for the mixture of binary phase states generated by $D$ with noise specified by $\mc N$. 
The ideal state is given by $\rho_D(\{\id\}) = D \projx 0 D$.
Our central quantity of interest is now the \emph{second-moment operator }
\begin{align}
\label{eq:iqp second moment}
    M_2(\mc N, \mc N') = \mb E_{D \sim \mf I} \left[ \rho_D(\mc N'), \rho_D(\mc N) \right].
\end{align}
To be most general, we allow different types of noise on the two copies. 

The average collision probability $\overline C$, average XEB fidelity $\overline \chi$, average fidelity $\overline F$, and average purity $\overline P$ can be written in terms of \cref{eq:iqp second moment} as 
\begin{align}
    \overline C & = \tr[ \Xx^{\otimes n} M_2(\{\id\}, \{\id\})] \label{eq:collision second moment}\\
     \overline \chi & = 2^n \tr[ \Xx^{\otimes n} M_2(\{\id\}, \mc N)]-1 \label{eq:xeb second moment}\\
    \overline F & = \tr[ \mb S^{\otimes n} M_2(\{\id\}, \mc N)] \label{eq:fidelity second moment}\\
    \overline P & = \tr[\mb S^{\otimes n} M_2(\mc N, \mc N)] \label{eq:purity second moment}
\end{align}
where we have defined $\mb S = \sum_{x,y \in \bin} \ket {xy}_a \bra{yx}_a$ as the single-qubit swap operator between the two copies, and $\mb X_a = \sum_{x \in \bin} \ket {xx}_a \bra{xx}_a$ as the single-qubit `delta projector' in the $a$-basis for $a \in \{x,z\}$.
Note that while the collision probability involves no noise on either copy in $M_2$, the average fidelity and average XEB fidelity involve noise on one copy and no noise on the other copy. 
In contrast, the purity involves noise on both copies.

Let us now derive the formalism. 
To start out, we consider $D$ chosen uniformly at random from degree-2 IQP circuits. 
As discussed in \cref{ssec:phase states}, each such circuit is characterized by a degree-$2$ polynomial $f(x) = \sum_{ij}b_{ij} x_i x_j + \sum_i a_i x_i $ with uniformly random choices of $a \in \bin^n,b \in \bin^{n \times n}$ specifying the locations of the $Z$ and $CZ$ gates. 

\subsection{The global second moment}

In order to  compute $M_2(\{\id\}, \{\id\})$ we build on the proof of anticoncentration of globally random degree-$2$ IQP circuits by \textcite{Bremner.2016} (Lemma 11 in App.~F).  
Let $\mf I_n$ be the family of uniformly random degree-2 IQP circuits on $n$ qubits, and $\mb P = \sum_{xy \in \bin} \ket{xx}\bra{yy}$ the unnormalized Bell state projector.

\begin{lemma}[Global projector]
\label{lem:global iqp proj}
    Consider the moment operator $M_2$ for random degree-$2$ IQP circuits in $\mf I_n$. It holds that
    \begin{align}
        M_2(\{\id\},\{\id\}) = \frac 1{2^{2n}}  (\id^{\otimes n} + \mb S^{\otimes n} + \mb P^{\otimes n} - 2 \Xz^{\otimes n}). 
    \end{align}
\end{lemma}
Before we prove the \lcnamecref{lem:global iqp proj}, we would like to observe the analogy with the second moment operator of the Haar measure, which is given by 
\begin{align}
\label{eq:second moment haar}
    \mb E_{U \sim U(d)}U^{\otimes 2} \proj {0} (U^\dagger)^{\otimes 2} = \frac {1}{2d_{[2]}}(\id^{\otimes n} + \mb S^{\otimes n}) ,
\end{align}
 where $d_{[2]}$ is the dimension of the symmetric subspace of two $d$-dimensional systems.  
While the diagonal action of the full unitary group is only invariant under a global swap of the two copies, giving rise to the swap operator in the moment formula \eqref{eq:second moment haar}, the diagonal action of random IQP circuits on binary phase states has additional invariances, namely, precisely those captured by the $\mb P$ and $\Xz$ operators.

\begin{proof}
To show \cref{lem:global iqp proj}, we expand
\begin{align}
\label{eq:iqp second moment expanded}
    M_2(\{\id\},\{\id\}) =  \frac 1 {2^{2n}}\sum_{z_1, \ldots, z_4} \mb E_{f} \left[(-1)^{\sum_i f(z_{i})}\right] \ket{z_1 z_3} \braz{z_2z_4},
\end{align}
and observe that 
\begin{multline}
\label{eq:average phase state coefficient}
    \mb E_{f} \left[(-1)^{\sum_i f(z_i)} \right]\\\propto \sum_{a \in \bin^n} \sum_{b \in \bin^{n \times n}} (-1)^{\sum_i \left(\sum_{kl}b_{kl} z_{ik} z_{il} + \sum_k a_k z_{ik} \right)}. 
\end{multline}
Now, we find criteria on the $z_i$ that have to be satisfied for the coefficient \eqref{eq:average phase state coefficient} to be nonzero.  
Following Ref.~\cite{Bremner.2016} this average results in three constraints on the sums, 
\begin{align}
   (\id) \quad&z_1=z_2\eqqcolon x \ 
    {\rm and}\ z_3=z_4 \eqqcolon y,\quad {\rm or}\nn \\ 
   (\mb S)\quad &z_1=z_4 \eqqcolon x\ {\rm and}\ z_2=z_3 \eqqcolon y ,\quad {\rm or}\nn \\
  (\mb P)\quad  &z_1=z_3 \eqqcolon x \ {\rm and}\  z_2=z_4 \eqqcolon y . \label{seq:z_constr}
\end{align}

For the sake of completeness, we recall the derivation here.
Consider the $k$-th bits $z_{1k},z_{2k},z_{3k},z_{4k}$ of the vectors $z_i$. 
Then any nonzero term in the exponent will result in the sum to vanish. 

Hence, the sum over $a_k$ combined with the linear term $\sum_{k} a_{k} z_{ik}$ enforces $z_{1k} = z_{2k}+z_{3k}+z_{4k}$.  
Substituting this, the sum over $b_{kl}$ now enforces that
\begin{equation}
\begin{split}
    0 &= z_{1k}z_{1l}+z_{2k}z_{2l}+z_{3k}z_{3l}+z_{4k}z_{4l} \label{seq:2constr} \\
    &= (z_{2k}+z_{3k}+z_{4k})(z_{2l}+z_{3l}+z_{4l})  \\
    &+z_{2k}z_{2l}+z_{3k}z_{3l}+z_{4k}z_{4l} \\
    & = z_{2k}(z_{3l}+z_{4l}) + z_{3k}(z_{2l}+z_{4l}) + z_{4k}(z_{2l}+z_{3l}),  
\end{split}
\end{equation}
where we removed terms that appeared an even number of times ($z_{2k}z_{2l}$ etc.) since they evaluate to $0$.  From the pigeonhole principle, we know that at least two of the $\{z_{2k},z_{3k},z_{4k}\}$ must be the same.

To show constraint $(\id)$, choose $z_{3k}\,{=}\,z_{4k}\,{=}\,y_k$ and rewrite the above condition as,
\begin{align*}
    (z_{2k}+y_k)(z_{3l}+z_{4l}) = 0,
\end{align*}
where we again removed the term $2z_{2l}y_k $. 
There are now two cases.  
In the first case, if $z_{2k}\,{\neq}\,y_k$ we immediately obtain $z_{3l}\,{=}\,z_{4l}$ for all $l$ and thus $z_3\,{=}\,z_4$ and $z_{1}\,{=}\,z_{2}$. 
Otherwise, if $z_{2k}\,{=}\,y_k$, then we have $z_{1k}\,{=}\,z_{2k}\,{=}\,z_{3k}\,{=}\,z_{4k}$ which is covered by all of the conditions in \cref{seq:z_constr}. 
The equality constraint  $z_{1l}\,{=}\,z_{2l}\,{=}\,z_{3l}\,{=}\,z_{4l}$ then enforces one of the conditions in \cref{seq:z_constr} on the $l$-th bit.
If the $z_{il}$ are not all equal, we can apply the first case to obtain one of the conditions in \cref{seq:z_constr} for all bits. 
Constraints $(\mb S)$ and $(\mb P)$ follow from the same argument starting with the choice $z_{2k} = z_{3k}$ and $z_{2k} = z_{4k}$, respectively.

The only overlap between the conditions is when 
\begin{align}
\label{eq:delta condition}
    (\Xz) \quad z_1\,{=}\,z_2\,{=}\,z_3\,{=}\,z_4\,{=}\,x. 
\end{align}
Thus, we can evaluate these three conditions independently and then subtract the over-counted terms corresponding to all four bitstrings being equal. 

Each of the conditions $(\id), (\mb S), (\mb P), (\Xz)$ in Eqs.~\eqref{seq:z_constr} and \eqref{eq:delta condition} results in the averaged exponents in \cref{eq:average phase state coefficient} to be trivially even-valued, since every term will appear in pairs.
Thus, the average coefficient evaluates to $1$. 
The statement of the lemma follows from inserting the conditions in \cref{eq:iqp second moment expanded}.
\end{proof}

\cref{lem:global iqp proj} directly implies the anticoncentration lemma of \textcite{Bremner.2016}.
\begin{corollary}[Anticoncentration of random IQP]
\label{cor:bremner anticoncentration}
    The average collision probability of random degree-2 IQP circuits is given by $3 \cdot 2^{-n} - 2^{-2n+1}$. 
\end{corollary}
\begin{proof}
We apply \cref{lem:global iqp proj} to \cref{eq:collision second moment} to find
\begin{align}
    \overline C & = \frac 1 {2^{2n}} \tr[\Xx^{\otimes n}  (\id^{\otimes n} + \mb S^{\otimes n} + \mb P^{\otimes n} - 2 \Xz^{\otimes n})]\\
    & = \frac 1 {2^{2n}}(2^n + 2^n + 2^n - 2 ) = 3 \cdot 2^{-n} - 2^{-2n+1},
\end{align}
where we have used that $\tr[\Xx \mb O] = \tr[\Xz \mb O]$ for \mbox{$\mb O = \id, \mb S, \mb P$} since the Hadamards cancel, and $\tr[\Xz \Xx] = 1$.
\end{proof}

\subsection{Statistical mechanics model for locally random degree-2 circuits}
\label{app:stat mech mapping}

In the following, we will use the properties of IQP circuits with locally restricted random degree-2 parts. 
To this end, we follow a line of thought developed in Refs.~\cite{dalzell_random_2024,gao_limitations_2024,Ware.2023}. 
There, the observation is, that layers of single-qubit Haar-random gates in the circuit project the corresponding second-moment operator locally onto the space spanned by the $\id$ and $\mb S$ operators. 
The effect of any gate between those layers of single-qubit Haar-random gates can thus be projected onto this space. 
As it turns out, the update rules of the $\id$ and $\mb S$ operators under such two-qubit gates are quite simple. 

We will proceed to derive an analogous model that lets us analyze the second-moment properties of noisy and noiseless IQP circuits with certain non-random gates. 
The relevant state space, in the most general case, is locally spanned by the invariant subspaces of those random IQP circuits.

\begin{lemma}[IQP invariants]
\label{lem:invariants}
Let $O = \sum_{x,y,z,w=0}^{2^n-1} q(x,y,z,w) \ket{xy}\bra{zw}$ be an operator on two copies of $n$ qubits. Then
\begin{multline}
\label{eq:random iqp general rho}
    \mb E_{D \sim \mf I_n} [ D^{\otimes 2 } O D^{\otimes 2}] \\ = \sum_{x\neq y=0}^{2^n-1} \big[ q(x,y,x,y) \ket{xy}\bra{xy} + q(x,y,y,x) \ket{xy}\bra{yx} \\+ q(x,x,y,y) \ket{xx}\bra{yy} \big]+ \sum_{x=0}^{2^n-1} q(x,x,x,x) \proj{xx}. 
\end{multline}
\end{lemma}
\begin{proof}
We follow the same argument as in the proof of \cref{lem:global iqp proj}. 
\end{proof}
\cref{lem:invariants} implies that there are $3 \cdot (2^{2n} - 2^n ) + 2^n  $ linearly independent operators $\ket{xy}\bra{xy}, \ket{xy}\bra{yx}, \ket{xx}\bra{yy}$ for all $x,y \in \bin^n$ which are invariant under the average over random IQP circuits. 
The IQP-average projects an input operator $O$ onto the space spanned by these operators. 

However, for the circuits that we will be interested in, we can significantly reduce this number. 
To this end, we consider circuits in which there are layers of minimal units of random IQP circuits on two qubits interleaved with other, non-random gates.  
This is a toy model of the hIQP circuits, where we apply layers of random IQP circuits on blocks of three qubits interleaved with layers of CNOT gates. 

From \cref{lem:invariants}, it follows that there are $3 \cdot (16 - 4) + 4 = 40$ linearly independent operators on the two qubits. 
We reduce this number to $16$ operators which are written as arbitrary products of $4$ linearly independent operators per qubit, exploiting that our circuits are quite restrictive. 

Specifically, we consider the mutually orthogonal operators on two copies of a single qubit
\begin{align}
    \ix & = \id - \Xz \label{eq:ix}\\
    & = \proj{01} + \proj{10} \nonumber \\
    \sx & = \mb S - \Xz  \label{eq:sx}\\
    &= \ket{01}\bra{10} + \ket{10}\bra{01} \nonumber \\
    \px & = \mb P - \Xz  \label{eq:px}\\
    &= \ket{00}\bra{11} + \ket{11}\bra{00}  \nonumber\\
    \xx & =\Xz  \label{eq:xx}\\
    &= \proj{00} + \proj{11} \nonumber. 
\end{align}
We define the local state space to be $\mc S = \{ \ix, \sx, \px, \xx\}$. 
It will also be useful to define $\mc S_\notx = \mc S \setminus \{\xx\}$.  
The operators $\ix, \sx, \px, \xx$ directly correspond to mutually exclusive constraints $(\ix), (\sx), (\px), (\xx)$ analogous to the (non-exclusive) constraints $(\id), (\mb S), (\mb P), (\Xz)$ defined in \cref{seq:z_constr,eq:delta condition}.

The circuit components we now need to analyze are 
\begin{enumerate}
    \item state preparation of $\proj{+^n}^{\otimes 2}$,
    \item layers of random IQP circuits $\prod_{i=1}^{n/2} D_{2i,2i+1}^{\otimes 2}$, where $D_{2i,2i+1} \in \mf I_2 $ are random degree-$2$ IQP circuits on two qubits, 
    \item $\cnot_{i,j}^{\otimes 2}$ gates between any two sites $i$ and $j$, 
    \item $X$ and $Z$ noise, that is operators $X \otimes \id$, $Z \otimes \id$ applied to the two-copy state.
\end{enumerate}

Let us first consider state preparation. We find that under a $z$-twirl the two-copy state $\proj{+^n}^{\otimes 2}$ can be expressed as
\begin{align}
   \mb E_{z\in \bin} &(Z^z)^{\otimes 2} \proj{+}^{\otimes 2} (Z^z)^{\otimes 2} \\
   & = \frac{1}{8} \sum_{u,v,w,x,z \in \bin} Z^z \ket u Z^z \ket v \bra w Z^z \bra x Z^z
   \\
   & = \frac 1 {4} (\ix + \sx + \px + \xx)^{\otimes n}.
\end{align}
In the following analysis we will therefore always assume that the first layer of the circuit is a layer of random $Z$ gates. 
This is no restriction in all cases we consider since this layer can be absorbed in the random $Z$ gates throughout the circuit.

\begin{lemma}[Projection update]
\label{lem:iqp projection}
Let $Q \in \mc S$, and $P, R \in \mc S_\notx$. Then 
\begin{align}
    \mb E_{D \in \mf I_2} \left[ D^{\otimes 2} (Q \otimes \xx) D^{\otimes 2} \right] &= Q \otimes \xx \label{eq:qx}\\
    \mb E_{D \in \mf I_2} \left[ D^{\otimes 2} (\xx \otimes Q) D^{\otimes 2} \right] &= \xx \otimes Q\label{eq:xq}\\
    \mb E_{D \in \mf I_2} \left[ D^{\otimes 2} (P \otimes R) D^{\otimes 2} \right] &= \delta_{P,R} P \otimes R \label{eq:pr}, 
\end{align}
where $\delta_{P,R} = 1$ if $P=R$ and $0$ otherwise.
\end{lemma}

\begin{proof}
The proof follows immediately from \cref{lem:invariants}, observing that all operators $\xx \otimes Q$ and $Q \otimes \xx$ satisfy one of the constraints $(\ix), (\sx), (\px), (\xx)$, while $P \otimes R $ for $P \neq R$ does not satisfy any. 
\end{proof}

In the next step, we analyze the behaviour of the operators under CNOT operators. 
\begin{lemma}[CNOT update]
\label{lem:cnot update}
Let $P,Q \in \mc S$. 
Then 
\begin{align}
    \cnot\tc (P \otimes Q) \cnot\tc = P \otimes \varepsilon_P(Q), 
\end{align}
where we define the permutation $\varepsilon_P: \mc S \rightarrow \mc S$ according to the following table 
\begin{center}
\begin{tabular}{c >{\ } c<{\quad}  c  >{\quad} c>{\quad}c }
    \toprule
   $Q$ & $\xx$ & $\ix$ & $\sx$ & $\px$  \\ \midrule\midrule
   $\varepsilon_{\xx}(Q)$ & $\xx$ & $\ix$ & $\sx$ & $\px$ \\
   $\varepsilon_{\ix}(Q)$ & $\ix$ & $\xx$ & $\px$ & $\sx$ \\
   $\varepsilon_{\sx}(Q)$ & $\sx$ & $\px$ & $\xx$ & $\ix$ \\
   $\varepsilon_{\px}(Q)$ & $\px$ & $\sx$ & $\ix$ & $\xx$ \\
   \bottomrule
\end{tabular}    . 
\end{center}
\end{lemma}
Notice that for any input $Q \in \mc S_\notx$, the permutation $\varepsilon_Q$ swaps $Q$ with $\xx$ as well as the other two states in $\mc S$, while $\epsilon_{\xx}(Q) = Q$. 
In order to prove \cref{lem:cnot update}, and to show the behaviour under $X$ noise, the following lemma will be  convenient. 

\begin{lemma}[$X$ updates]
\label{lem:one copy x}
Let $P \in \mc S$. Then 
\begin{align}
    (X \otimes X ) P (X \otimes X) =  (\id \otimes \id ) P (X \otimes X) = \varepsilon_{\xx}(P)\label{eq:x xx}\\
    (\id \otimes X ) P (\id \otimes X) =  (X \otimes \id ) P (X \otimes \id) = \varepsilon_{\ix}(P) \label{eq:x ix}\\
    (X \otimes \id ) P (\id \otimes X) =  (\id \otimes X) P (X \otimes \id) = \varepsilon_{\sx}(P) \label{eq:x sx}\\
    (X \otimes X ) P (\id \otimes \id) =  (\id \otimes \id ) P (X \otimes X) = \varepsilon_{\px}(P) \label{eq:x px}
\end{align}
\end{lemma}
To prove the lemma, we just expand $P$ in the computational basis. 

\begin{proof}[Proof of \cref{lem:cnot update}]    
Given \cref{lem:one copy x}, the \lcnamecref{lem:cnot update} follows immediately from expanding $P$ in the computational basis and observing that every $1$-entry in $P$ results in an $X$ applied to the corresponding entry of $Q$. 
For each $P$, the patterns o $1$ entries in $P$, correspond exactly to one of the cases in \cref{eq:x xx,eq:x ix,eq:x sx,eq:x px}. 
For instance, for $P \otimes Q = \ix \otimes Q$, we have
\begin{multline}
\cnot\tc (\ix \otimes Q) \cnot \tc \\
= \proj{01} (\id \otimes X ) P (\id \otimes X) + \proj{10}(X \otimes \id ) P (X \otimes \id)\\
=  \proj{01} \varepsilon_{\ix} (Q) + \proj{10} \varepsilon_{\ix}(Q)  = \ix \otimes \varepsilon_{\ix}(Q),
\end{multline}
and likewise for the other cases. 
\end{proof}

Finally, we need to analyze noise on one copy of the second moment operator. We do so for single-qubit $X$ and $Z$ noise
The case of $X$ noise follows immediately from \cref{lem:one copy x}. 
For $Z$ noise we find 
\begin{lemma}[$Z$ noise]
\label{lem:z noise}
We have that  
\begin{align}
\label{eq:z noise}
    ( Z \otimes \id ) P (Z \otimes \id) = \begin{cases}
        P & \text{if } P \in \{\ix, \xx\}\\
        - P & \text{if } P \in \{\sx, \px \}. 
    \end{cases}
\end{align}
\end{lemma}
The lemma follows from direct calculation in the computational basis. 

One might wonder whether the states in $\mc S$ remain invariant also under Toffoli gates, since those allow implementing arbitrary permutations. 
However, this is not the case. 
To see this, consider the behaviour of the operator $\xx \otimes P \otimes Q$ for any $P, Q \in \mc S$ under two-copy Toffoli gates        
\begin{multline}
    \tof\tc( \xx \otimes P \otimes Q)\tof\tc = \\
    \proj{00} \otimes P \otimes Q + \proj{11} \otimes P \otimes \varepsilon_{Q}(R),
\end{multline}
which cannot be decomposed into a linear combination of operators in $\mc S^{\otimes 3}$.

Given \cref{lem:iqp projection,lem:cnot update,lem:one copy x,lem:z noise}, we can now identify every operator $\sigma \in \mc S^{\otimes n}$ with a classical state $\ket{\sigma}$. 
The second moment of a circuit with the above components can then be written as 
\begin{align}
    \ket{M_2} = \sum_{\sigma \in \mc S^{\otimes n}} q(\sigma) \ket \sigma.
\end{align}
Given the second moment $\ket {M_2}$ of a certain random circuit $D$ in terms of the states $\ket S$ we can deduce the properties of the same circuit with a local gate or random IQP circuits applied to certain qubits $(i,j)$ as an update $T^E_{ij} \ket{M_2} = \sum_S q'(S) \ket S$, where $T_{ij}^E$ just acts on the tensor copies $i$ and $j$. 
We can thus analyze the second moment properties of circuits as an evolution of the distribution $q$ under the local update rules derived above.
The general update rule under a single parallel circuit layer with label $l$ is then given by 
\begin{align}
    \mc T(\sigma',\sigma) = \prod_i N^{(i)}_{\sigma_i', \sigma_i} \prod_{\langle i, j \rangle} T^{(ij)}_{\sigma_i', \sigma_j', \sigma_i, \sigma_j},
\end{align}
where the $Z$ and $X$ noise update matrices $N^Z$ and $N^X \in \bin^{4 \times 4}$ are defined by the update rules \eqref{eq:z noise} and \eqref{eq:x ix}, respectively, and the CNOT and random two-qubit IQP circuit update matrices $M^{\cnot}, M^{\mf I_2} \in \bin^{16 \times 16}$ by \cref{lem:cnot update,lem:iqp projection}, respectively. 

In order to compute our properties of interest in \cref{eq:collision second moment,eq:xeb second moment,eq:fidelity second moment,eq:purity second moment}, we need to compute the overlap of the resulting state with the $\Xx$ and $\mb S$ states. 
Defining $\braket{Q}{P} = \tr[QP]$, we observe that 
\begin{align}
\label{eq:swap trace}
    \braket{\mb S}{Q} &= \begin{cases}
        2 & \text{if } Q \in \{\sx, \xx\} \\
        0 & \text{if } Q \in \{\ix, \px\}\\
    \end{cases}\\
    \label{eq:x trace}
    \braket{\Xx}{Q} & =1  \quad\forall Q \in \mc S. 
\end{align}

\section{Anticoncentration of degree-$D$ circuits}
\label{app:second moment behaviour}

Using the rules derived in this section, we can analyze the second-moment properties of arbitrary noiseless circuits with some random IQP degree-$2$ parts and $X$, CNOT gates, as well as their behaviour under single-qubit Pauli noise. 
To start out, let us apply the model to compute some simple quantities of which we already know. 

First, we compute the fidelity of the ideal state with itself, which of course evaluates to one. 
We begin with the initial state $ \ket{+^n}$. 
We find 
\begin{align}
    \tr[\mb S\nc M_2(\{\id\},\{\id\})] = \frac 1 {4^n} \sum_{Q\in \mc S\nc} \prod_{i=1}^n\tr[Q_i \mb S]. 
\end{align}
But $\tr[Q_i \mb S] =0 $ unless $Q_i = \xx$ or $Q_i = \sx$ by \cref{eq:swap trace}. Hence, only all products of $\xx$ and $\sx$ are nonzero, and evaluate equally to $2^n$.
But there are exactly $2^n$ such terms and hence 
\begin{align}
    \overline F = \frac 1 {4^n} \sum_{Q\in \mc S\nc} \prod_{i=1}^n\tr[Q_i \mb S] = \frac 1 {4^n} \sum_{S \in \{\xx, \sx\}^{\otimes n}} 2^n = 1. 
\end{align}
Since the states in $ \{\xx, \sx\}^{\otimes n}$ are invariant under the two-qubit IQP circuit, and swap and entangling gates just permute them, the fidelity remains unchanged under the action of a circuit. 

Next, let us evaluate the XEB score. 
For XEB every term in $\mc S\nc$ contributes equally with a value $1$. 
For the state $\ket{+^n}$ we thus find 
\begin{align}
    \overline \chi &= 2^n\tr[M_2(\{\id\},\{\id\}) \xx\nc] -1 \\
    &= \frac {2^n} {4^n} \sum_{Q \in \mc S\nc} \prod_i \tr[Q_i \xx]-1\\
    & = \frac 1 {2^n} 4^n -1= 2^n -1, 
\end{align}
an exponentially large value. 

Applying a uniformly random IQP circuit on the to $\ket {+^n}$, however, should reduces the XEB score to the value found in \cref{cor:bremner anticoncentration}. 
To see how this score comes about in our model, we write a uniformly random $n$-qubit IQP circuit as $n(n-1)/2$ independent random two-qubit IQP circuits $D_{(i,j)} $ on every (unordered) pair of qubits $(i,j)$. 
This is equivalent to an $n$-qubit random IQP circuit, since the Z-gates just commute to the end, yielding a uniformly random string. 
But every two-qubit random IQP circuit $D_{(i,j)}$ acts as a constraint on the states of the model at positions $(i,j)$, namely, any state $S \in \mc S\nc$ with $S_i, S_j \in \mc S_{\neg \xx}$ and $S_i \neq S_j$ is annihilated by the constraint. 
We thus find 
\begin{multline}
    M_2(\{\id\},\{\id\})  = \mb E_{\{D_{(i,j)}\}}[ (\prod_{i,j} D_{(i,j)}) \sum_{S \in \mc S\nc} S (\prod_{i,j} D_{(i,j)})]\\
      = \sum_{S \in \{\xx, \ix\}\nc \cup \{\xx, \sx\}\nc \cup  \{\xx, \px\}\nc} S - 2 \xx\nc\\
     = \id\nc + \mb S\nc + \mb P\nc - 2\Xz\nc,
\end{multline}
which recovers the result of \cref{lem:iqp projection} since the sum over all elements in $\{\xx,\sx\}\nc$ equals $\mb S$.

\subsection{Sparse degree-2 IQP circuits}
\label{app:noiseless iqp + swap}

Let us now consider the first nontrivial example, namely, the following circuit model.
Consider a random IQP circuit with $\ell$ layers such that every layer is given by a random two-qubit IQP circuit $D_{(i,j)}$ on a random qubit pair $i,j$. 
We would now like to derive the anticoncentration properties of those circuits as a function of $d$. 

\begin{theorem}[Anticoncentration of sparse random IQP] \label{lem:ac}
Consider random degree-$2$ IQP circuits on $n$ qubits with uniformly random $Z$ gates and random CZ gates acting on $\ell$ random pairs. 
Then 
\begin{align}
    \overline \chi & \in 2 - 2^{-n +1} + p(n) 2^{-2\ell/n} & \text{if } \ell \geq \frac {11} 4 n \log(n) \\
    \label{eq:diverging}
    \overline \chi &\in \Omega(2^n) & \text{if } \ell \in O(n),
\end{align}
where $p(n) = n^{11/2}/(2\pi)^{3/2}$. 
\end{theorem}
\begin{proof}
For $Q \in \mc S$, define the $Q$-Hamming weight of a string $S \in \mc S$ as 
\begin{align}
    n_Q(S) = |\{i: S_i = Q\}|. 
\end{align}
We now observe that the action of a circuit layer on the states $\ket S$ depends only on the number of pairs $S_i\neq S_j$ such that $S_i, S_j \in \mc S_\notx$. 
Let us call a string $S$ with the property of having $k$ of such pairs a string with \emph{$PQ$-weight $k$}, written as $\pqw(S) = k$. 
We have that 
\begin{multline}
    \pqw(S) \equiv N(n_\ix(S),n_\sx(S),n_\px(S))\\
    \coloneqq  n_\ix(S)n_\sx(S) + n_\ix(S)n_\px(S) + n_\sx(S)n_\px(S).
\end{multline}
The probability that a state $S \in \mc S\nc$ is annihilated in a single circuit layer is then given by $\pqw(S)/\binom n 2$ since there are $\binom n 2 $ pairs in total. 

Let us write the random circuit in layer $i$ with a random degree-$2$ circuit in a random location as $D_i$. 
Let us also define the sets 
\begin{align}
    W_n & \coloneqq \{ (k,l,m) \in [n]^{\times 3}: k + l + m \le n\}\\
    W_n^* & \coloneqq W_n \setminus  \{ (k,l,m): \\& \quad k = l = 0 \vee k = m = 0 \vee l = m = 0\}\nonumber\\
    \mc S_n^{(k,l,m)} &\coloneqq \{ S \in \mc S\nc: \\
    &\quad n_\ix(S) = k, n_\sx(S) = l, n_\px(S) = m\}\nonumber
\end{align}
We now write 
\begin{align}
    \overline \chi +1 & = \frac 1{2^n} \sum_{S \in \mc S} \mb E_{\{D_i\}}  \tr[\left(\prod_{i=1}^\ell D_i\right)S \left(\prod_{i=1}^\ell D_i\right) \Xx]\\
    & = \frac 1 {2^n} \sum_{S \in \mc S} \left ( 1- \frac{\pqw(S)}{\binom n 2}\right)^\ell\\
& = \frac 1 {2^n} \sum_{(k,l,m) \in W_n  } |\mc S_n^{(k,l,m)}|  \left ( 1- \frac{N(k,l,m)}{\binom n 2}\right)^\ell,
\end{align}
and observe that the strings with generalized Hamming weight in $W_n^*$ contribute exactly the ideal score of $3 - 2^{-n + 1} $ to the XEB score so that 
\begin{align}
\label{eq:chi scaling sparse}
    \overline \chi&= 2 - \frac 2 {2^n} +  \frac 1 {2^n} \sum_{(k,l,m) \in W_n^*  } |\mc S_n^{(k,l,m)}|  \left ( 1- \frac{N(k,l,m)}{\binom n 2}\right)^\ell.
\end{align}
Let us now bound the terms in \cref{eq:chi scaling sparse}. To this end, we observe that 
\begin{align}
|W_n^*| &=  n^3- 3n\\
|\mc S_n^{(k,l,m)}| &=  \binom n {k} \binom {n - k} {l} \binom{n - k - l}{m} .
\end{align}
Using Stirling's formula for $k = p n$ we find the exact scaling of the binomial coefficients to be 
\begin{align}
   \binom n {pn} \in \left(\frac 1 {2 \pi n p(1-p)}\right)^{\frac 12} 2^{n H(p)}, 
\end{align}
where $H(p) = - p \log p - (1-p) \log(1-p)$ is the binary entropy function ($\log$ has base $2$). 
We thus find that 
\begin{align}
\label{eq:sparse anticoncentration bound}
    &\overline \chi - 2 + 2^{-n + 1}\leq  \frac{n^{7/2}} {(2\pi)^{3/2}} \times\\
    &\times \max_{(k,l,m)\in W_n^*} 2^{-\left[n + \frac {\ell}{\binom n 2 } N(k,l,m) - n\left[ H(\frac k n) + H(\frac l n) + H(\frac m n) \right] \right]}. \nonumber
\end{align} 

Now, since $2^{-x}$ is monotonously decreasing in $x$ we want to find the minimum exponent
\begin{multline}
\label{eq:minimum exponent sparse anticoncentration}
    \min_{(k,l,m)\in W_n^*} \bigg\{ n + \frac {\ell}{\binom n 2 } N(k,l,m) \\- n\left[ H\left(\frac k n\right) + H\left(\frac l n\right) + H\left(\frac m n\right) \right]\bigg\}.
\end{multline}

To show the claim, observe that the entropy function $H(p) \le 1$ and attains this value at $p = 1/2$. This implies that $n [H(\epsilon) + H(\delta) + H(\theta)]\leq 3 n $. 
More generally, let $k/n = \epsilon, l/n = \delta,m/n = \theta$ for arbitrary constants $\epsilon, \delta, \theta \in (0,1)$.
For this choice, $ 2 \ell N(k,l,m)/(n(n-1)) \geq 2 \ell (\epsilon \delta + \epsilon \theta + \delta \theta)$. 
For sufficiently large $\ell$ we therefore have $ 2 \ell N( \epsilon,  \delta, \theta) \geq n[ H(\epsilon) + H(\delta) + H(\theta)]$  for any choice of $\epsilon, \delta, \theta \in \Theta(1)$. 
One might think $ \ell = Cn$ for a large enough constant $C$ is sufficient, but observe that for any constant $C$, we can find constants $\epsilon, \delta$ such that $2 C \epsilon \delta < H(\epsilon) + H(\delta) $. 
Thus, we need to choose $\ell/n = f(n)$ to follow an arbitrary monotonously increasing function $f(n)$ that diverges with $n$. 
Otherwise, the corresponding terms in \cref{eq:chi scaling sparse} will diverge exponentially, showing the claim \eqref{eq:diverging}.

The same argument holds if only two of $k,l,m$ are chosen proportional to $n$, and the third one sublinear. 

Let $k/n \leq  1/f(n)$ for an arbitrary montonously increasing function $f$ that diverges with $n$. 
We have 
\begin{align}
    \label{eq:bound H(k/n)}
    H(k/n) &\leq \frac 1{f(n)} \log f(n) - \Big[1- \frac1 {f(n)} \Big]\log \Big[1-\frac 1{f(n)}\Big] \nonumber \\
    &= -\log \Big[1-\frac 1{f(n)}\Big] + \frac{1}{f(n)}\log[f(n)-1] \\
    & \leq \frac 1 {f(n)}\log[f(n)].\nonumber
\end{align}

Consider $k,l,m \in o(n)$.
In this case, by the bound \eqref{eq:bound H(k/n)} the entropy-terms decrease as $n H(k/n) \leq n o(1)$, which is asymptotically dominated by the linearly growing term. 

Now, consider $k = \Theta(n), l = m = o(n)$---the only remaining case.
Let us first fix a choice $k = n/2, l=m=1$.
In this case $H(k/n) = 1$, $H(l/n) = H(m/n) \leq \log(n)/n$ and $N(k,l,m) = n + 1$ so that by the bound \eqref{eq:bound H(k/n)} the right hand side of \cref{eq:sparse anticoncentration bound} is at most 
\begin{align}
   p(n)2^{- 2 \ell/n }. 
\end{align} 
This bound decays for $\ell > (11/4) \cdot n \log n $. 

We can vary the parameters around the bound by increasing or decreasing $k$, and increasing $l,m$. 
Let us fix $k=n/2$ and increase $l,m > 1$. 
In this case, $n (H(k/n) + H(l/n) + H(m/n)) \leq n + (l + m)\log (n)  $ and $N(k,l,m) = (l + m ) n /2 + l m $, in which case for $d > 11/4$ the bound \eqref{eq:sparse anticoncentration bound} decays. 
Next, let us decrease $k/n \leq  1/2 - \epsilon$ for constant $\epsilon > 0$. Since the entropy function is symmetric about $1/2$ and $N(k,l,m)$ increases in $k$, this is the worst case. 
Now, $H(k/n) \leq 1 - \epsilon^2$, and hence $n - n H(k/n) \geq n\epsilon^2$ which dominates all other terms. 
\end{proof}

\subsection{Sparse degree-$D$ IQP circuits}
\label{sub:sparse_degree_d_iqp_circuits}

\cref{lem:ac} provides upper and lower bounds for sparse random degree-$2$ IQP circuits.

\begin{corollary}[Anticoncentration of sparse degree-$D$ circuits]
\label{cor:degree-d anticoncentration}
Random IQP circuits with $\ell$ uniformly random degree-$2$ gates and any additional diagonal gates satisfy 
\begin{align}
 \overline \chi \in 2 - 2^{-n+1} + p(n) 2^{-2 \ell/n}, 
\end{align}
where $p(n) = n^{11/2}/(2\pi)^{3/2}$. 
\end{corollary}
\begin{proof}
    The corollary follows immediately from the proof of \cref{lem:ac}, by grouping all $\ell$ random degree-$2$ gates to the beginning of the circuit.
    The remaining gates can only decrease the XEB since $\tr[\mb X_x \proj{xyzw}] > 0$ for any $x,y,z,w \in \bin$ satisfying $x + y + z+w = 0$ and gates cannot generate new states. 
\end{proof}

\subsection{Coupled block-random IQP circuits}
\begin{figure}
    \centering
    \includegraphics[width=\linewidth]{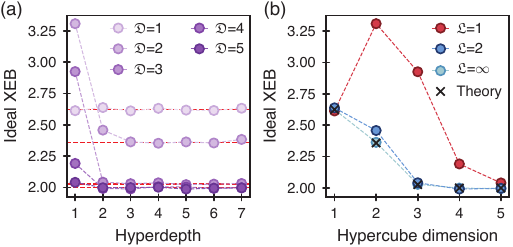}
    \caption{\textbf{Ideal XEB values of random hIQP circuits without in-block CNOTs.} 
    For each data point, we average ${\sim}\,10^6$ random degree-$2$ hIQP circuits.
    (a) Convergence of the ideal XEB score of hIQP circuits without in-block CNOT gates to the asymptotic value (dashed lines) $\overline \chi_{\mathrm{hIQP}}(\hdim)$ given in \cref{eq:xeb score parallel cnot hypercube} as a function of hyperdepth for various $\hdim$ (circles).
    (b) Convergence of the ideal XEB score of hIQP circuits to the uniform value as a function of hypercube dimension $\hdim$ for various hyperdepths $\hdepth$ (circles). }
    \label{fig:xeb_no_inblock}
\end{figure}
Let us now consider the block-random model of IQP in which we compose uniformly random degree-$2$ IQP circuits in blocks of $m$ qubits with fixed entangling gates between the blocks. 
In the hIQP circuits, considered in the main text, there are $m=3$ qubits per each of $2^D$ blocks, where $D$ is the hypercube dimension. 
The blocks are then entangled by CNOT gates with random direction according to the hypercube connectivity. 
Crucially, while the direction of the CNOTs on every edge of the hypercube is random, $b=3$ CNOT gates are applied \emph{in parallel} on every edge, that is, all three gates have the same direction, and pair up with the same qubit in every circuit layer. 
More precisely, the block-random-IQP model is defined by a block size $b$, a connectivity graph $G = (V,E)$ with $b \coloneqq |V|$ vertices, and an edge-coloring $C \subset \mc P(E)$, where $\mc P(E)$ is the power set of $E$.
A \emph{block-random IQP circuit with hyperdepth $\hdepth$ on $G$} then proceeds as follows 
\begin{enumerate}
    \item 
        Prepare $\ket {+^{mb}}$. 
    \item In every hyperlayer $l = 1, \ldots, \hdepth$, for every color $c \in C$: \vspace{-1ex}
    \begin{enumerate}
        \item[i.]
            apply a uniformly random degree-$2$ IQP circuit in every block. 
        \item[ii] 
            apply (potentially randomized) transversal two-qubit gates between the blocks connected by the edges in $c$, 
     \end{enumerate} 
    \item
        apply a uniformly random degree-$2$ IQP circuit in every block. 
    \item 
        Measure in the $X$-basis.
\end{enumerate}

In the hIQP circuits we also apply in-block uniformly random CNOT gates, but for the sake of completeness we will compute the properties of the block-random model first before turning to hIQP. 
We will consider different types of transversal two-qubit gates between blocks. 
The first type is what we call \emph{randomized parallel CNOT} gates. 
A randomized parallel CNOT gate applied on an edge of the block-graph is given by $(\cnot^{\otimes m})^d (\notc^{\otimes m})^{1 - d}$ depending on a random choice of direction $d \in \{0,1\}$, where $\notc = \mb S \, \cnot\,  \mb S$. 
A \emph{parallel randomized CNOT} gate, on the other hand, is specified by $d \in \{0,1\}^n$ and given by $\prod_{i=1}^n \cnot^{d_i} \notc^{1 - d_i}$. 
A \emph{parallel randomized SWAP} gate, is given by $\prod_{i=1}^m \mb S^{d_i}$.

The following \lcnamecref{thm:block iqp xeb} gives the ideal XEB score for such circuits in the limit of a large number of circuit layers. 

\begin{theorem}
\label{thm:block iqp xeb}
    Let $G$ be a connected graph on $b$ vertices. 
    The average XEB of block-random degree-$2$ IQP circuits acting on $b$ $m$-qubit blocks connected by $G$ and hyperdepth $\hdepth$ satisfies 
    \begin{multline}
    \label{eq:xeb score parallel cnot hypercube}
        \overline \chi \xrightarrow{\hdepth \rightarrow \infty} \overline \chi_{\mathrm{hIQP}}(\hdim) \\\coloneqq 2 - 2^{-mb +1} + \frac {(2^m -1 )(4^{b} - 3 \cdot 2^b  + 2 )} {2^{mb}} 
    \end{multline}
    if the two-qubit gates are randomized parallel CNOT gates, and 
    \begin{align}
    \label{eq:xeb score random cnot hypercube}
        \overline \chi  \xrightarrow{\hdepth \rightarrow \infty} 2 - 2^{-mb +1} + \frac{3 (4^{b} - 3 \cdot 2^b  + 2 ) }{2^{mb}} 
    \end{align}
    if the two-qubit gates are parallel randomized CNOT or parallel randomized SWAP gates. 
    
    In either case, if a uniform CNOT circuit is applied to every block before the inter-block gates, we have 
    \begin{align}
    \label{eq:xeb score random cnot}
        \overline \chi  \xrightarrow{\hdepth \rightarrow \infty} 2 - 2^{-mb +1}
    \end{align}
\end{theorem}
We show the convergence to the asymptotic value for the case of parallel randomized CNOT gates in \cref{fig:xeb_no_inblock}. 

\begin{proof}
To show the theorem we proceed in two steps. First, we show that the limits are lower bounds. Second, we show that they are attained in the limit of high depth.     
Our key tool will again be the statistical-mechanics mapping described above. 

Let us first consider the case of randomized parallel CNOT gates, and consider two strings $P,Q \in \mc S^{m}$ (we omit the tensor product in the following, understanding $sp$ to be the string composed of $s, p \in \mc S$).
Let $U = \cnot^n$. 
Then it follows from \cref{lem:cnot update} that under the action of $U$ the string $PQ \mapsto PS$, where $S_i = Q_i$ iff $P_i = \xx$ and $S_i \neq Q_i$ if $P_i \in \mc S_\notx$.  
Recall also that a string $P$ is annihilated by a random degree-$2$ circuit if there are at least two indices $i,j$ such that $P_i \neq P_j \in \mc S_\notx$, and remains if all $P_i \in \{ \xx, Q\}$ for a $Q \in \mc S_\notx$. 

Let the \emph{$\xx$-type of a state} $S \in \mc S^m$ be the locations $T_\xx(S) \coloneqq \{i \in [m]: S_i = \xx\}$ at which it has an $\xx$ state.
Let the \emph{$s$-type of a state} $S \in \{\xx,\ix\}^m \cup \{\xx,\sx\}^m \cup  \{\xx,\px\}^m$ be the non-\xx-string $T_s(S) \in \{\ix, \sx, \px\}$. 
We observe that the state $S = S^1 S^2$ of two neighbouring blocks coupled by a randomized parallel CNOT gate will survive iff $T_s(S^1) = T_s(S^2)$ or $T_\xx(S^1) = T_\xx(S^2)$. 
Hence, the only configurations $S = S^1 \cdots S^b$ of block-states $S^i\in \mc S^n$ which are guaranteed to survive arbitrarily many circuit rounds are given by 
\begin{enumerate}
    \item states which are immortal under global random IQP circuits, that is, states $S \in \mc S_{\text{imm}} \coloneqq \{\xx, \ix\}^{mb} \cup\{\xx, \sx\}^{mb} \cup\{\xx, \px\}^{mb}$ which have the same $s$-type in every block, and

    \item states $S \in \mc S_{t_\xx}\coloneqq \{ Q: T_\xx(Q^i) \in \{t_\xx , [m]\}\,  \forall i \in [b]\}$ with the same or trivial $\xx$-type on all blocks for some $t_\xx \subset[m]$.
\end{enumerate}
Let us now count those strings. We will do so by removing the \emph{trivial string} ($\xx^{m b}$) in all calculations and adding it back at the end. 
We have  
\begin{align}
|\mc S_{\text{imm}} \setminus \{\xx^{mb}\}| & = 3 \cdot (2^{m b} -1)  \\
|\mc S_{t_\xx} \setminus \{\xx^{mb}\}| & = 4^b - 1  \\
\{t_\xx \subset [m]\} \setminus \{ [m]\} & = 2^m -1 \\
\bigg(\mc S_{\text{imm}} \cap \bigcup_{t_\xx \subset [m]} (\mc S_{t_\xx} \setminus \{\xx^{mb}\}) \bigg)  &= 3 (2^b -1 )(2^m -1). 
\end{align}
This implies that the number of surviving strings is given by at least 
\begin{multline} 
3 \cdot (2^{m b} -1)  +  (4^b - 1) (2^m -1) - 3 (2^b -1 )(2^m -1) +1 \\
= (4^{b} - 3 \cdot 2^b  + 2 )(2^m -1 ) + 3 \cdot 2^{m b } - 2. 
\end{multline}

Now we show that all strings which are not contained in $\mc S_{\text{surv}} = \mc S_{\text{imm}} \cup \bigcup_{t_\xx \subset [m]} \mc S_{t_\xx}$ will die under the circuit evolution.
To see this, we take a string $Q \in \mc S_{\text{die}} = \mc S^{md} \setminus \mc S_{\text{surv}}$. 
By definition, there will be two substrings $Q^i, Q^j \in \mc S^n$ which have a different $\xx$-type, and contain two elements $Q^i_k\neq  Q^j_l \in \mc S_\notx$ such that $k \in T_\xx(Q^i), k \notin T_\xx(Q^j)$, and vice versa for $l$. 
We need to show that this state of affairs cannot survive. 
Wlog.\ let $Q^i_k = \ix, Q^j_l = \sx$.
First, we show that the circuit evolution of such a string will never leave $\mc S_{\text{die}}$.
Consider a gate between $Q^i$ and a neighbouring substring $Q^{o}$ with a different $\xx$-type. 
If $Q^i_k$ is the target of a $\cnot$-gate, which is controlled by $\xx$, it is unchanged. 
If it is the control of a $\cnot$-gate, it flips its target to $\ix$, changing the $\xx$-type of the target substring to a type with increased cardinality, and likewise for $Q^j_l$. 
At infinite depth, there will be a sequence of CNOT-gates in a circuit layer such that at then end of the sequence there will be a block with substring containing both~$\ix$ and~$\sx$. 
Then, the string will die. 

To show \cref{eq:xeb score random cnot hypercube}, we observe that strings with the same \emph{singleton $\xx$-types}, i.e., $\xx$-type with cardinality $|t_\xx| =  m-1$, are invariant under the circuit evolution, while all other strings will die by an analogous argument.

To show \cref{eq:xeb score random cnot}, we observe that a uniform CNOT circuit is a uniform permutation on the set of \xx-types. 
Hence, for every string in which there are two blocks with distinct $s$-type, there will be a sequence of in-block CNOT gates and parallel CNOT gates such that at the end of the sequence there will be to neighbouring blocks $i,j$ whose substrings $S^i,S^j \in \mc S^{b}$ have both different $\xx$-type and different $s$-type. Then the string will die. 
The only surviving strings are therefore immortal strings, and we have 
$|\mc S_{\text{imm}}| = 3 \cdot 2^{mb} -2$. 
\end{proof}

\section{Noisy IQP circuit dynamics}\label{app:xebproof}

In this section, we study the XEB, fidelity and purity in noisy IQP circuits whose noiseless version is composed of unitary CPSP channels.  
The goal of this study is to identify conditions under which the XEB can be used as an estimator of fidelity.  

We first consider the case of general noise models at low-noise rates and present results under increasing levels of assumptions about the noise and the random circuit ensemble.  
We then specialize to locally random degree-2 IQP circuits with SWAP gates and mid-circuit iid Pauli noise in order to understand the behaviour of the transversal circuits in color codes.  
Here, we can use the statistical mechanics mapping derived in \cref{app:statmech}, in order to  compute upper and lower bounds on the XEB.  
At very low noise rates, we find that XEB is a good fidelity proxy in noisy random IQP circuits, while at high noise rates XEB ceases to approximate fidelity, similar to the behavior observed recently in Haar-random circuits \cite{gao_limitations_2024,Morvan.2023,Ware.2023}.

\subsection{General noise at low error rates}
\label{app:low error rates}

 To begin, we consider a general IQP circuit $D^{(\ell)} = D_d D_{\ell-1} \cdots D_1$ composed of $\ell$ layers of unitary CPSP channels with corresponding conjugate action $\mc D_i =  D_i \cdot  D_i^\dagger$ and $\mc D^{(\ell)} =  D^{(\ell)} \cdot  (D^{(\ell)})^\dagger$.  We model the noisy circuit in terms of the noiseless circuit and a sequence of quantum instruments  $\{ \mathcal{E}_{i\mu_i} \}_{\mu_i}$ that implements a larger quantum instrument with elements
\begin{equation}
\mathcal{D}_\mu^{(\ell)} = \mathcal{E}_{\ell\mu_\ell } \circ \mathcal{D}_\ell \circ \cdots \circ \mathcal{E}_{1\mu_1} \circ \mathcal{D}_1 \circ \mathcal{E}_{0\mu_0},
\end{equation}
where $\mu=(\mu_0,\ldots,\mu_\ell)$ indexes the measurement outcomes in the quantum instruments. Writing as above $\ketz x $ as the $z$-basis state with label $x\in \bin^n$ and $\ketx x = H^{\otimes n} \ketz x$ as the $x$-basis state with label $x$ the fidelity $F_\mu$ and XEB $\chi_\mu$ of the noisy circuit  conditioned on syndrome outcomes $\mu$ are defined with respect to the noiseless circuit as
\begin{align}
	F_\mu & =\frac{1}{p_\mu} \brax 0 (\mc D^{(\ell)})^\dagger \circ \mathcal{D}_\mu^{(\ell)} (\ketx0 \brax 0) \ketx 0, \\
	\chi_\mu & = \frac{2^n}{p_\mu} \sum_{x} |\brax x D^{(\ell)} \ketx 0 |^2 
    \brax x  \mathcal{D}_\mu^{(\ell)} (\ketx 0\brax0)  \ketx x -1,
\end{align}
where $p_\mu = \mathrm{Tr}[\mc D_\mu^{(\ell)}(\ketx0\brax 0)]$ is the probability of measurement outcome $\mu = (\mu_0,\ldots,\mu_\ell)$.  

We begin by considering a general model in which noise events are constrained to occur only at the beginning of the circuit, and a fully random, noiseless IQP circuit follows the noisy evolution. 
We think of this model as a proxy for very low noise rates.

\begin{theorem} \label{thm:iqp_fid}
    Let $(\mathcal{E}_{i \mu_i})_{i = 1, \ldots, \ell}$ be a sequence of quantum instruments such that $\mathcal{E}_{\ell \mu_\ell}$ is diagonal, $\mc D^{(\ell)}= \mc D_\ell \cdots \mc D_1$ a $\ell$-layer sequence of unitary CPSP channels, $\mathcal{D}_{\mu}^{(\ell)}$  the associated quantum instrument for the noisy evolution, and  $D_\ell$ a random degree-2 IQP circuit.
    Then
    \begin{equation}
        \mathbb{E}_{\ell}[\chi_\mu] = 2\mathbb{E}_{\ell} [F_\mu] - 1/2^{n-1},
    \end{equation}
    where $\mathbb{E}_\ell$ denotes the average over the random gates in~$D_\ell$.
\end{theorem}
\begin{proof}
    We let $\mu$ be a fixed set of measurement outcomes.  Since  $\mathcal{E}_{\ell \mu_\ell}$ is diagonal (so it commutes with $D_\ell$) and $D_\ell$ is real, we can write
    \begin{equation}
        \mathcal{D}_\mu^{(\ell)}(|0\rangle_x \langle 0|_x) = \sum_{x y} c_{x\mu}  c^*_{y \mu} D_\ell |x \rangle_x \langle y |_x D_\ell,
    \end{equation}
for a set of complex coefficients $c_{x \mu}$. We utilize the properties of  phase states described in Appendix \ref{app:prelim} to write
\begin{equation}
\label{seq:ovp_exp}
    \brax x  D_t \ketx y =  \frac{1}{2^n}\sum_z (-1)^{f_t(z) + (x+y)\cdot z},
\end{equation}
where $(\cdot)$ represents the usual bitstring scalar product
and $f_t$ is a degree-3 polynomial defined by the circuit in layer $t$.

We can then write a formula for $F_\mu$ and $\chi_\mu$ as
\begin{align*}
    F_\mu =\sum_{xyz_i} &\frac{c_{x\mu}c_{y\mu}^*}{p_\mu 4^n} (-1)^{F_{\ell}(z_1)+F_{\ell}(z_2)+z_1\cdot x+ z_2 \cdot y},\\
    \chi_\mu = \sum_{x' xyz_i} &\frac{c_{x\mu}c_{y\mu}^*}{p_\mu 8^n} (-1)^{F_{\ell}(z_1)+F_{\ell}(z_2)+z_3 \cdot x + z_4 \cdot y}  \\
    &\times (-1)^{\sum_i f_\ell(z_i) +x' \cdot z_i }-1 \\
    =\sum_{xyz_i} &\frac{c_{x\mu}c_{y\mu}^*}{p_\mu 4^n} (-1)^{F_{\ell}(z_1)+F_{\ell}(z_2)+x\cdot z_3 + y \cdot z_4}  \\
    &\times (-1)^{\sum_i f_\ell(z_i)} \delta_{z_4,z_1+z_2+z_3}-1,
\end{align*}
where $p_\mu = \sum_{x} |c_{x\mu}|^2$, $F_t(z) = \sum_{i=1}^{t-1}f_i(z)$ is the cumulative polynomial of the noiseless circuit from time $0$ to $t-1$, and in the last line we used the sum over $x'$ to impose the constraint $z_1+z_2+z_3+z_4=0 \mod 2$.
When averaging over the last layer of gates, we notice that fidelity is independent of this averaging, while XEB requires the evaluation of
\begin{equation}
    \mathbb{E}_\ell \big[ (-1)^{\sum_i f_\ell(z_i)} \delta_{z_4,z_1+z_2+z_3} \big].
\end{equation}
To evaluate the circuit-average $\mathbb{E}_{\ell}$ over the degree-$2$ part of $f_\ell$, we follow the same arguments as in the proof of \cref{lem:global iqp proj} (following \textcite{Bremner.2016}), yielding the constraints $(\id), (\mb S), (\mb P)$ in \cref{seq:z_constr} with overlap $(\Xz)$ in \cref{eq:delta condition} on the bit string quadruples $(z_1, z_2, z_3,z_4)$.

Using these simplifications of the allowed quadruples, we arrive at a simplified formula for the conditional XEB
\begin{equation*}
\begin{split}
    \mathbb{E}_\ell[\chi_\mu] &= 2 \sum_{xyz_i} \frac{ c_{x\mu}c_{y \mu}^*  }{4^n p_\mu }(-1)^{F_\ell(z_1)+F_\ell(z_2)+x\cdot z_1 +y\cdot z_2} \\
    &+ \sum_{xyz}\frac{ c_{x\mu}c_{y \mu}^*  }{2^n p_\mu } (-1)^{x\cdot z+ y \cdot z} -1\\
    &- \frac{1}{ 2^{n-1}}\sum_{xyz}\frac{c_{x\mu}c_{y \mu}^*  }{2^n p_\mu } (-1)^{x\cdot z+ y \cdot z}   \\
    &= 2 F_\mu - \frac{1}{2^{n-1}}.
    \end{split}
\end{equation*}
    This identity completes the proof because $\mathbb{E}_\ell[F_\mu]=F_\mu.$
\end{proof}

\cref{thm:iqp_fid} establishes the equivalence between fidelity and XEB under a broad class of noise models.  Crucially, it also allows for noise models that are not quantum channels, but conditioned on measurement outcomes.  
Such noise arises at the level of logical qubits when studying the fidelity conditioned on a particular syndrome outcome, as is done in error detection.  However, the theorem makes the crucial assumption that the measurement error channel is diagonal and that the final layer of gates is a uniformly random degree-$2$ IQP circuit.  
In \cref{lem:ac}, we saw that using long-range connectivity, it is possible to achieve such an ensemble in depth $\log n$. 

To see why this assumption is necessary, we now show that choosing the noise adversarially breaks the correspondence between fidelity and XEB.

\begin{theorem}
    Consider a noisy IQP circuit 
    $\mc E_\ell \circ \mc D_\ell \circ \cdots \circ \mc E_1 \circ \mc D_1 \circ \mc E_0$
    with $\mathcal{E}_\ell$  the bit-flip channel on the first qubit $\mathcal{E}_\ell(\rho)=(1-p) \rho + p X_{1} \rho X_1$ with $0<p\le 1/2$, $D_\ell$ a random degree-2 IQP circuit, and  $\mathcal{E}_i = \id$, $D_i = \id$for $i<\ell$. 
    For this noisy IQP circuit, we have the identities $\mathbb{E}_\ell[\chi] = 1$ and $\mathbb{E}_\ell[F] = (1-p)+p/2^n$.
\end{theorem}
\begin{proof}
    The $x$-errors in the final layer have no effect on the XEB, which is why $\mathbb{E}_\ell[\chi_\mu] = 1$ in this noise model (i.e., the average XEB is equal to its noiseless value).  For the fidelity, we can make use of the formula
    \begin{equation}
    F =(1-p)+\frac{p}{ 4^n}\sum_{z_1,z_2}  (-1)^{f_{\ell}(z_1)+f_{\ell}(z_1+b)+f_\ell(z_2)+f_\ell(z_2+b)},
    \end{equation}
    where $b=(1,0,\ldots,0)$ is a shift arising from the single-site $X$ operator.
    Averaging the argument in the sum over the random degree-2 IQP circuit ensemble, for it to be non-zero we arrive at the constraint that $z_1 = z_2$ or $b = 0$.  
    Since $b \ne 0$, we only get a contribution from the $2^n$ terms in the sum for which $z_1= z_2$, which completes the proof.
\end{proof}
This breakdown of the correspondence between fidelity and XEB for the adversarial noise is reminiscent of what happens in noisy Haar random circuits \cite{gao_limitations_2024}.  

We now show that it is sufficient to consider the $z$-twirl of the noise channels when comparing the average fidelity and XEB.
\begin{theorem}
    Consider a noisy IQP circuit with diagonal gate layers $D_i= Z^{a_i} \tilde{D}_i $ containing independent Pauli-$Z$ gates $Z^{a_i} = \prod_j Z_j^{a_{i,j}}$ with uniformly chosen $a_{i,j} \leftarrow \bin$, then the circuit averaged fidelity and circuit averaged XEB are equal to their values with $\mathcal{E}_{i \mu_i}$ replaced by its $z$-twirl $\mathcal{E}_{z i \mu_i}$ for $0<i<\ell$.
\end{theorem}
\begin{proof}
The proof is analogous to the arguments that show one can twirl the noise into a depolarizing channel in Haar random circuits or Clifford randomized benchmarking.  For each CP map $\mathcal{E}_{i\mu_i}$ we take a Kraus decomposition with Kraus operators $K_{i\mu_i \nu_i}$. 
Define the new random variables
\begin{equation}
    b_t = a_1 + \cdots + a_t.
\end{equation}
Now, we can write the noisy state as
\begin{align}
\mathcal{D}_\mu^{(\ell)}&(\ketx 0 \brax 0) = \sum_{\nu} K_{\mu\nu} \ketx 0 \brax 0 K_{\mu\nu}^\dagger, \\
   K_{\mu\nu} &= K_{\ell\mu_\ell \nu_\ell} Z^{b_\ell} \tilde{D}_\ell \\ \nonumber
   &\times \left[\prod_{i=1}^{\ell-1} Z^{b_{\ell-i}}K_{\ell-i \mu_{\ell-i} \nu_{\ell-i}} Z^{b_{\ell-i}} \tilde{D}_{\ell-i} \right] K_{0\mu_0 \nu_0},
    \end{align}
    where $\nu=(\nu_0,\ldots,\nu_\ell)$ indexes the Kraus operators.  The random variables $b_i$ are independent uniformly random bit strings. 
    When we average the noisy state over $b_i$ for $i<\ell$ we can therefore apply the twirled channel.  
    In the fidelity and XEB, the noiseless part of the circuit depends only on the random variable $b_\ell$.  
    As a result, we can average over $b_i$ for $i<\ell$ independently of what appears in the noiseless circuit.  This simplification allows us to replace $\mathcal{E}_{i\mu_i}$ by its $z$-twirl for all $0<i<\ell$ when evaluating average XEB and average fidelity. 
\end{proof}
The proof can be readily extended to the case where the $D_i$ are unitary CPSP channels and not just diagonal gates by including random $Z$ gates at both the end and beginning of each layer.

Recall that the $z$-twirl of a channel has Kraus operators of the form
\begin{equation*}
    X^b K_{\nu b},~K_{\nu b} = \sum_a c_{ab}^\nu Z^a,
\end{equation*}
where $a$ and $b$ run over bitstrings of length $n$ and $c_{ab}^\nu$ are some complex coefficients.  Interestingly, for any mixture of phase states $\rho$ the probability $p_{\mu b}$ that a particular  bit-flip error $b$ occurr  is independent of the initial state
\begin{equation}
    p_{\mu b} = \sum_{\nu} \mathrm{Tr}[ K_{\mu \nu b}^\dag K_{\mu \nu b} \rho] = \sum_{a,\nu} |c_{ a b}^{\mu\nu} |^2.
\end{equation}
This identity holds because all phase states have zero expectation value of all $Z$-type operators.  
When the state-preparation errors are given by a PSP channel, we can thus think of bit flip errors of the $z$-twirled channel as occurring in the circuit at a rate that is independent of the current quantum state in the circuit.  
As a result, it is natural to consider noise models satisfying the following independence property. 
\begin{definition}
    A $z$-twirled CP map $\mathcal{E}_\mu$ has independent $x$-noise if the probability of a bit-flip error $p_{\mu b}$ on phase states satisfies
    \begin{equation}
        p_{\mu b} = (1-p_\mu)^{n-|b|} p_\mu^{|b|},
    \end{equation}
    for some parameter $0\le p_\mu \le 1$, where $|b|$ is the Hamming weight of $b$ (i.e., the number of nonzero entries in $b$).
\end{definition}
The parameter $p_\mu$ thus quantifies the effective local bit-flip rate of $\mc E_\mu$ applied to phase states. 

We now make some observations about noisy IQP circuits with $z$-twirled noise with independent $x$-noise bit-flip rates:
\begin{itemize}
    \item 
        When the ensemble of $D$ converges to random degree-$2$ IQP circuits and the local bit-flip rate   in the $\mu$-th layer $p_\mu=0$ for $\mu>0$, then $\mathbb{E}[\chi] = 2 \mathbb{E}[F]-1/2^{n-1}$. This follows from the fact that without bit-flip noise the noise channels commute with $D$. 
        As a result, we can bring the full channel into a form satisfying the conditions of \cref{thm:iqp_fid}.  

\item 
    Take bit-flip noise rates $p_i = c/n\log n$ for some $c>0$ and a noiseless circuit ensemble with ideal XEB value bounded by $2 - 1/2^{n-1} + o(1)$ in depth $\Omega(\log n)$, e.g., the one of \cref{lem:ac}.
    For sufficiently small $c$, with probability $1-O(c)$ over noise realizations, the errors  occur in the first layers of the circuit, with $\Omega(\log n)$ noiseless layers before the measurement. 
    As a result, for sufficiently small local bit-flip noise rates $c/n \log n$, we can follow the proof of \cref{thm:iqp_fid} and  find that $\mathbb{E}[\chi] = 2 \mathbb{E}[F] - 1/2^{n-1} + O(c) + o(1)$ for depths $\ell \in \Omega(\log n)$.
\end{itemize}
This second result is analogous to the white-noise approximation in noisy Haar random circuits established rigorously by \textcite{dalzell_random_2024}.  It implies that at sufficiently low-noise rates, there is an exact correspondence between fidelity and XEB.  Moreover, since fidelity will typically decay exponentially in the circuit size, the XEB will exhibit a similar decay at these noise rates.  A central question is how this behavior changes as the noise rate increases to larger values $c/n$ for increasing $c > 0$. 
In noisy Haar random circuits, there is a phase transition in XEB where it goes from a decay linear in the circuit volume to one that depends on the depth $\ell$ and not $n$~\cite{gao_limitations_2024,Morvan.2023,Ware.2023}.

\subsection{Upper and lower bounds on XEB}
\label{app:xeb bounds}

To more directly analyze the behavior of XEB in noisy IQP circuits at large noise rates, we specialize in this subsection to the random degree-2 circuit model with SWAP gates.  
Specifically, we consider circuits in which in every layer of the circuit, we first apply a 2-local random degree-2 IQP circuit to a random pair of qubits, and then apply the noise channel to all qubits. 
Let us call this model the \emph{local random degree-2 IQP model}.

To organize the calculations it is helpful to represent the initial state as
\begin{multline} \label{eqn:init}
    \frac{1}{4^n}(\ix + \sx + \px + \xx )^{n} \\= \frac{1}{4^n} \sum_k \binom{n}{k}(\ix + \xx)^{n-k}(\sx + \px)^{k},
\end{multline}
where we have removed the spatial dependence of the strings due to the permutation symmetry of the initial state and our circuits, 
and we write $S P $ as shorthand for $S \otimes P$ for states $S, P \in \{\ix, \sx, \px, \xx\}$.  
We consider the noise to be an \emph{$X$-$Y$-symmetric} Pauli channel, i.e., a Pauli channel with symmetric $X/Y$ noise rates $q_\perp/2 = p_x = p_y$.  It is convenient think of  the total Pauli noise channel as the action of a pure $Z$ noise channel with rate $ p_z/(1- q_\perp)$ followed by a symmetric Pauli-$X/Y$ noise channel with  $X/Y$-noise rates $q_\perp/2$. 
Defining $q = q_\perp+2p_z$, the combined noise channel will act on $\sx/\px$ states as 
\begin{align}
   \sx / \px &\to (1- q)\sx / \px,\\
    \ix / \xx & \to (1- q_\perp) \ix / \xx + q_\perp \xx / \ix ,
\end{align}
i.e., it leads to a pure damping of the string depending on the $\sx$ and $\px$ weight, while it flips each $\xx$ and $\ix$ component of the string at rate $q_\perp$.
We can commute all the $z$ noise in the circuit to the initial state, which adds a decay term
\begin{equation}
    (\sx + \px)^{k} \to (1-q)^{
    \ell k} (\sx + \px)^{k} .
\end{equation}
As a result, the only non-trivial effect of the noise in the statistical mechanics model is from the $X/Y$-noise acting on the $\ix/\xx$ operators of the strings.  Solving for the full dynamics of this model is not obviously analytically tractable. 
Instead we prove an upper and lower bound on $\overline{\chi}$ that illustrate the two regimes of XEB and are tight with respect to the scaling with depth.
\begin{theorem}
\label{thm:xeb bounds appendix}
    The average XEB for the local random degree-$2$ IQP model on $n$ qubits with $\ell$ gates and $\ell$ layers of $X$-$Y$-symmetric Pauli noise is bounded as 
    \begin{align}
        \overline{\chi} &\ge 2^{1-n} [ (1-q_\perp)^\ell+(1-q)^\ell]^n \\
        &+ 2^{-2\ell/n}[(1-q)(1-q_\perp)]^\ell,\\
        \overline{\chi} & \le 2^{1-n}[1+(1-q)^\ell]^n + p(n) 2^{-2\ell/n} ,\label{eq:upper bound noisy iqp}
    \end{align}
    where $p(n)= n^{11/2}/(2\pi)^{3/2}$.
\end{theorem}
\begin{proof}
    We start by proving the lower bound.  
    Recall that the initial state is a uniform superposition over all $\ix, \sx, \px, \xx$ strings and XEB is computed by summing up the contributions of each of these strings independent of their composition.  
    As a result, we can lower bound the XEB by neglecting the transitions from $X$ noise that change the string composition.  
    The lower bound then follows by computing the contribution to XEB from the strings
    \begin{multline}
        (\ix +\xx)^{n}+(\sx + \xx)^{n} + (\px + \xx)^{n} -2 \xx^{n} \\+(\ix + \xx)^{n-2}\ix \sx.
    \end{multline}
    The total XEB amplitude of the strings $(\sx + \xx)^{n}$ and $(\px + \xx)^{n}$ decays as 
     $2^{-n} [(1-q_\perp)^\ell + (1-q)^\ell]^n$.  
    On the other hand, the total weight of the string with one $\sx$ operator decays at least as $2^{-2\ell/n}[(1-q)(1-q_\perp)]^\ell$. 
    To see this, observe that the state $(\ix + \xx)$ is an eigenstate of the $X$ decay, and that the $PQ$-weight of the string $P = \sx \ix^{n-1}$ is given by $\pqw(P) = n-1$.

    To compute the XEB contribution $2^{-n} \tr[\Xx S]$ of $S = (\ix + \sx)^{n-2}\ix \sx$, we split the contributions from the noise-induced decay and the circuit-induced decay.
    In total, $S$ consists of $2^{n-2}$ individual strings. 
    Following the argument above, the noise-induced decay given by $[(1-q_\perp)(1- q)]^\ell$. 
    The circuit-induced decay is given by the last term in \cref{eq:chi scaling sparse}, $(1 - \pqw(S)/\binom n2)^\ell \geq 2^{-2\ell/n- 2}$ which follows from $(1- 2/n)^\ell \xrightarrow{n \rightarrow \infty} e^{-2\ell/n}$.

     To compute the upper bound we recall the decomposition of the initial state in Eq.~\eqref{eqn:init}.  The $X$ noise leaves the $\ix + \xx$ components unaffected, while a gate can transform $(\ix + \xx)  \sx \to \xx \sx$, i.e., it removes the $X$ eigenstate and replaces it by the $\xx$ state.   
     We now observe that $X$ noise acting on this site can only increase the $PQ$-weight of this string, which increases the decay rate of the string.  
    As a result, let us change the model such that gates acting on $\ix (\sx/ \px)$ pairs do not lead to a decay, while maintaining the rule that gates acting on $(\ix + \xx) (\sx/\px)$ pairs still send it to $\xx (\sx/\px)$.  
    This change in the model will only increase the XEB so that we can upper bound the XEB using it.  
    The model is then particularly simple to solve because the only effect of the noise is to lead to the decay of $\sx/\px$ components at rate $(1-q)$.  
    To derive the upper bound we then apply this decay process to the strings 
     \begin{equation}
         (\ix +\xx)^n+(\sx + \xx)^n + (\px + \xx)^n.
     \end{equation}
     We can upper-bound the decay of the contribution from the other strings  by following the calculation in \cref{lem:ac}.  
\end{proof}

The lower bound on XEB shows directly that the XEB has two regimes depending on the strength of the noise.  When $(p_x+p_y+p_z) < 2  \log 2/n^2$, then the XEB is dominated by the exponential decay that is linear in the circuit volume $n \ell$.  Whereas for larger noise rates, the XEB decay is dominated by the decay that is linear in the effective depth $\ell/n$ (recall that one layer of gates requires $\ell = n/2$ in this model). This sharp change in behavior of the XEB is exactly analogous to what has been found in noisy Haar random circuits \cite{gao_limitations_2024,Morvan.2023,Ware.2023}.  In both cases, it indicates that the XEB ceases to be good proxy for the fidelity, which always decays at a rate linear in the circuit volume until it reaches a value near $1/2^n$.  The upper bound in the XEB shows that the two scaling regimes observed in the lower bound are actually tight at low and large enough noise rates.

\subsection{Extension to degree-$D$ circuits}
\label{ssec:extension_to_degree_3_circuits}

To extend \cref{thm:xeb bounds appendix} to degree-$D$ circuits, we consider the following model of sparse random degree-$D$ circuits: place uniformly random degree-$D$ IQP gates on random subsets of $D$ qubits, each of which is given by a uniformly random degree-$D$ IQP circuit on $D$ qubits, i.e., a circuit comprising a $C^{D-1}Z$ gate, and a random degree-$(D - 1)$ IQP circuit. 
The noise is, like in the degree-$2$ case---applied in layers after every gate. 

\begin{corollary}
    \label{cor:noisy xeb degree d}
    The average XEB for the local random degree-$D$ IQP model on $n$ qubits with $\ell$ gates and $\ell$ layers of $X$-$Y$-symmetric Pauli noise is bounded as 
    \begin{align}
        \overline{\chi} &\ge 2^{1-n} [ (1-q_\perp)^\ell+(1-q)^\ell]^n \\
        &+ 2^{-D! \cdot \ell/n}[(1-q)(1-q_\perp)]^\ell,\\
        \overline{\chi} & \le 2^{1-n}[1+(1-q)^\ell]^n + p(n) 2^{-2\ell/n} ,\label{eq:upper bound noisy iqp}
    \end{align}
    where $p(n)= n^{11/2}/(2\pi)^{3/2}$.
\end{corollary}

\begin{proof}
The proof of the upper bound follows directly from the arguments in \cref{thm:xeb bounds appendix} together with \cref{cor:degree-d anticoncentration}, since there are at least $\ell$ of those gates in a random sparse degree-$D$ circuit with $\ell$ gates circuit. 

Since the random degree-$D$ gates have the same global invariant states, the proof of the first term of the lower bound remains the same, too.
The only change comes from the circuit-induced decay of the string $\ix^{n-1}\sx$, which gives rise to the second term of the lower bound. 
To see this, we observe that the probability that a string $S$ surives a random degree-$D$  gate on $D$ qubits is given by one minus the probability that a uniformly random length-$D$ substring of $S$ contains two elements in $\mc S_\neg$, which we can compute as 
\begin{align}
    1 - \frac{\pqw_D(S)}{\binom n D },
\end{align}
defining $\pqw_D(S)$ as the number of length-$D$ substrings of $S$ that contain at least two elements from $\mc S_\neg$. Note that $\pqw(S) = \pqw_2(S)$. 
We can upper bound 
\begin{align}
    \pqw_D(S) \le \pqw_2(S) \cdot (n-2) \cdots (n-D+1),
\end{align}
where we are neglecting double-counts of substrings containing more than two elements from $\mc S_\neg$. 
Since, $\pqw(\ix^{n-1}\sx) = n-1$ its survival probability under a  random sparse degree-$D$ circuit with $\ell$ gates is given by 
\begin{align}
    \left(1 - \frac{\pqw_D(S)}{\binom n D }\right)^\ell \ge 2^{- \ell \cdot \pqw_D(S)/\binom n D} \ge 2^{D! \cdot \ell/n }. 
\end{align}

\end{proof}

\subsection{Adding fixed gates}
\label{ssec:fixed gates}

Let us now show that the addition of one fixed gate layer to the circuit can shift the location of the transition in the scaling of the XEB decay by a factor of $3/2$ in the depth. 
\begin{lemma}
\label{lem:shift appendix}
     The average XEB for the local random degree-2 IQP model on an even number $n$ of qubits with $\ell$ gates and $\ell$ layers of $X$-$Y$-symmetric Pauli noise with a last layer of parallel CZ gates between neighbouring qubits is bounded as 
    \begin{align}
        \overline{\chi} &\ge 2^{1-n} [ (1-q_\perp)^\ell+(1-q)^\ell]^n \\
        &+ 2^{-4\ell/n}[(1-q)(1-q_\perp)]^{2\ell}. \\
        \overline \chi & \le 2^{1-n}[1+(1-q)^\ell]^n + p(n) 2^{-3\ell/n} 
    \end{align}
    Furthermore, the total weight of all strings with a single excitation, i.e., strings of the form $S \in \mc S_\notx \times \mc S_{\text{imm}}$, is given by $0$.
\end{lemma}

\begin{proof}
    Let us first show the final statement. 
    Wlog.\ consider a string $S= \ix \sx \ix^{n-2}$ and a CZ gate acting on the first two qubits. 
    Then $S \xmapsto{CZ} - S$, and hence the weight of the string is $-1$. 
    But there is another string with the same weight given by $S' = \ix \xx \ix^{n-2}$. 
    This string evolves as $S' \xmapsto{CZ} S'$ under the last circuit layer and hence contributes with $+1$ so their contributions cancel.

    The lower bound follows from the argument in the proof of \cref{thm:xeb bounds appendix} but now applied to the strings $\px\px \ix (\ix + \xx)^{n-3}$.  
    These strings are invariant under the final CZ layer since $\px \px \xmapsto{CZ} \px \px$. 
    Their maximum $PQ$-weight is given by $2(n-2)$ and hence their circuit-induced decay is lower-bounded by $(1- 4(n-2)/n (n-1)^\ell \geq 2^{-4\ell/n -2}$.
    Meanwhile, their noise-induced decay rate doubles. 

    For the upper bound, we argue that the second term in the upper bound \eqref{eq:upper bound noisy iqp} in the proof of \cref{thm:sparse iqp anticoncentration main}, $p(n) 2^{- 2 \ell/n}$ decreases by adding a last layer of parallel CZ gates. 
    We consider strings with $n_\ix(S) = 1,n_\sx(S) = 1, n_\px(S) = n/2$ of the form $S = \ix \sx \px^{n/2}\xx^{n/2 - 2}$, letting $n/2$ be even. Then $S \xmapsto{CZ} -S$, and there is another string  $S' = \ix \xx \px^{n/2}\sx \xx^{n/2 - 3}$ obtained from $S$ by a single swap of operators between the non-$\xx$ part of the string and the $\xx$-part of the string, for which $S' \xmapsto{CZ} S'$.
    Hence, the contributions of all strings of this type cancel out, leaving us with strings $S$ with $(n_\ix(S), n_\sx(S),n_\px(S)) = (2,1,n/2)$ giving total $PQ$-weight $\pqw(S) = 3n/2 + 2 $ saturating the upper bound. Using these strings, the second term in \cref{eq:minimum exponent sparse anticoncentration} increases to $\sim 3\ell/n$ and hence the upper bound follows, shifting the transition by a factor of $3/2$. 
\end{proof}
It is clear that the argument can be repeated any number of times, applying CZ gates on new sets of pairs. Eventually, as we apply it $n$ times, we retrieve the volume-like decay of \cref{thm:iqp_fid}.

\section{Hardness of degree-$4$ Bell sampling}
\label{app:bell sampling}

Let us in this section analyze in more detail the output distribution of uniform degree-$D$ IQP circuits measured in the Bell basis. 
We will argue that sampling from these circuits is classically intractable unless the polynomial hierarchy collapses. 
The argument will be somewhat more intricate compared to sampling in the standard-basis. 

Recall from the main text that we can write the output probability of a Bell measurement on a degree-$D$ circuit $C$ as 
\begin{align}
    P_C(x,z) & = \frac 1 {2^n} | \bra {+^n} C X_x Z_z C \ket{+^n} |^2 \\
    & = \frac 1 {2^n} | \bra {+^n} C_x Z_z \ket{+^n} |^2 , 
\end{align}
where $C_x = C X_x C$ is the degree-$D-1$ circuit induced by the $X$ outcomes and the circuit $C$. 
Denote by $\supp(x) \coloneqq \{ i \in [n]: x_i = 1\}$ the support of the string $x$. 
Then the support of the circuit $C_x$, i.e., the qubits on which it acts nontrivially, is given by $\supp(C_x) = \supp(x)^c$, since the one-locations of $x$ remove qubits from the support of $C$. 

Let us observe a few properties of this output distribution. 
\begin{lemma}
\label{lem:bell sampling probs}
We have the following facts about the output probabilities $P_C(x,z)$ of Bell sampling from degree-$D$ IQP circuits.
\begin{enumerate}[label=\roman*.]
    \item For all $(x,z)$ such that there is an $i \in [n]$ with $x_i = z_i = 1$, $P_C(x,z) = 0$. 
    \item The marginal probability $P_C(x) = \sum_z P_C(x,z) = 2^{-n}.$
\end{enumerate}
\end{lemma}
\begin{proof}
To see claim (i), observe that the zero-entries of $x$ determine $\supp(C_x)$. 
So for all $i \in [n]$ for which $x_i = 1$, $C_x$ acts trivially, and hence, $\bra {+^n} C_x Z_i \ket{ +^n} = \bra + Z \ket + \bra {+^{n-1}} (C_x)|_{[n]\setminus \{i\}} \ket{ +^{n-1}} = 0 $. 

To see claim (ii), we simply compute 
$P_C(x) = \sum_z P_C(x,z) = 2^{-n} \sum_z |\bra {z} H^{\otimes n} C_x H^{\otimes n} \ket{0^n}|^2 = 2^{-n}$ by the normalization of the output distribution of the degree-$D-1$ circuit $C_x$.
\end{proof}

The output distribution of degree-$D$ Bell sampling therefore `shatters' into separate sectors for each fixed outcome $x$ of the $X$ measurement. 
All of those sectors are equally likely, and the output distribution within each sector is given by the output distribution of the induced IQP circuit $C_x$ acting nontrivially on $n - |x|$ qubits. 

Let us now consider a uniformly random degree-$D$ circuit $C$. 
Then for every $x$, $C_x$ is a uniformly random degree-$(D-1)$ circuit acting on $\supp(C_x)= \supp(x)^c$. 
To see this in each $x$-sector, fix the indices of the gates in $C$ corresponding to the $x_i=1$ outcomes and observe that the marginal distribution on the remaining indices remains uniform. 
This motivates grouping the $x$-sectors further according to the size of their support, since the statistics of the output distribution in each sector are now solely determined by the size of the uniformly random IQP circuit in that sector. 
Hence, there are $n+1$ distinct sectors $S_k\coloneqq \{ x: |x| = n-k,  z \in \bin^k\} \subset \bin^{2n}$ oon which $P_C$ has nontrivial support, and the statistics of the outcomes in each sector are the same. 

This state of affairs is reminiscent of that in boson sampling. 
For instance, in Gaussian boson sampling the outcomes fall into distinct photon number sectors \cite{kruse_detailed_2019,deshpande_quantum_2022}, and in boson sampling with a linear number of modes there are sectors corresponding to the collision patterns of the outputs \cite{bouland_complexity-theoretic_2023}. 
This analogy will also guide us in the following argument for hardness of sampling from degree-$D$ Bell sampling with $D \geq 4$. 

\begin{theorem}[Hardness of degree-$D$ Bell sampling]
\label{thm:bell sampling hardness}
    Sampling from the Bell sampling output distribution $P_C:\bin^{2n} \rightarrow[0,1]$ from uniformly random degree-$D$ IQP circuits for $D \geq 4$ up to constant total-variation distance is classically intractable unless the polynomial hierarchy collapses and the approximate average-case hardness conjecture \ref{conj:aach degreeD} is false.
\end{theorem}

\begin{conjecture}[\cite{Bremner.2016}]
    \label{conj:aach degreeD}
    Approximating the output probabilities of uniformly random degree-$D$ IQP circuits for $D \geq 3$ up to relative error $1/4$ is \#P-hard.
\end{conjecture}

Since the output distribution of random degree-$D$ Bell sampling shatters into $n+1$ sectors, we need to consider those sectors individually. 
Observe, for instance, that the output probabilities of the $k=1$ sector are very easy to compute. 
However, that sector has negligible probability weight as the following lemma shows. 

\begin{lemma}
\label{lem:bell probability weight}
    The probability weight of $P_C$ outside of the sectors $S_k$ with $k \in [n/2 - O(\sqrt n), n/2 + O(\sqrt n)]$ is negligible. 
\end{lemma}
\begin{proof}
    We use (ii) of \cref{lem:bell sampling probs} to see that 
    \begin{align}
        P_C(S_k) &= \binom n k /2^n,
    \end{align}
    and letting $p(k,n) = \sqrt{2\pi (n/2 +k)(1/2- k/n)}$
    \begin{align}
        P_C(S_{n/2 + k}) &= \binom n {n/2 + k} /2^n\\
        & = \frac 1 {2^n p(k,n)} 2^{n H(1/2 + k/n) + \Theta(1/n)}\\
        & = \frac 1 {p(k,n)} 2^{  - 2 k^2/(n\ln 2)+ \Theta(1/n)},
        \label{eq:probability weight final}
    \end{align}
    since $H(1/2 + \epsilon) =  1 - 2\epsilon^2/\ln 2 + O(\epsilon^3)$. 
    Therefore for $k \in n/2 + \omega(\sqrt n)$, we have $p(k,n) P_C(S_{n/2+k}) \in o(1)$ and we have
    \begin{align}
        \sum_{k \in n/2 \pm \tilde \omega(\sqrt n)} P_C(S_k) \in o(1).
    \end{align}
    Note that we have chosen $k \in n/2 + \tilde \omega (\sqrt n)$ since the sum adds a factor of $n- O(\sqrt n)$ which needs to be accounted for with a $\log \sqrt n $ shift in the exponent of \cref{eq:probability weight final}. 
\end{proof}
\cref{lem:bell probability weight} implies that the dominant contributions to the output distribution are the sectors $k \in n/2 \pm \tilde O(\sqrt n) \subset \Theta(n)$.
For each of these sectors approximating the output probabilities is a \#P-hard task in the worst case since the respective degree-$D-1$ circuits on $k$ qubits are arbitrary \cite{Bremner.2016}. 
Likewise, we inherit the approximate-average case hardness conjecture \ref{conj:aach degreeD} from \textcite{Bremner.2016}.

Next, we show the hiding property in each sector. 
The hiding property is an important ingredient in the reduction from computing probabilities to sampling. 
The hiding property asserts that the distribution over input circuits $C$
 is invariant under a `hiding procedure'. 
This procedure removes the dependency on a particular outcome $x$ by hiding that outcome in the probability of obtaining a different outcome $y$ of a different random circuit $C_y$. This allows us to restrict our attention to the distribution over circuits of a fixed outcome.
\begin{lemma}[Hiding]
\label{lem:hiding bell sampling}
    For each $k \in [n]$, the Bell output distribution $P_C\restriction_{S_k}$ of a uniform degree-$D$ circuit $C$ satisfies the hiding property.
\end{lemma}
\begin{proof}
    First, we observe that the hiding property trivially holds with respect to the $z$ outcomes in every $k$-sector. 
    We can translate one $z$ outcome to a different outcome $z'$ by applying $Z_{z'+z}$ gates to the IQP circuit. But this leaves the distribution over circuits invariant since $Z$ gates are applied with uniform probability. 
    Next, consider the $x$-outcomes in every sector. These are all related by a permutation. Therefore, we can obtain an outcome $(x,z)$ of a degree-$D$ circuit $C$ by from an $(x',z')$ outcome of a circuit $C(x',z')$ which is obtained from $C$ by permuting the labels of the gates according to a permutation $\Pi$ such that $x' = \Pi x $, and then applying $Z_{z'+z}$.
\end{proof}

Finally, we consider anticoncentration of the Bell sampling distribution from uniform degree-$D$ IQP circuits with $D \geq 3$.
Again, we find that, restricted to each $k$-sector, anticoncentration holds. 
To this end, we compare the squared first moment (given by the squared uniform distribution) in every sector with the second moment in the sector.

The correct uniform distribution on the sample space is given by sampling a uniformly random $x \in \bin^n$ and then a uniformly random $z \in \bin^k$. 
This translates into the uniform distribution in every $k$-sector as
\begin{align}
    \mc U_{k} = P_C(S_k) \cdot \frac 1 {|S_k|} =\frac{\binom n k}{2^n}  \frac{1}{2^k  \binom n k } = 2^{-nk}. 
\end{align}

\begin{lemma}[Anticoncentration of degree-$D$ Bell sampling]
\label{lem:anticoncentration bell sampling}
    The second moment of the outcome probabilities of uniform degree-$D$ Bell sampling for $D \geq 3$ with $(x,z) \in S_k$ satisfies
    \begin{align}
        \mc U_k^{-2} \cdot \mb E_C P_C(x,z)^2 &= 3 - 2^{-k +1}, 
    \end{align}
\end{lemma}
\begin{proof}
    Let $\mathrm{uIQP}(k,D)$ be the family of uniform degree-$D$ circuits on $k$ qubits. 
    Furthermore let $z_x \in \bin^k$ be the restriction of $z$ to $\supp(x)$. 
     \begin{align}
        \mb E_C P_C(x,z)^2 &= \frac 1 {2^{2n}} \mb E_{C \sim \mathrm{uIQP}(k,D-1)} | \bra {+^n} C \ket{z_x} |^4\\
        & = \frac 1 {2^{2n}} \frac {3 - 2^{-k +1}}{2^{2k}}.
    \end{align}
\end{proof}

\begin{proof}[Proof of \cref{thm:bell sampling hardness}]
    To prove the theorem, we follow the route for hardness arguments of Gaussian boson sampling, see e.g. \cite{deshpande_quantum_2022}.
    We apply the argument due to Stockmeyer individually in every sector $S_k$.
    In each sector, the hiding argument holds, as well as anticoncentration for the postselected distribution.
    Hence, we can follow argue for hardness of sampling for the distributions in each sector, see Ref.~\cite{hangleiter_computational_2023} for details.
    For the outcomes with $k \in n/2 + O(\sqrt n)$, we have worst-case hardness as well as conjectured average case hardness (\cref{conj:aach degreeD}) due to Ref.~\cite{Bremner.2016}.
    Moreover, by \cref{lem:bell probability weight} these outcomes dominate the probability distribution.

    Now suppose there was an efficient sampling algorithm $\mc A$ that samples from a distribution $Q_C$ which satisfies $d_{\mathrm{TVD}}(Q_C - P_C) \leq \epsilon$. 
    Then we can use Stockmeyer's algorithm on input $\mc A, (x,z)$ to compute any outcome probability up to constant relative error, since in the worst case the error on the full distribution is concentrated in the sector $S_k\ni (x,z)$.
\end{proof}

\section{3D color code lattices}
\label{app:codes}
In this appendix, we provide a detailed description of the lattices that give rise to the family of 3D color codes with a growing code distance and a logical $\overline{\rm CCZ}$ gate implemented via a transversal $T$ gate.

We start by defining $\mathcal L'$ to be a region of a cubic lattice with vertices corresponding to the coordinates $(x,y,z)\in \mathbb Z_L^3$, where $L\equiv 0 \mod 2$.
We color each vertex $(x,y,z)$ of $\mathcal L'$ as follows:
red iff $x\equiv y \equiv z \mod 2$;
green iff $x\equiv y \equiv z + 1 \mod 2$;
yellow iff $x\equiv y + 1 \equiv z \mod 2$;
blue iff $x +1 \equiv y \equiv z \mod 2$;
see Fig.~\ref{fig_lattices_all}(c) for an illustration.
Each cube in $\mathcal L'$ is split into five tetrahedra spanned by the vertices of $\mathcal L'$.
Namely, for the cube centered at $\left(x+\tfrac{1}{2},y+\tfrac{1}{2},z+\tfrac{1}{2}\right)$ if $x+y\equiv z \mod 2$, then we split it accordingly to Fig.~\ref{fig_lattices_all}(d), where one of the tetrahedra is spanned by the vertices $(x+1,y,z)$,$(x,y+1,z)$,$(x,y,z+1)$ and $(x+1,y+1,z+1)$; otherwise, we split it accordingly to Fig.~\ref{fig_lattices_all}(e), where one of the tetrahedra is spanned by the vertices $(x,y,z)$,$(x+1,y+1,z)$,$(x+1,y,z+1)$ and $(x,y+1,z+1)$.

We obtain $\mathcal L$ from $\mathcal L'$ by including additional red vertices and tetrahedra.
Namely, for every green vertex $(x,y,z)$ at the top $z=L$ or  bottom $z=0$ boundaries of $\mathcal L'$ we add a vertex at either $(x,y,L+1)$ or $(x,y,-1)$;
for every yellow vertex $(x,y,z)$ at the front $y=0$ or rear $y=L$ boundaries of $\mathcal L'$ we add a vertex at either $(x,-1,z)$ or $(x,L+1,z)$;
for every blue vertex $(x,y,z)$ at the left $x=0$ or right $x=L$ boundaries of $\mathcal L'$ we add a vertex at either $(-1,y,z)$ or $(L+1,y,z)$.
Subsequently, for each newly added red vertex we add four tetrahedra spanned that vertex and neighboring vertices belonging to the boundary of $\mathcal L'$, as illustrated in \cref{fig_lattices_all}(f,g).

\section{Numerical simulation techniques}
\label{app:numerics}

In this appendix, we outline the numerical methods used to produce \cref{fig:random hiqp,fig:classical,fig:noisy xeb,fig:encoded_xeb_fidelity,fig:code_comparison} in the main text. 

\subsection{XEB and fidelity of Clifford circuits}

We simulate all Clifford circuits using the {\tt stim} package~\cite{gidney_stim_2021}. To evaluate the XEB, we perform Gaussian elimination on the noiseless stabilizers of the target state, effectively finding the subset of stabilizer generators that are diagonal in the measurement basis. For $n$ qubits and $k$ diagonal stabilizers, every non-zero bitstring probability is $2^{k-n}$ and, thus, the ideal XEB of the circuit is
\begin{equation}
    \chi_{\rm ideal} = 2^{k}-1, 
\end{equation}
which we use to obtain~\cref{fig:random hiqp}(a) by sampling random hIQP circuits.

For noisy simulati  ons, we employ two separate approaches. In \cref{fig:encoded sampling}(a) and \cref{fig:code_comparison}, we use {\tt stim}'s mega-sampling capabilities to quickly obtain large numbers of noisy bitstrings whose probabilities are then evaluated based on diagonal, pre-computed logical-circuit stabilizers. In \cref{fig:noisy xeb} and \cref{fig:encoded_xeb_fidelity}, we need to also evaluate fidelity (either encoded or not) so we perform multiple tableau simulations, one for each noise realizations; the fidelity is 0 if any of the circuit-level stabilizers is violated, and 1 otherwise. In \cref{fig:encoded_xeb_fidelity}, we additionally calculate the corrected (EC) logical fidelity where we perform error correction in the $Z$ basis (directly in the tableau) before evaluating the fidelity; we use a simple lookup-table decoder constructed by enumerating up to weight-2 physical errors.

Since the numerical studies performed here are based on sampling various circuits and noise realizations, the accessible system sizes and noise rates are severely limited due to the uncertainty associated with a finite sample size. A more sophisticated numerical study based on statistical-mechanics models would enable, e.g., detailed studies of the phase transition in the XEB-to-fidelity ratio.
\subsection{TVD of IQP-circuit connectivity}
A degree-2 IQP circuit is described by the polynomial function
\begin{equation}
    f(x) = \sum_{ij} b_{ij} x_i x_j + \sum_j a_i x_i,
\end{equation}
as discussed in \cref{eq:poly_phase_state}. Thus, we can capture the circuit connectivity via an $n{\times}n$ matrix whose off-diagonal and diagonal terms are given by the $b_{ij}$ and $a_i$ variables, respectively. We look at averaged eigenvalue distribution of such adjacency matrices to compare various circuit ensembles and calculate the total variation distance (TVD) of the histograms as the figure of merit.

For a uniformly random IQP circuit, $b_{ij}$ and $a_i$  are drawn from a random distribution with uniform probability $1{/}2$, whereas for the hIQP circuits they are sampled as described in \cref{fig:iqp circuits}(a). The CNOT gate layers are incorporated through a frame transformation on the adjacency matrix $A$,
\begin{equation}
    A\to M^\T A \,M,
\end{equation}
which is then added modulo 2. The CNOT matrix $M$ is constructed based on the CNOT gate action that transforms bitstrings as
\begin{equation}
    {\rm CNOT}(x_1,x_2) = \begin{pmatrix}
         1 & 0\\ 1 & 1
    \end{pmatrix}\begin{pmatrix}
        x_1 \\ x_2
    \end{pmatrix},
\end{equation}
which is embedded for many CNOTs in parallel in the larger $M$ matrix.

In \cref{fig:random hiqp}(b), we average over $10^6$ uniform and hIQP circuit realizations for each data point.

\subsection{Estimating classical hardness}

In \cref{fig:classical}, we compute the treewidth and minimal vertex cover of the effective IQP hypergraphs realized by random hyperdepth-$\hdepth$ hIQP circuits for various hypercube dimensions $\hdim$. 
To do this, for each pair $(\hdepth,\hdim) \in [2] \times [5]$ we sample $100$ random hIQP circuits and compute the effective IQP hypergraph using the relation \eqref{eq:iqp+cnot = iqp}.
The result is a hypergraph on $3 \cdot 2^\hdim$ vertices. 

To compute the treewidth of this hypergraph, we translate the hypergraph to a simple graph $G$ by adding a $3$-clique for every degree-$3$ hyperedge. 
We then use the \texttt{treewidth\_min\_fill\_in} heuristic implemented in the python \texttt{networkx} package to compute the treewidth \cite{networkx_developers_networkx_2024}. 

To compute the minimum vertex cover of the degree-$3$ part of the hypergraph, we write the optimization problem as an integer linear program ILP
\begin{align}
    &\min_{x \in \bin^n} \sum_{i=1}^{n} x_i\\
    \text{subject to }& x_{e_1} + x_{e_2} + x_{e_3} \geq 1  \, \, \forall e \in E_3\nonumber
\end{align}
where $E_3 \subset [n]^{\times 3}$ denotes the set of degree-$3$ hyperedges of the IQP graph. 
We then use the \texttt{Gurobi Optimizer}, a state-of-the-art solver to solve this ILP \cite{gurobi}.


\end{document}